\documentclass[oneside,12pt,table]{./IIScthesisPSnPDF}
\usepackage{etex}
\reserveinserts{28}

\usepackage{amssymb,amsmath,amsthm}
\usepackage{makecell}
\usepackage{multirow}
\usepackage{mathtools}
\usepackage{thmtools,thm-restate}
\usepackage{datetime}
\usepackage{relsize}
\usepackage{charter}
\usepackage{eulervm}

\usepackage{xpatch}
\makeatletter
\xpatchcmd{\@endpart}{\vfil\newpage}{}{}{}
\xpatchcmd{\@endpart}{\newpage}{}{}{}
\makeatother

\newcommand\blfootnote[1]{%
  \begingroup
  \renewcommand\thefootnote{}\footnote{#1}%
  \addtocounter{footnote}{-1}%
  \endgroup
}

\newcommand{\eqdef}{\overset{\mathrm{def}}{=\joinrel=}}

\newcommand{\specialcell}[2][c]{%
\begin{tabular}[#1]{@{}c@{}}#2\end{tabular}}

\usepackage{tikz}
\usetikzlibrary{shapes.geometric, arrows}

\usepackage{algpseudocode,algorithm,algorithmicx}

\newcommand*\Let[2]{\State #1 $\gets$ #2}
\algrenewcommand\algorithmicrequire{\textbf{Input:}}
\algrenewcommand\algorithmicensure{\textbf{Output:}}
\algnewcommand{\Initialize}[1]{%
  \State \textbf{Initialize:}
  \Statex \hspace*{\algorithmicindent}\parbox[t]{.8\linewidth}{\raggedright #1}
}
\algdef{SE}[FUNCTION]{Function}{EndFunction}%
   [2]{\algorithmicfunction\ \textproc{#1}\ifthenelse{\equal{#2}{}}{}{(#2)}}%
   {\algorithmicend\ \algorithmicfunction}%

\renewcommand{\Return}{\State \textbf{return} }

\usepackage{xcolor}

\usepackage{url}

\makeatletter
\newcommand{\mathleft}{\@fleqntrue\@mathmargin\parindent}
\newcommand{\mathcenter}{\@fleqnfalse}
\makeatother

\newcommand{\MOV}[1][r]{M_{#1}(\mathcal{E})}
\newcommand{\MV}{\textsc{Margin of Victory}\xspace}

\usepackage{xspace}

\newcommand{\el}{\ensuremath{\ell}\xspace}
\newcommand{\suc}{\ensuremath{\succ}\xspace}
\newcommand{\Query}{\textsc{Query}\xspace}

\DeclareMathOperator*{\argmin}{arg\!min}
\DeclareMathOperator*{\argmax}{arg\!max}

\DeclarePairedDelimiter{\ceil}{\lceil}{\rceil}
\DeclarePairedDelimiter{\floor}{\lfloor}{\rfloor}

\newcommand{\defproblem}[3]{
  \vspace{1mm}
\begin{center}
\noindent\fbox{

  \begin{minipage}{\textwidth}
  \begin{tabular*}{\textwidth}{@{\extracolsep{\fill}}l} \textsc{\underline{#1}} \\ \end{tabular*}\vspace{1ex}
  {\bf{Input:}} #2  \\
  {\bf{Question:}} #3
  \end{minipage}
 
  }
\end{center}
  \vspace{1mm}
}

\renewcommand{\leq}{\leqslant}
\renewcommand{\geq}{\geqslant}
\renewcommand{\ge}{\geqslant}
\renewcommand{\le}{\leqslant}

\newcommand{\E}[1]{\mathbb{E}\left[#1\right]}

\newcommand\numberthis{\addtocounter{equation}{1}\tag{\theequation}}

\newcommand{\YES}{\textsc{Yes}\xspace}
\newcommand{\NO}{\textsc{No}\xspace}

\newcommand{\BigO}{\ensuremath{\mathcal{O}}\xspace}
\newcommand{\true}{\textsc{true}\xspace}
\newcommand{\false}{\textsc{false}\xspace}
\newcommand{\pr}{\ensuremath{\prime}\xspace}

\newcommand{\prm}{{\ensuremath{\prime}}\xspace}
\newcommand{\PE}{\textsc{Preference Elicitation}\xspace}
\newcommand{\CW}{\textsc{Weak Condorcet Winner}\xspace}
\newcommand{\WCW}{\textsc{Weak Condorcet Winner}\xspace}
\newcommand{\WD}{\textsc{Winner Determination}\xspace}
\newcommand{\WM}{\textsc{Weak Manipulation}\xspace}
\newcommand{\SM}{\textsc{Strong Manipulation}\xspace}
\newcommand{\OM}{\textsc{Opportunistic Manipulation}\xspace}
\newcommand{\CM}{\textsc{Coalitional Manipulation}\xspace}
\newcommand{\PW}{\textsc{Possible Winner}\xspace}
\newcommand{\NW}{\textsc{Necessary Winner}\xspace}
\newcommand{\PS}{\textsc{Permutation sum}\xspace}

\newcommand{\NPCshort}{\ensuremath{\mathsf{NP}}\text{-complete}\xspace}

\newcommand{\Pshort}{\ensuremath{\mathsf{P}}\xspace}
\newcommand{\NPshort}{\ensuremath{\mathsf{NP}}\xspace}
\newcommand{\NP}{\ensuremath{\mathsf{NP}}\xspace}
\newcommand{\Pb}{\ensuremath{\mathsf{P}}\xspace}
\newcommand{\NPCb}{\ensuremath{\mathsf{NP}}\text{-complete}\xspace}

\newcommand{\coNPH}{\ensuremath{\mathsf{co}}-\ensuremath{\mathsf{NP}}-hard\xspace}
\newcommand{\NPH}{\ensuremath{\mathsf{NP}}-hard\xspace}
\newcommand{\NPC}{\ensuremath{\mathsf{NP}}-complete\xspace}

\newcommand{\caveat}{\ensuremath{\mathsf{CoNP\subseteq NP/Poly}}\xspace}

\newcommand{\RB}{\ensuremath{\mathbb R}\xspace}
\newcommand{\SB}{\ensuremath{\mathbb S}\xspace}
\newcommand{\NB}{\ensuremath{\mathbb N}\xspace}

\newcommand{\EB}{\ensuremath{\mathbb E}\xspace}

\newcommand{\SF}{\ensuremath\ensuremath{\mathfrak S}\xspace}
\newcommand{\PF}{\ensuremath\ensuremath{\mathfrak P}\xspace}

\renewcommand{\AA}{\ensuremath{\mathcal A}\xspace}
\newcommand{\BB}{\ensuremath{\mathcal B}\xspace}
\newcommand{\CC}{\ensuremath{\mathcal C}\xspace}
\newcommand{\DD}{\ensuremath{\mathcal D}\xspace}
\newcommand{\EE}{\ensuremath{\mathcal E}\xspace}
\newcommand{\FF}{\ensuremath{\mathcal F}\xspace}
\newcommand{\GG}{\ensuremath{\mathcal G}\xspace}
\newcommand{\HH}{\ensuremath{\mathcal H}\xspace}
\newcommand{\II}{\ensuremath{\mathcal I}\xspace}

\newcommand{\LL}{\ensuremath{\mathcal L}\xspace}

\newcommand{\NN}{\ensuremath{\mathcal N}\xspace}
\newcommand{\OO}{\ensuremath{\mathcal O}\xspace}
\newcommand{\PP}{\ensuremath{\mathcal P}\xspace}
\newcommand{\QQ}{\ensuremath{\mathcal Q}\xspace}
\newcommand{\RR}{\ensuremath{\mathcal R}\xspace}
\renewcommand{\SS}{\ensuremath{\mathcal S}\xspace}
\newcommand{\TT}{\ensuremath{\mathcal T}\xspace}
\newcommand{\UU}{\ensuremath{\mathcal U}\xspace}
\newcommand{\VV}{\ensuremath{\mathcal V}\xspace}
\newcommand{\WW}{\ensuremath{\mathcal W}\xspace}
\newcommand{\XX}{\ensuremath{\mathcal X}\xspace}
\newcommand{\YY}{\ensuremath{\mathcal Y}\xspace}

\newcommand{\ppp}{\ensuremath{\mathfrak p}\xspace}
\newcommand{\qqq}{\ensuremath{\mathfrak q}\xspace}

\newcommand{\vvv}{\ensuremath{\mathfrak v}\xspace}

\newcommand{\xxx}{\ensuremath{\mathfrak x}\xspace}

\newcommand{\bCC}{\overline{\mathcal{C}}\xspace}

\usepackage{nicefrac}

\newcommand{\nfrac}{\nicefrac}

\newcommand{\defparproblem}[4]{
  \vspace{1mm}
\begin{center}
\noindent\fbox{

  \begin{minipage}{.9\linewidth}
  \begin{tabular*}{\linewidth}{@{\extracolsep{\fill}}lr} \textsc{#1}  & {\bf{Parameter:}} #3 \\ \end{tabular*}
  {\bf{Input:}} #2  \\
  {\bf{Question:}} #4
  \end{minipage}
 
  }
\end{center}
  \vspace{1mm}
}

\newcommand{\eps}{\varepsilon}
\renewcommand{\epsilon}{\eps}

\usepackage{cleveref}

\newtheorem{proposition}{\bf Proposition}
\newtheorem{observation}{\bf Observation}
\newtheorem{theorem}{\bf Theorem}
\newtheorem{lemma}{\bf Lemma}
\newtheorem{claim}{\bf Claim}
\newtheorem{corollary}{\bf Corollary}
\newtheorem{definition}{\bf Definition}
\newtheorem{reductionrule}{\bf Reduction rule}
\newtheorem{example}{\bf Example}

\crefname{example}{Example}{Example}
\crefname{theorem}{Theorem}{Theorem}
\crefname{observation}{Observation}{Observation}
\crefname{lemma}{Lemma}{Lemma}
\crefname{corollary}{Corollary}{Corollary}
\crefname{proposition}{Proposition}{Proposition}
\crefname{definition}{Definition}{Definition}
\crefname{claim}{Claim}{Claim}
\crefname{reductionrule}{Reduction rule}{Reduction rule}
\crefname{chapter}{Chapter}{Chapter}
\crefname{ineq}{inequality}{Inequality}

\numberwithin{example}{chapter}
\numberwithin{proposition}{chapter}
\numberwithin{claim}{chapter}
\numberwithin{theorem}{chapter}
\numberwithin{corollary}{chapter}
\numberwithin{observation}{chapter}
\numberwithin{lemma}{chapter}
\numberwithin{definition}{chapter}
\numberwithin{reductionrule}{chapter}

\newdateformat{monthyeardate}{ \monthname[\THEMONTH], \THEYEAR}

\newcommand{\blankpage}{
\newpage
\thispagestyle{empty}
\mbox{}
\newpage
}

\crefname{observation}{observation}{observations}
\crefname{algorithm}{algorithm}{algorithms}
\crefname{align}{equation}{equations}
\crefname{eqnarray}{equation}{equations}

\newcommand{\thesistitle}{Resolving the Complexity of Some Fundamental Problems in Computational Social Choice}

\begin{document}

\title{\thesistitle}	

\submitdate{August 2016} 
\phd
\dept{Computer Science and Automation}
\faculty{Faculty of Engineering}
\author{Palash Dey}


\maketitle

%
%
%
%
%
%
%
%
%
%
		
\blankpage

\vspace*{\fill}
\begin{center}
\large\bf \textcopyright \ Palash Dey\\
\large\bf August 2016\\
\large\bf All rights reserved
\end{center}
\vspace*{\fill}
\thispagestyle{empty}

\blankpage

\vspace*{\fill}
\begin{center}
DEDICATED TO \\[2em]
\Large\it My teachers
\end{center}
\vspace*{\fill}
\thispagestyle{empty}


\setcounter{secnumdepth}{3}
\setcounter{tocdepth}{3}
\frontmatter 
\blankpage
\pagenumbering{roman}
\prefacesection{Acknowledgements}
I am grateful to my Ph.D. advisors Prof. Y. Narahari and Prof. Arnab Bhattacharyya for providing me a golden opportunity to work with them. I thank them for the all encompassing support they ushered upon me throughout my Ph.D. I convey my special thanks to Prof. Y. Narahari for all his academic and nonacademic supports which made my Ph.D. life very easy and enjoyable. I have throughly enjoyed working with Prof. Arnab Bhattacharyya. I convey my very special thanks to Prof. Neeldhara Misra for many excellent and successful research collaborations, useful discussions, and playing an all encompassing super-active role in my Ph.D. life. Prof. Neeldhara Misra was always my constant source of inspiration and she guided me throughout my Ph.D. as a supervisor, advisor, friend, and elder sister. I am indebted to Prof. Dilip P. Patil for teaching me how the practice of writing statements rigorously and precisely in mathematics automatically clarifies many doubts. I am extremely grateful to have a teacher like him at the very early stage of my stay in IISc who not only significantly improved my mathematical thinking process but also acted like a friend, philosopher, and guide throughout my stay in IISc. I throughly enjoyed all the courses that he taught during my stay in IISc and learned a lot from those courses. Specially his courses on linear algebra, Galois theory, commutative algebra have immensely influenced my thinking process. 

I thank Prof. Sunil Chandran for teaching me the art of proving results in graph theory in particular and theoretical computer science in general. I am grateful to Dr. Deeparnab Chakrabarty and Prof. Arnab Bhattacharyya for teaching me approximation algorithms and randomized algorithms, Prof. Manjunath Krishnapur for teaching me probability theory and martingales, Prof. Chandan Saha for teaching me complexity theory and algebraic geometry, Prof. Y. Narahari for teaching me game theory. I thank all the faculty members of Department of Computer Science and Automation (CSA) for all of their supports. I convey my special thank to all the staffs of CSA specially Ms. Suguna, Ms. Meenakshi, Ms. Kushael for their helping hands. I thank Prof. Ashish Goel for giving me an excellent opportunity to work with him and for hosting me in Stanford University for three months. I thank Dr. David Woodruff for collaborating with me on important problems.

I feel myself extremely lucky to get such a wonderful and friendly lab members in the Game Theory lab. I want to thank all the members of the Game Theory lab with special mention to Dr. Rohith D. Vallam, Satyanath Bhat, Ganesh Ghalme, Shweta Jain, Swapnil Dhamal, Divya Padmanabhan, Dr. Pankaj Dayama, Praful Chandra, Shourya Roy, Tirumala Nathan, Debmalya Mandal, Arupratan Ray, and Arpita Biswas. I want to thank Dr. Minati De, Abhiruk Lahiri, Achintya Kundu, Aniket Basu Roy, Prabhu Chandran, Prasenjit Karmakar, R Kabaleeshwaran, Rohit Vaish, Sayantan Mukherjee, Srinivas Karthik, Suprovat Ghoshal, Vineet Nair for all the important discussions. I greatly acknowledge support from Ratul Ray and Uddipta Maity during my internship in Stanford University.

I thank Google India for providing me fellowship during second and third years of my Ph.D. I acknowledge financial support from Ministry of Human Resource Development (MHRD) in India during my first year of Ph.D.

Last but not the least, I express my gratitude to my parents, brother, wife, and all my friends for their love, blessing, encouragement, and support.
\prefacesection{Abstract}
In many real world situations, especially involving multiagent systems and
artificial intelligence, participating agents often need to agree upon a
common alternative even if they have differing preferences over the
available alternatives. Voting is one of the tools of choice in these
situations. Common and classic applications of voting in modern
applications include collaborative filtering and recommender systems,
metasearch engines, coordination and planning among multiple automated
agents etc. Agents in these applications usually have computational power
at their disposal. This makes the study of computational aspects of voting
crucial. This thesis is devoted to a study of computational complexity of
several fundamental algorithmic and complexity-theoretic problems arising in the context of voting theory.

The typical setting for our work is an ``election''; an election consists of
a set of voters or agents, a set of alternatives, and a voting rule. The
vote of any agent can be thought of as a ranking (more precisely, a
complete order) of the set of alternatives. A voting profile comprises a
collection of votes of all the agents. Finally, a voting rule is a mapping
that takes as input a voting profile and outputs an alternative, which is
called  the ``winner'' or ``outcome'' of the election. Our contributions in
this thesis can be categorized into three parts and are described below.

\paragraph*{Part I: Preference Elicitation.} In the first part of the
thesis, we study the problem of eliciting the preferences of a set of
voters by asking a small number of comparison queries (such as who a voter
prefers between two given alternatives) for various interesting
domains of preferences.

We commence with considering the domain of single peaked preferences on
trees in \Cref{chap:pref_elicit_peak}. This domain is a significant
generalization of the classical well studied domain of single peaked
preferences. The domain of single peaked preferences and its
generalizations are hugely popular among political and social
scientists. We show tight dependencies between query complexity of
preference elicitation and various parameters of the single peaked tree,
for example, number of leaves, diameter, path width, maximum degree of a
node etc.

We next consider preference elicitation for the domain of single crossing
preference profiles in \Cref{chap:pref_elicit_cross}. This domain has also been
studied extensively by political scientists, social choice theorists, and
computer scientists. We establish that the query complexity of preference
elicitation in this domain crucially depends on how the votes are accessed
and on whether or not  any single crossing ordering is a priori known.

\paragraph*{Part II: Winner Determination.} In the second part of the
thesis, we undertake a study of the computational complexity of several
important problems related to determining winner of an election.

We begin with a study of the following problem: Given an election, predict
the winner of the election under some fixed voting rule by sampling as few
votes as possible. We establish optimal or almost optimal bounds on the
number of votes that one needs to sample for many commonly used voting rules
when the margin of victory is at least $\eps n$ ($n$ is the number of
voters and $\eps$ is a parameter). We next study efficient sampling based algorithms for
estimating the margin of victory of a given election for many common voting
rules. The margin of victory of an election is a useful measure that captures the robustness
of an election outcome. The above two works are presented in \Cref{chap:winner_prediction_mov}.

In \Cref{chap:winner_stream}, we design an optimal algorithm for
determining the plurality winner of an election when the votes are arriving
one-by-one in a streaming fashion. This resolves an intriguing question on
finding heavy hitters in a stream of items, that has remained open for more than $35$ years in
the data stream literature. We also provide near optimal algorithms for determining the winner of a stream
of votes for other popular voting rules, for example, veto, Borda, maximin
etc.

Voters' preferences are often partial orders instead of complete orders.
This is known as the incomplete information setting in computational social
choice theory. In an incomplete information setting, an extension of the
winner determination problem which has been studied extensively is the
problem of determining possible winners. We study the kernelization
complexity (under the complexity-theoretic framework of parameterized
complexity) of the possible winner problem in \Cref{chap:kernel}. We show
that there do not exist kernels of size that is polynomial in the number of
alternatives for this problem for commonly used voting rules under a plausible 
complexity theoretic assumption. However, we also
show that the problem of coalitional manipulation which is an important
special case of the possible winner problem admits a kernel whose size is
polynomially bounded in the number of alternatives for common voting rules.

\paragraph{Part III: Election Control.} In the final part of the thesis, we study the computational complexity of
various interesting aspects of strategic behavior in voting.

First, we consider the impact of partial information in the context of
strategic manipulation in \Cref{chap:partial}. We show that lack of
complete information makes the computational problem of manipulation
intractable for many commonly used voting rules.

In \Cref{chap:detection}, we initiate the study of the computational
problem of detecting possible instances of election manipulation. We show
that detecting manipulation may be computationally easy under certain
scenarios even when manipulation is intractable.

The computational problem of bribery is an extensively studied problem in
computational social choice theory. We study computational complexity of
bribery when the briber is ``frugal'' in nature.
We show for many common voting rules
that the bribery problem remains intractable even when the briber's
behavior is
restricted to be frugal,  thereby strengthening the
intractability results from the literature. This forms the subject of
\Cref{chap:frugal_bribery}.
\blankpage
\prefacesection{Publications based on this Thesis}
\begin{enumerate}

 \item \textit{Palash Dey}, Neeldhara Misra, and Yadati Narahari. ``\textit{Kernelization Complexity of Possible Winner and Coalitional Manipulation Problems in Voting}". In \textit{Theoretical Computer Science}, volume 616, pages 111-125, February 2016.
\\Remark: A preliminary version of this work was presented in the $14^{th}$ \textit{International Conference on Autonomous Systems and Multiagent Systems (AAMAS-15), 2015.}

 \item \textit{Palash Dey}, Neeldhara Misra, and Yadati Narahari. ``\textit{Frugal Bribery in Voting}''. Accepted in \textit{Theoretical Computer Science}, March 2017.
 \\Remark: A preliminary version of this work was presented in the \textit{$30^{th}$ AAAI Conference on Artificial Intelligence (AAAI-16)}, 2016.

 \item Arnab Bhattacharyya, \textit{Palash Dey}, and David P. Woodruff. ``\textit{An Optimal Algorithm for $\ell_1$-Heavy Hitters in Insertion Streams and Related Problems}''. Proceedings of the $35^{th}$ \textit{ACM SIGMOD-SIGACT-SIGAI Symposium on Principles of Database Systems (PODS-16)}, pages 385-400, San Francisco, USA, 26 June - 1 July, 2016.
 
 \item \textit{Palash Dey} and Neeldhara Misra. ``\textit{Elicitation for Preferences Single Peaked on Trees}''. Proceedings of the \textit{$25^{th}$ International Joint Conference on Artificial Intelligence (IJCAI-16)}, pages 215-221, New York, USA, 9-15 July, 2016.
 
 \item \textit{Palash Dey} and Neeldhara Misra. ``\textit{Preference Elicitation For Single Crossing Domain}''. Proceedings of the \textit{$25^{th}$ International Joint Conference on Artificial Intelligence (IJCAI-16)}, pages 222-228, New York, USA, 9-15 July, 2016.
 
 \item \textit{Palash Dey}, Neeldhara Misra, and Yadati Narahari. ``\textit{Complexity of Manipulation with Partial Information in Voting}''. Proceedings of the \textit{$25^{th}$ International Joint Conference on Artificial Intelligence (IJCAI-16)}, pages 229-235, New York, USA, 9-15 July, 2016.
 
 \item \textit{Palash Dey}, Neeldhara Misra, and Yadati Narahari. ``\textit{Frugal Bribery in Voting}''. Proceedings of the \textit{$30^{th}$ AAAI Conference on Artificial Intelligence (AAAI-16)}, vol. 4, pages 2466-2672, Phoenix, Arizona, USA, 2016.
 
 \item \textit{Palash Dey} and Yadati Narahari. ``\textit{Estimating the Margin of Victory of Elections using Sampling}". Proceedings of the \textit{$24^{th}$ International Joint Conference on Artificial Intelligence (IJCAI-15)}, pages 1120-1126, Buenos Aires, Argentina, 2015.
 
 \item \textit{Palash Dey} and Arnab Bhattacharyya. ``\textit{Sample Complexity for Winner Prediction in Elections}". Proceedings of the \textit{$14^{th}$ International Conference on Autonomous Systems and Multiagent Systems (AAMAS-15)}, pages 1421-1430, Istanbul, Turkey, 2015.
 
 \item \textit{Palash Dey}, Neeldhara Misra, and Yadati Narahari. ``\textit{Detecting Possible Manipulators in Elections}". Proceedings of the \textit{$14^{th}$ International Conference on Autonomous Systems and Multiagent Systems(AAMAS-15)}, pages 1441-1450, Istanbul, Turkey, 2015.
 
 \item \textit{Palash Dey}, Neeldhara Misra, and Yadati Narahari ``\textit{Kernelization Complexity of Possible Winner and Coalitional Manipulation Problems in Voting}". Proceedings of the \textit{$14^{th}$ International Conference on Autonomous Systems and Multiagent Systems (AAMAS-15)}, pages 87-96, Istanbul, Turkey, 2015.
 
 
 \end{enumerate}

\tableofcontents
\listoffigures
\listoftables
\mainmatter 
\blankpage
\setcounter{page}{1}
\chapter{Introduction}
\label{chap:intro}

\begin{quote}
{\small In this chapter, we provide an informal introduction to the area called social choice theory and motivate its computational aspects. We then provide an overview of the thesis along with further motivations for the problems we study in this thesis.} 
\end{quote}

People always have different opinions on almost everything and hence they often have to come to a decision collectively on things that affect all of them. There are various ways for people with different opinions to come to a decision and the use of {\em voting} is the most natural and common method for this task. Indeed, the use of voting can be traced back to Socrates (470-399 BC) who was sentenced to death by taking a majority vote. We also see the mention of voting rules in Plato's (424-348 BC) unfinished book {\em ``Laws''}. Then, in the Middle Ages, the Catalan philosopher, poet, and missionary Ramon Llull (1232-1316) mentioned Llull voting in his manuscripts Ars notandi, Ars eleccionis, and Alia ars eleccionis, which were lost until 2001. The German philosopher, theologian, jurist, and astronomer Nicholas of Cusa (1401-1464) suggested the method for electing the Holy Roman Emperor in 1433. The first systematic study of the theory of voting was carried out by two French mathematicians, political scientists, and philosophers Marie Jean Antoine Nicolas de Caritat, marquis de Condorcet (1743-1794) and Jean-Charles, chevalier de Borda (1733-1799). A more involved study of voting was pioneered by Kenneth Arrow (a co-recipient of the {\em Nobel Memorial Prize in Economics} in 1972). The celebrated {\em Arrow's impossibility theorem} in 1951 along with {\em Gibbard-Satterthwaite's impossibility theorem} in 1973 and 1977 laid the foundation stone of modern voting theory.

\section{Important Aspects of Voting}

In a typical setting for voting, we have a set of alternatives or {\em candidates}, a set of agents called {\em voters} each of whom has an ordinal preference called {\em votes} over the candidates, and a {\em voting rule} which chooses one or a set of candidates as {\em winner(s)}. We call a set of voters along with their preferences over a set of candidates and a voting rule an {\em election}. In \Cref{tbl:simple_example_election} we exhibit an election with a set of voters \{Satya, Shweta, Swapnil, Divya, Ganesh\} and a set of five ice cream flavors \{Chocolate, Butterscotch, Pesta, Vanilla, Kulfi\} as a set of candidates. Suppose we use the plurality voting rule which chooses the candidate that is placed at the first positions of the votes most often as winner. Then the winner of the plurality election in \Cref{tbl:simple_example_election} is Chocolate. Plurality voting rule is among the simplest voting rules and only observes the first position of every vote to determine the winner. Mainly due to its simplicity, the plurality voting rule frequently finds its applications in voting scenarios where human beings are involved, for example, political elections. Indeed many countries including India, Singapore, United Kingdom, Canada use plurality voting at various levels of elections~\cite{plurality_wiki}.

\begin{table}[!htbp]
\begin{center}
{\renewcommand*{\arraystretch}{1.5}
\begin{tabular}{|c|c|}\hline
 Voters & Preferences\\\hline\hline
 Satya & Chocolate $\succ$ Kulfi $\succ$ Butterscotch $\succ$ Vanilla $\succ$ Pesta\\\hline
 Shweta & Butterscotch $\succ$ Kulfi $\succ$ Chocolate $\succ$ Vanilla $\succ$ Pesta\\\hline
 Swapnil & Pesta $\succ$ Butterscotch $\succ$ Kulfi $\succ$ Vanilla $\succ$ Chocolate\\\hline
 Divya & Chocolate $\succ$ Vanilla $\succ$ Kulfi $\succ$ Pesta $\succ$ Butterscotch\\\hline
 Ganesh & Kulfi $\succ$ Butterscotch $\succ$ Chocolate $\succ$ Vanilla $\succ$ Pesta\\\hline
\end{tabular}
}
\end{center}
\caption{Example of an election. Chocolate is the winner if the plurality voting rule is used. Kulfi is the winner if the Borda voting rule is used. Butterscotch is the Condorcet winner.}\label{tbl:simple_example_election}
\end{table}

However, there are more sophisticated voting rules which takes into account the entire preference ordering of the voters to choose a winner. One such voting rule which finds its widespread use specially in sports and elections in many educational institutions is the Borda voting rule~\cite{borda_wiki}. The Borda voting rule, first proposed by Jean-Charles, chevalier de Borda in 1770, gives a candidate a score equal to the number of other candidates it is preferred over for every vote. In \Cref{tbl:simple_example_election} for example, the candidate Kulfi receives a score of 3 from Satya, 3 from Shweta, 2 from Swapnil, 2 from Divya, and 4 from Ganesh. Hence the Borda score of Kulfi is 14. It can be verified that the Borda score of every other candidate is less than 14 which makes Kulfi the Borda winner of the election. We can already observe how two widely used and intuitively appealing voting rules select two different winners for the same set of votes.

A closer inspection of \Cref{tbl:simple_example_election} reveals something peculiar about what the plurality and Borda voting rules choose as winners. The plurality voting rule chooses Chocolate as the winner whereas a majority (three out of five) of the voters prefer Butterscotch over Chocolate. The same phenomenon is true for the Borda winner (Kulfi) too --- a majority of the voters prefer Butterscotch over Kulfi. More importantly Butterscotch is preferred over every other alternative (ice cream flavors) by a majority of the voters. Such an alternative, if it exists, which is preferred over every other alternative by a majority of the voters, is called the {\em Condorcet winner} of the election, named after the French mathematician and philosopher Marie Jean Antoine Nicolas de Caritat, marquis de Condorcet. We provide a more detailed discussion to all these aspects of voting including others in \Cref{chap:prelim}.

Now let us again consider the example in \Cref{tbl:simple_example_election}. Suppose Shweta misreports her preference as ``Butterscotch $\succ$ Pesta $\succ$ Chocolate $\succ$ Vanilla $\succ$ Kulfi.'' Observe that the misreported vote of Shweta when tallied with the other votes makes Butterscotch the Borda winner with a Borda score of 12. Moreover, misreporting her preference results in an outcome (that is Butterscotch) which Shweta prefers over the honest one (that is Kulfi). This example demonstrate a fundamental notion in voting theory called {\em manipulation} --- there can be instances where an individual voter may have an outcome by misreporting her preference which she prefers over the outcome of honest voting (according to her preference). Obviously manipulation is undesirable since it leads to suboptimal societal outcome. Hence one would like to use voting rules that do not create such possibilities of manipulation. Such voting rules are called {\em strategy-proof voting rules}. When we have only two candidates, it is known that the majority voting rule, which selects the majority candidate (breaking ties arbitrarily) as the winner, is a strategy-proof voting rule. There exist strategy-proof voting rules even if we have more than two candidates. For example, consider the voting rule $f_\vvv$ that always outputs the candidate which is preferred most by some fixed voter say \vvv. Obviously $f_\vvv$ is strategy-proof since any voter other than \vvv cannot manipulate simply because they can never change the outcome of the election and the winner is already the most preferred candidate of \vvv. However, voting rules like $f_\vvv$ are undesirable intuitively because it is not performing the basic job of voting which is aggregating preferences of all the voters. The voting rule $f_\vvv$ is called a {\em dictatorship} since it always selects the candidate preferred most by some fixed voter \vvv; such a voter \vvv is called a {\em dictator}. Hence an important question to ask is the following: do there exist voting rules other than dictatorship that are strategy-proof with more than two candidates? Unfortunately classical results in social choice theory by Gibbard and Satterthwaite prove that there does not exist any such {\em onto} voting rule~\cite{gibbard1973manipulation,satterthwaite1975strategy} if we have more than two candidates. A voting rule is onto if every candidate is selected for some set of votes. We formally introduce all these fundamental aspects of voting including others in \Cref{chap:prelim}.

\section{Applications and Challenges of Voting Theory Relevant to Computer Science}

With rapid improvement of computational power, the theory of voting found its applications not only in decision making among human beings but also in aggregating opinions of computational agents. Indeed, voting is commonly used whenever any system with multiple autonomous agents wish to take a common decision. Common and classic applications of voting in multiagent systems in particular and artificial intelligence in general include collaborative filtering and recommender systems~\cite{PennockHG00}, planning among multiple automated agents~\cite{Ephrati}, metasearch engines~\cite{Dwork}, spam detection~\cite{Cohen}, computational biology~\cite{JacksonSA08}, winner determination in sports competitions~\cite{BetzlerBN14}, similarity search and classification of high dimensional data~\cite{Fagin}. The extensive use of voting by computational agents makes the study of computational aspects of voting extremely important. In many applications of voting in artificial intelligence, one often has a huge number of voters and candidates. For example, voting has been applied in the design of metasearch engines~\cite{Dwork} where the number of alternatives the agents are voting for is in the order of trillions if not more --- the alternatives are the set of web pages in this case. In a commercial recommender system, for example Amazon.com, we often have a few hundred millions of users and items. In these applications, we need computationally efficient algorithms for quickly finding a winner of an election with so many candidates and voters.

Moreover, the very fact that the agents now have computational power at their disposal allows them to easily strategize their behavior instead of acting truthfully. For example, the problem of manipulation is much more severe in the presence of computational agents since the voters can easily use their computational power to find a manipulative vote. The field known as {\em computational social choice theory (COMSOC)} studies various computational aspects in the context of voting. We refer the interested readers to a recent book on computational social choice theory for a more elaborate exposition~\cite{moulin2016handbook}. 

Other than practical applications, the theory of voting turns out to be crucial in proving various fundamental theoretical results in computer science. Indeed, a notion known as {\em noise stability} of majority voting has been beautifully utilized for proving a lower bound on the quality of any polynomial time computable approximate solution of the maximum cut problem in a undirected graph~\cite{Khot,MosselOO05}. We refer interested reader to~\cite{ODonnell} for a more detailed discussion.

\section{Thesis Overview}

In this thesis, we resolve computational questions of several fundamental problems arising in the context of voting theory. We provide a schematic diagram of the contribution of the thesis in \Cref{fig:thesis_overview}. We begin our study with the first step in any election system -- eliciting the preferences of the voters. We consider the problem of preference elicitation in the first part of the thesis and develop (often) optimal algorithms for learning the preferences of a set of voters. In the second part of the thesis, we develop (often) optimal sublinear time algorithms for finding the winner of an election. The problem of determining winner of an election is arguably the most fundamental problem that comes to our mind once the preferences of all the voters have been elicited. In the third part of the thesis, we exhibit complexity-theoretic results for various strategic aspects of voting.

\tikzstyle{textbox} = [rectangle, rounded corners, minimum width=3cm, minimum height=1cm,text centered, draw=black, thick]
\tikzstyle{arrow} = [thick,->,>=stealth]

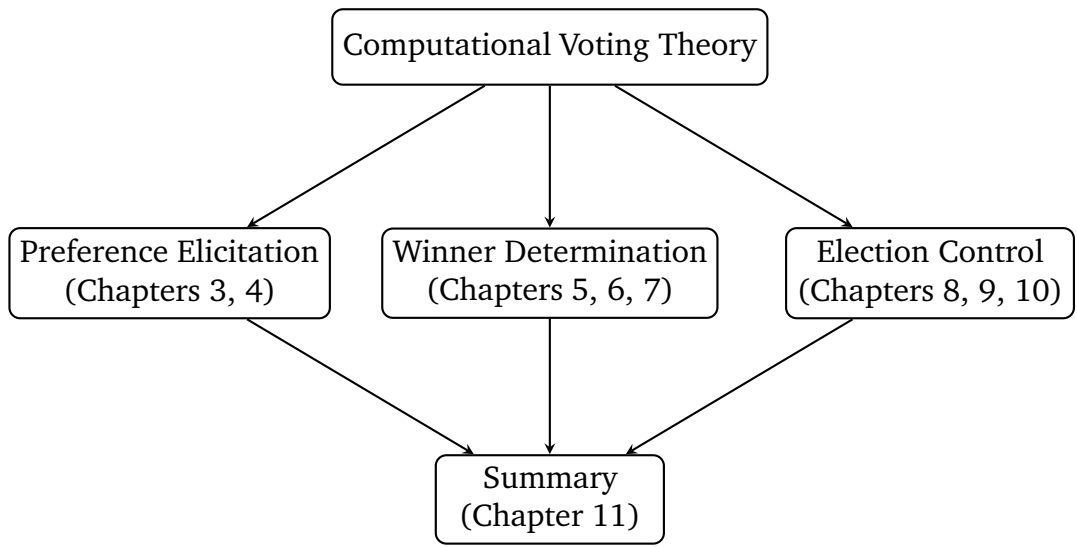
\begin{figure}
\centering
\begin{tikzpicture}[node distance=3cm]

\node (root) [textbox] {Computational Voting Theory};

\node (elicit) [textbox, below of=root, xshift=-5cm, align=center] {Preference Elicitation\\(Chapters 3, 4)};
\node (winner) [textbox, below of=root, align=center] {Winner Determination\\(Chapters 5, 6, 7)};
\node (control) [textbox, below of=root, xshift=5cm, align=center] {Election Control\\(Chapters 8, 9, 10)};

\node (sum) [textbox, below of=winner, align=center] {Summary\\(Chapter 11)};

\draw [arrow] (root) -- (elicit);
\draw [arrow] (root) -- (winner);
\draw [arrow] (root) -- (control);

\draw [arrow] (elicit) -- (sum);
\draw [arrow] (winner) -- (sum);
\draw [arrow] (control) -- (sum);

\end{tikzpicture}
\caption{Structural overview of this dissertation}\label{fig:thesis_overview}
\end{figure}

\subsection{Part I: Preference Elicitation}

The first step in any election system is to receive preferences of the voters as input. However, it may often be impractical to simply ask the voters for their preferences due to the presence of a huge number of alternatives. For example, asking users of Amazon.com to provide a complete ranking of all the items does not make sense. In such scenarios, one convenient way to know these agents' preferences is to ask the agents to compare (a manageable number of) pairs of items. However, due to the existence of a strong query complexity lower bound (from sorting), we can hope to reduce the burden of eliciting preference from the agents only by assuming additional structure on the preferences. Indeed, in several application scenarios commonly considered, it is rare that preferences are entirely arbitrary, demonstrating no patterns whatsoever. For example, the notion of \textit{single peaked preferences} forms the basis of several studies in the analytical political sciences. Intuitively, preferences of a set of agents is single peaked if the alternatives and the agents can be arranged in a linear order and the preferences ``respect'' this order. We defer the formal definition of single peaked profiles to \Cref{chap:pref_elicit_peak}.

We show in \Cref{chap:pref_elicit_peak,chap:pref_elicit_cross} that the query complexity for preference elicitation can be reduced substantially by assuming practically appealing restrictions on the domain of the preferences.

\subsubsection*{\Cref{chap:pref_elicit_peak}: Preference Elicitation for Single Peaked Preferences on Trees}
 Suppose we have $n$ agents and $m$ alternatives. Then it is known that if the preferences are single peaked, then elicitation can be done using $\Theta(mn)$ comparisons~\cite{Conitzer09}. The notion of single peaked preferences has been generalized to single peaked preferences in a tree --- intuitively, the preferences should be single peaked on every path of the tree. We show that if the single peaked tree has $\ell$ leaves, then elicitation can be done using $\Theta(mn\log\ell)$ comparisons~\cite{deypeak}. We also show that other natural parameters of the single peaked tree, for example, its diameter, path width, minimum degree, maximum degree do not decide the query complexity for preference elicitation.

\subsubsection*{\Cref{chap:pref_elicit_cross}: Preference Elicitation for Single Crossing Preferences}

Another well studied domain is the domain of single crossing profiles. Intuitively, a set of preferences is single crossing if the preferences can be arranged in a linear order such that all the preferences who prefer an alternative $x$ over another alternative $y$ are consecutive, for every $x$ and $y$. We consider in~\cite{deycross} two distinct scenarios: when an ordering of the voters with respect to which the profile is single crossing is {\em known} a priori versus when it is {\em unknown}. We also consider two different access models: when the votes can be accessed at random, as opposed to when they arise in any arbitrary sequence. In the sequential access model, we distinguish two cases when the ordering is known: the first is that the sequence in which the votes appear is also a single-crossing order, versus when it is not.
The main contribution of our work is to provide polynomial time algorithms with low query complexity for preference elicitation in all the above six cases. Further, we show that the query complexities of our algorithms are optimal up to constant factor for all but one of the above six cases.

\subsection{Part II: Winner Determination} 

Once we have elicited the preferences of the voters, the next important task in an election is to determine the winner of this election. We show interesting complexity theoretic results on determining the winner of an election under a variety of  application scenarios.

\subsubsection*{\Cref{chap:winner_prediction_mov}: Winner Prediction}

We begin by studying the following problem: Given an election, predict the winner of the election under some fixed voting rule by sampling as few preferences as possible. Our results show that the winner of an election can be determined (with high probability) by observing only a few preferences that are picked uniformly at random from the set of preferences. We show this result  for many common voting rules when the margin of victory is at least $\epsilon n$~\cite{DeyB15}, where $n$ is the number of voters and $\epsilon$ is a parameter. The margin of victory of an election is the smallest number of preferences that must be modified to change the winner of an election. We also establish optimality (in terms of the number of preferences our algorithms sample) of most of our algorithms by proving tight lower bounds on the sample complexity of this problem.

Another important aspect of an election is the robustness of the election outcome. A popular measure of robustness of an election outcome is its margin of victory. We develop efficient sampling based algorithms for estimating the margin of victory of a given election for many common voting rules~\cite{DeyN15}. We also show optimality of our algorithms for most of the cases under appealing practical scenarios.

\subsubsection*{\Cref{chap:winner_stream}: Winner Determination in Streaming}

An important, realistic model of data in big data research is the streaming model. In the setting that we consider, the algorithms are allowed only one pass over the data. We give the first optimal bounds for returning the $\ell_1$-heavy hitters in an insertion only data stream, together with their approximate frequencies, thus settling a long line of work on this problem. For a stream of $m$ items in $\{1, 2, \ldots, n\}$ and parameters $0 < \epsilon < \phi \leq 1$, let $f_i$ denote the frequency of item $i$, i.e., the number of times item $i$ occurs in the stream in~\cite{deystream}. With arbitrarily large constant probability, our algorithm returns all items $i$ for which $f_i \geq \phi m$, returns no items $j$ for which $f_j \leq (\phi -\epsilon)m$, and returns approximations $\tilde{f}_i$ with $|\tilde{f}_i - f_i| \leq \epsilon m$ for each item $i$ that it returns. Our algorithm uses $O(\epsilon^{-1} \log\phi^{-1} + \phi^{-1} \log n + \log \log m)$ bits of space, processes each stream update in $O(1)$ worst-case time, and can report its output in time linear in the output size. We also prove a lower bound, which implies that our algorithm is optimal up to a constant factor in its space complexity. A modification of our algorithm can be used to estimate the maximum frequency up to an additive $\epsilon m$ error in the above amount of space, resolving an open question on algorithms for data streams for the case of $\ell_1$-heavy hitters. We also introduce several variants of the heavy hitters and maximum frequency problems, inspired by rank aggregation and voting schemes, and show how our techniques can be applied in such settings. Unlike the traditional heavy hitters problem, some of these variants look at comparisons between items rather than numerical values to determine the frequency of an item.

\subsubsection*{\Cref{chap:kernel}: Kernelization Complexity for Determining Possible Winners}

Voters' preferences are often partial orders instead of complete orders. This is known as the incomplete information setting in computational social choice theory. In an incomplete information setting, an extension of the winner determination problem which has been studied extensively is the problem of determining possible winners. An alternative $x$ is a possible winner in a set of incomplete preferences if there exists a completion of these incomplete preferences where the alternative $x$ wins. Previous work has provided, for many
common voting rules, fixed parameter tractable algorithms 
for the {\em Possible winner\/} problem, with
number of candidates as the parameter. 
However, the corresponding kernelization question is still
open, and in fact, has been mentioned as a key research challenge~\cite{ninechallenges}. 
In this work, we settle this open question for many common voting rules.

We show in~\cite{DeyMN15,journalsDeyMN16} that the {\em Possible winner\/} problem for
maximin, Copeland, Bucklin, ranked pairs, and a class of scoring rules that
includes the Borda voting rule do not admit a polynomial kernel with the number
of candidates as the parameter. We show however that the {\em Coalitional manipulation\/}
problem which is an important special case of the {\em Possible winner\/} problem
does admit a polynomial kernel for maximin, Copeland, ranked pairs, and a class of
scoring rules that includes the Borda voting rule, when the number of
manipulators is polynomial in the number of candidates. 
A significant conclusion of this work is that the {\em Possible winner\/} problem 
is computationally harder than the {\em Coalitional manipulation\/} problem since the {\em Coalitional manipulation\/} problem admits a polynomial kernel whereas the {\em Possible winner\/} problem does not admit 
a polynomial kernel.

\subsection{Part III: Election Control} 

When agents in a multiagent system have conflicting goals, they can behave strategically to make the system produce an outcome that they favor. For example, in the context of voting, there can be instances where an agent, by misreporting her preference, can cause the election to result in an alternative which she prefers more than the alternative that would have been the outcome in case she had reported her preference truthfully.  This phenomenon is called {\em manipulation}. Certainly, we would like to design systems that are robust to such manipulation. Unfortunately, classical results in social choice theory establish that any reasonable voting rule will inevitably create such opportunities for manipulation. In this context, computer scientists provide some hope by showing that the computational problem of manipulation can often be intractable. Indeed, if a computationally bounded agent has to solve an intractable problem to manipulate the election system, then the agent would, in all likelihood, fail to manipulate.

\subsubsection*{\Cref{chap:partial}: Manipulation with Partial Information}

The computational problem of manipulation has classically been studied in a complete information setting -- the manipulators know the preference of every other voter. We extend this line of work to the more practical, incomplete information setting where agents have only a partial knowledge about the preferences of other voters~\cite{deypartial}.  In our framework,
the manipulators know a partial order for each voter that is consistent with
the true preference of that voter. 
We say that an extension of a partial order is
\textit{viable} if there exists a manipulative vote for that extension. We propose the following notions of manipulation when manipulators have incomplete information about the votes of other voters.

\begin{enumerate}
\item \textsc{Weak Manipulation}: the manipulators seek to
vote in a way that makes their preferred candidate win in \textit{at least one
extension} of the partial votes of the non-manipulators.
\item \textsc{Opportunistic Manipulation}: the manipulators seek to
vote in a way that makes their preferred candidate win\textit{ in every
viable extension} of the partial votes of the non-manipulators.
\item \textsc{Strong Manipulation}: the manipulators seek to
vote in a way that makes their preferred candidate win \textit{in every
extension} of the partial votes of the non-manipulators.
\end{enumerate}

We consider several scenarios for which the traditional manipulation problems are easy (for instance, the Borda voting rule with a single manipulator). For many of them, the corresponding manipulative questions that we propose turn out to be computationally intractable. Our hardness results often hold even when very little information is missing, or in other words, even when the instances are very close to the complete information setting. Our overall conclusion is that computational hardness continues to be a valid obstruction or deterrent to manipulation, in the context of the more realistic setting of incomplete information.

\subsubsection*{\Cref{chap:detection}: Manipulation Detection}

Even if the computational problem of manipulation is almost always intractable, there have been many instances of manipulation in real elections. In this work, we initiate the study of the computational problem of detecting possible instances of manipulation in an election~\cite{DeyMN15a}. We formulate two pertinent computational 
problems in the context of manipulation detection - Coalitional Possible Manipulators (CPM) and 
Coalitional Possible Manipulators given Winner (CPMW), where a 
suspect group of voters is provided as input and we have to determine whether 
they can form a potential coalition of manipulators.
In the absence of any suspect group, we formulate two more computational 
problems namely Coalitional Possible Manipulators Search (CPMS) and Coalitional Possible Manipulators Search 
given Winner (CPMSW). We provide polynomial time algorithms
for these problems, for several popular voting rules. For a few other voting rules,
we show that these problems are NP-complete. We observe that 
detecting possible instances of manipulation may be easy even when the actual manipulation problem is computationally intractable, as seen
for example, in the case of the Borda voting rule.

\subsubsection*{\Cref{chap:frugal_bribery}: Frugal Bribery}

Another fundamental problem in the context of social choice theory is bribery. Formally, the computational problem of bribery is as follows: given (i) a set of votes, (ii) a cost model for changing the votes, (iii) a budget, and (iv) a candidate $x$, is it possible for an external agent to bribe the voters to change their votes (subject to the budget constraint) so that $x$ wins the election? We introduce and study two important special cases of the classical \textsc{\$Bribery} problem, namely, \textsc{Frugal-bribery} and \textsc{Frugal-\$bribery} where the briber is frugal in nature in~\cite{frugalDeyMN16}. By this, we mean that the briber is 
only able to influence voters who benefit from the suggestion of the briber.
More formally, a voter is {\em vulnerable} if the outcome of the election  
improves according to her own preference when she accepts the suggestion of the 
briber. In the \textsc{Frugal-bribery} problem, the goal of the briber is to make a certain
candidate win the election by changing {\em only} the vulnerable votes. In the
\textsc{Frugal-\$bribery} problem, the vulnerable votes have prices and the
goal is to make a certain candidate win the election by changing only the 
vulnerable votes, subject to a budget constraint. We further
formulate two natural variants of the \textsc{Frugal-\$bribery} problem namely
\textsc{Uniform-frugal-\$bribery} and \textsc{Nonuniform-frugal-\$bribery}
where the prices of the vulnerable votes are, respectively, all same or different.

We observe that, even if we have only a small number of
candidates, the problems are intractable for all voting rules studied here for
weighted elections, with the sole exception of the \textsc{Frugal-bribery}
problem for the plurality voting rule. In contrast, we have polynomial time
algorithms for the \textsc{Frugal-bribery} problem for plurality, veto,
$k$-approval, $k$-veto, and plurality with runoff voting rules for unweighted
elections. However, the \textsc{Frugal-\$bribery} problem is intractable for
all the voting rules studied here barring the plurality and the veto voting
rules for unweighted elections.
These intractability results demonstrate that bribery is a hard computational 
problem, in the sense that several special cases of this problem continue to be
computationally intractable. This strengthens the view that bribery, although a possible attack on an election in principle, may be infeasible in practice.

We believe that the thesis work has attempted to break fresh ground by resolving the computational complexity of many long-standing canonical problems in computational social choice. We finally conclude in \Cref{chap:summary} with future directions of research.
\chapter{Background}
\label{chap:prelim}

\begin{quote}
{\small This chapter provides an overview of selected topics required to understand the technical content in this thesis.
We cover these topics, up to the requirement of following this thesis, in the following order: \nameref{back:voting}, \nameref{bck:cc}, \nameref{bck:prob}, \nameref{bck:inf}, \nameref{bck:tree}, \nameref{bck:hash}.} 
\end{quote}

We present the basic preliminaries in this chapter. We denote, for any natural number $\el\in \NB$, the set $\{1, 2, \ldots, \el\}$ by $[\el]$. We denote the set of permutations of $[\el]$ by $\SF_\el$. For a set $\XX$, we denote the set of subsets of \XX of size $k$ by $\PF_k(\XX)$ and the power set of \XX by $2^\XX$.

\section{Voting and Elections}\label{back:voting}

In this section, we introduce basic terminologies of voting.

\subsection{Basic Setting}

Let $\CC = \{c_1, c_2, \ldots, c_m\}$ be a set of candidates or alternatives and $\VV = \{v_1, v_2, \ldots, v_n\}$ a set of voters. If not mentioned otherwise, we denote the set of candidates by \CC, the set of voters by \VV, the number of candidates by $m$, and the number of voters by $n$. Every voter $v_i$ has a preference or vote $\suc_i$ which is a complete order over \CC. A complete order over any set \XX is a relation on \XX which is reflexive, transitive, anti-symmetric, and total. A relation \RR on \XX is a subset of $\XX\times\XX$. A relation \RR is called reflexive if $(x,x)\in\RR$ for every $x\in\XX$, transitive if $(x,y)\in\RR \text{ and } (y,z)\in\RR$ implies $(x,z)\in\RR$ for every $x,y,z\in\XX$, anti-symmetric is $(x,y)\in\RR \text{ and } (y,x)\in\RR$ implies $x=y$ for every $x,y\in\XX$, and total is either $(x,y)\in\RR$ or $(y,x)\in\RR$ for every $x,y\in\XX$. We denote the set of complete orders over \CC by $\LL(\CC)$. We call a tuple of $n$ preferences $(\suc_1, \suc_2, \cdots, \suc_n)\in\LL(\CC)^n$ an $n$-voter preference profile. Given a preference $\suc = c_1 \suc c_2 \suc \cdots \suc c_m \in\LL(\CC)$, we call the order $c_m \suc c_{m-1} \suc \cdots \suc c_2 \suc c_1$ the reverse order of \suc and denote it by $\overleftarrow{\suc}$. We say that a candidate $x\in\CC$ is at the $i^{th}$ position of a preference \suc if there exists exactly $i-1$ other candidates in \CC who are preferred over $x$ in \suc, that is $|\{y\in\CC\setminus\{x\}:y\suc x\}|=i-1$. It is often convenient to view a preference as a subset of $\CC\times\CC$ --- a preference \suc corresponds to the subset $\AA = \{(x, y)\in\CC\times\CC: x\suc y\}$. For a preference \suc and a subset $\AA\subseteq\CC$ of candidates, we define $\suc(\AA)$ be the preference \suc restricted to \AA, that is $\suc(\AA) = \suc \cap (\AA\times\AA)$.

A voting correspondence is a function $\tilde{r}: \cup_{n\in\NN} \LL(\CC)^n \longrightarrow 2^\CC\setminus\{\emptyset\}$ which selects, from a preference profile, a nonempty set of candidates as the winners. A tie breaking rule $\tau: 2^\CC\setminus\{\emptyset\} \longrightarrow \CC$ with $\tau(\AA)\in \AA$ for every $\AA\in2^\CC\setminus\{\emptyset\}$ selects one candidate from a nonempty set of (tied) candidates. A voting rule $r: \cup_{n\in\NN} \LL(\CC)^n \longrightarrow \CC$ selects one candidate as the winner from a preference profile. An important class of tie breaking rules is the class of {\it lexicographic} tie breaking rules where the ties are broken in a fixed complete order of candidates. More formally, given a complete order $\suc\in\LL(\CC)$, we can define a tie breaking rule $\tau_\suc$ as follows: for every nonempty subset $A\subseteq\CC$ of \CC, $\tau_\suc(A) = x$ such that $x\in A$ and $x\suc y$ for every $y\in A\setminus\{x\}$. A natural way for defining a voting rule is to compose a voting correspondence $\tilde{r}$ with a tie breaking rule $\tau$; that is $r = \tau\circ\tilde{r}$ is a voting rule. We call a preference profile over a set of candidates along with a voting rule an election. The winner (or winners) of an election is often called the outcome of the election.

\subsection{Axioms of Voting Rules}

Let us now look at important axiomatic properties of voting rules that one may wish to satisfy.

\begin{itemize}
 \item {\it Onto:} A voting rule $r$ is called onto if every candidate wins for at least one voting profile. More formally, a voting rule $r$ is called onto if for every $x\in\CC$ there exists a $(\suc_1, \suc_2, \ldots, \suc_n)\in\LL(\CC)^n$ such that $r(\suc_1, \suc_2, \ldots, \suc_n)=x$.
 
 \item {\it Anonymity:} A voting rule is anonymous if the ``names'' of the voters does not affect the outcome. More formally, a voting rule $r$ is said to satisfy anonymity if for every $(\suc_1, \suc_2, \ldots, \suc_n)\in\LL(\CC)^n$, we have $r(\suc_1, \suc_2, \ldots, \suc_n) = r(\suc_{\sigma(1)}, \suc_{\sigma(2)}, \ldots, \suc_{\sigma(n)})$ for every permutation $\sigma\in\SF_n$.
 
 \item {\it Neutrality:} A voting rule is called neutral if the ``names'' of the candidates are immaterial for determining election outcome. That is, a voting rule $r$ is called neutral if for every $(\suc_1, \suc_2, \ldots, \suc_n)\in\LL(\CC)^n$, we have $r(\suc_1, \suc_2, \ldots, \suc_n) = \sigma(r(\sigma(\suc_1), \sigma(\suc_2), \ldots, \sigma(\suc_n)))$ for every permutation $\sigma\in\SF_m$. Given a preference $\suc = c_1 \suc c_2 \suc \cdots \suc c_m$ and a permutation $\sigma\in\SF_m$, we denote the preference $c_{\sigma(1)} \suc c_{\sigma(2)} \suc \cdots \suc c_{\sigma(m)}$ by $\sigma(\suc)$.
 
 \item {\it Homogeneity:} A voting rule is called homogeneous if the outcome of the election solely depends on the fraction of times (of the number of voters $n$) every complete order $\suc\in\LL(\CC)$ appears in the preference profile.
 
 \item {\it Dictatorship:} A voting rule $r$ is called a dictatorship if there exists an integer $i\in[n]$ such that $r(\suc_1, \suc_2, \ldots, \suc_n) = \suc_i\hspace{-.7ex}(1)$ for every $(\suc_1, \suc_2, \ldots, \suc_n)\in\LL(\CC)^n$, where $\suc_i\hspace{-.7ex}(1)$ the candidate placed at the first position of the preference $\suc_i$.
 
 \item {\it Manipulability:} A voting rule $r$ is called manipulable if there exists a preference profile $(\suc_1, \suc_2, \ldots, \suc_n)\in\LL(\CC)^n$ and a preference $\suc\in\LL(\CC)$ such that the following holds.
 
 \[r(\suc_1, \suc_2, \ldots, \suc_{i-1},\suc,\suc_{i+1}, \ldots, \suc_n) \suc_i r (\suc_1, \suc_2, \ldots, \suc_n)\]
 
 That it, voter $v_i$ prefers the election outcome if she reports $\suc$ to be her preference than the election outcome if she reports her true preference $\suc_i$.
 
 \item {\it Condorcet consistency:} A candidate is called the Condorcet winner of an election if it defeats every other candidate in pairwise election. More formally, given an election \EE, let us define $N_\EE(x,y) = |\{ i: x \suc_i y \}|$ for any two candidates $x, y\in\CC$. A candidate $c$ is called the Condorcet winner of an election if $N_\EE(c,x) > \nfrac{n}{2}$ for every candidate $x\in\CC\setminus\{c\}$ other than $c$. A voting rule is called Condorcet consistent is it selects the Condorcet winner as the outcome of the election whenever such a candidate exists.
 
 \item {\it Weak Condorcet consistency: } A candidate is called a weak Condorcet winner of an election if it does not lose to any other candidate in pairwise election. More formally, given an election \EE, let us define $N_\EE(x,y) = |\{ i: x \suc_i y \}|$ for any two candidates $x, y\in\CC$. A candidate $c$ is called the Condorcet winner of an election if $N_\EE(c,x) \ge \nfrac{n}{2}$ for every candidate $x\in\CC\setminus\{c\}$ other than $c$. A voting rule is called weak Condorcet consistent is it selects a weak Condorcet winner as the outcome of the election whenever one such candidate exists.
\end{itemize}

Given an election $\EE = (\suc, \CC)$ and two candidates $x, y\in\CC$, we say that the candidate $x$ defeats the candidate $y$ in {\em pairwise election} if $N_\EE(x,y) > N_\EE(y,x)$.

\subsection{Majority Graph}

Given an election $\EE = (\suc = (\suc_1, \suc_2, \ldots, \suc_n), \CC)$, we can construct a weighted directed graph $\GG_\EE = (U = \CC, E)$ as follows. The vertex set of the graph $\GG_\EE$ is the set of candidates $\CC$. For any two candidates $x, y\in\CC$ with $x\ne y$, let us define the margin $\DD_\EE(x, y)$ of $x$ from $y$ to be $N_\EE(x,y) - N_\EE(y,x)$. We have an edge from $x$ to $y$ in $\GG_\EE$ if $\DD_\EE(x,y)>0$. Moreover, in that case, the weight $w(x,y)$ of the edge from $x$ to $y$ is $\DD_\EE(x,y)$. Observe that, a candidate $c$ is the Condorcet winner of an election \EE if and only if there is an edge from $c$ to every other vertices in the weighted majority graph $\GG_\EE$.

\subsection{Condorcet Paradox}

Given an election $\EE = (\suc, \CC)$, there may exist candidates $c_1, c_2, \ldots, c_k \in \CC$ for some integer $k\ge 3$ such that the candidate $c_i$ defeats $c_{i+1\mod k}$ in pairwise election for every $i\in[k]$. This phenomenon was discovered by Marie Jean Antoine Nicolas de Caritat, marquis de Condorcet and is called Condorcet paradox or Condorcet cycles. \Cref{fig:cond_paradox} shows an example of a Condorcet paradox where $a$ defeats $b$, $b$ defeats $c$, and $c$ defeats $a$.

\begin{figure}[!htbp]
\begin{center}
\begin{tikzpicture}
 \node[draw,align=left] at (0,0) {$a$ \suc $b$ \suc $c$ \\ $b$ \suc $c$ \suc $a$ \\ $c$ \suc $a$ \suc $b$};
\end{tikzpicture}
\end{center}
\caption{Condorcet paradox.}\label{fig:cond_paradox}
\end{figure}

\subsection{Incomplete Information Setting}

A more general setting is an {\it election} where the votes are only 
\emph{partial orders} over candidates. A \emph{partial order} is a relation that is \emph{reflexive, 
antisymmetric}, and \emph{transitive}. A partial vote can be extended to possibly more than one linear votes depending on how we fix the order for the unspecified pairs of candidates. For example, in an election with the set of candidates $\mathcal{C} = \{a, b, c\}$, 
a valid partial vote can be $a \succ b$. This partial vote can be extended to three linear votes namely, $a \succ b \succ c$, $a \succ c \succ b$, $c \succ a \succ b$. However, the voting rules always take as input a set of votes that are complete orders over the candidates.

\subsection{Voting Rules}

Examples of common voting correspondences are as follows. We can use any tie breaking rule, for example, any lexicographic tie breaking rule, with these voting correspondences to get a voting rule.

\begin{itemize}
 \item {\bf Positional scoring rules:} A collection of $m$-dimensional vectors $\overrightarrow{s_m}=\left(\alpha_1,\alpha_2,\dots,\alpha_m\right)\in\mathbb{R}^m$ 
 with $\alpha_1\ge\alpha_2\ge\dots\ge\alpha_m$ and $\alpha_1>\alpha_m$ for every $m\in \mathbb{N}$ naturally defines a voting rule --- a candidate gets score $\alpha_i$ from a vote if it is placed at the $i^{th}$ position, and the  score of a candidate is the sum of the scores it receives from all the votes. 
 The winners are the candidates with maximum score. Scoring rules remain unchanged if we multiply every $\alpha_i$ by any constant $\lambda>0$ and/or add any constant $\mu$. Hence, we assume without loss of generality that for any score vector $\overrightarrow{s_m}$, there exists a $j$ such that $\alpha_j - \alpha_{j+1}=1$ and $\alpha_k = 0$ for all $k>j$. We call such a $\overrightarrow{s_m}$ a normalized score vector. A scoring rule is called {\em strict} if $\alpha_1 > \alpha_2 > \cdots > \alpha_m$ for every natural number $m$. If $\alpha_i$ is $1$ for $i\in [k]$ and $0$ otherwise, then we get the $k$-approval voting rule. For the $k$-veto voting rule, $\alpha_i$ is $0$ for $i\in [m-k]$ and $-1$ otherwise. $1$-approval is called the plurality voting rule and $1$-veto is called the veto voting rule. For the Borda voting rule, we have $\alpha_i=m-i$ for every $i\in[m]$.
 
 \item {\bf Bucklin and simplified Bucklin:} Let \el be the minimum integer such that at least one candidate gets majority within top \el positions of the votes. The winners under the simplified Bucklin voting rule are the candidates having more than $\nfrac{n}{2}$ votes within top \el positions. The winners under the Bucklin voting rule are the candidates appearing within top \el positions of the votes highest number of times. However, for brevity, other than \Cref{chap:partial}, we use Bucklin to mean simplified Bucklin only.
 
 \item {\bf Fallback and simplified Fallback:} For these voting rules, each voter $v$ ranks a subset $\XX_v\subset\CC$ of candidates and disapproves the rest of the candidates~\cite{brams2009voting}. Now for the Fallback and simplified Fallback voting rules, we apply the Bucklin and simplified Bucklin voting rules respectively to define winners. If there is no integer \el for which at least one candidate gets more than $\nfrac{n}{2}$ votes, both the Fallback and simplified Fallback voting rules output the candidates with most approvals as winners. We assume, for simplicity, that the number of candidates each partial vote approves is known.
 
 \item {\bf Maximin:} The maximin score of a candidate $x$ in an election $E$ is $\min_{y\ne x} D_E(x,y)$. The winners are the candidates with maximum maximin score.
 
 \item {\bf Copeland$^{\alpha}$:} Given $\alpha\in[0,1]$, the Copeland$^{\alpha}$ score of a candidate $x$ is $|\{y\ne x:D_E(x,y)>0\}|+\alpha|\{y\ne x:D_E(x,y)=0\}|$. The winners are the candidates with maximum Copeland$^{\alpha}$ score. If not mentioned otherwise, we will assume $\alpha$ to be zero.
 
 \item {\bf Ranked pairs:} Given an election $E$, we pick a pair $(c_i, c_j) \in \CC \times \CC$ such that $D_E(c_i,c_j)$ is maximum. We fix the ordering between $c_i$ and $c_j$ to be 
 $c_i \succ c_j$ unless it contradicts previously fixed orders. We continue this process 
 until all pairwise elections are considered. At this point, we have a complete order $\succ$ over the candidates. Now the top candidate of $\succ$ is chosen as the winner.
 
 \item \textbf{Plurality with runoff:} The top two candidates according to
 plurality score are selected first. The pairwise winner of these two
 candidates is selected as the winner of the election. This rule is
 often called the {\em runoff} voting rule.
 
 \item {\bf Single transferable vote:} In single transferable vote (STV), 
 a candidate with least plurality score is dropped out of the election and its votes 
 are transferred to the next preferred candidate. If two or more candidates receive least plurality score, then some predetermined tie breaking rule is used. The candidate that remains after $(m-1)$ rounds is the winner. The single transferable vote is also called the instant runoff vote.
\end{itemize}

We use the parallel-universes tie breaking~\cite{conitzer2009preference,brill2012price} 
to define the winning candidate for
the ranked pairs voting rule. In this setting, 
a candidate $c$ is a winner if and only if there exists a way to 
break ties in all of the steps such that $c$ is the winner.

\section{Computational Complexity}\label{bck:cc}

\subsection{Parameterized Complexity}

Preprocessing, as a strategy for coping with hard problems, is 
universally applied in practice. The main goal here is \emph{instance compression} - the objective is to output a smaller instance while maintaining equivalence. In the classical setting, \NPH problems are unlikely to have efficient compression algorithms (since repeated application would lead to an efficient solution for the entire problem, which is unexpected). However, the breakthrough notion of \emph{kernelization} in parameterized complexity provides a mathematical framework for analyzing the quality of preprocessing strategies. In parameterized 
complexity, each problem instance comes 
with a parameter $k$, and the central notion is \emph{fixed parameter 
tractability} (FPT) which means, for a 
given instance $(x,k)$, solvability in time $f(k) \cdot p(|x|)$, 
where $f$ is an arbitrary function of $k$ and 
$p$ is a polynomial in the input size $|x|$. 
We use the notation $O^*(f(k))$ to denote $O(f(k)poly(|x|))$.
A parameterized problem $\Pi$ is a 
subset of $\Gamma^{*}\times
\mathbb{N}$, where $\Gamma$ is a finite alphabet. An instance of a
parameterized problem is a tuple $(x,k)$, where $k$ is the
parameter. We refer the reader to the books~\cite{guo2007invitation,downey1999parameterized,flum2006parameterized} for a detailed introduction to this paradigm, and below we state only the definitions that are relevant to our work.

A {\em kernelization} algorithm is a set of preprocessing rules that runs in 
polynomial time and reduces the instance size with a guarantee on the output instance size. This notion is 
formalized below.
\begin{definition}{\rm \bf[Kernelization]}~\cite{niedermeier2006invitation,flum2006parameterized}
A kernelization algorithm for a parameterized problem   $\Pi\subseteq \Gamma^{*}\times \mathbb{N}$ is an 
algorithm that, given $(x,k)\in \Gamma^{*}\times \mathbb{N} $, outputs, in time polynomial in $|x|+k$, a pair 
$(x',k')\in \Gamma^{*}\times
  \mathbb{N}$ such that (a) $(x,k)\in \Pi$ if and only if
  $(x',k')\in \Pi$ and (b) $|x'|,k'\leq g(k)$, where $g$ is some
  computable function.  The output instance $x'$ is called the
  kernel, and the function $g$ is referred to as the size of the
  kernel. If $g(k)=k^{O(1)}$, then we say that
  $\Pi$ admits a polynomial kernel.
\end{definition}

For many parameterized problems, it is well established that the existence of a polynomial kernel would 
imply the collapse of the polynomial hierarchy to the third level (or more precisely, \caveat{}). 
Therefore, it is considered unlikely that these problems would admit polynomial-sized kernels. 
For showing kernel lower bounds, we simply establish reductions from these problems. 

\begin{definition}{\rm \bf[Polynomial Parameter Transformation]}
\label{def:ppt-reduction} {\rm\cite{BodlaenderThomasseYeo2009}}
Let $\Gamma_1$ and $\Gamma_2$ be parameterized problems. We say that $\Gamma_1$ is
polynomial time and parameter reducible to $\Gamma_2$, written
$\Gamma_1\le_{Ptp} \Gamma_2$, if there exists a polynomial time computable
function $f:\Sigma^{*}\times\mathbb{N}\to\Sigma^{*}\times\mathbb{N}$, and a
polynomial $p:\mathbb{N}\to\mathbb{N}$, and for all
$x\in\Sigma^{*}$ and $k\in\mathbb{N}$, if
$f\left(\left(x,k\right)\right)=\left(x',k'\right)$, then
$\left(x,k\right)\in \Gamma_1$ if and only if $\left(x',k'\right)\in \Gamma_2$,
and $k'\le p\left(k\right)$. We call $f$ a polynomial parameter
transformation (or a PPT) from $\Gamma_1$ to $\Gamma_2$.
\end{definition}

This notion of a reduction is useful in showing kernel lower bounds
because of the following theorem.

\begin{theorem}\label{thm:ppt-reduction}~{\rm\cite[Theorem 3]{BodlaenderThomasseYeo2009}}
  Let $P$ and $Q$ be parameterized problems whose derived
  classical problems are $P^{c},Q^{c}$, respectively. Let $P^{c}$
  be $\NPCb{}$, and $Q^{c}\in$ \NPshort{}. Suppose there exists a PPT from
  $P$ to $Q$.  Then, if $Q$ has a polynomial kernel, then $P$ also
  has a polynomial kernel. 
\end{theorem}

\subsection{Approximation Factor of an Algorithm}

For a minimization problem $\mathcal{P}$, we say an algorithm $\mathcal{A}$ archives an approximation factor of $\alpha$ if $\mathcal{A}(\mathcal{I}) \le \alpha OPT(\mathcal{I})$ for every problem instance $\mathcal{I}$ of $\mathcal{P}$. In the above, $OPT(\mathcal{I})$ denotes the value of the optimal solution of the problem instance $\mathcal{I}$.

\subsection{Communication Complexity}

Communication complexity of a function measures the number of bits that need to be exchanged between two players to compute a function whose input is split among those two players~\cite{yao1979some}. In a more restrictive {\it one-way communication model}, Alice, the first player, sends only one message to Bob, the second player, and Bob outputs the result. A protocol is a method that the players follow to compute certain functions of their input. Also the protocols can be randomized; in that case, the protocol needs to output correctly with probability at least $1-\delta$, for $\delta\in(0,1)$ (the probability is taken over the random coin tosses of the protocol). The randomized one-way communication complexity of a function $f$ with error probability $\delta$ is denoted by $\mathcal{R}_\delta^{\text{1-way}}(f)$. \cite{Kushilevitz} is a standard reference for communication complexity.

\section{Probability Theory}\label{bck:prob}

\subsection{Statistical Distance Measures}

Given a finite set $X$, a distribution $\mu$ on $X$ is defined as a function $\mu : X \longrightarrow [0,1]$, such that $\sum_{x\in X} \mu(x) = 1$. 
The finite set $X$ is called the base set of the distribution
$\mu$. We use the following distance measures among distributions in
our work.

\begin{definition}
 The {\em KL divergence}~\cite{kullback1951information} and the {\em Jensen-Shannon divergence}~\cite{lin1991divergence} between two distributions $\mu_1$ and $\mu_2$ on $X$ are defined as follows.
 \[ D_{KL}(\mu_1 || \mu_2) = \sum_{x\in X} \mu_1(x) \ln \frac{\mu_1(x)}{\mu_2(x)} \]
 \[ JS(\mu_1, \mu_2) = \frac{1}{2} \left( D_{KL}\left(\mu_1 || \frac{\mu_1 + \mu_2}{2} \right) + D_{KL}\left(\mu_2 || \frac{\mu_1 + \mu_2}{2} \right) \right) \]
\end{definition}

The Jensen-Shannon divergence has subsequently been generalized to measure the mutual distance among more than two distributions as follows.

\begin{definition}
 Given $n$ distributions $\mu_1, \ldots, \mu_n$ over the same base
 set, the {\em generalized Jensen-Shannon divergence}\footnote{The
   generalized Jensen-Shannon divergence is often formulated with
   weights on each of the $n$ distributions. The definition here puts equal
 weight on each distribution and is sufficient for our purposes.}
among them is:
 \[ JS (\mu_1, \ldots, \mu_n) = \frac{1}{n} \sum_{i=1}^n D_{KL}\left(\mu_i || \frac{1}{n}\sum_{j=1}^n \mu_j\right) \]
\end{definition}

\subsection{Concentration Inequalities of Probability}

Suppose we have $n$ events namely $A_1, A_2, \ldots, A_n$ in a probability space. Then the probability of occurring any of these $n$ events can be bounded by what is known as union bound.

\begin{theorem}(Union Bound)\\
 The probability of happening any of a given $n$ events $A_1, A_2, \ldots, A_n$ can be upper bounded as follows.
 \[ \Pr[\cup_{i=1}^n A_i] \le \sum_{i=1}^n \Pr[A_i] \]
\end{theorem}

Given a probability space $(\Omega, \FF, \mu)$, a random variable $X: \Omega \longrightarrow \RB$ is a function from the sample space $\Omega$ to the set of real numbers $\RB$ such that $\{\omega\in\Omega: X(\omega) \le c\} \in \FF$ for every $c\in\RB$. We refer to \cite{durrett2010probability} for elementary notions of probability theory.

Given a positive random variable $X$ with finite mean $\EB[X]$, the classical Markov inequality bounds the probability by which the random variable $X$ deviates from its mean.

\begin{theorem}(Markov Inequality)\\
 For a positive random variable $X$ with mean $\EB[X]$, the probability that the random variable $X$ takes value more than $c\EB[X]$ can be upper bounded as follows.
 \[ \Pr[X \ge c\EB[X]] \le \frac{1}{c} \]
\end{theorem}

If the variance of a random variable $X$ is known, the Chebyshev inequality provides a sharper concentration bound for a random variable.

\begin{theorem}(Chebyshev Inequality)\\
 For a random variable with mean $\EB[X]$ and variance $\sigma^2$, the probability that the random variable $X$ takes value more than $c\sigma$ away from its mean $\EB[X]$ is upper bounded as follows.
 \[ \Pr[|X-\EB[X]| \ge c\sigma] \le \frac{1}{c^2} \]
\end{theorem}

For a sum of independent bounded random variables $X$, the Chernoff bound gives a much tighter concentration around the mean of $X$.

\begin{theorem}(Chernoff Bound)\\\label{thm:chernoff}
Let $X_1, \dots, X_\ell$ be a sequence of $\ell$ independent
random variables in $[0,1]$ (not necessarily identical). Let $S = \sum_i X_i$ and
let $\mu = \E{S}$. Then, for any $0 \leq \delta \leq 1$:
$$\Pr[|S - \mu| \geq \delta \mu] < 2\exp(-\delta^2\mu/3)$$
\end{theorem}

\section{Information Theory}\label{bck:inf}

For a discrete random variable $X$ with possible values $\{x_1,x_2,\ldots,x_n\}$, the
Shannon entropy of $X$ is defined as
$H(X)=-\sum_{i=1}^{n}\Pr(X=x_i)\log_2 \Pr(X=x_i)$. Let $H_b(p)=-p\log_2 p-(1-p)\log_2
(1-p)$ denote the binary entropy function when $p \in (0,1)$.
For two random variables $X$ and $Y$ with possible values $\{x_1,x_2,\ldots,x_n\}$ and
$\{y_1,y_2,\ldots,y_m\}$, respectively, the conditional entropy of $X$ given $Y$ is defined
as $H(X\ |\ Y)=\sum_{i,j}\Pr(X=x_i,Y=y_j)\log_2\frac{\Pr(Y=y_j)}{\Pr(X=x_i,Y=y_j)}$.
Let $I(X; Y) = H(X) - H(X\ |\ Y) = H(Y) - H(Y\ |\ X)$
denote the mutual information between two random variables $X, Y$.
Let $I(X; Y\ |\ Z)$ denote the mutual information between two random variables $X, Y$
conditioned on $Z$, i.e.,
$I(X ; Y\ |\ Z) = H(X\ |\ Z) - H(X\ |\ Y, Z)$.
The following summarizes several basic properties of entropy and mutual information.
\begin{proposition}\label{prop:mut}
Let $X, Y, Z, W$ be random variables.
\begin{enumerate}
\item If $X$ takes value in $\{1,2, \ldots, m\}$, then $H(X) \in [0, \log m]$.
\item $H(X) \ge H(X\ |\ Y)$ and $I(X; Y) = H(X) - H(X\ |\ Y) \ge 0$.
\item If $X$ and $Z$ are independent, then we have $I(X; Y\ |\ Z) \ge I(X; Y)$.
Similarly, if $X, Z$ are independent given $W$, then $I(X; Y\ |\ Z, W) \ge I(X; Y\ |\ W)$.
\item (Chain rule of mutual information)
$I(X, Y; Z) = I(X; Z) + I(Y; Z\ |\ X).$
More generally, for any random variables $X_1, X_2, \ldots, X_n, Y$,
$\textstyle I(X_1, \ldots, X_n; Y) = \sum_{i = 1}^n I(X_i; Y\ |\ X_1, \ldots, X_{i-1})$.
Thus, $I(X, Y; Z\ |\ W) \ge I(X; Z\ |\ W)$.
\item (Fano's inequality) Let $X$ be a random variable chosen from domain $\mathcal{X}$
according to distribution $\mu_X$, and $Y$ be a random variable chosen from domain
$\mathcal{Y}$ according to distribution $\mu_Y$. For any reconstruction function $g :
\mathcal{Y} \to \mathcal{X}$ with error $\delta_g$,
$$H_b(\delta_g) + \delta_g \log(|\mathcal{X}| - 1) \ge H(X\ |\ Y).$$
\end{enumerate}
\end{proposition}

We refer readers to \citep{cover2012elements} for a nice introduction to information theory.

\section{Relevant Terminology for Trees}\label{bck:tree}

A tree $\TT = (\VV, \EE)$ is a set of vertices \VV along with a set of edges $\EE\subset\PF_2(\VV)$ such that for every two vertices $x, y\in\VV$, there exists exactly one path between $x$ and $y$ in \TT. The following definitions pertaining to the structural aspects of trees will be useful.
\begin{itemize}
	\item The {\em pathwidth} of \TT is the minimum width of a {\em path decomposition} of \TT \cite{heinrich1993path}. 
	\item A set of disjoint paths $\QQ = \{Q_1 = (\XX_1, \EE_1), \ldots, Q_k = (\XX_k, \EE_k)\}$ is said to cover a tree $\TT = (\XX, \EE)$ if $\XX = \cup_{i\in[k]} \XX_i, \EE_i\subseteq \EE, \XX_i \cap \XX_j = \emptyset, \EE_i \cap \EE_j = \emptyset$ for every $i, j\in[k]$ with $i\ne j$. The {\em path cover number} of \TT is the cardinality of the smallest set \QQ of {\em disjoint} paths that cover \TT.
	\item The {\em distance of a tree \TT from a path} is the smallest number of nodes whose removal makes the tree a path.
	\item The {\em diameter} of a tree \TT is the number of edges in the longest path in \TT.
\end{itemize}

We also list some definitions of \textit{subclasses} of trees (which are special types of trees, see also~\Cref{fig:tree-classes}). 
\begin{itemize}
\item A tree is a {\em star} if there is a {\em center vertex} and every other vertex is a neighbor of this vertex.
\item A tree is a {\em subdivision of a star} if it can be constructed by replacing each edge of a star by a path. 
\item A subdivision of a star is called {\em balanced} if there exists an integer \el such that the distance of every leaf node from the center is \el. 
\item A tree is a {\em caterpillar} if there is a {\em central path} and every other vertex is at a distance of one it. 
\item A tree is a {\em complete binary tree} rooted at $r$ if every nonleaf node has exactly two children and there exists an integer $h$, called the height of the tree, such that every leaf node is at a distance of either $h$ or $h-1$ from the root node $r$.	
\end{itemize}

\begin{figure}[t]
	\centering
	\includegraphics[scale=0.35]{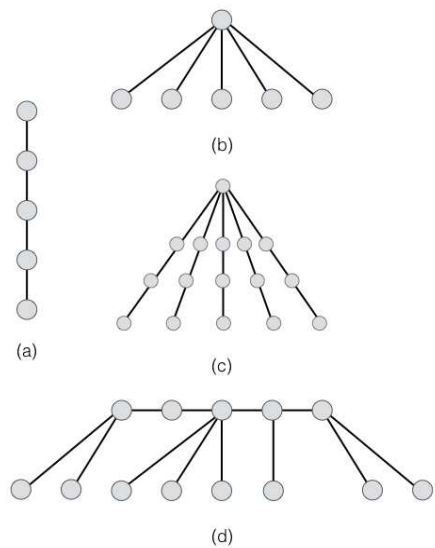}
	\caption{Depicting classes of trees: (a) a path, (b) a star, (c) a (balanced) subdivision of a star, (d) a caterpillar.}
	\label{fig:tree-classes}
\end{figure}

\section{Universal Family of Hash Functions}\label{bck:hash}

In this section, we discuss universal family of hash functions. Intuitively, a family of hash functions is universal if a hash function picked uniformly at random from the family behaves like a completely random function. The formal definition of universal family of hash functions is as follows. We use this notion crucially in \Cref{chap:winner_stream}.

\begin{definition}(Universal family of hash functions)\\
 A family of functions $\mathcal{H} = \{ h | h:A\rightarrow B \}$ is called a {\em universal family of hash functions} if for all $a \ne b \in A$, $\Pr\{ h(a)= h(b) \} = \nfrac{1}{|B|}$, where $h$ is picked uniformly at random from $\mathcal{H}$.
 
 \[ \Pr_{h \text{ picked uniformly at random from } \mathcal{H}}\{ h(a)= h(b) \} = \nfrac{1}{|B|} \]
\end{definition}

We know that there exists a universal family of hash functions $\mathcal{H}$ from $[k]$ to $[\ell]$ for every positive integer $\ell$ and every prime $k$~\cite{cormen2009introduction}. Moreover, $|\mathcal{H}|$, the size of $\mathcal{H}$, is $O(k^2)$.
\blankpage
\part{Preference Elicitation}
\vspace{5ex}
This part of the thesis consists of the following two chapters.

\begin{itemize}
 \item In \Cref{chap:pref_elicit_peak} -- \nameref{chap:pref_elicit_peak} -- we present efficient algorithms for preference elicitation when the preferences belong to the domain of single peaked preferences on a tree. We show interesting connections between query complexity for preference elicitation and various common parameters of trees.
 \item In \Cref{chap:pref_elicit_cross} -- \nameref{chap:pref_elicit_cross} -- we show optimal algorithms (except for one scenario) for eliciting single crossing preference profiles. Our results show that the query complexity for preference elicitation for single crossing preference profiles crucially depends on the way the votes are allowed to access.
\end{itemize}

\blankpage
\chapter{Preference Elicitation for Single Peaked Profiles on Trees}
\label{chap:pref_elicit_peak}

\blfootnote{A preliminary version of the work in this chapter was published as \cite{deypeak}: Palash Dey and Neeldhara Misra. Elicitation for preferences single peaked on trees. In
Proc. Twenty-Fifth International Joint Conference on Artificial Intelligence,
IJCAI 2016, New York, NY, USA, 9-15 July 2016, pages 215-221, 2016.}

\begin{quotation}
{\small In multiagent systems, we often have a set of agents each of which has a preference ordering over a set of items. One would often like to know these preference orderings for various tasks like data analysis, preference aggregation, voting etc. However, we often have a large number of items which makes it impractical to directly ask the agents for their complete preference ordering. In such scenarios, one convenient way to know these agents' preferences is to ask the agents to compare (hopefully a small number of) pairs of items.

Prior work on preference elicitation focuses on unrestricted domain and the domain of single peaked preferences. They show that the preferences in single peaked domain can be elicited with much less number of queries compared to unrestricted domain. We extend this line of research and study preference elicitation for single peaked preferences on trees which is a strict superset of the domain of single peaked preferences. We show that the query complexity for preference elicitation crucially depends on the number of leaves, the path cover number, and the minimum number of nodes that should be removed to turn the single peaked tree into a path. We also observe that the other natural parameters of the single peaked tree like maximum and minimum degrees of a node, diameter, pathwidth do not play any direct role in determining the query complexity of preference elicitation. We then investigate the query complexity for finding a weak Condorcet winner (we know that at least one weak Condorcet winner always exists for profiles which are single peaked on a tree) for preferences single peaked on a tree and show that this task has much less query complexity than preference elicitation. Here again we observe that the number of leaves in the underlying single peaked tree and the path cover number of the tree influence the query complexity of the problem.}
\end{quotation}

\section{Introduction}
In multiagent systems, we often have scenarios where a set of agents have to arrive at a consensus although they possess different opinions over the alternatives available. Typically, the agents have preferences over a set of items, and the problem of aggregating these preferences in a suitable manner is one of the most well-studied problems in social choice theory~\cite{brandt2015handbook}. There are many ways of expressing preferences over a set of alternatives. One of the most comprehensive ways is to specify a complete ranking over the set of alternatives. However, one of the downsides of this model is the fact that it can be expensive to solicit the preferences when a large number of alternatives and agents are involved. 

Since asking agents to provide their complete rankings is impractical, a popular notion is one of \textit{elicitation}, where we ask agents simple \textit{comparison queries}, such as if they prefer an alternative $x$ over another alternative $y$. This naturally gives rise to the problem of \textit{preference elicitation}, where we hope to recover the complete ranking (or possibly the most relevant part of the ranking) based on a small number of comparison queries. 

In the context of a fixed voting rule, we may also want to query the voters up to the point of determining the winner (or the aggregate ranking, as the case may be). Yet another refinement in this setting is when we have prior information about how agents are likely to vote, and we may want to determine which voters to query first, to be able to quickly rule out a large number of alternatives, as explored by~\cite{conitzer2002vote}. 

When our goal is to elicit preferences that have no prior structure, one can demonstrate scenarios where it is imperative to ask each agent (almost) as many queries as would be required to determine an arbitrary ranking. However, in recent times, there has been considerable interest in voting profiles that are endowed with additional structure. The motivation for this is two-fold. The first is that in several application scenarios commonly considered, it is rare that votes are ad-hoc, demonstrating no patterns whatsoever. For example, the notion of \textit{single-peaked preferences}, which we will soon discuss at length, forms the basis of several studies in the analytical political sciences~\cite{hinich1997analytical}. In his work on eliciting preferences that demonstrate the ``single-peaked'' structure, Conitzer argues that the notion of single-peakedness is also a reasonable restriction in applications outside of the domain of political elections~\cite{Conitzer09}.

The second motivation for studying restricted preferences is somewhat more technical, but is just as compelling. To understand why structured preferences have received considerable attention from social choice theorists, we must first take a brief detour into some of the foundational ideas that have shaped the landscape of voting theory as we understand them today. As it turns out, the axiomatic approach of social choice involves defining certain ``properties'' that formally capture the quality of a voting rule. For example, we would not want a voting rule to be, informally speaking, a \textit{dictatorship}, which would essentially mean that it discards all but one voter's input. Unfortunately, a series of cornerstone results establish that it is impossible to devise voting rules which respect some of the simplest desirable properties. Indeed, the classic work of Arrow \cite{arrow1950difficulty} and Gibbard-Satterthwaite \cite{gibbard1973manipulation,satterthwaite1975strategy} show that there is no straight forward way to simultaneously deal with properties like voting paradoxes, strategy-proofness, nondictatorship, unanimity etc. We refer to \cite{moulin1991axioms} for a more elaborate discussion on this. Making the matter worse, many classical voting rules, for example the Kemeny voting rule~\cite{kemeny1959mathematics,black1958theory,bartholdi1989voting,hemaspaandra2005complexity}, the Dodgson voting rule~\cite{dodgson1876method,hemaspaandra1997exact}, the Young voting rule~\cite{young1977extending} etc., turn out to be computationally intractable.

This brings us to the second reason for why structured preferences are an important consideration. The notion of single-peakedness that we mentioned earlier is an excellent illustration (we refer the reader to~\Cref{peak:sec:prelim} for the formal definition). Introduced in~\cite{black1948rationale}, it not only captures the essence of structure in political elections, but also turns out to be extremely conducive to many natural theoretical considerations. To begin with, one can devise voting rules that are ``nice'' with respect to several properties, when preferences are single-peaked. Further, they are structurally elegant from the point of view of winner determination, since they always admit a {\em weak Condorcet winner} --- a candidate which is not defeated by any other candidate in pairwise election --- thus working around the Condorcet paradox which is otherwise a prominent concern in the general scenario. In a landmark contribution, Brandt et al. \cite{brandt2015bypassing} show that several computational problems in social choice theory that are intractable in the general setting become polynomial time solvable when we consider single-peaked preferences. 

A natural question at this point is if the problem of elicitation becomes any easier --- that is, if we can get away with fewer queries --- by taking advantage of the structure provided by single-peakedness. It turns out that the answer to this is in the affirmative, as shown in a detailed study by Conitzer~\cite{Conitzer09}. The definition of single-peakedness involves an ordering over the candidates (called the \textit{harmonious ordering} by some authors). The work of Conitzer~\cite{Conitzer09} shows that $\BigO(mn)$ queries suffice, assuming either that the harmonious ordering is given, or one of the votes is known, where $m$ and $n$ are the number of candidates and voters respectively.

We now return to the theme of structural restrictions on preferences. As it turns out, the single peaked preference domain has subsequently been generalized to single peakedness on trees (roughly speaking, these are profiles that are single peaked on every path on a given tree) \cite{demange1982single,trick1989recognizing}. This is a class that continues to exhibit many desirable properties of single peaked domain. For example, there always exists a weak Condorcet winner and further, many voting rules that are intractable in an unrestricted domain are polynomial time computable if the underlying single peaked tree is ``nice'' \cite{yu2013multiwinner,peters2016preferences}. We note the class of profiles that are single peaked on trees are substantially more general than the class of single peaked preferences. Note that the latter is a special case since a path is, in particular, a tree. Our work here addresses the issue of elicitation on profiles that are single-peaked on trees, and can be seen as a significant generalization of the results in~\cite{Conitzer09}. We now detail the specifics of our contributions. 

\begin{table*}[htbp]
 \begin{center}
 \renewcommand{\arraystretch}{1.1}
  \begin{tabular}{|c|c|c|}\hline
   Parameter & Upper Bound & Lower Bound \\\hline\hline
   
   Path width ($w$) & \specialcell{$\BigO(mn\log m)$\\\relax[\Cref{obs:trivial}]} & \specialcell{$\Omega(mn\log m)$ even for $w = 1, \log m$\\\relax[\Cref{cor:pwidth_elicit_lb}]} \\\hline
      
   Maximum degree ($\Delta$) & \specialcell{$\BigO(mn\log m)$\\\relax\Cref{obs:trivial}]} & \specialcell{$\Omega(mn\log m)$ even for $\Delta = 3, m-1$\\\relax[\Cref{cor:maxdeg_elicit_lb}]} \\\hline
   
   Path cover number ($k$) & \specialcell{$\BigO(mn\log k)$\\\relax[\Cref{thm:pcover_elicit_ub}]} & \specialcell{$\Omega(mn\log k)$ \\\relax[\Cref{cor:pcover_elicit_lb}]} \\\hline
   
   Number of leaves (\el) & \specialcell{$\BigO(mn\log \el)$\\\relax[\Cref{cor:leaves_elicit_ub}]} & \specialcell{$\Omega(mn\log \el)$\\\relax[\Cref{thm:leaves_elicit_lb}]} \\\hline
   
   Distance from path ($d$) & \specialcell{$\BigO(mn + nd\log d)$\\\relax[\Cref{thm:pdist_elicit_ub}]} & \specialcell{$\Omega(mn + nd\log d)$ \\\relax[\Cref{thm:pdist_elicit_lb}]} \\\hline
   
   Diameter ($\omega$) & \specialcell{$\BigO(mn\log m)$\\\relax[\Cref{obs:trivial}]} & \specialcell{$\Omega(mn\log m)$ even for $\omega = 2, \nfrac{m}{2}$\\\relax[\Cref{cor:dia_elicit_lb}]} \\\hline
  \end{tabular}
 \end{center}
\caption{Summary of query complexity bounds for eliciting preferences which are single peaked on trees.}\label{tbl:elicit}
\end{table*}

\subsection{Our Contribution}

We study the query complexity for preference elicitation when the preference profile is single peaked on a tree. We provide {\em tight} connections between various parameters of the underlying single peaked tree and the query complexity for preference elicitation. Our broad goal is to provide a suitable generalization of preference elicitation for single peaked profiles to profiles that are single peaked on trees. Therefore, we consider various ways of quantifying the ``closeness'' of a tree to a path, and reflect on how these measures might factor into the query complexity of an algorithm that is actively exploiting the underlying tree structure.

We summarize our results for preference elicitation in \Cref{tbl:elicit}, where the readers will note that most of the parameters (except diameter) chosen are small constants (typically zero, one or two) when the tree under consideration is a path. Observe that in some cases --- such as the number of leaves, or the path cover number --- the dependence on the parameter is transparent (and we recover the results of \cite{Conitzer09} as a special case), while in other cases, it is clear that the perspective provides no additional mileage (the non-trivial results here are the matching lower bounds).

In terms of technique, our strategy is to ``scoop out the paths from the tree'' and use the algorithm form \cite{Conitzer09} to efficiently elicit the preference on the parts of the trees that are paths. We then efficiently merge this information across the board, and that aspect of the algorithm varies depending on the parameter under consideration. The lower bounds typically come from trees that provide large ``degrees of freedom'' in reordering candidates, typically these are trees that don't have too many long paths (such as stars). The arguments are often subtle but intuitive. 

We then study the query complexity for finding a weak Condorcet winner of a preference profile which is single peaked on a tree. Here, we are able to show that a weak Condorcet winner can be found with far fewer queries than the corresponding elicitation problem. In particular, we establish that a weak Condorcet winner can be found using $\BigO(mn)$ many queries for profiles that are single peaked on trees [\Cref{{thm:cw_gen_ub}}], and we also show that this bound is the best that we can hope for [\Cref{thm:cw_gen_lb}]. We also consider the problem for the special case of single peaked profiles. While Conitzer in~\cite{Conitzer09} showed that $\Omega(mn)$ queries are necessary to determine the \textit{aggregate ranking}, we show that only $\BigO(n\log m)$ queries suffice if we are just interested in (one of the) weak Condorcet winners. Moreover, we show this bound is tight under the condition that the algorithm does not interleave queries to different voters [\Cref{thm:con_sp_lb}] (our algorithm indeed satisfies this condition). 

Finally, for expressing the query complexity for determining a weak Condorcet winner in terms of a measure of closeness to a path, we show an algorithm with query complexity $\BigO(nk\log m)$ where $k$ is the path cover number of \TT [\Cref{thm:cw_pc_ub}] or the number of leaves in \TT [\Cref{cor:cw_leaf_ub}]. We now elaborate further on our specific contributions for preference elicitation.

\begin{itemize}
 \item We design novel algorithms for preference elicitation for profiles which are single peaked on a tree with \el leaves with query complexity $\BigO(mn\log \el)$ [\Cref{cor:leaves_elicit_ub}]. Moreover, we prove that there exists a tree \TT with \el leaves such that any preference elicitation algorithm for profiles which are single peaked on tree \TT has query complexity $\Omega(mn\log\el)$ [\Cref{thm:leaves_elicit_lb}]. We show similar results for the parameter path cover number of the underlying tree [\Cref{thm:pcover_elicit_ub,cor:pcover_elicit_lb}]. 
 
 \item We provide a preference elicitation algorithm with query complexity $\BigO(mn + nd\log d)$ for single peaked profiles on trees which can be made into a path by deleting at most $d$ nodes [\Cref{thm:pdist_elicit_ub}]. We show that our query complexity upper bound is tight up to constant factors [\Cref{thm:pdist_elicit_lb}]. These results show that the query complexity for preference elicitation tightly depends on the number of leaves, the path cover number, and the distance from path of the underlying tree.
 
 \item We then show that there exists a tree \TT with pathwidth one or $\log m$ [\Cref{cor:pwidth_elicit_lb}] or maximum degree is $3$ or $m-1$ [\Cref{cor:maxdeg_elicit_lb}] or  diameter is $2$ or $\nfrac{m}{2}$ [\Cref{cor:dia_elicit_lb}] such that any preference elicitation algorithm for single peaked profiles on the tree \TT has query complexity $\Omega(mn\log m)$. These results show that the query complexity for preference elicitation does not directly depend on the parameters above.
\end{itemize}

We next investigate the query complexity for finding a weak Condorcet winner for profiles which are single peaked on trees and we have the following results.

\begin{itemize}
 \item We show that a weak Condorcet winner can be found using $\BigO(mn)$ queries for profiles that are single peaked on trees [\Cref{{thm:cw_gen_ub}}] which is better than the query complexity for preference elicitation. Moreover, we prove that this bound is tight in the sense that any algorithm for finding a weak Condorcet winner for profiles that are single peaked on stars has query complexity $\Omega(mn)$ [\Cref{thm:cw_gen_lb}].
 
 \item On the other hand, we can find a weak Condorcet winner using only $\BigO(n\log m)$ queries for single peaked profiles [\Cref{thm:con_sp_ub}]. Moreover, we show that this bound is tight under the condition that the algorithm does not interleave queries to different voters [\Cref{thm:con_sp_lb}] (our algorithm indeed satisfies this condition). For any arbitrary underlying single peaked tree \TT, we provide an algorithm for finding a weak Condorcet winner with query complexity $\BigO(nk\log m)$ where $k$ is the path cover number of \TT [\Cref{thm:cw_pc_ub}] or the number of leaves in \TT [\Cref{cor:cw_leaf_ub}].
\end{itemize}

To summarize, we remark that our results non-trivially generalize earlier works on query complexity for preference elicitation in \cite{Conitzer09}. We believe revisiting the preference elicitation problem in the context of profiles that are single peaked on trees is timely, and that this work also provides fresh algorithmic and structural insights on the domain of preferences that are single peaked on trees.

\subsection{Related Work}

We have already mentioned the work in~\cite{Conitzer09} addressing the question of eliciting preferences in single-peaked profiles, which is the closest predecessor to our work. Before this, Conitzer and Sandholm addressed the computational hardness for querying minimally for winner determination \cite{conitzer2002vote}. They also prove that one would need to make $\Omega(mn\log m)$ queries even to decide the winner for many commonly used voting rules \cite{conitzer2005communication} which matches with the trivial $\BigO(mn\log m)$ upper bound for preference elicitation in unrestricted domain (due to sorting lower bound). Ding and Lin study preference elicitation under partial information setting and show interesting properties of what they call a deciding set of queries \cite{ding2013voting}. Lu and Boutilier provide empirical study of preference elicitation under probabilistic preference model \cite{lu2011vote} and devise several novel heuristics which often work well in practice \cite{lu2011robust}.

\section{Domain of Single Peaked Profiles on Trees}\label{peak:sec:prelim}

A preference $\succ\thinspace\in\LL(\CC)$ over a set of candidates \CC is called {\em single peaked} with respect to an order $\succ^\pr\in\LL(\CC)$ if, for every candidates $x, y\in\CC$, we have $x\succ y$ whenever we have either $c\succ^\pr x\succ^\pr y$ or $y\succ^\pr x\succ^\pr c$, where $c\in\CC$ is the candidate at the first position of $\succ$. A profile $\PP=(\succ_i)_{i\in[n]}$ is called single peaked with respect to an order $\succ^\pr\in\LL(\CC)$ if $\succ_i$ is single peaked with respect to $\succ^\pr$ for every $i\in[n]$. Notice that if a profile \PP is single peaked with respect to an order $\succ^\pr\in\LL(\CC)$, then \PP is also single peaked with respect to the order $\overleftarrow{\succ}^\pr$. \Cref{ex:single_peak} exhibits an example of a single peaked preference profile.

\begin{example}(Example of single peaked preference profile)\label{ex:single_peak}
 Consider a set \CC of $m$ candidates, corresponding $m$ distinct points on the Real line, a set \VV of $n$ voters, also corresponding $n$ points on the Real line, and the preference of every voter are based on their distance to the candidates -- given any two candidates, every voter prefers the candidate nearer to her (breaking the tie arbitrarily). Then the set of $n$ preferences is single peaked with respect to the ordering of the candidates according to the ascending ordering of their positions on the Real line.
\end{example}

Given a path $\QQ = (x_1, x_2, \ldots, x_\el)$ from a vertex $x_1$ to another vertex $x_\el$ in a tree \TT, we define the order induced by the path \QQ to be $x_1\succ x_2\succ \cdots\succ x_\el$. Given a tree $\TT = (\CC, \EE)$ with the set of nodes as the set of candidates \CC, a profile \PP is called single peaked on the tree \TT if \PP is single peaked on the order induced by every path of the tree \TT; that is for every two candidates $x, y\in\CC$, the profile $\PP(\XX)$ is single peaked with respect to the order $\succ$ on the set of candidates \XX induced by the unique path from $x$ to $y$ in \TT. We call the tree \TT the {\em underlying single peaked tree} of the preference profile \PP. It is known (c.f.~\cite{demange1982single}) that there always exists a weakly Condorcet winner for a profile \PP which is single peaked on a tree \TT.

\section{Problem Definitions and Known Results}

Suppose we have a profile \PP with $n$ voters and $m$ candidates. For any pair of distinct candidates $x$ and $y$, and a voter $\el \in [n]$, we introduce the boolean-valued function $\text{\Query}(x \succ_\el y)$ as follows. The output of this function is \true if the voter \el prefers the candidate $x$ over the candidate $y$ and \false otherwise. We now formally state the two problems that we consider in this work.

\begin{definition}\PE\\\label{def:query_model}
 Given a tree $\TT = (\CC, \EE)$ and an oracle access to the function \Query($\cdot$) for a profile \PP which is single peaked on \TT, find \PP.
\end{definition}

Suppose we have a set of candidates $\CC = \{c_1, \ldots, c_m\}$. We say that an algorithm \AA makes $q$ queries if there are exactly $q$ distinct tuples $(\el, c_i, c_j)\in[n]\times\CC\times\CC$ with $i<j$ such that \AA calls \Query($c_i\succ_\el c_j$) or \Query($c_j\succ_\el c_i$). We call the maximum number of queries made by an algorithm \AA for any input its {\em query complexity}. 

We state some known results that we will appeal to later. The first observation employs a sorting algorithm like merge sort to elicit every vote with $\BigO(m\log m)$ queries, while the second follows from the linear-time merge subroutine of merge sort (\cite{cormen2009introduction}).

\begin{observation}\label{obs:trivial}
 There is a \PE algorithm with query complexity $\BigO(mn\log m)$.
\end{observation}

\begin{observation}\label{obs:merge_single_peak}
 Suppose $\CC_1, \CC_2\subseteq\CC$ form a partition of \CC and $\succ$ is a ranking of the candidates in \CC. Then there is a linear time algorithm that finds $\succ$ given $\succ(\CC_1)$ and $\succ(\CC_2)$ with query complexity $\BigO(|\CC|)$.
\end{observation}

\begin{theorem}\label{thm:con}\cite{Conitzer09} There is a \PE algorithm with query complexity $\BigO(mn)$ for single peaked profiles.
\end{theorem}

We now state the \CW problem, which asks for eliciting only up to the point of determining a weak Condercet winner (recall that at least one such winner is guaranteed to exist on profiles that are single-peaked on trees). 
\begin{definition}\CW\\
 Given a tree $\TT = (\CC, \EE)$ and an oracle access to the function \Query($\cdot$) for a profile \PP which is single peaked on \TT, find a weak Condorcet winner of \PP.
\end{definition}

\section{Results for Preference Elicitation}
In this section, we present our results for \PE for profiles that are single peaked on trees. Recall that we would like to generalize~\Cref{thm:con} in a way to profiles that are single peaked on trees. Since the usual single peaked profiles can be viewed as profiles single peaked with respect to a path, we propose the following measures of how much a tree resembles a path.

 \begin{itemize}
 	\item \textit{Leaves.} Recall any tree has at least two leaves, and paths are the trees that have exactly two leaves. We consider the class of trees that have $\el$ leaves, and show an algorithm with query complexity of $\BigO(mn\log \el)$. 
 	\item \textit{Path Cover.} Consider the notion of a \textit{path cover number} of a tree, which is the smallest number of disjoint paths that the tree can be partitioned into. Clearly, the path cover number of a path is one; and for trees that can be covered with $k$ paths, we show an algorithm with query complexity $\BigO(mn\log k)$.
 	\item \textit{Distance from Paths.} Let $d$ be the size of the smallest set of vertices whose removal makes the tree a path. Again, if the tree is a path, then the said set is simply the empty set. For trees that are at a distance $d$ from being a path (in the sense of vertex deletion), we provide an algorithm with query complexity $\BigO(mn\log d)$.
 	\item \textit{Pathwidth and Maximum Degree.} Finally, we note that paths are also trees that have pathwidth one, and maximum degree two. These perspectives turn out to be less useful: in particular, there are trees where these parameters are constant, for which we show that elicitation is as hard as it would be on an arbitrary profile, and therefore the easy algorithm from~\Cref{obs:trivial} is actually the best that we can hope for. 
 \end{itemize}

For the first three perspectives that we employ, that seemingly capture an appropriate aspect of paths and carry it forward to trees, the query complexities that we obtain are tight --- we have matching lower bounds in all cases. Also, while considering structural parameters, it is natural to wonder if there is a class of trees that are incomparable with paths but effective for elicitation. Our attempt in this direction is to consider trees of bounded diameter. However, again, we find that this is not useful, as we have examples to show that there exist trees of diameter two that are as hard to elicit as general profiles. 

We remark at this point that all these parameters are polynomially computable for trees, making the algorithmic results viable. For example, the distance from path can be found by finding the longest path in the single peaked tree which can be computed in polynomial amount of time \cite{cormen2009introduction}. Also, for the parameters of pathwidth, maximum degree and diameter, we show lower bounds on trees where these parameters are large (such as trees with pathwidth $O(\log m)$, maximum degree $m-1$, and diameter $m/2$), which --- roughly speaking --- also rules out the possibility of getting a good inverse dependence. As a concrete example, motivated by the $\BigO(mn)$ algorithm for paths, which have diameter $m$, one might wonder if there is an algorithm with query complexity $\BigO(\frac{mn \log m}{\log \omega})$ for single peaked profiles on a tree with diameter $\omega$. This possibility, in particular, is ruled out. 
We are now ready to discuss the results in~\Cref{tbl:elicit}. 

We begin with showing a structural result about trees: any tree with \el leaves can be partitioned into \el paths. The idea is to fix some non-leaf node as root and iteratively find a path from low depth nodes (depth of a node is its distance from root) to some leaf node which is disjoint from all the paths chosen so far. We formalize this idea below.

\begin{lemma}\label{lem:path_decom}
 Let $\TT = (\XX, \EE)$ be a tree with \el leaves. Then there is a polynomial time algorithm which partitions \TT into \el disjoint paths $\QQ_i = (\XX_i, \EE_i), i\in[\el]$; that is we have $\XX_i \cap \XX_j = \emptyset, \EE_i \cap \EE_j = \emptyset$ for every $i, j\in[\el]$ with $i\ne j$, $\XX = \cup_{i\in[\el]} \XX_i$, $\EE = \cup_{i\in[\el]} \EE_i$, and $\QQ_i$ is a path in \TT for every $i\in[\el]$.
\end{lemma}

\begin{proof}
 We first make the tree \TT rooted at any arbitrary nonleaf node $r$. We now partition the tree \TT into paths iteratively as follows. Initially every node of the tree \TT is {\em unmarked} and the set of paths \QQ we have is empty. Let $Q_1 = (\XX_1, \EE_1)$ be the path in \TT from the root node to any leaf node. We put $Q_1$ into \QQ and {\em mark} all the nodes in $Q_1$. More generally, in the $i^{th}$ iteration we pick an unmarked node $u$ that is closest to the root $r$, breaking ties arbitrarily, and add any path $\QQ_i$ in \TT from $u$ to any leaf node in the subtree $\TT_u$ rooted at $u$, and {\em mark} all the nodes in $\QQ_i = (\XX_i, \EE_i)$. Since $u$ is unmarked, we claim that every node in $\TT_u$ is unmarked. Indeed, otherwise suppose there is a node $w$ in $\TT_u$ which is already marked. Then there exists two paths from $r$ to $w$ one including $u$ and another avoiding $u$ since $u$ is currently unmarked and $w$ is already marked. This contradicts the fact that \TT is a tree. Hence every node in $\TT_u$ is unmarked. We continue until all the leaf nodes are marked and return \QQ. Since \TT has \el leaves and in every iteration at least one leaf node is marked (since every $\QQ_i$ contains a leaf node), the algorithm runs for at most \el iterations. Notice that, since the algorithm always picks a path consisting of unmarked vertices only, the set of paths in \QQ are pairwise disjoint. We claim that \QQ forms a partition of \TT. Indeed, otherwise there must be a node $x$ in \TT that remains unmarked at the end. From the claim above, we have all the nodes in the subtree $\TT_x$ rooted at $x$ unmarked which includes at least one leaf node of \TT. This contradicts the fact that the algorithm terminates when all the leaf nodes are marked.
\end{proof}

Using \Cref{lem:path_decom}, and the fact that any path can account for at most two leaves, we have that the path cover number of a tree is the same as the number of leaves up to a factor of two.

\begin{lemma}\label{lem:path_leaf}
 Suppose the path cover number of a tree \TT with \el leaves is $k$. Then we have $\nfrac{\el}{2}\le k\le \el$.
\end{lemma}

\begin{proof}
 The inequality $\nfrac{\el}{2}\le k$ follows from the fact that any path in \TT can involve at most two leaves in \TT and there exists $k$ paths covering all the leaf nodes. The inequality $k\le \el$ follows from the fact from \Cref{lem:path_decom} that any tree \TT with \el leaves can be partitioned into \el paths.
\end{proof}

\subsection{Algorithmic Results for \PE}

We now present our main algorithmic results. We begin with generalizing the result of \Cref{thm:con} to any single peaked profiles on trees whose path cover number is at most $k$. The idea is to partition the tree into $k$ disjoint paths, use the algorithm from \Cref{thm:con} on each paths to obtain an order of the candidates on each path of the partition, and finally merge these suborders intelligently. We now formalize this idea as follows.

\begin{theorem}\label{thm:pcover_elicit_ub}
 There is a \PE algorithm with query complexity $\BigO(mn\log k)$ for profiles that are single peaked on trees with path cover number at most $k$.
\end{theorem}

\begin{proof}
 Since the path cover number is at most $k$, we can partition the tree $\TT = (\CC, \EE)$ into $t$ disjoint paths $\PP_i = (\CC_i, \EE_i), i\in[t]$, where $t$ is at most $k$. We now show that we can elicit any preference $\succ$ which is single peaked on the tree \TT by making $\BigO(m\log t)$ queries which in turn proves the statement. We first find the preference ordering restricted to $\CC_i$ using \Cref{thm:con} by making $\BigO(|\CC_i|)$ queries for every $i\in[t]$. This step needs $\sum_{i\in[t]} \BigO(|\CC_i|) = \BigO(m)$ queries since $\CC_i, i\in[t]$ forms a {\em partition} of \CC. We next merge the $t$ orders $\succ(\CC_i), i\in[t],$ to obtain the complete preference $\succ$ by using a standard divide and conquer approach for $t$-way merging as follows which makes $O(m \log t)$ queries \cite{hopcroft1983data}. 
 
 Initially we have $t$ orders to merge. We arbitrarily pair the $t$ orders into $\lceil \nfrac{t}{2}\rceil$ pairs with at most one of them being singleton (when $t$ is odd). By renaming, suppose the pairings are as follows: $(\succ(\CC_{2i-1}), \succ(\CC_{2i})), i\in[\lfloor\nfrac{t}{2}\rfloor]$. Let us define $\CC_i^\pr = \CC_{2i-1} \cup \CC_{2i}$ for every $i\in[\lfloor\nfrac{t}{2}\rfloor]$ and $\CC_{\lceil\nfrac{t}{2}\rceil}^\pr = \CC_t$ if $t$ is an odd integer. We merge $\succ(\CC_{2i-1})$ and $\succ(\CC_{2i})$ to get $\succ(\CC_i^\pr)$ for every $i\in[\lfloor\nfrac{t}{2}\rfloor]$ using \Cref{obs:merge_single_peak}. The number queries the algorithm makes in this iteration is $\sum_{i\in[\lfloor\nfrac{t}{2}\rfloor]} \BigO(|\CC_{2i-1}| + |\CC_{2i}|) = \BigO(m)$ since $(\CC_i)_{i\in[t]}$ forms a partition of \CC. At the end of the first iteration, we have $\lceil\nfrac{t}{2}\rceil$ orders $\succ(\CC_i^\pr), i\in[\lceil\nfrac{t}{2}\rceil]$ to merge to get $\succ$. The algorithm repeats the step above $\BigO(\log t)$ times to obtain $\succ$ and the query complexity of each iteration is $\BigO(m)$. Thus the query complexity of the algorithms is $\BigO(m + m\log t) = \BigO(m\log k)$.
\end{proof}

Using the fact from \Cref{lem:path_leaf} that a tree with $\el$ leaves can be partitioned into at most $\el$ paths, we can use the algorithm from~\Cref{thm:pcover_elicit_ub} to obtain the following bound on query complexity for preference elicitation in terms of leaves.
 
\begin{corollary}\label{cor:leaves_elicit_ub}
 There is a \PE algorithm with query complexity $\BigO(mn\log \el)$ for profiles that are single peaked on trees with at most $\el$ leaves.
\end{corollary}

Finally, if we are given a subset of candidates whose removal makes the single peaked tree a path, then we have an elicitation algorithm that makes $\BigO(mn + nd\log d)$ queries. As before, we first determine the ordering among the candidates on the path (after removing those $d$ candidates) with $\BigO(m-d)$ queries. We then determine the ordering among the rest in $\BigO(d\log d)$ queries using~\Cref{obs:trivial} and merge using~\Cref{obs:merge_single_peak} these two orderings. This leads us to the following result. 

\begin{theorem}\label{thm:pdist_elicit_ub}
 There is a \PE algorithm with query complexity $\BigO(mn + nd\log d)$ for profiles that are single peaked on trees with distance $d$ from path.
\end{theorem}

\begin{proof}
 Let \XX be the smallest set of nodes of a tree $\TT = (\CC, \EE)$ such that $\TT\setminus\XX$, the subgraph of \TT after removal of the nodes in \XX, is a path. We have $|\XX|\le d$. For any preference $\succ$, we make $\BigO(d\log d)$ queries to find $\succ(\XX)$ using \Cref{obs:trivial}, make $\BigO(|\CC\setminus\XX|) = \BigO(m-d)$ queries to find $\succ(\CC\setminus\XX)$ using \Cref{thm:con}, and finally make $\BigO(m)$ queries to find $\succ$ by merging $\succ(\XX)$ and $\succ(\CC\setminus\XX)$ using \Cref{obs:merge_single_peak}. This gives an overall query complexity of $\BigO(mn + nd\log d)$.
\end{proof}

\subsection{Lower Bounds for \PE}

We now turn to query complexity lower bounds for preference elicitation. Our first result is based on a counting argument, showing that the query complexity for preference elicitation in terms of the number of leaves, given by \Cref{cor:leaves_elicit_ub}, is tight up to constant factors. Indeed, let us consider a subdivision of a star \TT with $\el$ leaves and let $t$ denote the distance from the center to the leaves, so that we have a total of $t\el + 1$ candidates. One can show that if the candidates are written out in level order, the candidates that are distance $i$ from the star center can be ordered arbitrarily within this ordering. This tells us that the number of possible preferences $\succ$ that are single peaked on the tree \TT is at least $(\el!)^t$. We obtain the lower bound by using a decision tree argument, wherein we are able to show that it is always possible for an oracle to answer the comparison queries asked by the algorithm in such a way that the total number of possibilities for the preference of the current voter decreases by at most a factor of two. Since the decision tree of the algorithm must entertain at least $(\el!)^t$ leaves to account for all possibilities, we obtain the claimed lower bound.

\begin{theorem}\label{thm:leaves_elicit_lb}
 Let $\TT = (\CC, \EE)$ be a balanced subdivision of a star on $m$ nodes with $\el$ leaves. Then any deterministic \PE algorithm for single peaked profiles on \TT has query complexity $\Omega(mn\log \el)$.
\end{theorem}

\begin{proof}
 Suppose the number of candidates $m$ be $(t\el + 1)$ for some integer $t$. Let $c$ be the center of \TT. We denote the shortest path distance between any two nodes $x, y\in\CC$ in \TT by $d(x, y)$. We consider the partition $(\CC_0, \ldots, \CC_t)$ of the set of candidates \CC where $\CC_i = \{ x\in\CC : d(x, c) = i\}$. We claim that the preference $\succ = \pi_0 \succ \pi_1 \succ \cdots \succ \pi_t$ of the set of candidates \CC is single peaked on the tree \TT where $\pi_i$ is any arbitrary order of the candidates in $\CC_i$ for every $0\le i\le t$. Indeed; consider any path $\QQ = (\XX, \EE^\pr)$ in the tree \TT. Let $y$ be the candidate closest to $c$ among the candidates in \XX; that is $y = \argmin_{x\in\XX} d(x,c)$. Then clearly $\succ(\XX)$ is single peaked with respect to the path \QQ having peak at $y$. We have $|\CC_i| = \el$ for every $i\in[t]$ and thus the number of possible preferences $\succ$ that are single peaked on the tree \TT is at least $(\el!)^t$. 
 
 Let \AA be any \PE algorithm for single peaked profiles on the tree \TT. We now describe our oracle to answer the queries that the algorithm \AA makes. For every voter $v$, the oracle maintains the set $\RR_v$ of possible preferences of the voter $v$ which is consistent with the answers to all the queries that the algorithm \AA have already made for the voter $v$. At the beginning, we have $|\RR_v| \ge ( \el!)^t$ for every voter $v$ as argued above. Whenever the oracle receives a query on $v$ for any two candidates $x$ and $y$, it computes the numbers $n_1$ and $n_2$ of orders in $\RR_v$ which prefers $x$ over $y$ and $y$ over $x$ respectively; the oracle can compute the integers $n_1$ and $n_2$ since the oracle has infinite computational power. The oracle answers that the voter $v$ prefers the candidate $x$ over $y$ if and only if $n_1\ge n_2$ and updates the set $\RR_v$ accordingly. Hence, whenever the oracle is queried for a voter $v$, the size of the set $\RR_v$ decreases by a factor of at most two. On the other hand, we must have, from the correctness of the algorithm, $\RR_v$ to be a singleton set when the algorithm terminates for every voter $v$ -- otherwise there exists a voter $v$ (whose corresponding $\RR_v$ is not singleton) for which there exist two possible preferences which are single peaked on the tree \TT and are consistent with all the answers the oracle has given and thus the algorithm \AA fails to output the preference of the voter $v$ correctly. Hence every voter must be queried at least $\Omega(\log((\el!)^t)) = \Omega(t \el\log \el) = \Omega(m\log \el)$ times.
\end{proof}

Since the path cover number of a subdivided star on $\el$ leaves is at least 
$\el/2$, we also obtain the following.

\begin{corollary}\label{cor:pcover_elicit_lb}
 There exists a tree \TT with path cover number $k$ such that any deterministic \PE algorithm for single peaked profiles on \TT has query complexity $\Omega(mn\log k)$.
\end{corollary}

Mimicking the level order argument above on a generic tree with $\el$ leaves, and using the connection between path cover and leaves, we obtain lower bounds that are functions of $(n,\el)$ and $(n,k)$, as given below. This will be useful for our subsequent results.

\begin{theorem}\label{thm:pcover_elicit_lb_any}
 Let $\TT = (\CC, \EE)$ be any arbitrary tree with $\el$ leaves and path cover number $k$. Then any deterministic \PE algorithm for single peaked profiles on \TT has query complexity $\Omega(n\el\log \el)$ and $\Omega(nk\log k)$.
\end{theorem}

\begin{proof}
 Let \XX be the set of leaves in \TT. We choose any arbitrary nonleaf node $r$ as the root of \TT. We denote the shortest path distance between two candidates $x, y\in\CC$ in the tree \TT by $d(x,y)$. Let $t$ be the maximum distance of a node from $r$ in \TT; that is $t = \max_{y\in\CC\setminus\XX} d(r, x)$. We partition the candidates in $\CC\setminus\XX$ as $(\CC_0, \CC_1, \ldots, \CC_t)$ where $\CC_i = \{y\in\CC\setminus\XX : d(r, y) = i\}$ for $0\le i\le t$. We claim that the preference $\succ = \pi_0 \succ \pi_1 \succ \cdots \succ \pi_t \succ \pi$ of the set of candidates \CC is single peaked on the tree \TT where $\pi_i$ is any arbitrary order of the candidates in $\CC_i$ for every $0\le i\le t$ and $\pi$ is an arbitrary order of the candidates in $\CC\setminus\XX$. Indeed, otherwise consider any path $\QQ = (\YY, \EE^\pr)$ in the tree \TT. Let $y$ be the candidate closest to $r$ among the candidates in \YY; that is $y = \argmin_{x\in\YY} d(x,r)$. Then clearly $\succ(\YY)$ is single peaked with respect to the path \QQ having peak at $y$. We have the number of possible preferences $\succ$ that are single peaked on the tree \TT is at least $|\XX|! = \el!$. Again using the oracle same as in the proof of \Cref{thm:leaves_elicit_lb}, we deduce that any \PE algorithm \AA for profiles that are single peaked on the tree \TT needs to make $\Omega(n\el\log \el)$ queries. The bound with respect to the path cover number now follows from \Cref{lem:path_leaf}.
\end{proof}

The following results can be obtained simply by applying~\Cref{thm:pcover_elicit_lb_any} on particular graphs. For instance, we use the fact that stars have $(m-1)$ leaves and have pathwidth one to obtain the first part of \Cref{cor:pwidth_elicit_lb}, while appealing to complete binary trees that have $O(m)$ leaves and pathwidth $O(\log m)$ for the second part. These examples also work in the context of maximum degree, while for diameter we use stars and caterpillars with a central path of length $m/2$ in \Cref{cor:dia_elicit_lb}.

\begin{corollary}\label{cor:maxdeg_elicit_lb}
 There exist two trees \TT and $\TT^\pr$ with maximum degree $\Delta=3$ and $m-1$ respectively such that any deterministic \PE algorithm for single peaked profiles on \TT and $\TT^\pr$ respectively has query complexity $\Omega(mn\log m)$.
\end{corollary}

\begin{proof}
 Using \Cref{thm:pcover_elicit_lb_any}, we know that any \PE algorithm for profiles which are single peaked on a complete binary tree has query complexity $\Omega(mn\log m)$ since a complete binary tree has $\Omega(m)$ leaves. The result now follows from the fact that the maximum degree $\Delta$ of a node is three for any binary tree. The case of $\Delta = m-1$ follows immediately from \Cref{thm:pcover_elicit_lb_any} applied on stars.
\end{proof}

\begin{corollary}\label{cor:dia_elicit_lb}
 There exists two trees \TT and $\TT^\pr$ with diameters $\omega=2$ and $\omega=\nfrac{m}{2}$ respectively such that any deterministic \PE algorithm for profiles which are single peaked on \TT and $\TT^\pr$ respectively has query complexity $\Omega(mn\log m)$.
\end{corollary}

\begin{proof}
 The $\omega=2$ and $\omega=\nfrac{m}{2}$ cases follow from \Cref{thm:pcover_elicit_lb_any} applied on star and caterpillar graphs with a central path of length $\nfrac{m}{2}$ respectively.
\end{proof}

We next consider the parameter pathwidth of the underlying single peaked tree. We immediately get the following result for \PE on trees with pathwidths one or $\log m$ from \Cref{thm:pcover_elicit_lb_any} and the fact that the pathwidths of a star and a complete binary tree are one and $\log m$ respectively.

\begin{corollary}\label{cor:pwidth_elicit_lb}
 There exist two trees \TT and $\TT^\pr$ with pathwidths one and $\log m$ respectively such that any deterministic \PE algorithm for single peaked profiles on \TT and $\TT^\pr$ respectively has query complexity $\Omega(mn\log m)$.
\end{corollary}


Our final result, which again follows from \Cref{thm:pcover_elicit_lb_any} applied of caterpillar graphs with a central path of length $m-d$, shows that the bound in~\Cref{thm:pdist_elicit_ub} is tight. 

\begin{theorem}\label{thm:pdist_elicit_lb}
 For any integers $m$ and $d$ with $1\le d\le\nfrac{m}{4}$, there exists a tree \TT with distance $d$ from path such that any deterministic \PE algorithm for profiles which are single peaked on \TT has query complexity $\Omega(mn + nd\log d)$.
\end{theorem}

\begin{proof}
 Consider the caterpillar graph where the length of the central path \QQ is $m-d$; there exists such a caterpillar graph since $d\le\nfrac{m}{4}$. Consider the order $\succ = \pi \succ \sigma$ of the set of candidates \CC where $\pi$ is an order of the candidates in \QQ which is single peaked on \QQ and $\sigma$ is any order of the candidates in $\CC\setminus\QQ$. Clearly, $\succ$ is single peaked on the tree \TT. Any elicitation algorithm \AA needs to make $\Omega(m-d)$ queries involving only the candidates in \QQ to elicit $\pi$ due to \cite{Conitzer09} and $\Omega(d\log d)$ queries to elicit $\sigma$ due to sorting lower bound for every voter. This proves the statement.
\end{proof}

\section{Results for Weak Condorcet Winner}

In this section, we present our results for the query complexity for finding a weak Condorcet winner in profiles that are single peaked on trees. 

\subsection{Algorithmic Results for Weak Condorcet Winner}

We now show that we can find a weak Condorcet winner of profiles that are single peaked on trees using fewer queries than the number of queries needed to find the profile itself. We note that if a Condorcet winner is guaranteed to exist for a profile, then it can be found using $\BigO(mn)$ queries --- we pit an arbitrary pair of candidates $x,y$ and use $\BigO(n)$ queries to determine if $x$ defeats $y$. We push the winning candidate forward and repeat the procedure, clearly requiring at most $m$ rounds. Now, if a profile is single peaked with respect to a tree, and there are an odd number of voters, then we have a Condorcet winner and the procedure that we just described would work. Otherwise, we simply find a Condorcet winner among the first $(n-1)$ voters. It can be easily shown that such a winner is one of the weak Condorcet winners for the overall profile, and we therefore have the following results. We begin with the following general observation.

\begin{observation}\label{obs:cond_trivial_ub}
 Let \PP be a profile where a Condorcet winner is guaranteed to exist. Then we can find the Condorcet winner of \PP by making $\BigO(mn)$ queries.
\end{observation}

\begin{proof}
 For any two candidates $x, y\in\CC$ we find whether $x$ defeats $y$ or not by simply asking all the voters to compare $x$ and $y$; this takes $\BigO(n)$ queries. The algorithms maintains a set \SS of candidates which are {\em potential} Condorcet winners. We initialize \SS to \CC. In each iteration we pick any two candidates $x, y\in\SS$ from \SS, remove $x$ from \SS if $x$ does not defeat $y$ and vice versa using $\BigO(n)$ query complexity until \SS is singleton. After at most $m-1$ iterations, the set \SS will be singleton and contain only Condorcet winner since we find a candidate which is not a Condorcet winner in every iteration and thus the size of the set \SS decreases by at least one in every iteration. This gives a query complexity bound of $\BigO(mn)$.
\end{proof}

Using \Cref{obs:cond_trivial_ub} we now develop a \CW algorithm with query complexity $\BigO(mn)$ for profiles that are single peaked on trees.

\begin{theorem}\label{thm:cw_gen_ub}
 There is a \CW algorithm with query complexity $\BigO(mn)$ for single peaked profiles on trees.
\end{theorem}

\begin{proof}
 Let $\PP = (\succ_i)_{i\in[n]}$ be a profile which is single peaked on a tree \TT. If $n$ is an odd integer, then we know that there exists a Condorcet winner in \PP since no two candidates can tie and there always exists at least one weak Condorcet winner in every single peaked profile on trees. Hence, if $n$ is an odd integer, then we use \Cref{obs:cond_trivial_ub} to find a weak Condorcet winner which is the Condorcet winner too. Hence let us now assume that $n$ is an even integer. Notice that $\PP_{-1} = (\succ_2, \ldots, \succ_n)$ is also single peaked on \TT and has an odd number of voters and thus has a Condorcet winner. We use \Cref{obs:cond_trivial_ub} to find the Condorcet winner $c$ of $\PP_{-1}$ and output $c$ as a weak Condorcet winner of \PP. We claim that $c$ is a weak Condorcet winner of \PP. Indeed otherwise there exists a candidate $x$ other than $c$ who defeats $c$ in \PP. Since $n$ is an even integer, $x$ must defeat $c$ by a margin of at least two (since all pairwise margins are even integers) in \PP. But then $x$ also defeats $c$ by a margin of at least one in $\PP_{-1}$. This contradicts the fact that $c$ is the Condorcet winner of $\PP_{-1}$.
\end{proof}

For the special case of single peaked profiles, we can do even better. Here we take advantage of the fact that a ``median candidate''~\cite{mas1995microeconomic}  is guaranteed to be a weak Condorcet winner. We make $\OO(\log m)$ queries per vote to find the candidates placed at the first position of all the votes using the algorithm in \cite{Conitzer09} and find a median candidate to have an algorithm for finding a weak Condorcet winner. If a profile is single peaked (on a path), then there is a \CW algorithm with query complexity $\BigO(n\log m)$ as shown below. Let us define the frequency $f(x)$ of a candidate $x\in\CC$ to be the number of votes where $x$ is placed at the first position. Then we know that a median candidate according to the single peaked ordering of the candidates along with their frequencies as defined above is a weak Condorcet winner for single peaked profiles \cite{mas1995microeconomic}.

\begin{theorem}\label{thm:con_sp_ub}
 There is a \CW algorithm with query complexity $\BigO(n\log m)$ for single peaked profiles (on a path).
\end{theorem}

\begin{proof}
 Let \PP be a profile that is single peaked with respect to an ordering $\succ\in\LL(\CC)$ of candidates. Then we find, for every voter $v$, the candidate the voter $v$ places at the first position using $\BigO(\log m)$ queries using the algorithm in \cite{Conitzer09} and return a median candidate.
\end{proof}

The next result uses~\Cref{thm:con_sp_ub} on paths in a path cover of the single peaked tree eliminating the case of even number of voters by the idea of setting aside one voter that was used in~\Cref{thm:cw_gen_ub}.

\begin{theorem}\label{thm:cw_pc_ub}
 Let \TT be a tree with path cover number at most $k$. Then there is an algorithm for \CW for profiles which are single peaked on \TT with query complexity $\BigO(nk\log m)$.
\end{theorem}

Recalling that the number of leaves bounds the path cover number, we have the following consequence. 

\begin{corollary}\label{cor:cw_leaf_ub}
 Let \TT be a tree with \el leaves. Then there is an algorithm for \CW for profiles which are single peaked on \TT with query complexity $\BigO(n\el\log m)$.
\end{corollary}

From \Cref{thm:cw_gen_ub,thm:con_sp_ub} we have the following result for any arbitrary tree.

\begin{theorem}\label{thm:cw_pc_ub}
 Let \TT be a tree with path cover number at most $k$. Then there is an algorithm for \CW for profiles which are single peaked on \TT with query complexity $\BigO(nk\log m)$.
\end{theorem}

\begin{proof}
 Let \PP be the input profile and $\QQ_i = (\XX_i, \EE_i) ~i\in[t]$ be $t(\le k)$ disjoint paths that cover the tree \TT. Here again, if the number of voters is even, then we remove any arbitrary voter and the algorithm outputs the Condorcet winner of the rest of the votes. The correctness of this step follows from the proof of \Cref{thm:cw_gen_ub}. Hence we assume, without loss of generality, that we have an odd number of voters. The algorithm proceeds in two stages. In the first stage, we find the Condorcet winner $w_i$ of the profile $\PP(\XX_i)$ for every $i\in[t]$ using \Cref{thm:con_sp_ub}. The query complexity of this stage is $\BigO(n\sum_{i\in[t]} \log|\XX_i|) = \BigO(nt\log (\nfrac{m}{t}))$. In the second stage, we find the Condorcet winner $w$ of the profile $\PP(\{w_i : i\in[t]\})$ using \Cref{thm:cw_gen_ub} and output $w$. The query complexity of the second stage is $\BigO(nt\log t)$. Hence the overall query complexity of the algorithm is $\BigO(nt\log(\nfrac{m}{t})) + \BigO(nt\log t) = \BigO(nk\log m)$.
\end{proof}

\subsection{Lower Bounds for Weak Condorcet Winner}

We now state the lower bounds pertaining to \CW. First, we show that any algorithm for \CW for single peaked profiles on stars has query complexity $\Omega(mn)$, showing that the bound of \Cref{thm:cw_gen_ub} is tight. 

\begin{theorem}\label{thm:cw_gen_lb}
 Any deterministic \CW algorithm for single peaked profiles on stars has query complexity $\Omega(mn)$.
\end{theorem}

\begin{proof}
 Let \TT be a star with center vertex $c$. We now design an oracle that will ``force'' any \CW algorithm \AA for single peaked profiles on \TT to make $\Omega(mn)$ queries. For every voter $v$, the oracle maintains a set of {\em ``marked''} candidates which can not be placed at the first position of the preference of $v$. Suppose the oracle receives a query to compare two candidates $x$ and $y$ for a voter \el. If the order between $x$ and $y$ for the voter \el follows from the answers the oracle has already provided to all the queries for the voter \el, then the oracle answers accordingly. Otherwise it answers $x\succ_\el y$ if $y$ is unmarked and marks $y$; otherwise the oracle answers $y\suc_\el x$ and marks $x$. Notice that the oracle marks at most one unmarked candidate every time it is queried. We now claim that there must be at least $\nfrac{n}{10}$ votes which have been queried at least $\nfrac{m}{4}$ times. If not, then there exists $n-\nfrac{n}{10} = \nfrac{9n}{10}$ votes each of which has at least $m-\nfrac{m}{4} = \nfrac{3m}{4}$ candidates unmarked. In such a scenario, there exists a constant $N_0$ such that for every $m, n>N_0$, we have at least two candidates $x$ and $y$ who are unmarked in at least $(\lfloor\nfrac{n}{2}\rfloor + 1)$ votes each. Now if the algorithm outputs $x$, then we put $y$ at the first position in at least $(\lfloor\nfrac{n}{2}\rfloor + 1)$ votes and at the second position in the rest of the votes and this makes $y$ the (unique) Condorcet winner. If the algorithm does not output $x$, then we put $x$ at the first position in at least $(\lfloor\nfrac{n}{2}\rfloor + 1)$ votes and at the second position in the rest of the votes and this makes $x$ the (unique) Condorcet winner. Hence the algorithm fails to output correctly in both the cases contradicting the correctness of the algorithm. Also the resulting profile is single peaked on \TT with center at $y$ in the first case and at $x$ in the second case. Therefore the algorithm \AA must have query complexity $\Omega(mn)$.
\end{proof}

Our next result uses an adversary argument, and shows that the query complexity for \CW for single peaked profiles in \Cref{thm:con_sp_ub} is essentially optimal, provided that the queries to different voters are not interleaved, as is the case with our algorithm.

\begin{theorem}\label{thm:con_sp_lb}
 Any deterministic \CW algorithm for single peaked profiles which does not interleave the queries to different voters has query complexity $\Omega(n\log m)$.
\end{theorem}

\begin{proof}
 Let a profile \PP be single peaked with respect to the ordering of the candidates $\suc = c_1\succ c_2\succ \cdots \succ c_m$. The oracle maintains two indices $\el$ and $r$ for every voter such that any candidate from $\{c_\el, c_{\el+1}, \ldots, c_r\}$ can be placed at the first position of the preference of the voter $v$ and still be consistent with all the answers provided by the oracle for $v$ till now and single peaked with respect to \suc. The algorithm initializes $\el$ to one and $r$ to $m$ for every voter. The oracle answers any query in such a way that maximizes the new value of $r-\el$. More specifically, suppose the oracle receives a query to compare candidates $c_i$ and $c_j$ with $i< j$ for a voter $v$. If the ordering between $c_i$ and $c_j$ follows, by applying transitivity, from the answers to the queries that have already been made so far for this voter, then the oracle answers accordingly. Otherwise the oracle answers as follows. If $i < \el$, then the oracle answers that $c_j$ is preferred over $c_i$; else if $j > r$, then the oracle answers that $c_i$ is preferred over $c_j$. Otherwise (that is when $\el \le i < j\le r$), if $j-\el > r-i$, then the oracle answers that $c_i$ is preferred over $c_j$ and changes $r$ to $j$; if $j-\el \le r-i$, then the oracle answers that $c_j$ is preferred over $c_i$ and changes $\el$ to $i$. Hence whenever the oracle answers a query for a voter $v$, the value of $r-\el$ for that voter $v$ decreases by a factor of at most two. Suppose the election instance has an odd number of voters. Let \VV be the set of voters. Now we claim that the first $\lfloor\nfrac{n}{5}\rfloor$ voters must be queried $(\log m - 1)$ times each. Suppose not, then consider the first voter $v^\pr$ that is queried less than $(\log m - 1)$ times. Then there exist at least two candidates $c_t$ and $c_{t+1}$ each of which can be placed at the first position of the vote $v^\pr$. The oracle fixes the candidates at the first positions of all the votes that have not been queried till $v^\pr$ is queried (and there are at least $\lceil\nfrac{4n}{5}\rceil$ such votes) in such a way that $\lfloor\nfrac{n}{2}\rfloor$ voters in $\VV\setminus\{v^\pr\}$ places some candidate in the left of $c_t$ at the first positions and $\lfloor\nfrac{n}{2}\rfloor$ voters in $\VV\setminus\{v^\pr\}$ places some candidate in the right of $c_{t+1}$. If the algorithm outputs $c_t$ as the Condorcet winner, then the oracle makes $c_{t+1}$ the (unique) Condorcet winner by placing $c_{t+1}$ at the top position of $v^\pr$, because $c_{t+1}$ is the unique median in this case. If the algorithm does not output $c_t$ as the Condorcet winner, then the oracle makes $c_t$ the (unique) Condorcet winner by placing $c_t$ at the top position of $v^\pr$, because $c_t$ is the unique median in this case. Hence the algorithm fails to output correctly in both the cases thereby contradicting the correctness of the algorithm.
\end{proof}

\section{Conclusion}

In this work, we present algorithms for eliciting preferences of a set of voters when the preferences are single peaked on a tree thereby significantly extending the work of \cite{Conitzer09}. Moreover, we show non-trivial lower bounds on the number of comparison queries any preference elicitation algorithm would need to ask in a domain of single peaked profiles on a tree. From this, we conclude that our algorithms asks optimal number of comparison queries up to constant factors. Our main finding in this work is the interesting dependencies between the number of comparison queries any preference elicitation algorithm would ask and various parameters of the single peaked tree. For example, our results show that the query complexity for preference elicitation is a monotonically increasing function on the number of leaf nodes in the single peaked tree. On the other hand, the query complexity for preference elicitation does not directly depend on other tree parameters like, maximum degree, minimum degree, path width, diameter etc.

We then move on to study query complexity for finding a weak Condorcet winner of a set of votes which is single peaked on a tree. Here, our results show that a weak Condorcet winner can be found with much less number of queries compared to preference elicitation.

In the next chapter, we explore the preference elicitation problem again, but for another well-studied domain known as single crossing profiles.
\chapter{Preference Elicitation for Single Crossing Profiles}
\label{chap:pref_elicit_cross}

\blfootnote{A preliminary version of the work in this chapter was published as \cite{deycross}: Palash Dey and Neeldhara Misra. Preference elicitation for single crossing domain. In
Proc. Twenty-Fifth International Joint Conference on Artificial Intelligence, IJCAI 2016,
New York, NY, USA, 9-15 July 2016, pages 222-228, 2016.}

\begin{quotation}
{\small In this chapter, we consider the domain of single crossing preference profiles and study the query complexity of preference elicitation under various situations. We consider two distinct scenarios: when an ordering of the voters with respect to which the profile is single crossing is {\em known} a priori versus when it is {\em unknown}. We also consider two different access models: when the votes can be accessed at random, as opposed to when they are coming in a predefined sequence. In the sequential access model, we distinguish two cases when the ordering is known: the first is that the sequence in which the votes appear is also a single-crossing order, versus when it is not. The main contribution of our work is to provide polynomial time algorithms with low query complexity for preference elicitation in all the above six cases. Further, we show that the query complexities of our algorithms are optimal up to constant factors for all but one of the above six cases.}
\end{quotation}

\section{Introduction}

Eliciting the preferences of a set of agents is a nontrivial task since we often have a large number of candidates (ranking restaurants for example) and it will be infeasible for the agents to rank all of them. Hence it becomes important to elicit the preferences of the agents by asking them (hopefully a small number of) comparison queries only - ask an agent $i$ to compare two candidates $x$ and $y$. 

Unfortunately, it turns out that one would need to ask every voter $\Omega(m\log m)$ queries to know her preference (due to sorting lower bound). However, if the preferences are not completely arbitrary, and admit additional structure, then possibly we can do better. Indeed, an affirmation of this thought comes from the work of \cite{Conitzer09,deypeak}, who showed that we can elicit preferences using only $O(m)$ queries per voter if the preferences are single peaked. The domain of single peaked preferences has been found to be very useful in modeling preferences in political elections. We, in this work, study the problem of eliciting preferences of a set of agents for another popular domain namely the domain of single crossing profiles. A profile is called {\em single crossing} if the voters can be arranged in a complete order $\suc$ such that for every two candidates $x$ and $y$, all the voters who prefer $x$ over $y$ appear consecutively in $\suc$ \cite{mirrlees1971exploration,RePEc:eee:pubeco:v:8:y:1977:i:3:p:329-340}. Single crossing profiles have been extensively used to model income of a set of agents \cite{RePEc:eee:pubeco:v:8:y:1977:i:3:p:329-340,RePEc:ucp:jpolec:v:89:y:1981:i:5:p:914-27}. The domain is single crossing profiles are also popular among computational social scientists since many computationally hard voting rules become tractable if the underlying preference profile is single crossing \cite{moulin2016handbook}.

\subsection{Related Work}

Conitzer and Sandholm show that determining whether we have enough information at any point of the elicitation process for finding a winner under some common voting rules is computationally intractable \cite{conitzer2002vote}. They also prove in their classic paper \cite{conitzer2005communication} that one would need to make $\Omega(mn\log m)$ queries even to decide the winner for many commonly used voting rules which matches with the trivial $\BigO(mn\log m)$ upper bound (based on sorting) for preference elicitation in unrestricted domain.

A natural question at this point is if these restricted domains allow for better elicitation algorithms as well. The answer to this is in the affirmative, and one can indeed elicit the preferences of the voters using only $\BigO(mn)$ many queries for the domain of single peaked preference profiles \cite{Conitzer09}. Our work belongs to this kind of research-- we investigate the number of queries one has to ask for preference elicitation in single crossing domains. When some partial information is available about the preferences, Ding and Lin prove interesting properties of what they call a deciding set of queries \cite{ding2013voting}. Lu and Boutilier empirically show that several heuristics often work well \cite{lu2011vote,lu2011robust}.

\subsection{Single Crossing Domain}

A preference profile $\PP = (\succ_1, \ldots, \succ_n)$ of $n$ agents or voters over a set \CC of candidates is called a single crossing profile if there exists a permutation $\sigma\in\SB_n$ of $[n]$ such that, for every two distinct candidates $x, y\in\CC$, whenever we have $x\succ_{\sigma(i)} y$ and $x\succ_{\sigma(j)} y$ for two integers $i$ and $j$ with $1\le i< j\le n$, we have $x\succ_{\sigma(k)} y$ for every $i\le k\le j$. \Cref{ex:single_cross} exhibits an example of a single crossing preference profile.

\begin{example}(Example of single crossing preference profile)\label{ex:single_cross}
 Consider a set \CC of $m$ candidates, corresponding $m$ distinct points on the Real line, a set \VV of $n$ voters, also corresponding $n$ points on the Real line, and the preference of every voter are based on their distance to the candidates -- given any two candidates, every voter prefers the candidate nearer to her (breaking the tie arbitrarily). Then the set of $n$ preferences is single crossing with respect to the ordering of the voters according to the ascending ordering of their positions on the Real line.
\end{example}

The following observation is immediate from the definition of single crossing profiles.

\begin{observation}
 Suppose a profile $\PP$ is single crossing with respect to an ordering $\sigma\in\SB_n$ of votes. Then \PP is single crossing with respect to the ordering $\overleftarrow{\sigma}$ too.
\end{observation}

\subsection{Single Crossing Width}

A preference $\PP = (\succ_1, \ldots, \succ_n)$ of $n$ agents or voters over a set \CC of candidates is called a single crossing profile with width $w$ if the set of candidates \CC can be partitioned into $(\CC_i)_{i\in[k]}$ such that $|\CC_i|\le w$ for every $i$ and for every two candidates $x\in\CC_\el$ and $y\in\CC_t$ with $\el\ne t$ from two different subsets of the partition, whenever we have $x\succ_{\sigma(i)} y$ and $x\succ_{\sigma(j)} y$ for two integers $i$ and $j$ with $1\le i< j\le n$, we have $x\succ_{\sigma(k)} y$ for every $i\le k\le j$.

\subsection{Our Contribution}

In this work we present novel algorithms for preference elicitation for the domain of single crossing profiles in various settings.  We consider two distinct situations: when an ordering of the voters with respect to which the profile is single crossing is {\em known} versus when it is {\em unknown}. We also consider different access models: when the votes can be accessed at random, as opposed to when they are coming in a pre-defined sequence. In the sequential access model, we distinguish two cases when the ordering is known: the first is that sequence in which the votes appear is also a single-crossing order, versus when it is not.  We also prove lower bounds on the query complexity of preference elicitation for the domain of single crossing profiles; these bounds match the upper bounds up to constant factors (for a large number of voters) for all the six scenarios above except the case when we know a single crossing ordering of the voters and we have a random access to the voters; in this case, the upper and lower bounds match up to a factor of $\OO(m)$. We summarize our results in \Cref{tbl:summary}.

\begin{table*}[!htbp]
 \begin{center}
  \begin{tabular}{|c|c|c|c|}\hline
   \multirow{2}{*}{Ordering} & \multirow{2}{*}{Access model} & \multicolumn{2}{c|}{Query Complexity}\\\cline{3-4}
    & & Upper Bound & Lower Bound \\\hline\hline
   
   \multirow{3}{*}{Known} & Random & \makecell{$\BigO(m^2 \log n)$\\\relax[\Cref{lem:sc_random_ub}]} & \makecell{$\Omega(m\log m + m\log n)$\\\relax[\Cref{thm:sc_random_lb}]} \\\cline{2-4}
   
   & Sequential single crossing order & \makecell{$\BigO(mn + m^2)$\\\relax[\Cref{thm:sc_seq_known_ub}]} & \multirow{3}{*}{\makecell{$\Omega(m\log m + mn)$\\\relax[\Cref{thm:sc_seq_known_lb}]}} \\\cline{2-3}
   
   & Sequential any order & \makecell{$\BigO(mn + m^2\log n)$\\\relax[\Cref{thm:sc_seq_any_ub}]} & \\\cline{1-3}
   
   \multirow{2}{*}{Unknown} &  Sequential any order & \makecell{$\BigO(mn + m^3\log m)$\\\relax[\Cref{thm:sc_seq_unknown_ub}]} & \\\cline{2-4}
   
   & Random & \makecell{$\BigO(mn + m^3\log m)$\\\relax[\Cref{cor:sc_random_unknown_ub}]} & \makecell{$\Omega(m\log m + mn)$\\\relax[\Cref{thm:sc_random_unknown_lb}]} \\\hline
   
  \end{tabular}
 \end{center}
\caption{Summary of Results for preference elicitation for single crossing profiles.}\label{tbl:summary}
\end{table*}

We then extend our results to domains which are single crossing of width $w$ in \Cref{subsec:sc_width}. We also prove that a weak Condorcet winner and the Condorcet winner (if it exists) of a single crossing preference profile can be found out with much less number of queries in \Cref{sec:sc_cw}.

\section{Problem Formulation}

The query and cost model is the same as \Cref{def:query_model} in \Cref{chap:pref_elicit_peak}. We recall it below for ease of access.

Suppose we have a profile \PP with $n$ voters and $m$ candidates. Let us define a function $\text{\Query}(x \succ_\el y)$ for a voter \el and two different candidates $x$ and $y$ to be \true if the voter \el prefers the candidate $x$ over the candidate $y$ and \false otherwise. We now formally define the problem.

\begin{definition}\PE\\
 Given an oracle access to \Query($\cdot$) for a single crossing profile \PP, find \PP.
\end{definition}

For two distinct candidates $x, y\in \CC$ and a voter \el, we say a \PE algorithm \AA {\em compares} candidates $x$ and $y$ for voter \el, if \AA makes a call to either $\text{\Query}(x \succ_\el y)$ or $\text{\Query}(y \succ_\el x)$. We define the number of queries made by the algorithm \AA, called the {\em query complexity} of \AA, to be the number of distinct tuples $(\el, x, y)\in \VV\times\CC\times\CC$ with $x\ne y$ such that the algorithm \AA compares the candidates $x$ and $y$ for voter \el. Notice that, even if the algorithm \AA makes multiple calls to \Query($\cdot$) with same tuple $(\el, x, y)$, we count it only once in the query complexity of \AA. This is without loss of generality since we can always implement a wrapper around the oracle which memorizes all the calls made to the oracle so far and whenever it receives a duplicate call, it replies from its memory without ``actually'' making a call to the oracle. We say two query complexities $\qqq(m,n)$ and $\qqq^\pr(m,n)$ are tight up to a factor of \el ~{\em for a large number of voters} if $\nfrac{1}{\el} \le \lim_{n\to\infty} \nfrac{\qqq(m,n)}{\qqq^\pr(m,n)} \le \el$.

Note that by using a standard sorting routine like merge sort, we can fully elicit an unknown preference using $\BigO(m \log m)$ queries. We state this explicitly below, as it will be useful in our subsequent discussions. 

\begin{observation}\label{obs:naive_single_crossing}
 There is a \PE algorithm with query complexity $\BigO(m\log m)$ for eliciting one vote from a single crossing preference profile.
\end{observation}

\subsection{Model of Input}

We study two models of input for \PE for single crossing profiles.

\begin{itemize}
 \item {\bf Random access to voters:} In this model, we have a set of voters and we are allowed to ask any voter to compare any two candidates at any point of time. Moreover, we are also allowed to interleave the queries to different voters. Random access to voters is the model of input for elections within an organization where every voter belongs to the organization and can be queried any time.
 
 \item {\bf Sequential access to voters:} In this model, voters are arriving in a sequential manner one after another to the system. Once a voter \el arrives, we can query voter \el as many times as we like and then we ``release'' the voter \el from the system to access the next voter in the queue. Once voter \el is released, it can never be queried again. Sequential access to voters is indeed the model of input in many practical elections scenarios such as political elections, restaurant ranking etc.
\end{itemize}

\section{Results for Preference Elicitation}

In this section, we present our technical results. We first consider the (simpler) situation when one single crossing order is known, and then turn to the case when no single crossing order is a priori known. In both cases, we explore all the relevant access models.

\subsection{Known Single Crossing Order}

We begin with a simple \PE algorithm when we are given a random access to the voters and one single crossing ordering is known.

\begin{lemma}\label{lem:sc_random_ub}
 Suppose a profile \PP is single crossing with respect to a known permutation of the voters. Given a random access to voters, there is a \PE algorithm with query complexity $\BigO(m^2 \log n)$.
\end{lemma}

\begin{proof}
 By renaming, we assume that the profile is single crossing with respect to the identity permutation of the votes. Now, for every ${m\choose 2}$ pair of candidates $\{x, y\}\subset\CC$, we perform a binary search over the votes to find the index $i(\{x, y\})$ where the ordering of $x$ and $y$ changes. We now know how any voter $j$ orders any two candidates $x$ and $y$ from $i(\{x, y\})$ and thus we have found $\mathcal{P}$.
\end{proof}

Interestingly, the simple algorithm in \Cref{lem:sc_random_ub} turns out to be optimal up to a multiplicative factor of $\OO(m)$ as we prove next. The idea is to ``pair up'' the candidates and design an oracle which ``hides'' the vote where the ordering of the two candidates in any pair $(x, y)$ changes unless it receives at least $(\log m - 1)$ queries involving only these two candidates $x$ and $y$. We formalize this idea below. Observe that, when a single crossing ordering of the voters in known, we can assume without loss of generality, by renaming the voters, that the preference profile is single crossing with respect to the identity permutation of the voters.

\begin{theorem}\label{thm:sc_random_lb}
 Suppose a profile \PP is single crossing with respect to the identity permutation of votes. Given random access to voters, any deterministic \PE algorithm has query complexity $\Omega(m\log m + m\log n)$.
\end{theorem}

\begin{proof}
 The $\Omega(m\log m)$ bound follows from the query complexity lower bound of sorting and the fact that any profile consisting of only one preference $\succ\in\LL(\CC)$ is single crossing. Let $\CC = \{c_1, \ldots, c_m\}$ be the set of $m$ candidates where $m$ is an even integer. Consider the ordering $Q = c_1 \succ c_2 \succ \cdots \succ c_m \in\LL(\CC)$ and the following pairing of the candidates: $\{c_1, c_2\}, \{c_3, c_4\}, \ldots, \{c_{m-1}, c_m\}$. Our oracle answers \Query($\cdot$) as follows. The oracle fixes the preferences of the voters one and $n$ to be $Q$ and $\overleftarrow{Q}$ respectively. For every odd integer $i\in[m]$, the oracle maintains $\theta_i$ (respectively $\beta_i$) which corresponds to the largest (respectively smallest) index of the voter for whom $(c_i, c_{i+1})$ has already been queried and the oracle answered that the voter prefers $c_i$ over $c_{i+1}$ ($c_{i+1}$ over $c_i$ respectively). The oracle initially sets $\theta_i = 1$ and $\beta_i = n$ for every odd integer $i\in[m]$. Suppose oracle receives a query to compare candidates $c_i$ and $c_j$ for $i,j\in[m]$ with $i<j$ for a voter $\ell$. If $i$ is an even integer or $j-i\ge 2$ (that is, $c_i$ and $c_j$ belong to different pairs), then the oracle answers that the voter \el prefers $c_i$ over $c_j$. Otherwise we have $j=i+1$ and $i$ is an odd integer. The oracle answers the query to be $c_i \succ c_{i+1}$ and updates $\theta_i$ to \el keeping $\beta_i$ fixed if $|\el-\theta_i| \le |\el-\beta_i|$ and otherwise answers $c_{i+1} \succ c_i$ and updates $\beta_i$ to \el keeping $\theta_i$ fixed (that is, the oracle answers according to the vote which is closer to the voter \el between $\theta_i$ and $\beta_i$ and updates $\theta_i$ or $\beta_i$ accordingly). If the pair $(c_i, c_{i+1})$ is queried less than $(\log n - 2)$ times, then we have $\beta_i - \theta_i \ge 2$ at the end of the algorithm since every query for the pair $(c_i, c_{i+1})$ decreases $\beta_i - \theta_i$ by at most a factor of two and we started with $\beta_i - \theta_i = n-1$. Consider a voter $\kappa$ with $\theta_i < \kappa < \beta_i$. If the elicitation algorithm outputs that the voter $\kappa$ prefers $c_i$ over $c_{i+1}$ (respectively $c_{i+1}$ over $c_i$), then the oracle sets all the voters $\kappa^\pr$ with $\theta_i < \kappa^\pr < \beta_i$ to prefer $c_{i+1}$ over $c_i$ (respectively $c_i$ over $c_{i+1}$). Clearly, the algorithm does not elicit the preference of the voter $\kappa$ correctly. Also, the profile is single crossing with respect to the identity permutation of the voters and consistent with the answers of all the queries made by the algorithm. Hence, for every odd integer $i\in[m]$, the algorithm must make at least $(\log n - 1)$ queries for the pair $(c_i, c_{i+1})$ thereby making $\Omega(m\log n)$ queries in total.
\end{proof}

We now present our \PE algorithm when we have a sequential access to the voters according to a single crossing order. We elicit the preference of the first voter using \Cref{obs:naive_single_crossing}. From second vote onwards, we simply use the idea of {\em insertion sort} relative to the previously elicited vote \cite{cormen2009introduction}. Since we are using insertion sort, any particular voter may be queried $\OO(m^2)$ times. However, we are able to bound the query complexity of our algorithm due to two fundamental reasons: (i) consecutive preferences will often be almost similar in a single crossing ordering, (ii) our algorithm takes only $\BigO(m)$ queries to elicit the preference of the current voter if its preference is indeed the same as the preference of the voter preceding it. In other words, every time we have to ``pay'' for shifting a candidate further back in the current vote, the relative ordering of that candidate with all the candidates that it jumped over is now fixed, because for these pairs, the one permitted crossing is now used up. We begin with presenting an important subroutine called Elicit($\cdot$) which finds the preference of a voter \el given another preference \RR by performing an insertion sort using \RR as the order of insertion.

\begin{algorithm}[ht]
 \caption{Elicit(\CC, \RR, \el)}\label{alg:elicit}
\begin{algorithmic}[1]
 \Require{A set of candidates $\CC = \{c_i : i\in[m]\}$, an ordering $\RR = c_1 \succ \cdots \succ c_m$ of \CC, a voter \el}
 \Ensure{Preference ordering $\succ_\el$ of voter \el on \CC}
 \State $\QQ\leftarrow c_1$ \Comment{\QQ will be the preference of the voter \el}
 \For{$i \gets 2 \textrm{ to } m$}\label{elicit_for}\Comment{$c_i$ is inserted in the $i^{th}$ iteration}
  \State Scan \QQ linearly {\em from index $i-1$ to $1$} to find the index $j$ where $c_i$ should be inserted according to the preference of voter \el and insert $c_i$ in \QQ at $j$\label{elicit_scan}
 \EndFor
 \Return \QQ
\end{algorithmic}
\end{algorithm}

For the sake of the analysis of our algorithm, let us introduce a few terminologies. Given two preferences $\succ_1$ and $\succ_2$, we call a pair of candidates $(x, y)\in\CC\times\CC, x\ne y,$ {\em good} if both $\succ_1$ and $\succ_2$ order them in a same way; a pair of candidates is called {\em bad} if it is not good. We divide the number of queries made by our algorithm into two parts: goodCost($\cdot$) and badCost($\cdot$) which are the number of queries made between good and respectively bad pair of candidates. In what follows, we show that goodCost($\cdot$) for Elicit($\cdot$) is small and the total badCost($\cdot$) across all the runs of Elicit($\cdot$) is small.

\begin{lemma}\label{lem:good}
 The goodCost(Elicit($\CC, \RR, \el$)) of Elicit($\CC, \RR, \el$) is $\OO(m)$ (good is with respect to the preferences \RR and $\succ_\el$).
\end{lemma}

\begin{proof}
 Follows immediately from the observation that in any iteration of the for loop at line \ref{elicit_for} in \Cref{alg:elicit}, only one good pair of candidates are compared.
\end{proof}

We now use \Cref{alg:elicit} iteratively to find the profile. We present the pseudocode in \Cref{alg:final} which works for the more general setting where a single crossing ordering is known but the voters are arriving in any arbitrary order $\pi$. We next compute the query complexity of \Cref{alg:final} when voters are arriving in a single crossing order.

\begin{algorithm}[ht]
\caption{PreferenceElicit($\pi$)}\label{alg:final}
 \begin{algorithmic}[1]
  \Require{$\pi\in\SB_n$}
  \Ensure{Profile of all the voters}
  \State $\QQ[\pi(1)] \leftarrow$ Elicit $\succ_{\pi(1)}$ using \Cref{obs:naive_single_crossing} \Comment{\QQ stores the profile}
  \State $\XX \leftarrow \{\pi(1)\}$ \Comment{Set of voters' whose preferences have already been elicited}
  \For{$i \gets 2 \textrm{ to } n$} \Comment{Elicit the preference of voter $\pi(i)$ in iteration $i$}
   \State $k \leftarrow \min_{j\in\XX} |j-i|$\label{final_k} \Comment{Find the closest known preference}
   \State $\RR \leftarrow \QQ[k], \XX \leftarrow \XX\cup\{\pi(i)\}$
   \State $\QQ[\pi(i)] \leftarrow \text{ Elicit}(\CC, \RR, \pi(i))$
  \EndFor
  \Return \QQ
 \end{algorithmic}
\end{algorithm}

\begin{theorem}\label{thm:sc_seq_known_ub}
 Assume that the voters are arriving sequentially according to an order with respect to which a profile \PP is single crossing. Then there is a \PE algorithm with query complexity $\BigO(mn + m^2)$.
\end{theorem}

\begin{proof}
 By renaming, let us assume, without loss of generality, that the voters are arriving according to the identity permutation $id_n$ of the voters and the profile \PP is single crossing with respect to $id_n$. Let the profile $\PP$ be $(P_1, P_2, \ldots, P_n) \in \LL(\CC)^n$. For two candidates $x, y\in \CC$ and a voter $i\in\{2, \ldots, n\}$, let us define a variable $b(x,y,i)$ to be one if $x$ and $y$ are compared for the voter $i$ by Elicit(\CC,$P_{i-1}$, $i$) and $(x, y)$ is a bad pair of candidates with respect to the preferences of voter $i$ and $i-1$; otherwise $b(x,y,i)$ is defined to be zero. Then we have the following. 
 \begin{eqnarray*}
  &&\text{CostPreferenceElicit}(id_n)\\ &=& \BigO(m\log m) +\sum_{i=2}^n\mathlarger{\mathlarger{(}}\text{goodCost(\Query}(\CC, P_{i-1}, i)) + \text{badCost(\Query}(\CC, P_{i-1}, i))\mathlarger{\mathlarger{)}}\\
  &\le& \BigO(m\log m + mn) + \sum_{i=2}^n\text{badCost(\Query}(\CC, P_{i-1}, i))\\
  &=& \BigO(m\log m + mn) + \mathlarger{\sum}_{(x, y)\in\CC\times\CC} \left(\sum_{i=2}^n b(x, y, i)\right)\\
  &\le& \BigO(m\log m + mn) + \sum_{(x, y)\in\CC\times\CC} 1\\
  &=& \BigO(mn + m^2)
 \end{eqnarray*} 
 The first inequality follows from \Cref{lem:good}, the second equality follows from the definition of $b(x, y, i)$, and the second inequality follows from the fact that $\sum_{i=2}^n b(x, y, i) \le 1$ for every pair of candidates $(x,y)\in\CC$ since the profile \PP is single crossing.
\end{proof}

We show next that, when the voters are arriving in a single crossing order, the query complexity upper bound in \Cref{thm:sc_seq_known_lb} is tight for a large number of voters up to constant factors. The idea is to pair up the candidates in a certain way and argue that the algorithm must compare the candidates in every pair for every voter thereby proving a $\Omega(mn)$ lower bound on query complexity.

\begin{theorem}\label{thm:sc_seq_known_lb}
 Assume that the voters are arriving sequentially according to an order with respect to which a profile \PP is single crossing. Then any deterministic \PE algorithm has query complexity $\Omega(m\log m + mn)$.
\end{theorem}

\begin{proof}
 The $\Omega(m\log m)$ bound follows from the fact that any profile consisting of only one preference $P\in\LL(\CC)$ is single crossing. By renaming, let us assume without loss of generality that the profile \PP is single crossing with respect to the identity permutation of the voters. Suppose we have an even number of candidates and $\CC = \{c_1, \ldots, c_m\}$. Consider the order $\QQ = c_1 \succ c_2 \succ \cdots \succ c_m$ and the pairing of the candidates $\{c_1, c_2\}, \{c_3, c_4\}, \ldots, \{c_{m-1}, c_m\}$. The oracle answers all the query requests consistently according to the order $Q$ till the first voter $\kappa$ for which there exists at least one odd integer $i\in[m]$ such that the pair $(c_i, c_{i+1})$ is not queried. If there does not exist any such $\kappa$, then the algorithm makes at least $\nfrac{mn}{2}$ queries thereby proving the statement. Otherwise, let $\kappa$ be the first vote such that the algorithm does not compare $c_i$ and $c_{i+1}$ for some odd integer $i\in[m]$. The oracle answers the queries for the rest of the voters $\{\kappa+1, \ldots, n\}$ according to the order $Q^\pr = c_1 \succ c_2 \succ \cdots \succ c_{i-1} \succ c_{i+1} \succ c_i \succ c_{i+2} \succ \cdots \succ c_m$. If the algorithm orders $c_i \succ_\kappa c_{i+1}$ in the preference of the voter $\kappa$, then the oracle sets the preference of the voter $\kappa$ to be $\QQ^\pr$. On the other hand, if the algorithm orders $c_{i+1} \succ_\kappa c_i$ in the preference of voter $\kappa$, then the oracle sets the preference of voter $\kappa$ to be \QQ. Clearly, the elicitation algorithm fails to correctly elicit the preference of the voter $\kappa$. However, the profiles for both the cases are single crossing with respect to the identity permutation of the voters and are consistent with the answers given to all the queries made by the algorithm. Hence, the algorithm must make at least $\nfrac{mn}{2}$ queries.
\end{proof}

We next move on to the case when we know a single crossing order of the voters; however, the voters arrive in an arbitrary order $\pi\in\SB_n$. The idea is to call the function Elicit($\CC, \RR, i$) where the current voter is the voter $i$ and \RR is the preference of the voter which is closest to $i$ according to a single crossing ordering and whose preference has already been elicited by the algorithm.

\begin{theorem}\label{thm:sc_seq_any_ub}
 Assume that a profile \PP is known to be single crossing with respect to a known ordering of voters $\sigma\in\SB_n$. However, the voters are arriving sequentially according to an arbitrary order $\pi\in\SB_n$ which may be different from $\sigma$. Then there is a \PE algorithm with query complexity $\BigO(mn + m^2\log n)$.
\end{theorem}

\begin{proof}
 By renaming, let us assume, without loss of generality, that the profile \PP is single peaked with respect to the identity permutation of the voters. Let the profile $\PP$ be $(P_1, P_2, \ldots, P_n) \in \LL(\CC)^n$. Let $f:[n]\longrightarrow [n]$ be the function such that $f(i)$ is the $k$ corresponding to the $i$ at line \ref{final_k} in \Cref{alg:final}. For candidates $x, y\in\CC$ and voter $\ell$, we define $b(x, y, \el)$ analogously as in the proof of \Cref{thm:sc_seq_known_ub}. We claim that $ B(x, y) = \sum_{i=2}^n b(x, y, i) \le \log n$. To see this, we consider any arbitrary pair $(x, y)\in\CC\times\CC$. Let the set of indices of the voters that have arrived immediately after the first time $(x,y)$ contributes to $B(x, y)$ be $\{i_1, i_2, \ldots, i_t\}$. Without loss of generality, let us assume $i_1 < i_2 < \cdots < i_t$. Again, without loss of generality, let us assume that voters $i_1, i_2, \ldots, i_j$ prefer $x$ over $y$ and voters $i_{j+1}, \ldots, i_t$ prefer $y$ over $x$. Let us define $\Delta$ to be the difference between smallest index of the voter who prefers $y$ over $x$ and the largest index of the voter who prefers $x$ over $y$. Hence, we currently have $\Delta = i_{j+1} - i_j$. A crucial observation is that if a new voter \el contributes to $B(x, y)$ then we must necessarily have $i_j < \el < i_{j+1}$. Another crucial observation is that whenever a new voter contributes to $B(x, y)$, the value of $\Delta$ gets reduced at least by a factor of two by the choice of $k$ at line \ref{final_k} in \Cref{alg:final}. Hence, the pair $(x,y)$ can contribute at most $(1+\log \Delta) = \BigO(\log n)$ to $B(x, y)$ since we have $\Delta\le n$ to begin with. Then we have the following. 
 \sloppypar
 \begin{eqnarray*}
  &&\text{CostPreferenceElicit}(\pi)\\  &=& \BigO(m\log m) +\sum_{i=2}^n\text{goodCost(\Query}(\CC, P_{f(i)}, i)) +  \text{badCost(\Query}(\CC, P_{f(i)}, i))\\
  &\le& \BigO(m\log m + mn) + \sum_{i=2}^n\text{badCost(\Query}(\CC, P_{f(i)}, i))\\
  &=& \BigO(m\log m + mn) + \sum_{(x, y)\in\CC\times\CC} \sum_{i=2}^n b(x, y, i)\\
  &\le& \BigO(m\log m + mn) + \sum_{(x, y)\in\CC\times\CC} \log n\\
  &=& \BigO(mn + m^2\log n)
 \end{eqnarray*} 
 The first inequality follows from \Cref{lem:good}, the second equality follows from the definition of $b(x, y, i)$, and the second inequality follows from the fact that $\sum_{i=2}^n b(x, y, i) \le \log n$.
\end{proof}

\subsection{Unknown Single Crossing Order}

We now turn our attention to \PE for single crossing profiles when no single crossing ordering is known. Before we present our \PE algorithm for this setting, let us first prove a few structural results about single crossing profiles which we will use crucially later. We begin with showing an upper bound on the number of distinct preferences in any single crossing profile.

\begin{lemma}\label{lem:bound}
 Let \PP be a profile on a set \CC of candidates which is single crossing. Then the number of distinct preferences in \PP is at most ${m\choose 2}+1$.
\end{lemma}

\begin{proof}
 By renaming, let us assume, without loss of generality, that the profile \PP is single crossing with respect to the identity permutation of the voters. We now observe that whenever the $i^{th}$ vote is different from the $(i+1)^{th}$ vote for some $i\in[n-1]$, there must exist a pair of candidates $(x, y)\in \CC\times\CC$ whom the $i^{th}$ vote and the $(i+1)^{th}$ vote order differently. Now the statement follows from the fact that, for every pair of candidates $(a,b)\in\CC\times\CC$, there can exist at most one $i\in[n-1]$ such that the $i^{th}$ vote and the $(i+1)^{th}$ vote order $a$ and $b$ differently.
\end{proof}

We show next that in every single crossing preference profile \PP where all the preferences are {\em distinct}, there exists a pair of candidates $(x, y)\in\CC\times\CC$ such that nearly half of the voters in \PP prefer $x$ over $y$ and the other voters prefer $y$ over $x$.

\begin{lemma}\label{lem:divide}
 Let \PP be a preference profile of $n$ voters such that all the preferences are distinct. Then there exists a pair of candidates $(x, y)\in\CC$ such that $x$ is preferred over $y$ in at least $\lfloor\nfrac{n}{2}\rfloor$ preferences and $y$ is preferred over $x$ in at least $\lfloor\nfrac{n}{2}\rfloor$ preferences in \PP.
\end{lemma}

\begin{proof}
 Without loss of generality, by renaming, let us assume that the profile \PP is single crossing with respect to the identity permutation of the voters. Since all the preferences in \PP are distinct, there exists a pair of candidates $(x, y)\in\CC\times\CC$ such that the voter $\lfloor\nfrac{n}{2}\rfloor$ and the voter $\lfloor\nfrac{n}{2}\rfloor + 1$ order $x$ and $y$ differently. Let us assume, without loss of generality, that the voter $\lfloor\nfrac{n}{2}\rfloor$ prefers $x$ over $y$. Now, since the profile \PP is single crossing, every voter in $[\lfloor\nfrac{n}{2}\rfloor]$ prefer $x$ over $y$ and every voter in $\{\lfloor\nfrac{n}{2}\rfloor + 1, \ldots, n\}$ prefer $y$ over $x$.
\end{proof}

Using \Cref{lem:bound,lem:divide} we now design a \PE algorithm when no single crossing ordering of the voters is known. The overview of the algorithm is as follows. At any point of time in the elicitation process, we have the set \QQ of all the distinct preferences that we have already elicited completely and we have to elicit the preference of a voter $\ell$. We first search the set of votes \QQ for a preference which is {\em possibly} same as the preference $\succ_\el$ of the voter $\ell$. It turns out that we can find a possible match $\succ\in\QQ$ using $\OO(\log|\QQ|)$ queries due to \Cref{lem:divide} which is $\OO(\log m)$ due to \Cref{lem:bound}. We then check whether the preference of the voter \el is indeed the same as $\succ$ or not using $\OO(m)$ queries. If $\succ$ is the same as $\succ_\el$, then we have elicited $\succ_\el$ using $\OO(m)$ queries. Otherwise, we elicit $\succ_\el$ using $\OO(m\log m)$ queries using \Cref{obs:naive_single_crossing}. Fortunately, \Cref{lem:bound} tells us that we would use the algorithm in \Cref{obs:naive_single_crossing} at most $\OO(m^2)$ times. We present the pseudocode of our \PE algorithm in this setting in \Cref{alg:unknown}. It uses \Cref{alg:same} as a subroutine which returns \true if the preference of any input voter is same as any given preference.

\begin{algorithm}[ht]
 \caption{Same(\RR, \el)}\label{alg:same}
 \begin{algorithmic}[1]
  \Require{$\RR = c_1\succ c_2\succ \cdots \succ c_m \in\LL(\CC) ,\el\in[n]$}
  \Ensure{\true if the preference of the $\el^{th}$ voter is \RR; \false otherwise}
  \For{$i \gets 1 \textrm{ to } m-1$}
   \If{\Query($c_i \succ_\el c_{i+1}$) = \false}
    \Return \false \Comment{We have found a mismatch.}
   \EndIf
  \EndFor
  \Return \true
 \end{algorithmic} 
\end{algorithm}

\begin{algorithm}[ht]
\caption{PreferenceElicitUnknownSingleCrossingOrdering($\pi$)}\label{alg:unknown}
 \begin{algorithmic}[1]
  \Require{$\pi\in\SB_n$}
  \Ensure{Profile of all the voters}
  \State $\RR, \QQ \leftarrow \emptyset$ \Comment{\QQ stores all the votes seen so far without duplicate. \RR stores the profile.}
  \For{$i \gets 1 \textrm{ to } n$} \Comment{Elicit preference of the $i^{th}$ voter in $i^{th}$ iteration of this for loop.}
   \State $\QQ^\pr \leftarrow \QQ$ 
   \While{$|\QQ^\pr|>1$} \Comment{Search \QQ to find a vote potentially same as the preference of $\pi(i)$}
    \State Let $x, y\in \CC$ be two candidates such that at least $\lfloor\nfrac{|\QQ^\pr|}{2}\rfloor$ votes in $\QQ^\pr$ prefer $x$ over $y$ and at least $\lfloor\nfrac{|\QQ^\pr|}{2}\rfloor$ votes in $\QQ^\pr$ prefer $y$ over $x$.
    \If{\Query($x \succ_{\pi(i)} y$) = \true}
     \State $\QQ^\pr \leftarrow \{ v\in\QQ^\pr : v \text{ prefers } x \text{ over } y \}$
    \Else
     \State $\QQ^\pr \leftarrow \{ v\in\QQ^\pr : v \text{ prefers } y \text{ oer } x \}$
    \EndIf
   \EndWhile
   \State Let $w$ be the only vote in $\QQ^\pr$\label{potential_match} \Comment{$w$ is potentially same as the preference of $\pi(i)$}
   \If{Same($w, \pi(i)$) = \true} \Comment{Check whether the vote $\pi(i)$ is potentially same as $w$}
    \State $\RR[\pi(i)] \leftarrow w$
   \Else 
    \State $\RR[\pi(i)] \leftarrow$ Elicit using \Cref{obs:naive_single_crossing}
    \State $\QQ \leftarrow \QQ\cup\{\RR[\pi(i)]\}$
   \EndIf
  \EndFor
  \Return \RR
 \end{algorithmic}
\end{algorithm}

\begin{theorem}\label{thm:sc_seq_unknown_ub}
 Assume that a profile \PP is known to be single crossing. However, no ordering of the voters with respect to which \PP is single crossing is known a priori. The voters are arriving sequentially according to an arbitrary order $\pi\in\SB_n$. Then there is a \PE algorithm with query complexity $\BigO(mn + m^3\log m)$.
\end{theorem}

\begin{proof}
 We present the pesudocode in \Cref{alg:unknown}. We maintain two arrays in the algorithm. The array \RR is of length $n$ and the $j^{th}$ entry stores the preference of voter $j$. The other array \QQ stores all the votes seen so far after removing duplicate votes; more specifically, if some specific preference $\suc$ has been seen \el many times for any $\el>0$, \QQ stores only one copy of $\suc$. Upon arrival of voter $i$, we first check whether there is a preference in \QQ which is ``potentially'' same as the preference of voter $i$. At the beginning of the search, our search space $\QQ^\pr=\QQ$ for a potential match in \QQ is of size $|\QQ|$. We next iteratively keep halving the search space as follows. We find a pair of candidates $(x, y)\in\CC\times\CC$ such that at least $\lfloor\nfrac{|\QQ^\pr|}{2}\rfloor$ preferences in $\QQ^\pr$ prefer $x$ over $y$ and at least $\lfloor\nfrac{|\QQ^\pr|}{2}\rfloor$ preferences prefer $y$ over $x$. The existence of such a pair of candidates is guaranteed by \Cref{lem:divide} and can be found in $\OO(m^2)$ time by simply going over all possible pairs of candidates. By querying how voter $i$ orders $x$ and $y$, we reduce the search space $\QQ^\pr$ for a potential match in \QQ to a set of size at most $\lfloor\nfrac{|\QQ^\pr|}{2}\rfloor+1$. Hence, in $\BigO(\log m)$ queries, the search space reduces to only one preference since we have $|\QQ|\le m^2$ by \Cref{lem:bound}. Once we find a potential match $w$ in \QQ (line \ref{potential_match} in \Cref{alg:unknown}), we check whether the preference of voter $i$ is the same as $w$ or not using $\BigO(m)$ queries. If the preference of voter $i$ is indeed same as $w$, then we output $w$ as the preference of voter $i$. Otherwise, we use \Cref{obs:naive_single_crossing} to elicit the preference of voter $i$ using $\BigO(m\log m)$ queries and put the preference of voter $i$ in \QQ. Since the number of times we need to use the algorithm in \Cref{obs:naive_single_crossing} is at most the number of distinct votes in \PP which is known to be at most $m^2$ by \Cref{lem:bound}, we get the statement.
\end{proof}

\Cref{thm:sc_seq_unknown_ub} immediately gives us the following corollary in the random access to voters model when no single crossing ordering is known.

\begin{corollary}\label{cor:sc_random_unknown_ub}
 Assume that a profile \PP is known to be single crossing. However, no ordering of the voters with respect to which \PP is single crossing is known. Given a random access to voters, there is a \PE algorithm with query complexity $\BigO(mn + m^3\log m)$.
\end{corollary}

\begin{proof}
 \Cref{alg:unknown} works for this setting also and exact same bound on the query complexity holds.
\end{proof}

We now show that the query complexity upper bound of \Cref{cor:sc_random_unknown_ub} is tight up to constant factors for large number of voters.

\begin{theorem}\label{thm:sc_random_unknown_lb}
 Given a random access to voters, any deterministic \PE algorithm which do not know any ordering of the voters with respect to which the input profile is single crossing has query complexity $\Omega(m\log m + mn)$.
\end{theorem}

\begin{proof}
 The $\Omega(m\log m)$ bound follows from sorting lower bound and the fact that any profile consisting of only one preference $P\in\LL(\CC)$ is single crossing. Suppose we have an even number of candidates and $\CC = \{c_1, \ldots, c_m\}$. Consider the ordering $\QQ = c_1 \succ c_2 \succ \cdots \succ c_m$ and the pairing of the candidates $\{c_1, c_2\}, \{c_3, c_4\}, \ldots, \{c_{m-1}, c_m\}$. The oracle answers all the query requests consistently according to the ordering $Q$. We claim that any \PE algorithm \AA must compare $c_i$ and $c_{i+1}$ for every voter and for every odd integer $i\in[m]$. Indeed, otherwise, there exist a voter $\kappa$ and an odd integer $i\in[m]$ such that the algorithm \AA does not compare $c_i$ and $c_{i+1}$. Suppose the algorithm outputs a profile $\PP^\pr$. If the voter $\kappa$ prefers $c_i$ over $c_{i+1}$ in $\PP^\pr$, then the oracle fixes the preference $\suc_\kappa$ to be $c_1 \succ c_2 \succ \cdots \succ c_{i-1} \succ c_{i+1} \succ c_i \succ c_{i+2} \succ \cdots \succ c_m$; otherwise the oracle fixes $\suc_\kappa$ to be \QQ. The algorithm fails to correctly output the preference of the voter $\kappa$ in both the cases. Also the final profile with the oracle is single crossing with respect to any ordering of the voters that places the voter $\kappa$ at the end. Hence, \AA must compare $c_i$ and $c_{i+1}$ for every voter and for every odd integer $i\in[m]$ and thus has query complexity $\Omega(mn)$.
\end{proof}

\subsection{Single Crossing Width}\label{subsec:sc_width}

We now consider preference elicitation for profiles which are nearly single crossing. We begin with profiles with bounded single crossing width.

\begin{proposition}\label{prop:width}
 Suppose a profile \PP is single crossing with width $w$. Given a \PE algorithm \AA with query complexity $\qqq(m,n)$ for random (or sequential) access to the voters when a single crossing order is known (or unknown), there exists a \PE algorithm $\AA^\pr$ for the single crossing profiles with width $w$ which has query complexity $\OO(\qqq(\nfrac{m}{w},n) + mn\log w)$ under same setting as \AA.
\end{proposition}

\begin{proof}
 Let the partition of the set of candidates \CC with respect to which the profile \PP is single crossing be $\bCC_i, i\in[\lceil\nfrac{m}{w}\rceil]$. Hence, $\CC = \cup_{i\in[\lceil\nfrac{m}{w}\rceil]} \bCC_i$ and $\bCC_i \cap \bCC_j = \emptyset$ for every $i, j\in[\lceil\nfrac{m}{w}\rceil]$ with $i\ne j$. Let $\CC^\pr$ be a subset of candidates containing exactly one candidate from $\bCC_i$ for each $i\in[\lceil\nfrac{m}{w}\rceil]$. We first find $\PP(\CC^\pr)$ using \AA. The query complexity of this step is $\OO(\qqq(\nfrac{m}{w},n))$. Next we find $\PP(\bCC_i)$ using \Cref{obs:naive_single_crossing} for every $i\in[\nfrac{m}{w}]$ thereby finding \PP. The overall query complexity is $\OO(\qqq(\nfrac{m}{w},n) + (\nfrac{m}{w})nw\log w) = \OO(\qqq(\nfrac{m}{w},n) + mn\log w)$.
\end{proof}

From \Cref{prop:width,lem:sc_random_ub,thm:sc_seq_known_ub,thm:sc_seq_any_ub,thm:sc_seq_unknown_ub,cor:sc_random_unknown_ub} we get the following.

\begin{corollary}
 Let a profile $\PP$ be single crossing with width $w$. Then there exists a \PE algorithm with query complexity $\BigO((\nfrac{m^2}{w}) \log (\nfrac{n}{w}) + mn\log w)$ for known single crossing order and random access to votes, $\BigO(\nfrac{m^2}{w^2} + mn\log w)$ for sequential access to votes according to a single crossing order, $\BigO((\nfrac{m^2}{w^2})\log (\nfrac{n}{w}) + mn\log w)$ for known single crossing order but arbitrary sequential access to votes, $\BigO(\nfrac{m^3}{w^3}\log (\nfrac{m}{w}) + mn\log w)$ for unknown single crossing order and arbitrary sequential access to votes or random access to votes.
\end{corollary}

\section{Results for Condorcet Winner}\label{sec:sc_cw}

\begin{lemma}\label{lem:weak_cond}
 Let an $n$ voter profile $\PP$ is single crossing with respect to the ordering $v_1, v_2, \ldots, v_n$. Then the candidate which is placed at the top position of the vote $v_{\lceil \nfrac{n}{2}\rceil}$ is a weak Condorcet winner.
\end{lemma}

\begin{proof}
 Let $c$ be the candidate which is placed at the top position of the vote $v_{\lceil \nfrac{n}{2}\rceil}$ and $w$ be any other candidate. Now the result follows from the fact that either all the votes in $\{v_i: 1\le i\le \lceil \nfrac{n}{2}\rceil\}$ prefer $c$ to $w$ or all the votes in $\{v_i: \lceil \nfrac{n}{2}\rceil \le i \le n\}$ prefer $c$ to $w$.
\end{proof}

\Cref{lem:weak_cond} immediately gives the following.

\begin{corollary}
 Given either sequential or random access to votes, there exists a \WCW algorithm with query complexity $\BigO(m)$ when a single crossing ordering is known.
\end{corollary}

\begin{proof}
 Since we know a single crossing ordering $v_1, v_2, \ldots, v_n$, we can find the candidate at the first position of the vote $v_{\lceil \nfrac{n}{2}\rceil}$ (which is a weak Condorcet winner by \Cref{lem:weak_cond}) by making $\BigO(m)$ queries.
\end{proof}

We now move on to finding the Condorcet winner. We begin with the following observation regarding existence of the Condorcet winner of a single crossing preference profile.

\begin{lemma}\label{lem:cond}
 Let an $n$ voter profile $\PP$ is single crossing with respect to the ordering $v_1, v_2, \ldots, v_n$. If $n$ is odd, then the Condorcet winner is the candidate placed at the first position of the vote $v_{\lceil \nfrac{n}{2}\rceil}$. If $n$ is even, then there exists a Condorcet winner if and only there exists a candidate which is placed at the first position of both the votes $v_{\nfrac{n}{2}}$ and $v_{\nfrac{n}{2}+1}$.
\end{lemma}

\begin{proof}
 If $n$ is odd, the result follows from \Cref{lem:weak_cond} and the fact that a candidate $c$ is a Condorcet winner if and only if $c$ is a weak Condorcet winner.
 
 Now suppose $n$ is even. Let a candidate $c$ is placed at the first position of the votes $v_{\nfrac{n}{2}}$ and $v_{\nfrac{n}{2}+1}$. Then for every candidate $x\in\CC\setminus\{c\}$, either all the votes in $\{v_i: 1\le i\le \nfrac{n}{2} + 1\}$ or all the votes in $\{v_i: \nfrac{n}{2}\le i\le n\}$ prefer $c$ over $x$. Hence, $c$ is a Condorcet winner. Now suppose a candidate $w$ is a Condorcet winner. Then we prove that $w$ must be placed at the first position of both the votes $v_{\nfrac{n}{2}}$ and $v_{\nfrac{n}{2}+1}$. If not, then both the two candidates placed at the first positions of the votes $v_{\nfrac{n}{2}}$ and $v_{\nfrac{n}{2}+1}$ are weak Condorcet winners according to \Cref{lem:weak_cond} applied to single crossing orders $v_1, v_2, \ldots, v_n$ and $v_n, v_{n-1}, \ldots, v_1$.
\end{proof}

\Cref{lem:cond} immediately gives us the following.

\begin{corollary}
 Given either sequential or random access to votes, there exists a \CW algorithm with query complexity $\BigO(m)$ when a single crossing ordering is known.
\end{corollary}

\begin{proof}
 Since we know a single crossing ordering $v_1, v_2, \ldots, v_n$, we can find the candidates which are placed at the first positions of the votes $v_{\nfrac{n}{2}}$ and $v_{\nfrac{n}{2}+1}$ by making $\BigO(m)$ queries. Now the result follows from \Cref{lem:cond}.
\end{proof}

\section{Conclusion}

In this work, we have presented \PE algorithms with low query complexity for single crossing profiles under various settings. Moreover, we have proved that the query complexity of our algorithms are tight for a large number of voters up to constant factors for all but one setting namely when the voters can be accessed randomly but we do not know any ordering with respect to which the voters are single crossing. We then move on to show that a weak Condorcet winner and the Condorcet winner (if one exists) can be found from a single crossing preference profile using much less number of queries.

With this, we conclude the first part of the thesis. In the next part of the thesis, we study the problem of finding a winner of an election under various real world scenarios.
\part{Winner Determination}
\vspace{5ex}
In the second part of the thesis, we present our work on determining the winner of an election under various circumstances. This part consists of the following chapters.

\begin{itemize}
 \item In \Cref{chap:winner_prediction_mov} -- \nameref{chap:winner_prediction_mov} -- we present efficient algorithms based on sampling to predict the winner of an election and its robustness. We prove that both the problems of winner prediction and robustness estimation can be solved simultaneously by sampling only a few votes uniformly at random.
 
 \item In \Cref{chap:winner_stream} -- \nameref{chap:winner_stream} -- we develop (often) optimal algorithms for determining the winner of an election when the votes are arriving in a streaming fashion. Our results show that an approximate winner can be determined fast with a small amount of space.
 
 \item In \Cref{chap:kernel} -- \nameref{chap:kernel} -- we present interesting results on kernelization for determining possible winners from a set of incomplete votes. Our results prove that the problem of determining possible winners with incomplete votes does not have any efficient preprocessing strategies under plausible complexity theoretic assumptions even when the number of candidates is relatively small. We also present efficient kernelization algorithms for the problem of manipulating an election.
\end{itemize}

\blankpage
\chapter{Winner Prediction and Margin of Victory Estimation}
\label{chap:winner_prediction_mov}

\blfootnote{A preliminary version of the work on winner prediction in this chapter was published as \cite{DeyB15}: Palash Dey and Arnab Bhattacharyya. Sample complexity for winner prediction in elections. In Proceedings of the 2015 International Conference on Autonomous Agents and Multiagent Systems, AAMAS 2015, Istanbul, Turkey, May 4-8, 2015, pages 1421-1430, 2015.

A preliminary version of the work on the estimation of margin of victory in this chapter was published as \cite{DeyN15}: Palash Dey and Y. Narahari. Estimating the margin of victory of an election using sampling. In Proceedings of the Twenty-Fourth International Joint Conference on Artificial Intelligence, IJCAI 2015, Buenos Aires, Argentina, July 25-31, 2015, pages 1120-1126, 2015.}

\begin{quotation}
{\small Predicting the winner of an election and estimating the margin of victory of that election are favorite problems both for news media pundits and computational social choice theorists. Since it is often infeasible to elicit the preferences of all the voters in a typical prediction scenario, a common algorithm used for predicting the winner and estimating the margin of victory is to run the election on a small sample of randomly chosen votes and predict accordingly. We analyze the performance of this algorithm for many commonly used voting rules.

More formally, for predicting the winner of an election, we introduce the $(\epsilon, \delta)$-\WD problem, where given an election \EE on $n$ voters and $m$ candidates in which the margin of victory is at least $\epsilon n$ votes, the goal is to determine the winner with probability at least $1-\delta$ where $\eps$ and $\delta$ are parameters with $0< \eps, \delta<1$. The margin of victory of an election is the smallest number of votes that need to be modified in order to change the election winner. We show interesting lower and upper bounds on the number of samples needed to solve the $(\epsilon, \delta)$-\WD problem for many common voting rules, including all scoring rules, approval, maximin, Copeland, Bucklin, plurality with runoff, and single transferable vote. Moreover, the lower and upper bounds match for many common voting rules up to constant factors.

For estimating the margin of victory of an election, we introduce the \textsc{$(c, \epsilon, \delta)$--Margin of Victory} problem, where given an election $\mathcal{E}$ on $n$ voters, the goal is to estimate the margin of victory $M(\mathcal{E})$ of $\mathcal{E}$ within an additive error of $cM(\mathcal{E})+\eps n$ with probability of error at most $\delta$ where $\eps, \delta,$ and $c$ are the parameters with $0<\eps, \delta<1$ and $c>0$. We exhibit interesting bounds on the sample complexity of the \textsc{$(c, \epsilon, \delta)$--Margin of Victory} problem for many commonly used voting rules including all scoring rules, approval, Bucklin, maximin, and Copeland$^{\alpha}$. We observe that even for the voting rules for which computing the margin of victory is $\NPshort$-hard, there may exist efficient sampling based algorithms for estimating the margin of victory, as observed in the cases of maximin and Copeland$^{\alpha}$ voting rules.}
\end{quotation}

\section{Introduction}

In many situations, one wants to predict the winner without
holding the election for the entire population of voters. The
most immediate such example is an {\em election poll}. Here, the
pollster wants to quickly gauge public opinion in order to predict the
outcome of a full-scale election. For political elections,
exit polls (polls conducted on voters after they have
voted) are widely used by news media to predict the winner before 
official results are announced. In {\em surveys}, a full-scale
election is never conducted, and the goal is to determine the winner,
based on only a few sampled votes, for a hypothetical election on all the voters. For instance, it is
not possible to force all the residents of a city to fill out an
online survey to rank the local Chinese restaurants, and so only
those voters who do participate have their preferences aggregated. 

If the result of the poll or the survey has to reflect the true
election outcome, it is obviously necessary that the number of sampled
votes not be too small. Here, we investigate this fundamental
question: 

\begin{quote}
What is the minimum number of votes that need to be sampled
so that the winner of the election on the sampled votes is the
same as the winner of the election on all the votes?
\end{quote}

This question can be posed for any voting rule. The most immediate rule
to study is  the {\em plurality} voting rule, where each voter votes
for a single candidate and the candidate with most votes
wins. Although the plurality rule is the most common voting rule used
in political elections, it is important to extend the 
analysis to other popular voting rules. For example, the {\em single
  transferable vote} is used in political elections in Australia,
India and Ireland, and it was the subject of a nationwide referendum
in the UK in 2011. The {\em Borda} voting rule is used in the
Icelandic parliamentary elections. Outside politics, in private
companies and competitions, a wide variety of voting rules are
used. For example, the {\em approval} voting rule has been used by the
Mathematical Association of America, the American Statistical
Institute, and the Institute of Electrical and Electronics
Engineers, and {\em Condorcet consistent} voting rules are used
by many free software organizations.

A voting rule is called {\em anonymous} if the winner does not change after 
any renaming of the voters. All popular voting rules including the ones mentioned 
before are anonymous. For any anonymous voting rule, the question of finding the
minimum number of vote samples required becomes trivial if a single
voter in the election can change the winning candidate. In this
case, all the votes need to be counted, because otherwise that single
crucial vote may not be sampled. We get around this problem by
assuming that in the elections we consider, the winning candidate wins
by a considerable {\em margin of victory}. Formally, the margin of
victory of an election is defined as the minimum number of votes that
must be changed in order to change the election winner. Note that the
margin of victory depends not only on the votes cast but also on the
voting rule used in the election.

Other than predicting the winner of an election, one may also like to know how robust the election outcome is with respect to the changes in votes~\cite{shiryaev2013elections,caragiannis2014modal,regenwetter2006behavioral}. One way to capture robustness of an election outcome is the margin of victory of that election. An election outcome is considered to be robust if the margin of victory of that election is large.

In addition to formalizing the notion of robustness of an election outcome, the margin of victory of an election plays a crucial role in many practical applications. One such example is post election audits --- methods to observe a certain number of votes (which is often selected randomly) after an election to detect an incorrect outcome. There can be a variety of reasons for an incorrect election outcome, for example, software or hardware bugs in voting machine~\cite{norden2007post}, machine output errors, use of various clip-on devices that can tamper with the memory of the voting machine~\cite{wolchok2010security}, human errors in counting votes. Post election audits have nowadays become common practice to detect problems in electronic voting machines in many countries, for example, USA. As a matter of fact, at least thirty states in the USA have reported such problems by 2007~\cite{norden2007post}. Most often, the auditing process involves manually observing some sampled votes. Researchers have subsequently proposed various {\em risk limiting auditing} methodologies that not only minimize the cost of manual checking, but also limit the risk of making a human error by sampling as few votes as possible~\cite{stark2008conservative,stark2008sharper,stark2009efficient,sarwate2011risk}. The sample size in a risk limiting audit critically depends on the margin of victory of the election. 

Another important application where the margin of victory plays an important role is polling. One of the most fundamental questions in polling is: how many votes should be sampled to be able to correctly predict the outcome of an election? It turns out that the {\em sample complexity} in polling too crucially depends on the  margin of victory of the election from which the pollster is sampling~\cite{canetti1995lower,DeyB15}. Hence, computing (or at least approximating sufficiently accurately) the margin of victory of an election is often a necessary task in many practical applications. However, in many applications including the ones discussed above, one cannot observe all the votes. For example, in a survey or polling, one cannot first observe all the votes to compute the margin of victory and then sample the required number of votes based on the margin of victory computed. Hence, one often needs a {\it ``good enough''} estimate of the margin of victory by observing a small number votes. We precisely address this problem: estimate the margin of victory of an election by sampling as few votes as possible.

\subsection{Our Contribution}

In this work, we show nontrivial bounds for the sample complexity of predicting the winner of an election and estimating the margin of victory of an election for many commonly used voting rules.

\subsubsection{Winner Prediction}

Let $n$ be the number of voters, $m$ the number of candidates, and $r$ a voting rule.
We introduce and study the following problem in the context of winner prediction:
  
\begin{definition}{\sc($(\epsilon, \delta)$-\WD)}\\
 Given a $r$-election $\mathcal{E}$ whose margin of victory is at least $\epsilon n$,  
 determine the winner of the election with probability at least
 $1-\delta$. (The probability is taken over the internal coin tosses of
 the algorithm.) 
\end{definition}

We remind the reader that there is no assumption about the
distribution of votes in this problem. Our goal is to solve the
$(\eps, \delta)$-\WD problem by a randomized
algorithm that is allowed to query the votes of arbitrary voters. Each
query reveals the full vote of the voter. The minimum number of votes
queried by any algorithm that solves the $(\eps, \delta)$-\WD problem is called the {\em   sample complexity} of this problem. The
sample complexity  can of course depend on $\eps$, $\delta$, $n$, $m$,
and the voting rule in use. 

A standard result in \cite{canetti1995lower} implies that solving the above
problem for the majority rule on $2$ candidates requires at least
$(\nfrac{1}{4\epsilon^2})\ln (\frac{1}{8e\sqrt{\pi}\delta})$ samples (\Cref{thm:lb}). Also, a straightforward argument (Theorem \ref{thm:gen}) using Chernoff bounds shows that for any homogeneous voting rule, the sample complexity is at most $(\nfrac{9m!^2}{2\eps^2}) \ln(\nfrac{2m!}{\delta})$. So, when $m$ is a constant, the sample complexity is of the order $\Theta((\nfrac{1}{\eps^2}) \ln (\nfrac{1}{\delta}))$ for any homogeneous voting rule that reduces to the majority voting rule on $2$ candidates and this bound is tight up to constant factors. We note that all the commonly used voting rules including the ones we study here, are homogeneous and they reduce to the majority voting rule when we have only $2$ candidates. Note that this bound is independent of $n$ if $\eps$ and $\delta$ are independent of $n$.

Our main technical contribution is in understanding the dependence of
the sample complexity on the number of candidates $m$. Note
that the upper bound cited above has very bad dependence on $m$ and is
clearly unsatisfactory in situations when $m$ is large (such as in
online surveys about restaurants). 
\begin{itemize}
 \item We show that the sample complexity of the   $(\epsilon,
   \delta)$-\WD problem is at most
   $(\nfrac{9}{2\epsilon^2})\ln(\nfrac{2k}{\delta})$ for the
   $k$-approval voting rule (\Cref{thm:kapp}) and at most $(\nfrac{27}{\epsilon^2})\ln (\nfrac{4}{\delta})$ for the plurality with runoff voting rule (\Cref{thm:runoff}). In
   particular, for the plurality rule, the
   sample complexity is independent of $m$ as well as $n$.

 \item We show that the sample complexity of the $(\epsilon,
   \delta)$-\WD problem is at most $(\nfrac{9}{2\eps^2})\ln (\nfrac{2m}{\delta})$ and $\Omega\left((1-\delta)(\nfrac{1}{\epsilon^2})\ln m\right)$ for the Borda (\Cref{thm:strlwb}), approval (\Cref{thm:app}), maximin (\Cref{thm:maximin}), and
   Bucklin (\Cref{thm:bucklin}) voting rules. Note that when $\delta$
   is a constant, the upper and lower bounds match up to
   constant factors. 
   
 \item We show a sample complexity upper bound of $(\nfrac{25}{2\epsilon^2})\ln^3 (\nfrac{m}{\delta})$ for the $(\epsilon, \delta)$-\WD problem for the Copeland$^\alpha$ voting rule (\Cref{thm:copeland}) and $(\nfrac{3m^2}{\epsilon^2})(m\ln 2+\ln (\nfrac{2m}{\delta}))$ for the STV  voting rule (\Cref{thm:stv}).
\end{itemize}

We summarize these results in \Cref{table:wd}.

\begin{table}[htbp]
  \begin{center}
  {\renewcommand{\arraystretch}{1.7}
 \begin{tabular}{|c|c|c| }\hline
  \textbf{\multirow{2}{*}{Voting rules}}	& \multicolumn{2}{c|}{\textbf{Sample complexity for $(\epsilon, \delta)$-\WD}} \\\cline{2-3}
  & Upper bounds & Lower bounds \\\hline\hline
  
  $k$-approval	& $(\nfrac{9}{2\epsilon^2})\ln(\nfrac{2k}{\delta})$ [\Cref{thm:kapp}] & \specialcell{$\Omega\left((\nfrac{\ln (k+1)}{\epsilon^2})\left( 1 - \delta \right)\right)^\S$\\~[\Cref{thm:strlwb}]}   \\\hline
  
  $k$-veto & $(\nfrac{9}{2\epsilon^2})\ln(\nfrac{2k}{\delta})$ [\Cref{thm:kveto}] & \specialcell{$(\nfrac{1}{4\epsilon^2})\ln (\nfrac{1}{8e\sqrt{\pi}\delta})^\ast$\\~[\Cref{cor:lb}]}\\\hline
  
  Scoring rules	& \multirow{2}{*}{$(\nfrac{9}{2\eps^2})\ln (\nfrac{2m}{\delta})$ [\Cref{thm:scr}]} &  \\\cline{1-1}
  Borda	&  & \multirow{4}{*}{\specialcell{$\Omega\left((\nfrac{\ln m}{\epsilon^2})\left( 1 - \delta \right) \right) ^\dagger$\\~[\Cref{thm:strlwb}]}} \\\cline{1-2}
    
  Approval	& $(\nfrac{9}{2\eps^2})\ln (\nfrac{2m}{\delta})$ [\Cref{thm:app}] &  \\\cline{1-2}
  
  Maximin	& $(\nfrac{9}{2\epsilon^2})\ln (\nfrac{2m}{\delta})$ [\Cref{thm:maximin}]&  \\\cline{1-2}
  
  Copeland	& $(\nfrac{25}{2\epsilon^2}) \ln^3 (\nfrac{2m}{\delta})$ [\Cref{thm:copeland}]&		\\\cline{1-2}
  
  Bucklin	& $(\nfrac{9}{2\epsilon^2})\ln (\nfrac{2m}{\delta})$ [\Cref{thm:bucklin}]	& 	\\ \hline
  
  Plurality with runoff & $(\nfrac{27}{\epsilon^2})\ln (\nfrac{4}{\delta})$ [\Cref{thm:runoff}] & \multirow{3}{*}{\specialcell{$(\nfrac{1}{4\epsilon^2})\ln (\nfrac{1}{8e\sqrt{\pi}\delta})^\ast$\\~[\Cref{cor:lb}]}}\\\cline{1-2}
  
  STV		& $(\nfrac{3m^2}{\epsilon^2})(m\ln 2+\ln (\nfrac{2m}{\delta}))$ [\Cref{thm:stv}]&		\\\cline{1-2}
    
  Any homogeneous voting rule		& $(\nfrac{9m!^2}{2\eps^2}) \ln(\nfrac{2m!}{\delta})$ [\Cref{thm:gen}]&		\\\hline
 \end{tabular}
 }
  \caption{\normalfont Sample complexity of the $(\epsilon,
    \delta)$-\WD problem for common voting
    rules. $\dagger$--The lower bound of $\Omega ( (\nfrac{\ln
      m}{\epsilon^2})\left( 1 - \delta \right) )$ also applies to any voting rule that is 
    Condorcet consistent. ${\ast}$-- The lower bound of $(\nfrac{1}{4\epsilon^2})\ln (\nfrac{1}{8e\sqrt{\pi}\delta})$ holds for any voting rule
  that reduces to the plurality voting rule for elections with two candidates. $\S$-- The lower bound holds for $k\le .999m$.}
  \label{table:wd}
  \end{center}
\end{table} 

\subsubsection{Estimating Margin of Victory}

The margin of victory of an election is defined as follows. 

\begin{definition}Margin of Victory (MOV)\\
 Given an election \EE, the margin of victory of \EE is defined as the smallest number of votes that must be changed to change the winner of the election \EE.\label{def:margin of victory}
\end{definition}

We abbreviate margin of victory as \textsf{MOV}. We denote the \textsf{MOV} of an election \EE by $M(\EE)$. We introduce and study the following computational problem for estimating the margin of victory of an election: 

\begin{definition}\textsc{($(c, \epsilon, \delta)$--Margin of Victory(}\textsf{MOV}\textsc{))}\\\label{def:prob}
 Given a $r$-election $\mathcal{E}$, determine the margin of victory $\MOV$ of $\mathcal{E}$, within an additive error of at most $c\MOV + \epsilon n$ with probability at least $1-\delta$. The probability is taken over the internal coin tosses of the algorithm.
\end{definition}

Our goal here is to solve the $(c, \epsilon, \delta)$--\MV problem by observing as few sample votes as possible. Our main technical contribution is to come up with efficient sampling based {\em polynomial time} randomized algorithms to solve the $(c, \epsilon, \delta)$--\MV problem for common voting rules. Each sample reveals the entire preference order of the sampled vote. We summarize the results on the $(c, \epsilon, \delta)$--\MV problem in \Cref{table:margin of victory_summary}.

\begin{table}[htbp]
\centering
\renewcommand{\arraystretch}{1.7}
\begin{tabular}{|c|c|c|}\hline
  \textbf{\multirow{2}{*}{Voting rules}}	& \multicolumn{2}{c|}{\textbf{\multirow{1}{*}{Sample complexity for $(c, \epsilon, \delta)$--\MV}}} \\\cline{2-3}
  & Upper bounds & Lower bounds \\\hline\hline
  
  \multirow{1}{*}{Scoring rules}	& \multirow{1}{*}{$(\nfrac{1}{3}, \epsilon, \delta)$--\textsf{MOV}, $(\nfrac{12}{\epsilon ^2})\ln (\nfrac{2m}{\delta})$ [\Cref{thm:scr_mov}]} & \multirow{6}{*}{\specialcell{$(c, \epsilon, \delta)$--\textsf{MOV}$^\dagger$,\\$(\nfrac{(1-c)^2}{36\epsilon^2})\ln \left(\nfrac{1}{8e\sqrt{\pi}\delta}\right)$,\\~[\Cref{thm:lb_mov}]\\~[\Cref{cor:lb_mov}]}} \\\cline{1-2}
  
  \multirow{1}{*}{$k$-approval}	& \multirow{1}{*}{$(0, \epsilon, \delta)$--\textsf{MOV}, $(\nfrac{12}{\epsilon^2})\ln (\nfrac{2k}{\delta})$, [\Cref{thm:kapp_mov}]} &  \\\cline{1-2}
  
  \multirow{1}{*}{Approval}	& \multirow{1}{*}{$(0, \epsilon, \delta)$--\textsf{MOV}, $(\nfrac{12}{\epsilon^2})\ln (\nfrac{2m}{\delta})$, [\Cref{thm:app_mov}]} &  \\\cline{1-2}
  
  \multirow{1}{*}{Bucklin}	& \multirow{1}{*}{$(\nfrac{1}{3}, \epsilon, \delta)$--\textsf{MOV}, $(\nfrac{12}{\epsilon^2})\ln (\nfrac{2m}{\delta})$, [\Cref{thm:bucklin_mov}]}	& 	\\\cline{1-2}
  
  \multirow{1}{*}{Maximin}	& \multirow{1}{*}{$(\nfrac{1}{3}, \epsilon, \delta)$--\textsf{MOV}, $(\nfrac{24}{\epsilon^2})\ln (\nfrac{2m}{\delta})$, [\Cref{thm:maximin_mov}]} &  \\\cline{1-2}
  
  Copeland$^\alpha$	& \makecell{$\left(1-O\left(\nfrac{1}{\ln m}\right), \epsilon, \delta\right)$--\textsf{MOV}, $(\nfrac{96}{\epsilon^2})\ln (\nfrac{2m}{\delta})$,\\\relax[\Cref{thm:copeland_mov}]}  &		\\\hline
 \end{tabular}
  \caption{\normalfont Sample complexity for the $(c, \epsilon, \delta)$--\MV problem for common voting rules. $\dagger$The result holds for any $c \in [0,1).$}
  \label{table:margin of victory_summary}
\end{table} 

\Cref{table:margin of victory_summary} shows a practically appealing positive result --- {\em the sample complexity of all the algorithms presented here is independent of the number of voters}. Our lower bounds on the sample complexity of the $(c, \epsilon, \delta)$--\MV problem for all the voting rules studied here match with the upper bounds up to constant factors when we have a constant number of candidates. Moreover, the lower and upper bounds on the sample complexity for the $k$-approval voting rule match up to constant factors irrespective of number of candidates, when $k$ is a constant. 

\begin{itemize}
 \item We show a sample complexity lower bound of $(\nfrac{(1-c)^2}{36\epsilon^2})\ln \left(\nfrac{1}{8e\sqrt{\pi}\delta}\right)$ for the $(c, \epsilon, \delta)$--\MV problem for all the commonly used voting rules, where $c\in [0,1)$ (\Cref{thm:lb_mov} and \Cref{cor:lb_mov}).
 \item We show a sample complexity upper bound of $(\nfrac{12}{\epsilon ^2})\ln (\nfrac{2m}{\delta})$ for the $(\nfrac{1}{3}, \epsilon, \delta)$--\MV problem for arbitrary scoring rules (\Cref{thm:scr_mov}). However, for a special class of scoring rules, namely, the $k$-approval voting rules, we prove a sample complexity upper bound of $(\nfrac{12}{\epsilon^2})\ln (\nfrac{2k}{\delta})$ for the $(0, \epsilon, \delta)$--\MV problem (\Cref{thm:kapp_mov}).
\end{itemize}

One key finding of our work is that, there may exist efficient sampling based polynomial time algorithms for estimating the margin of victory, even if computing the margin of victory is $\NPshort$-hard for a voting rule~\cite{xia2012computing}, as observed in the cases of maximin and Copeland$^\alpha$ voting rules.  

\subsection{Related Work}

The subject of voting is at the heart of (computational) social choice
theory, and there is a vast amount of literature in this area.  Elections take place not
only in human societies but also in man made social networks
\cite{boldi2009voting,rodriguez2007smartocracy} and, generally, in
many multiagent systems \cite{ephrati1991clarke,PennockHG00}. The winner determination
problem is the task of finding the winner in an election, given the
voting rule in use and the set of all votes cast. It is known that
there are natural voting rules, e.g., Kemeny's rule and Dodgson's
method, for which the winner determination problem is
\textsf{NP}-hard \cite{bartholdi1989voting,hemaspaandra2005complexity,hemaspaandra1997exact}. 

The basic model of election has been generalized in several other ways to capture
real world situations. One important consideration is that the votes may be incomplete
rankings of the candidates and not a complete ranking. There can also
be uncertainty over which voters and/or candidates will eventually
turn up. The uncertainty may additionally come up from the voting rule that will be used eventually to select the winner. In these incomplete information settings, several winner models have been proposed, for example, robust winner~\cite{boutilier2014robust,lu2011robust,shiryaev2013elections}, 
multi winner~\cite{lu2013multi}, stable winner~\cite{falik2012coalitions}, approximate winner~\cite{doucetteapproximate}, 
probabilistic winner~\cite{bachrach2010probabilistic}, possible winner~\cite{moulin2016handbook,DeyMN15,journalsDeyMN16}. 
Hazon et al.~\cite{hazon2008evaluation} proposed useful methods to
evaluate the outcome of an election under various uncertainties. We do
not study the role of uncertainty in this work.

The general question of whether the outcome of an election can be
determined by less than the full set of votes is the subject of {\em
  preference elicitation}, a central category of problems in AI. The
$(\eps, \delta)$-\WD problem also falls in this area
when the elections are restricted to those having margin of victory
at least $\eps n$. For general elections, the preference elicitation
problem was studied by Conitzer and Sandholm \cite{conitzer2002vote},
who defined an elicitation policy as an adaptive sequence of questions
posed to voters. They proved that finding an efficient elicitation
policy is \textsf{NP}-hard for many common voting rules. Nevertheless,
several elicitation policies have been developed in later work
\cite{Conitzer09,lu2011robust,lu2011vote,ding2012voting,oren2013efficient,deypeak,deycross} that
work well in practice and have formal guarantees under various
assumptions on the vote distribution. Another related work is that of
Dhamal and Narahari \cite{dhamal2013scalable} who show that if the
voters are members of a social network where neighbors in the network
have similar candidate votes, then it is possible to elicit the
votes of only a few voters to determine the outcome of the full
election. 

In contrast, in our work on winner prediction, we posit no assumption on the vote
distribution other than that the votes create a substantial margin of
victory for the winner. Under this assumption, we show that even for
voting rules in which winner determination is \textsf{NP}-hard in the
worst case, it is possible to sample a small number of votes to
determine the winner. Our work falls inside the larger framework of
{\em property testing} \cite{ron2001property}, a class of problems studied in theoretical
computer science, where the inputs are promised to either satisfy some
property or have a ``gap'' from instances satisfying the property. In
our case, the instances are elections which either have some candidate $w$ as the
winner or are ``far'' from having $w$ being the winner (in the sense
that many votes need to be changed).

There have been quite a few work on computing the margin of victory of an election. Most prominent among them is the work of Xia~\cite{xia2012computing}. Xia presents polynomial time algorithms for computing the margin of victory of an election for various voting rules, for example the scoring rules, and proved intractability results for several other voting rules, for example the maximin and Copeland$^\alpha$ voting rules. Magrino et al.~\cite{magrino2011computing} present approximation algorithms to compute the margin of victory for the instant runoff voting (IRV) rule. Cary~\cite{cary2011estimating} provides algorithms to estimate the margin of victory of an IRV election. Endriss et al.~\cite{endriss2014margin} compute the complexity of exact variants of the margin of victory problem for Schulze, Cup, and Copeland voting rules. However, all the existing algorithms to either compute or estimate the margin of victory {\em need to observe all the votes,} which defeats the purpose in many applications including the ones we discussed. We, in this work, show that we can estimate the margin of victory of an election for many commonly used voting rules quite accurately by sampling a few votes only. Moreover, the accuracy of our estimation algorithm is good enough for many practical scenarios. For example, \Cref{table:margin of victory_summary} shows that it is enough to select only $3600$ many votes uniformly at random to estimate $\frac{\text{\textsf{MOV}}}{n}$ of a plurality election within an additive error of $0.1$ with probability at least $0.99$, where $n$ is the number of votes. We note that in all the sampling based applications we discussed, the sample size is inversely proportional to $\frac{\text{\textsf{MOV}}}{n}$~\cite{canetti1995lower} and thus it is enough to estimate $\frac{\text{\textsf{MOV}}}{n}$ accurately (see \Cref{table:wd}).

The problem of finding the margin of victory in an election is the same as the optimization version of the destructive bribery problem introduced by Faliszewski et al.~\cite{faliszewski2006complexity,faliszewski2009hard}. However, to the best of our knowledge, there is no prior work on estimating the cost of bribery by sampling votes.

\section{Results for Winner Prediction}\label{sec:wd}

In this section, we present our results for the $(\eps, \delta)$-\WD problem.

\subsection{Results on Lower Bounds}\label{sec:lwb}

We begin with presenting our lower bounds for the $(\eps, \delta)$-\WD problem for various voting rules. Our lower bounds for the sample complexity of the $(\eps, \delta)$-\WD problem are derived from information-theoretic lower bounds
for distinguishing distributions.

We start with the following basic observation.
Let $X$ be a random variable taking value $1$ with probability
$\frac{1}{2}-\epsilon$ and $0$ with probability
$\frac{1}{2}+\epsilon$; $Y$ be a random variable taking value $1$ with
probability $\frac{1}{2}+\epsilon$ and $0$ with probability
$\frac{1}{2}-\epsilon$.  Then, it is known that every algorithm needs at least
$\frac{1}{4\epsilon^2}\ln \frac{1}{8e\sqrt{\pi}\delta}$ many samples to
distinguish between $X$ and $Y$ with probability of making an error
being at most $\delta$~\cite{canetti1995lower,bar2001sampling}. We immediately have the following:

\begin{theorem}\label{thm:lb}
 The sample complexity of the $(\epsilon, \delta)$-\WD problem for the plurality voting rule is at least $\frac{1}{4\epsilon^2}\ln \frac{1}{8e\sqrt{\pi}\delta}$ even when the number of candidates is $2$.
\end{theorem}

\begin{proof}
 Consider an election with two candidates $a$ and $b$. Consider two
 vote distributions $X$ and $Y$. In $X$, exactly $\frac{1}{2} +
 \epsilon$ fraction of voters prefer $a$ to $b$ and thus $a$ is the
 plurality winner of the election. In $Y$, exactly $\frac{1}{2} +
 \epsilon$ fraction of voters prefer $b$ to $a$ and thus $b$ is the
 plurality winner of the election. Also, the margin of victory of both
 the elections corresponding to the vote distributions $X$ and $Y$ is
 $\epsilon n$, since each vote change can change the plurality score
 of any candidate by at most one. Observe that any $(\epsilon, \delta)$-\WD algorithm for plurality will give us a distinguisher between the
 distributions $X$ and $Y$ with probability of error at most $\delta$.
 and hence will need $\frac{1}{4\epsilon^2}\ln \frac{1}{8e\sqrt{\pi}\delta}$ samples. 
\end{proof}

\Cref{thm:lb} immediately gives us the following sample complexity lower bounds for the $(\epsilon, \delta)$-\WD problem for other voting rules.

\begin{corollary}\label{cor:lb}
 Every $(\epsilon, \delta)$-\WD algorithm needs at least
 $\frac{1}{4\epsilon^2}\ln \frac{1}{8e\sqrt{\pi}\delta}$ samples for
 any voting rule which reduces to the plurality rule for two
 candidates. In particular, the lower bound holds for approval,
 scoring rules, maximin, Copeland,  Bucklin, plurality with runoff,
 and STV voting rules.  
\end{corollary}

\begin{proof}
All the voting rules mentioned in the statement except the approval voting rule is same as the plurality voting rule for elections with two candidates. 
Hence, the result follows immediately from \Cref{thm:lb} for the above voting rules except the approval voting rule. 
The result for the approval voting rule follows from the fact that any arbitrary plurality election is also a valid approval election where 
every voter approves exactly one candidate.
\end{proof}

 We derive stronger lower bounds in terms of $m$ by explicitly viewing
 the $(\eps,\delta)$-\WD problem as a {\em statistical
   classification} problem. In this problem, we are given a black box
 that contains a distribution $\mu$ which is guaranteed to be one of
 $\ell$ known distributions $\mu_1, \dots, \mu_\ell$. A {\em
   classifier} is a randomized oracle which has to determine the
 identity of $\mu$, where each oracle call produces a sample from
 $\mu$. At the end of its execution, the classifier announces a guess
 for the identity of $\mu$, which has to be correct with probability
 at least $1-\delta$. Using information-theoretic methods, Bar-Yossef
 \cite{bar2003sampling} showed the following:
 
\begin{lemma}\label{lem:kldiv}
 The worst case sample complexity $q$ of a classifier $C$ for \el probability distributions $\mu_1,
 \ldots, \mu_\ell$ which does not make error with probability more
 than $\delta$ satisfies  following.
 \[  q \ge \Omega\left( \frac{\ln \ell}{JS\left( \mu_1, \ldots, \mu_\ell \right)} . \left( 1 - \delta \right) \right) \]
\end{lemma}

The connection with our problem is the following.   A
set $V$ of $n$ votes on a candidate set $\mathcal{C}$ generates a
probability distribution $\mu_V$ on $\mathcal{L}(\mathcal{C})$, where 
$\mu_V(\succ)$ is proportional to the number of voters who vote
$\succ$. Querying a random vote from $V$ is then equivalent to
sampling from the distribution $\mu_V$. The margin of victory is proportional to the minimum statistical distance between $\mu_V$ and $\mu_W$,
over all the voting profiles $W$ having a different winner than the
winner of $V$. 

Now suppose we have $m$ voting profiles $V_1, \dots, V_m$ having
different winners such that each $V_i$ has margin of 
victory at least $\eps n$. Any $(\eps,\delta)$-\WD
algorithm must also be a statistical classifier for $\mu_{V_1}, \dots,
\mu_{V_m}$ in the above sense. It then remains to construct such
voting profiles for various voting rules which we do in the proof of
the following theorem:

\begin{theorem}\label{thm:strlwb}
 Every $(\epsilon, \delta)$-\WD algorithm needs $\Omega \left( \frac{\ln m}{\epsilon^2}.\left( 1 - \delta \right) \right)$ samples for approval, Borda, Bucklin, and any Condorcet consistent voting rules, and $\Omega \left( \frac{\ln k}{\epsilon^2}.\left( 1 - \delta \right) \right)$ samples for the $k$-approval voting rule for $k\le cm$ for any constant $c\in(0,1)$.
\end{theorem}

\begin{proof}
 For each voting rule mentioned in the theorem, we will show $d$ ($d=k+1$ for the $k$-approval voting rule and $d=m$ for the rest of the voting rules) distributions $\mu_1, \ldots, \mu_d$ on the votes with the following properties. Let $V_i$ be an election where each vote $v\in \mathcal{L(C)}$ occurs exactly $\mu_i(v)\cdot n$ many times. Let $\mu = \frac{1}{d}\sum_{i=1}^d \mu_i$.
 \begin{enumerate}
  \item For every $i \ne j$, the winner in $V_i$ is different from the winner in $V_j$.
  \item For every $i$, the margin of victory of $V_i$ is $\Omega(\epsilon n)$.
  \item $D_{KL}(\mu_i || \mu) = O(\epsilon^2)$
 \end{enumerate}
 The result then follows from \Cref{lem:kldiv}. The distributions for different voting rules are as follows. Let the candidate set be $\mathcal{C} = \{ c_1, \ldots, c_m \}$. 
 
 \begin{itemize}

 \item \textbf{$k$-approval voting rule for $k\le cm$ for any constant $c\in(0,1)$:} Fix any arbitrary $M:= k+1$ many candidates $c_1, \ldots, c_M$. For $i \in [M]$, we define a distribution $\mu_i$ on all $k$ sized subsets of $\mathcal{C}$ (for the $k$-approval voting rule, each vote is a $k$-sized subset of $\mathcal{C}$) as follows. Each $k$ sized subset corresponds to top $k$ candidates in a vote.
 $$ \mu_i(x) = \begin{cases}
                (\nfrac{\epsilon}{{M-1 \choose k-1}}) + (\nfrac{1-\epsilon}{{M \choose k}})& \text{if } c_i\in x \text{ and } x\subseteq \{c_1, \ldots, c_{M}\}\\
                \nfrac{(1-\epsilon)}{{M \choose k}}& c_i\notin x \text{ and } x\subseteq \{c_1, \ldots, c_{M}\}\\
                0& \text{else}
               \end{cases}
 $$
 The score of $c_i$ in $V_i$ is $n\left( \eps + \nfrac{\left( 1 - \eps \right){M-1 \choose k-1}}{{M \choose k}} \right)$, the
 score of any other candidate $c_j \in \{c_1, \ldots, c_M\}\setminus\{c_i\}$ is $\nfrac{n\left( 1 - \eps \right){M-1 \choose k-1}}{{M \choose k}}$, and the score of the rest of the candidates is zero. Hence,
 the margin of victory is $\Omega(\epsilon n)$, since each
 vote change can reduce the score of $c_i$ by at most one and increase
 the score of any other candidate by at most one and $k\le cm$ for constant $c\in(0,1)$. This proves the
 result for the $k$-approval voting rule. Now we show that $D_{KL}(\mu_i || \mu)$ to be $O(\epsilon^2)$. 
 \begin{eqnarray*}
  D_{KL}(\mu_i || \mu) 
  &=& \left( \eps+\left(1-\eps\right)\frac{k}{M} \right)\ln\left( 1-\eps+\eps\frac{M}{k} \right) + \left( 1-\eps \right)\left( 1-\frac{k}{M} \right)\ln\left(1-\eps\right)\\
  &\le& \left( \eps+\left(1-\eps\right)\frac{k}{M} \right)\left( \eps\frac{M}{k} - \eps \right) - \left( 1-\eps \right)\left( 1-\frac{k}{M} \right)\eps\\
  &=& \eps^2 \left( \frac{M}{k}-1 \right)\\
  &\le& 2\eps^2
 \end{eqnarray*}
 
 \item \textbf{Approval voting rule:} The result follows from the fact that every $\frac{m}{2}$-approval election is also a valid approval election and \Cref{lem:approval}.
 
 \item \textbf{Borda and any Condorcet consistent voting rule:} The score vector for the Borda voting rule which we use in this proof is $(m, m-1, \ldots, 1)$. For $i \in [m]$, we define a distribution $\mu_i$ on all possible linear orders over $\mathcal{C}$ as follows.
 
 $$ \mu_i(x) = \begin{cases}
                \nfrac{2\epsilon}{m!} + \nfrac{(1-\epsilon)}{m!}& \text{if } c_i \text{ is within top } \frac{m}{2} \text{ positions in }x.\\
                \nfrac{(1-\epsilon)}{m!}& \text{else}
               \end{cases}
 $$
 
 The score of $c_i$ in $V_i$ is at least $n(\nfrac{3\eps m}{5}+\nfrac{(1-\eps)m}{2}) = n(\nfrac{m}{2}+\nfrac{\eps m}{10})$ whereas the score of any other candidate $c_j \ne c_i$ is $\nfrac{mn}{2}$. Hence, the margin of victory is at least $\nfrac{\epsilon n}{8}$, since each vote change can reduce the score of $c_i$ by at most $m$ and increase the score of any other candidate by at most $m$. Also, in the weighted majority graph for the election $V_i$, $w(c_i, c_j) \ge \nfrac{\epsilon n}{10}$. Hence, the margin of victory is at least $\nfrac{\epsilon n}{4}$, since each vote change can change the weight of any edge in the weighted majority graph by at most two. Now we show that $D_{KL}(\mu_i || \mu)$ to be $O(\epsilon^2)$.
 
 \begin{eqnarray*}
  D_{KL}(\mu_i || \mu) 
  &=& \frac{1+\eps}{2}\ln\left( 1+\eps \right) + \frac{1-\eps}{2}\ln\left( 1-\eps \right)\\
  &\le& \frac{1+\eps}{2}\eps - \frac{1-\eps}{2}\eps\\
  &=& \eps^2
 \end{eqnarray*}
 
 \item \textbf{Bucklin:} For $i \in [m]$, we define a distribution $\mu_i$ on all $\nfrac{m}{4}$ sized subsets of $\mathcal{C}$ as follows. Each $\nfrac{m}{4}$ sized subset corresponds to the top $\frac{m}{4}$ candidates in a vote.
 
 $$ \mu_i(x) = \begin{cases}
                (\nfrac{\eps}{{m-1 \choose \frac{m}{4}-1}}) + (\nfrac{(1-\epsilon)}{{m \choose \frac{m}{4}}})& \text{if } c_i\in x\\
                \nfrac{(1-\epsilon)}{{m \choose \frac{m}{4}}}& \text{else}
               \end{cases}
 $$
 
 The candidate $c_i$ occurs within the top $m(\frac{1}{2}-\frac{\eps}{10})$ positions at least $n(\frac{1}{2}-\frac{\eps}{10}+\eps) = n(\frac{1}{2}+\frac{9\eps}{10})$ times. On the other hand any candidate $c_j \ne c_i$ occurs within the top $m(\frac{1}{2}-\frac{\eps}{10})$ positions at most $n(\frac{1}{2}-\frac{\eps}{10})$ times. Hence the margin of victory is at least $\frac{\eps n}{30} = \Omega(\eps n)$. Now we show that $D_{KL}(\mu_i || \mu)$ to be $O(\epsilon^2)$. 
 \begin{eqnarray*}
  D_{KL}(\mu_i || \mu) 
  &=& \left(\eps + \frac{1-\eps}{4}\right)\ln\left(1+3\eps\right) + \frac{3}{4}\left(1-\eps\right)\ln\left(1-\eps\right)\\
  &=& \frac{1}{4}\left(3\eps\left(1+3\eps\right)-3\eps\left(1-\eps\right)\right)\\
  &=& 3\eps^2
 \end{eqnarray*}

%
%
 \end{itemize}
\end{proof}

\subsection{Results on Upper Bounds}\label{sec:upbd}

In this section, we present the upper bounds on the sample complexity
of the $(\epsilon, \delta)$-\WD problem for various
voting rules. The general framework for proving the upper bounds is as
follows. For each voting rule, we first prove a useful structural
property about the election when the margin of victory is known to be
at least $\epsilon n$. Then, we sample a few votes uniformly
at random to estimate either the  score of the candidates for score
based voting rules or weights of the edges in the weighted majority
graph for the voting rules which are defined using weighted majority graph (maximin and Copeland for example). Finally,  appealing to the structural
property that has been established, we argue that, the winner of the
election on the sampled votes will be the  same as the winner of the
election, if we are able to estimate either the scores of the candidates or
the weights of the edges in the weighted majority graph to a certain level of accuracy.  

Before getting into specific voting rules, we prove a straightforward
bound on the sample  complexity for the $(\epsilon, \delta)$-winner
determination problem for {\em any homogeneous} voting rule.  

\begin{theorem}\label{thm:gen}
 There is a $(\epsilon, \delta)$-\WD algorithm for every homogeneous voting rule with sample complexity $(\nfrac{9m!^2}{2\eps^2}) \ln(\nfrac{2m!}{\delta})$.
\end{theorem}

\begin{proof}
 We sample $\ell$ votes uniformly at random from the set of votes with
 replacement. For $x\in\LL(\CC)$, let $X_i^x$ be an indicator random variable that is $1$
 exactly when $x$ is the $i$'th sample, and let $g(x)$ be the total
 number of voters whose vote is $x$. Define $\hat{g}(x) =
 \frac{n}{\ell}\sum_{i=1}^{\ell}X_i^x$.  Using the Chernoff bound (\Cref{thm:chernoff}), we have the following:
 $$ \Pr\left[ |\hat{g}(x) - g(x)| \ge \frac{\epsilon n}{3m!} \right] \le
 2\cdot {\exp\left(- \frac{2\eps^2 \ell}{9m!^2}\right)}$$
 By using the union bound, we have the following,
 \begin{eqnarray*}
  \Pr\left[\exists x\in \mathcal{L(C)}, |\hat{g}(x) - g(x)| >
    \frac{\epsilon n}{3m!} \right] &\le& 2m!\cdot \exp\left(-\frac{2\eps^2 \ell}{9m!^2}\right)
 \end{eqnarray*}
 Since the margin of victory is $\epsilon n$ and the voting rule is anonymous, the winner of the $\ell$
 sample votes will be same as the winner of the election if
 $|\hat{g}(x) - g(x)| \le  \nfrac{\epsilon n}{3m!}$ for every linear
 order $x\in \mathcal{L(C)}$. Hence, it is enough to take $\ell =
 (\nfrac{9m!^2}{2\eps^2}) \ln(\nfrac{2m!}{\delta})$. 
\end{proof}

\subsubsection{Approval Voting Rule}

We derive the upper bound on the sample complexity for the $(\epsilon, \delta)$-\WD problem for the approval voting rule.

\begin{lemma}\label{lem:approval}
 If $ \textsf{\textsf{MOV}} \ge \epsilon n $ and $w$ be the winner of an approval election, then, 
 $ s(w) - s(x) \ge \epsilon n, $ for every candidate $ x \ne w $, where $s(y)$ is the number of approvals that a candidate $y$ receives.
\end{lemma}

\begin{proof}
 Suppose there is a candidate $ x \ne w $ such that $ s(w) - s(x) < \epsilon n$. Then there must exist 
 $\epsilon n - 1$ votes which does not approve the candidate $x$. We modify these votes to make it 
 approve $x$. This makes $w$ not the unique winner in the modified election. This contradicts the fact that the \textsf{MOV} 
 is at least $\epsilon n$.
\end{proof}

\begin{theorem}\label{thm:app}
 There is a $(\epsilon, \delta)$-\WD algorithm for the approval voting rule with sample complexity at most $\nfrac{9\ln (\nfrac{2m}{\delta})}{2\eps^2}$.
\end{theorem}

\begin{proof}
Suppose $w$ is the winner. We sample $\ell$ votes uniformly at random
from the set of votes with replacement. For a candidate $x$, let $X_i^x$
be a random variable  indicating whether the $i$'th vote sampled
approved $x$. Define $\hat{s}(x) =
\frac{n}{\ell}\sum_{i=1}^{\ell}X_i^x$. Then,  by an argument analogous to the
proof of \Cref{thm:gen}, $\Pr[\exists x \in \mathcal{C}, |\hat{s}(x)
-s(x)| \ge \nfrac{\eps n}{3}]\leq  2m\cdot \exp\left(-\nfrac{2\eps^2\ell}{9}\right)$. Thus
since \textsf{MOV}$\geq \eps n$ and by \Cref{lem:approval}, if we  take $\ell = \nfrac{9\ln (\nfrac{2m}{\delta})}{2	\eps^2}$, $\hat{s}(w)$ is
greater than $\hat{s}(x)$ for  all $x \neq w$.
\end{proof}

\subsubsection{Scoring Rules}

Now we move on to the scoring rules. Again, we first establish a
structural consequence of having large \textsf{MOV}.

\begin{lemma}\label{lem:scr}
 Let $\alpha = (\alpha_1, \dots, \alpha_m)$ be any normalized score vector (hence, $\alpha_m=0$). If $w$ and $z$ are the candidates that receive highest and second highest score respectively in a $\alpha$--scoring rule election instance $\mathcal{E}=(V, C)$ and $\MOV[\alpha]$ is the margin of victory of \EE, then,
 $$ \alpha_1(\MOV[\alpha]-1) \le s(w) - s(z) \le 2\alpha_1\MOV[\alpha] $$
\end{lemma}

\begin{proof}
 We claim that there must be at least $\MOV[\alpha]-1$ many votes $v\in V$ where $w$ is preferred over $z$. Indeed, otherwise, we swap $w$ and $z$ in  all the votes where $w$ is preferred over $z$. This makes $z$ win the election. However, we have changed at most $\MOV[\alpha]-1$ votes only. This contradicts the definition of margin of victory (see \Cref{def:margin of victory}). Let $v\in V$ be a vote where $w$ is preferred over $z$. Let $\alpha_i$ and $\alpha_j(\le \alpha_i)$ be the scores received by the candidates $w$ and $z$ respectively from the vote $v$. We replace the vote $v$ by $v^{\prime} = z \succ \cdots \succ c$. This vote change reduces the value of $s(w)-s(z)$ by $\alpha_1 + \alpha_i - \alpha_j$ which is at least $\alpha_1$. Hence, $\alpha_1(\MOV[\alpha]-1) \le s(w) - s(z)$. Each vote change reduces the value of $s(w)-s(z)$ by at most $2\alpha_1$ since $\alpha_m=0$. Hence, $s(w) - s(z) \le 2\alpha_1\MOV[\alpha]$.
\end{proof}

Using \Cref{lem:scr}, we prove the following sample complexity upper bound for the $(\epsilon, \delta)$-\WD problem for the scoring rules.

\begin{theorem}\label{thm:scr}
 Suppose $\alpha = (\alpha_1, \dots, \alpha_m)$ be a normalized score
 vector. There is a $(\epsilon, \delta)$-\WD
 algorithm for the  $\alpha$-scoring rule with sample complexity at most
 $\nfrac{9\ln (\nfrac{2m}{\delta})}{2\eps^2}$. 
\end{theorem}

\begin{proof}
 We sample $\ell$ votes uniformly
 at random from the set of votes with replacement. For a candidate
 $x$, define $X_i = \nfrac{\alpha_i}{\alpha_1}$ if
 $x$ gets a score of $\alpha_i$ from the  $i${th} sample vote, and let
 $\hat{s}(x) = \frac{n\alpha_1}{\ell}\sum_{i=1}^\ell X_i$.  Now using
 Chernoff bound (\Cref{thm:chernoff}), we have:
 $$ \Pr\left[ \left|\hat{s}(x) - s(x)\right| \ge \alpha_1 \epsilon n/3\right] \le
 2\exp\left(-\frac{2\epsilon^2  \ell}{9}\right)$$
The rest of the proof follows from  an argument analogous to the proof of
 \Cref{thm:app} using \Cref{lem:scr}. 
\end{proof}

From \Cref{thm:scr}, we have a $(\epsilon, \delta)$-winner
determination algorithm for the $k$-approval voting rule which needs
$\nfrac{9\ln (\nfrac{2m}{\delta})}{2\eps^2}$ many samples for any $k$. We now improve this bound to $\nfrac{9\ln(\nfrac{2k}{\delta})}{2\epsilon^2}$ for the $k$-approval voting rule. Before embarking  on the proof of the above fact, we prove
the following lemma which we will use crucially in \Cref{thm:kapp}.

\begin{lemma}\label{lem:funmax}
 Let $f : \mathbb{R} \longrightarrow \mathbb{R}$ be a function defined by $f(x) = e^{-\nfrac{\lambda}{x}}$. Then, 
 \[  f(x) + f(y) \le f(x+y), \text{ for } x,y > 0, \frac{\lambda}{x+y} > 2, x < y \]
\end{lemma}

\begin{proof} For the function $f(x)$, we have following.
 \begin{eqnarray*}
  f(x) &=& e^{-\nfrac{\lambda}{x}} \\
  \Rightarrow f^{\prime}(x) &=& \frac{\lambda}{x^2} e^{-\nfrac{\lambda}{x}}\\
  \Rightarrow f^{\prime\prime}(x) &=& \frac{\lambda^2}{x^4} e^{-\nfrac{\lambda}{x}} - \frac{2\lambda}{x^3} e^{-\nfrac{\lambda}{x}}
 \end{eqnarray*}
 Hence, for $x,y > 0, \frac{\lambda}{x+y} > 2, x < y$ we have $f^{\prime\prime}(x), f^{\prime\prime}(y), f^{\prime\prime}(x+y) > 0$. This implies following for $x < y$ and an infinitesimal positive $\delta$.
 \begin{eqnarray*}
  f^{\prime}(x) &\le& f^{\prime}(y)\\
  \Rightarrow \frac{f(x-\delta) - f(x)}{\delta} &\ge& \frac{ f(y) - f(y-\delta)}{\delta} \\
  \Rightarrow f(x) + f(y) &\le& f(x-\delta) + f(y+\delta)\\
  \Rightarrow f(x) + f(y) &\le& f(x+y)
 \end{eqnarray*}
\end{proof}

We now present our $(\epsilon, \delta)$-\WD algorithm for the $k$-approval voting rule.

\begin{theorem}\label{thm:kapp}
 There is a $(\epsilon, \delta)$-\WD algorithm for the $k$-approval voting rule with sample complexity at most $\nfrac{9\ln(\nfrac{2k}{\delta})}{2\epsilon^2}$.
\end{theorem}

\begin{proof}
We sample $\ell$ votes uniformly at random from the set of votes with
replacement. For a candidate $x$, let $X_i^x$ be a random variable
indicating whether $x$ is among the top $k$ candidates for the $i^{th}$
vote sample. Define $\hat{s}(x) = \frac{n}{\ell}\sum_{i=1}^{\ell}X_i^x$,
and let $s(x)$ be the actual score of $x$. Then by the multiplicative Chernoff bound
(\Cref{thm:chernoff}), we have:
 $$ \Pr\left[ |\hat{s}(x) - s(x)| > \nfrac{\epsilon n}{3} \right] \le
 2\exp\left(-\frac{2\epsilon^2 \ell n}{9 s(x)}\right)$$
 By union bound, we have the following,
 \begin{eqnarray*}
  && \Pr[ \exists x\in \mathcal{C}, |\hat{s}(x) - s(x)| > \nfrac{\epsilon n}{3} ]\\
  &\le& \sum_{x\in \mathcal{C}} 2\exp\left(-\nfrac{2\epsilon^2 \ell n}{9
      s(x)}\right) \\
  &\le& 2k\exp\left(-\nfrac{2\epsilon^2 \ell}{9}\right)
 \end{eqnarray*}
 Let the candidate $w$ be the winner of the election. The second
 inequality in the above derivation follows from the fact that, the
 function $\sum_{x\in \mathcal{C}} 2 {\exp\left(-\nfrac{2\epsilon^2 \ell n}{9 s(x)}\right)}$ is maximized in the domain, defined by the constraint: for every candidate $x\in \mathcal{C}$, $s(x) \in [0,n]$ and $\sum_{x\in\mathcal{C}} s(x) = kn$, by setting $s(x)=n$ for every $x \in \mathcal{C}^\prime$ and $s(y)=0$ for every $y \in \mathcal{C}\setminus\mathcal{C}^\prime$, for any arbitrary subset $\mathcal{C}^\prime \subset \mathcal{C}$ of cardinality $k$ (due to \Cref{lem:funmax}). The rest of the proof follows by an argument analogous to the proof of \Cref{thm:gen} using \Cref{lem:scr}.
\end{proof}

Notice that, the sample complexity upper bound in \Cref{thm:kapp} is independent of $m$ for the plurality voting rule. \Cref{thm:kapp} in turn implies the following Corollary which we consider to be of independent interest.

\begin{corollary}\label{cor:linfty}
 There is an algorithm to estimate the $\ell_\infty$ norm $\ell_\infty(\mu)$ of a distribution $\mu$ within an additive factor of $\eps$ by querying only $\nfrac{9\ln(\nfrac{2}{\delta})}{2\epsilon^2}$ many samples, if we are allowed to get i.i.d. samples from the distribution $\mu$.
\end{corollary}

Such a statement seems to be folklore in the statistics community \cite{DKW}. Recently in an independent and nearly simultaneous work, Waggoner \cite{Waggoner2015} obtained a sharp bound of $\frac{4}{\eps^2} \ln(\frac{1}{\delta})$ for the sample complexity in \Cref{cor:linfty}.

We now turn our attention to the $k$-veto voting rule. For the $k$-veto voting rule, we have the following result for the $(\epsilon, \delta)$-\WD problem.

\begin{theorem}\label{thm:kveto}
 There is a $(\epsilon, \delta)$-\WD algorithm for the $k$-veto voting rule with sample complexity at most $\nfrac{9\ln(\nfrac{2k}{\delta})}{2\epsilon^2}$.
\end{theorem}

\begin{proof}
 We first observe that, since the margin of victory of the input election is at least $\eps n$, every candidate other than the winner must receive at least $\eps n$ vetoes. Hence we have the following.
 
 $$\eps n (m-1) \ge kn \text{ i.e. } m-1\le \nfrac{k}{\eps}$$
 
 Let us sample $\ell$ votes uniformly at random from the set of votes with
replacement. For a candidate $x$, let $X_i^x$ be a random variable
indicating whether $x$ is among the bottom $k$ candidates for the $i^{th}$
vote sample. Define $\hat{s}(x) = \frac{n}{\ell}\sum_{i=1}^{\ell}X_i^x$,
and let $s(x)$ be the actual score of $x$. Then by the multiplicative Chernoff bound
(\Cref{thm:chernoff}), we have:
 $$ \Pr\left[ |\hat{s}(x) - s(x)| > \nfrac{\epsilon n}{3} \right] \le
 2\exp\left(-\frac{2\epsilon^2 \ell n}{9 s(x)}\right)$$
 By union bound, we have the following,
 \begin{eqnarray*}
  \Pr[ \exists x\in \mathcal{C}, |\hat{s}(x) - s(x)| > \nfrac{\epsilon n}{3} ]
  &\le& \sum_{x\in \mathcal{C}} 2\exp\left(-\nfrac{2\epsilon^2 \ell n}{9
      s(x)}\right) \\
  &\le& 2k\exp\left(-\nfrac{2\epsilon^2 \ell}{9}\right)
 \end{eqnarray*}
 Let the candidate $w$ be the winner of the election. The second
 inequality in the above derivation follows from the fact that, the
 function $\sum_{x\in \mathcal{C}} 2 {\exp\left(-\nfrac{2\epsilon^2 \ell n}{9 s(x)}\right)}$ is maximized in the domain, defined by the constraint: for every candidate $x\in \mathcal{C}$, $s(x) \in [0,n]$ and $\sum_{x\in\mathcal{C}} s(x) = kn$, by setting $s(x)=n$ for every $x \in \mathcal{C}^\prime$ and $s(y)=0$ for every $y \in \mathcal{C}\setminus\mathcal{C}^\prime$, for any arbitrary subset $\mathcal{C}^\prime \subset \mathcal{C}$ of cardinality $k$ (due to \Cref{lem:funmax}). The rest of the proof follows by an argument analogous to the proof of \Cref{thm:gen} using \Cref{lem:scr}.
\end{proof}

\subsubsection{Maximin Voting Rule}

We now turn our attention to the maximin voting rule. The idea here is to sample enough number of votes such that we are able to estimate the weights of the edges in the weighted majority graph with certain level of accuracy which in turn 
leads us to predict winner.

\begin{lemma}\label{lem:maximin}
 Let $\mathcal{E}=(V,C)$ be any instance of a maximin election. If $w$ and $z$ are the candidates that receive highest and second highest maximin score respectively in $\mathcal{E}$ and $\MOV[maximin]$ is the margin of victory of \EE, then,
 $$ 2\MOV[maximin] \le s(w) - s(z) \le 4\MOV[maximin] $$
\end{lemma}

\begin{proof}
 Each vote change can increase the value of $s(z)$ by at most two and decrease the value of $s(w)$ by at most two. Hence, we have $s(w) - s(z) \le 4\MOV[maximin]$. Let $x$ be the candidate that minimizes $D_{\mathcal{E}}(w,x)$, that is, $x\in \argmin_{x\in C\setminus \{w\}}\{D_{\mathcal{E}}(w,x)\}$. Let $v\in V$ be a vote where $w$ is preferred over $x$. We replace the vote $v$ by the vote $v^{\prime} = z \succ x \succ \cdots \succ w$. This vote change reduces the score of $w$ by two and does not reduce the score of $z$. Hence, $s(w) - s(z) \ge 2\MOV[maximin]$.
\end{proof}

We now present our $(\epsilon, \delta)$-\WD algorithm for the maximin voting rule.

\begin{theorem}\label{thm:maximin}
 There is a $(\epsilon, \delta)$-\WD algorithm for the maximin voting rule with sample complexity $(\nfrac{9}{2\epsilon^2})\ln (\nfrac{2m}{\delta})$.
\end{theorem}

\begin{proof}
 Let $x$ and $y$ be any two arbitrary candidates. We sample $\ell$
 votes uniformly at random from the set of votes with
 replacement. Let $X_i$ be a random variable  defined as follows.
 $$ X_i = \begin{cases}
           1,& \text{if } x\succ y \text{ in the } i^{th} \text{ sample}\\
           -1,& \text{else}
          \end{cases}
 $$
Define $\hat{D}(x,y) = \frac{n}{\ell}\sum_{i=1}^{\ell}X_i$. 
 We estimate $\hat{D}(x,y)$ within the closed ball of radius $\epsilon
 n/2$ around $D(x,y)$ for every candidates $x, y\in \mathcal{C}$ and
 the rest of the proof  follows from  by an argument analogous to the
 proof of \Cref{thm:app} using \Cref{lem:maximin}. 
\end{proof}

\subsubsection{Copeland Voting Rule}

Now we move on to the Copeland$^\alpha$ voting rule. The approach for the Copeland$^\alpha$ voting rule is similar
to the maximin voting rule. However, it turns out that we need to
estimate the edge weights of the weighted majority graph more
accurately for the Copeland$^\alpha$ voting rule. Xia introduced the quantity called the {\em relative margin of victory} (see Section 5.1 in \cite{xia2012computing}) which we will use crucially for showing sample complexity upper bound for the Copeland$^\alpha$ voting rule. Given an election, a candidate $x\in C$, and an integer (may be negative also) $t$, $s^\prime_t(V, x)$ is defined as follows. 

$$s^\prime_t(V, x) = |\{ y\in C: y\ne x, D(y,x)<2t \}| + \alpha|\{ y\in C: y\ne x, D(y,x)=2t \}|$$

For every two distinct candidates $x$ and $y$, the relative margin of victory, denoted by $RM(x,y)$, between $x$ and $y$ is defined as the minimum integer $t$ such that, $s^\prime_{-t}(V, x) \le s^\prime_t(V, y)$. Let $w$ be the winner of the election $\mathcal{E}$. We define a quantity $\Gamma(\mathcal{E})$ to be $\min_{x\in C\setminus\{w\}} \{RM(w,x)\}$. Notice that, given an election $\mathcal{E}$, $\Gamma(\mathcal{E})$ can be computed in polynomial amount of time. Now we have the following lemma.

\begin{lemma}\label{lem:copeland}
Suppose $ \textsf{MOV} \ge \epsilon n $ and $w$ be the winner of a Copeland$^\alpha$ election. Then, 
 $ RM(w, x) \ge \nfrac{\epsilon n}{(2(\ceil*{\ln m} +1))}, $ for every candidate $ x \ne w $.
\end{lemma}

\begin{proof}
 Follows from Theorem 11 in \cite{xia2012computing}.
\end{proof}

\begin{theorem}\label{thm:copeland}
 There is a $(\epsilon, \delta)$-\WD algorithm for the Copeland$^\alpha$ voting rule with sample complexity $(\nfrac{25}{2\epsilon^2}) \ln^3 (\nfrac{2m}{\delta})$.
\end{theorem}

\begin{proof}
 Let $x$ and $y$ be any two arbitrary candidates and $w$ the Copeland$^\alpha$ winner of the election. We estimate $D(x,y)$ within the closed ball of radius $\nfrac{\epsilon n}{(5(\ceil*{\ln m} +1))}$ around $D(x,y)$ for every candidates $x, y\in \mathcal{C}$ in a way analogous to the proof of \Cref{thm:maximin}. This needs $(\nfrac{25}{2\epsilon^2}) \ln^3 (\nfrac{2m}{\delta})$ many samples. The rest of the proof follows from \Cref{lem:copeland} by an argument analogous to the proof of \Cref{thm:gen}. 
\end{proof}

\subsubsection{Bucklin Voting Rule}

For the Bucklin voting rule, we will estimate how many times each
candidate occurs within the first 
$k$ position for every $k\in [m]$. This eventually leads us to predict the winner of the election due to the following lemma.

\begin{lemma}\label{lem:bucklin}
 Suppose \textsf{MOV} of a Bucklin election be at least $\epsilon n$. Let $w$ be the winner of the election and $x$ be any 
 arbitrary candidate other than $w$. Suppose
 $$ b_w = \min_i \{ i : w \text{ is within top i places in at least } \nfrac{n}{2} + \nfrac{\epsilon n}{3} \text{ votes} \} $$
 $$ b_x = \min_i \{ i : x \text{ is within top i places in at least } \nfrac{n}{2} - \nfrac{\epsilon n}{3} \text{ votes} \} $$
 Then, $ b_w < b_x $.
\end{lemma}

\begin{proof}
 We prove it by contradiction. So, assume $ b_w \ge b_x $. Now by changing $\nfrac{\epsilon n}{3}$ votes, we can make the Bucklin score of $w$ to be at least $b_w$. By changing another $\nfrac{\epsilon n}{3}$ votes, we can make the Bucklin score of $x$ to be at most $b_x$. Hence, by changing $\nfrac{2 \epsilon n}{3}$ votes, it is possible not to make $w$ the unique winner which contradicts the fact that the \textsf{MOV} is at least $\epsilon n$.
\end{proof}

Our $(\epsilon, \delta)$-\WD algorithm for the Bucklin voting rule is as follows.

\begin{theorem}\label{thm:bucklin}
 There is a $(\epsilon, \delta)$-\WD algorithm for the Bucklin voting rule with sample complexity $(\nfrac{9}{2\epsilon^2})\ln (\nfrac{2m}{\delta})$.
\end{theorem}

\begin{proof}
 Let $x$ be any arbitrary candidate and $1\le k\le m$. We sample $\ell$ votes uniformly at random from the set of votes with replacement. Let $X_i$  be a random variable defined as follows.
 $$ X_i = \begin{cases}
           1,& \text{if } x \text{ is within top } k \text{ places in } i^{th} \text{ sample}\\
           0,& \text{else}
          \end{cases}
 $$
 Let $\hat{s}_k(x)$ be the estimate of the number of times the candidate $x$ has been placed within top $k$ positions. 
 That is, $\hat{s}_k(x) = \frac{n}{\ell} \sum_{i=1}^{\ell} X_i$. Let
 $s_k(x)$ be the number of times the candidate $x$ been placed in top
 $k$ positions. Clearly, $E[\hat{s}_k(x)] = \frac{n}{\ell} \sum_{i=1}^{\ell} E[X_i] = s_k(x) $.
 We estimate $\hat{s}_k(x)$ within the closed ball of radius $\nfrac{\epsilon
 n}{3}$ around $s_k(x)$ for every candidate $x \in \mathcal{C}$ and
 every integer $k\in [m]$,  and the rest of the proof follows from  by
 an argument analogous to the proof of \Cref{thm:app} using \Cref{lem:bucklin}.
\end{proof}

\subsubsection{Plurality with Runoff Voting Rule}

Now we move on to the plurality with runoff voting rule. In this case, we first estimate the plurality score of each of the candidates. In the next round, we estimate the pairwise margin of victory of the two candidates that qualifies to the second round.

\begin{lemma}\label{lem:runoff}
 Suppose $ \textsf{MOV} \ge \epsilon n $, and $w$ and $r$ be the winner and runner up of a plurality with runoff election respectively, and $x$ be any arbitrary candidate other than and $r$. Then, following holds. Let $s(.)$ denote plurality score of candidates. Then following holds.
 \begin{enumerate}
  \item $D(w,r) > 2 \epsilon n$.
  \item For every candidate $x \in \mathcal{C}\setminus\{w,r\}$, $ 2 s(w) > s(x) + s(r) +\epsilon n$.
  \item If $s(x) > s(r) - \nfrac{\epsilon n}{2}$, then $D(w,x) > \nfrac{\epsilon n}{2}$.
 \end{enumerate}
\end{lemma}

\begin{proof}
 If the first property does not hold, then by changing $\epsilon n$ votes, we can make $r$ winner.
 If the second property does not hold, then by changing $\epsilon n$ votes, we can make both $x$ and $r$ qualify to the second round.
 If the third property does not hold, then by changing $\nfrac{\epsilon n}{2}$ votes, the candidate $x$ can be sent to the second round of the runoff election. By changing another $\nfrac{\epsilon n}{2}$ votes, $x$ can be made to win the election. This contradicts the \textsf{MOV} assumption. 
\end{proof}

Now we present our $(\epsilon, \delta)$-\WD algorithm for the plurality with runoff voting rule.

\begin{theorem}\label{thm:runoff}
 There is a $(\epsilon, \delta)$-\WD algorithm for the plurality with runoff voting rule with sample complexity $(\nfrac{27}{\epsilon^2})\ln (\nfrac{4}{\delta})$.
\end{theorem}

\begin{proof}
 Let $x$ be any arbitrary candidate. We sample $\ell$ votes uniformly at random from the set of votes with replacement. Let, $X_i$ be a random variable defined as follows.
 $$ X_i = \begin{cases}
           1,& \text{if } x \text{ is at first position in the } i^{th} \text{ sample}\\
           0,& \text{else}
          \end{cases}
 $$
 The estimate of the plurality score of $x$ be $\hat{s}(x)$. Then $\hat{s}(x) = \frac{n}{\ell}\sum_{i=1}^{\ell}X_i$. Let $s(x)$ be the actual plurality score of $x$. Then we have following,
 $$ E[X_i] = \frac{s(x)}{n}, E[ \hat{s}(x) ] = \frac{n}{\ell} \sum_{i=1}^{\ell} E[ X_i ] = s(x)$$
 By Chernoff bound, we have the following,
 $$ \Pr[ |\hat{s}(x) - s(x)| > \epsilon n ] \le \frac{2}{\exp\{\epsilon^2 \ell n/3s(x)\}}$$
 By union bound, we have the following,
 \begin{eqnarray*}
  \Pr[ \exists x\in \mathcal{C}, |\hat{s}(x) - s(x)| > \epsilon n ] &\le& \sum_{x\in \mathcal{C}} \frac{2}{\exp\{\epsilon^2 ln/3s(x)\}}\\
 &\le& \frac{2}{\exp\{\epsilon^2 \ell/3\}}
 \end{eqnarray*}
 The last line follows from \Cref{lem:funmax}. Notice that, we do not need the random variables $\hat{s}(x)$ and $\hat{s}(y)$ to be independent for any two candidates $x$ and $y$. Hence, we can use the same $\ell$ sample votes to estimate $\hat{s}(x)$ for every candidate $x$.
 
 Now let $y$ and $z$ be the two candidates that go to the second round.
 $$ Y_i = \begin{cases}
           1,& \text{if } y\succ z \text{ in the } i^{th} \text{ sample}\\
           -1,& \text{else}
          \end{cases}
 $$
 The estimate of $D(y,z)$ be $\hat{D}(y,z)$. Then $\hat{D}(y,z) = \frac{n}{\ell}\sum_{i=1}^{\ell}Y_i$. Then we have following,
 $$ E[Y_i] = \frac{D(y,z)}{n}, E[ \hat{D}(y,z) ] = \frac{n}{\ell} \sum_{i=1}^{\ell} E[ Y_i ] = D(y,z)$$
 By Chernoff bound, we have the following,
 $$ \Pr[ |\hat{D}(y,z) - D(y,z)| > \epsilon n ] \le \frac{2}{\exp\{\epsilon^2 \ell/3\}}$$
 Let $A$ be the event that $\forall x\in \mathcal{C}, |\hat{s}(x) - s(x)| \le \epsilon n$ and $ |\hat{D}(y,z) - D(y,z)| \le \epsilon n$. Now we have,
 $$\Pr[A ] \ge 1 - (\frac{2}{\exp\{\epsilon^2 \ell/3\}} + \frac{2}{\exp\{\epsilon^2 \ell/3\}})$$
 Since we do not need independence among the random variables $\hat{s}(a)$, $\hat{s}(b)$, $\hat{D}(w,x)$, $\hat{D}(y,z)$ for any candidates $a, b, w, x, y,$ and $z$, we can use the same $\ell$ sampled votes. Now from \Cref{lem:runoff}, if $|\hat{s}(x) - s(x)| \le \nfrac{\epsilon n}{3}$ for every candidate $x$ and $|\hat{D}(y,z) - D(y,z)| \le \nfrac{\epsilon n}{3}$ for every candidates $y$ and $z$, then the plurality with runoff winner of the sampled votes coincides with the actual runoff winner. The above event happens with probability at least $1-\delta$ by choosing an appropriate $ \ell = (\nfrac{27}{\epsilon^2})\ln (\nfrac{4}{\delta})$.
\end{proof}

\subsubsection{STV Voting Rule}

Now we move on the STV voting rule. The 
following lemma provides an upper bound on the number of votes that need to be changed to make some arbitrary candidate win the election. 
More specifically, given a sequence of $m$ candidates $\{x_i\}_{i=1}^m$ with $x_m$ not being the winner, \Cref{lem:stv} below proves an upper bound on the number of number of votes that need to be modified such that the candidate $x_i$ gets eliminated at the $i^{th}$ round in the STV voting rule.

\begin{lemma}\label{lem:stv}
 Suppose $\mathcal{V}$ be a set of votes and $w$ be the winner of a STV election. Consider the following chain with candidates $x_1\ne x_2\ne \ldots \ne x_m$ and $ x_m \ne w $.
 $$ \mathcal{C}\supset \mathcal{C}\setminus\{x_1\}\supset \mathcal{C}\setminus\{x_1,x_2\}\supset \ldots \supset\{x_m\} $$
 Let $s_{\mathcal{V}}(A,x)$ be the plurality score of a candidate $x$ when all the votes in $\mathcal{V}$ are restricted to the set of candidates $A\subset \mathcal{C}$. Let us define $\mathcal{C}_{-i} = \mathcal{C}\setminus \{x_1, \ldots, x_i\}$ and $s^*_{\mathcal{V}}(A) := \min_{x\in A} \{s_{\mathcal{V}}(A,x)\}$. Then, we have the following.
 $$ \sum_{i=0}^{m-1} \left(s_{\mathcal{V}}\left({\mathcal{C}_{-i}},
     x_{i+1}\right) -
   s^*_{\mathcal{V}}\left({\mathcal{C}_{-i}}\right)\right) \ge
 \textsf{MOV}$$
\end{lemma}

\begin{proof}
 We will show that by changing $ \sum_{i=0}^{m-1} \left(s_{\mathcal{V}}\left({\mathcal{C}_{-i}}, x_{i+1}\right) - s^*_{\mathcal{V}}\left({\mathcal{C}_{-i}}\right)\right) $ votes, we can make the candidate $x_m$ winner. If $x_1$ minimizes $s_{\mathcal{V}}(\mathcal{C},x)$ over $x\in \mathcal{C}$, then we do not change anything and define $\mathcal{V}_1 = \mathcal{V}$. Otherwise, there exist $s_{\mathcal{V}}(\mathcal{C},x_1) - s^*_{\mathcal{V}}(\mathcal{C})$ many votes of following type.
 $$ x_1\succ a_1\succ a_2\succ \ldots \succ a_{m-1}, a_i\in \mathcal{C}, \forall 1\le i\le m-1 $$
 We replace $s_{\mathcal{V}}(\mathcal{C},x_1) - s^*_{\mathcal{V}}(\mathcal{C})$ many votes of the above type by the votes as follows.
 $$ a_1\succ x_1\succ a_2\succ \ldots \succ a_{m-1} $$
 Let us call the new set of votes by $\mathcal{V}_1$. We claim that, $s_{\mathcal{V}}(\mathcal{C}\setminus {x_1}, x) = s_{\mathcal{V}_1}(\mathcal{C}\setminus {x_1}, x)$ for every candidate $x\in \mathcal{C}\setminus\{x_1\}$. Fix any arbitrary candidate $x\in \mathcal{C}\setminus\{x_1\}$. The votes in $\mathcal{V}_1$ that are same as in $\mathcal{V}$ contributes same quantity to both side of the equality. Let $v$ be a vote that has been changed as described above. If $x = a_1$ then, the vote $v$ contributes one to both sides of the equality. If $x \ne a_1$, then the vote contributes zero to both sides of the equality. Hence, we have the claim. We repeat this process for $(m-1)$ times. Let $\mathcal{V}_i$ be the set of votes after the candidate $x_i$ gets eliminated. Now in the above argument, by replacing $\mathcal{V}$ by $\mathcal{V}_{i-1}$, $\mathcal{V}_1$ by $\mathcal{V}_i$, the candidate set $\mathcal{C}$ by $\mathcal{C}\setminus \{x_1, \ldots, x_{i-1}\}$, and the candidate $x_1$ by the candidate $x_i$, we have the 
following.
 $$ s_{\mathcal{V}_{i-1}}(\mathcal{C}_{-i}, x) = s_{\mathcal{V}_i}(\mathcal{C}_{-i}, x) \forall x\in \mathcal{C}\setminus\{x_1, \ldots, x_i\}$$
 Hence, we have the following.
 $$ s_{\mathcal{V}}(\mathcal{C}_{-i}, x) = s_{\mathcal{V}_i}(\mathcal{C}_{-i}, x) \forall x\in \mathcal{C}\setminus\{x_1, \ldots, x_i\}$$
 In the above process, the total number of votes that are changed is $ \sum_{i=0}^{m-1} \left(s_{\mathcal{V}}\left({\mathcal{C}_{-i}}, x_{i+1}\right) - s^*_{\mathcal{V}}\left({\mathcal{C}_{-i}}\right)\right) $.
\end{proof}

We now use \Cref{lem:stv} to prove the following sample complexity upper bound for the $(\epsilon, \delta)$-\WD problem for the STV voting rule.

\begin{theorem}\label{thm:stv}
 There is a $(\epsilon, \delta)$-\WD algorithm for the STV voting rule with sample complexity $(\nfrac{3m^2}{\epsilon^2})(m\ln 2+\ln (\nfrac{2m}{\delta}))$.
\end{theorem}

\begin{proof}
 We sample $\ell$ votes uniformly at random from the set of votes with replacement and output the STV winner of those $\ell$ votes say $w^{\prime}$ as the winner of the election. Let, $w$ be the winner of the election. We will show that for $\ell =  (\nfrac{3m^2}{\epsilon^2})(m\ln 2+\ln (\nfrac{2m}{\delta}))$ for which $w=w^{\prime}$ with probability at least $1 - \delta$. Let $A$ be an arbitrary subset of candidates and $x$ be any candidate in $A$. Let us define a random variables $X_i, 1\le i\le \ell$ as follows.
 $$ X_i = \begin{cases}
           1,& \text{if } x \text{ is at top } i^{th} \text{ sample when restricted to } A\\
           0,& \text{else}
          \end{cases}
 $$
 Define another random variable $\hat{s}_{\mathcal{V}}(A,x) := \sum_{i=1}^\ell X_i$. Then we have, $E[\hat{s}_{\mathcal{V}}(A,x)] = s_{\mathcal{V}}(A,x)$. Now using Chernoff bound, we have the following,
 $$ \Pr[ |\hat{s}_{\mathcal{V}}(A,x) - s_{\mathcal{V}}(A,x)| > \frac{\epsilon n}{m} ] \le \frac{2}{\exp\{\nfrac{\epsilon^2 \ell}{3m^2}\}}$$
 Let $E$ be the event that $\exists A\subset \mathcal{C} \text{ and } \exists x\in A, |\hat{s}_{\mathcal{V}}(A,x) - s_{\mathcal{V}}(A,x)| > \frac{\epsilon n}{m}$. By union bound, we have,
 \begin{eqnarray*}
  \Pr[ \bar{E} ] &\ge& 1 - \frac{m2^{m+1}}{\exp\{\nfrac{\epsilon^2 \ell}{3m^2}\}}
 \end{eqnarray*}
 The rest of the proof follows by an argument analogous to the proof of \Cref{thm:gen} using \Cref{lem:stv}.
\end{proof}

\section{Results for Estimating Margin of Victory}\label{sec:mov}

In this section we present our results for the \textsc{$(c, \epsilon, \delta)$}--\MV problem.

\subsection{Results on Lower Bounds}\label{sec:lwb_mov}

Our lower bounds for the sample complexity of the $(c, \epsilon, \delta)$--\MV problem are derived from the information-theoretic lower bound for distinguishing two distributions. 

\begin{theorem}\label{thm:lb_mov}
 The sample complexity of the $(c, \epsilon, \delta)$--\MV problem for the plurality voting rule is at least $(\nfrac{(1-c)^2}{36\epsilon^2})\ln \left(\nfrac{1}{8e\sqrt{\pi}\delta}\right)$ for any $c \in [0,1)$.
\end{theorem}

\begin{proof}
Consider two vote distributions $X$ and $Y$, each over the candidate set $\{a,b\}$. In $X$, exactly $\frac{1}{2} +
\frac{6\eps + 2c/n}{1-c}$ fraction of voters prefer $a$ to $b$ and thus the margin of victory is $\frac{3\eps + c/n}{1-c} n$. 
In $Y$, exactly $\frac{1}{2}$ fraction of voters prefer $b$ to $a$ and thus the margin of victory is one. Any $(c, \epsilon, \delta)$--\MV algorithm $\mathcal{A}$ for the plurality voting rule gives us a distinguisher between $X$ and $Y$ with probability of error at most $2\delta$. This is so because, if the input to $\mathcal{A}$ is $X$ then, the output of $\mathcal{A}$ is less than $c+2\eps n$ with probability at most $\delta$, whereas, if the input to $\mathcal{A}$ is $Y$ then, the output of $\mathcal{A}$ is more than $c+\eps n$ with probability at most $\delta$. Now since $n$ can be arbitrarily large, we get the result.
\end{proof}

\Cref{thm:lb} immediately gives the following corollary.
\begin{corollary}\label{cor:lb_mov}
 For any $c \in [0,1)$, every $(c, \epsilon, \delta)$--\MV algorithm needs
 at least $(\nfrac{(1-c)^2}{36\epsilon^2})\ln \left(\nfrac{1}{8e\sqrt{\pi}\delta}\right)$ many samples for
 all voting rules which reduce to the plurality rule for two
 candidates. In particular, the lower bound holds for 
 scoring rules, approval, Bucklin, maximin, and Copeland$^\alpha$ voting rules.
\end{corollary}

We note that the lower bound results in \Cref{thm:lb} and \Cref{cor:lb} do not assume anything about the sampling strategy or the computational complexity of the estimator.

\subsection{Results on Upper Bounds}\label{sec:upbd_mov}

A natural approach for estimating the margin of victory of an election efficiently is to compute the margin of victory of a suitably small number of sampled votes. Certainly, it is not immediate that the samples chosen uniformly at random preserve the value of the margin of victory of the original election within some desired factor. Although it may be possible to formulate clever sampling strategies that tie into the margin of victory structure of the election, we will show that uniformly chosen samples are good enough to design algorithms for estimating the margin of victory for many common voting rules. Our proposal has the advantage that the sampling component of our algorithms are always easy to implement, and further, there is no compromise on the bounds in the sense that they are optimal for any constant number of candidates. 

Our algorithms involve computing a quantity (which depends on the voting rule under consideration) based on the sampled votes, which we argue to be a suitable estimate of the margin of victory of the original election. This quantity is not necessarily the margin of victory of the sampled votes. For scoring rules, for instance, we will use the sampled votes to estimate candidate scores, and we use the difference between the scores of the top two candidates (suitably scaled) as an estimate for the margin of victory. We also establish a relationship between scores and the margin of victory to achieve the desired bounds on the estimate. The overall strategy is in a similar spirit for other voting rules as well, although the exact estimates may be different. We now turn to a more detailed description.

\subsubsection{Scoring Rules and Approval Voting Rule}

We begin with showing that the margin of victory of any scoring rule based election can be estimated quite accurately by sampling only $\frac{12}{\epsilon^2}\ln\frac{2m}{\delta}$ many votes. An important thing to note is that, the sample complexity upper bound is independent of the score vector.

\begin{theorem}\label{thm:scr_mov}
 There is a polynomial time \textsc{$(\nfrac{1}{3}, \epsilon, \delta)$--MoV} algorithm for the scoring rules with sample complexity at most $(\nfrac{12}{\epsilon^2})\ln(\nfrac{2m}{\delta})$.
\end{theorem}

\begin{proof}
 Let $\alpha = (\alpha_1, \dots, \alpha_m)$ be any arbitrary normalized score vector and $\mathcal{E} = (V, C)$ an election instance. We sample $\ell$ (the value of $\ell$ will be chosen later) votes uniformly at random from the set of votes with replacement. For a candidate $x$, define a random variable $X_i(x) = \nfrac{\alpha_i}{\alpha_1}$ if $x$ gets a score of $\alpha_i$ from the  $i${th} sample vote. Define $\bar{s}(x) = \frac{n\alpha_1}{\ell}\sum_{i=1}^\ell X_i(x)$ the estimate of $s(x)$, the score of $x$.  Also define $\eps^{\prime} = \nfrac{\eps}{2}$. Now using Chernoff bound (\Cref{thm:chernoff}), we have the following.
 $$ \Pr\left[ \left|\bar{s}(x) - s(x)\right| \ge \alpha_1 \epsilon^{\prime} n\right] \le
 2\exp\left(-\frac{\epsilon^{\prime2}  \ell}{3}\right)$$
 We now use the union bound to get the following.
 
 \begin{align} \Pr[ \exists x\in C, |\bar{s}(x) - s(x)| > \alpha_1 \epsilon^{\prime} n ] \le 2m\exp\left(-\frac{\epsilon^{\prime2}  \ell}{3}\right) \label[ineq]{eqn:scr} 
 \end{align}
 
 Define $ \bar{M} \eqdef \nfrac{ (\bar{s}(\bar{w}) - \bar{s}(\bar{z}))}{1.5\alpha_1} $ the estimate of the margin of victory of the election $\mathcal{E}$ (and thus the output of the algorithm), where $\bar{w}\in \argmax_{x\in C}\{\bar{s}(x)\}$ and $\bar{z}\in \argmax_{x\in C\setminus\{\bar{w}\}}\{\bar{s}(x)\}$. We claim that, if $\forall x\in C, |\bar{s}(x) - s(x)| \le \epsilon^{\prime} n$, then $|\bar{M}-\MOV[\alpha]| \le \nfrac{\MOV[\alpha]}{3} + \eps n$. This can be shown as follows.
 
 \begin{align*}
  \bar{M} - \MOV[\alpha] &= \frac{ \bar{s}(\bar{w}) - \bar{s}(\bar{z}) }{1.5\alpha_1} - \MOV[\alpha]\\
  &\le \frac{ s(w) - s(z) }{1.5\alpha_1} + \frac{2\eps^{\prime}n}{1.5} - \MOV[\alpha]\\
  &\le \frac{1}{3}\MOV[\alpha] + \eps n
 \end{align*}
 
  The second inequality follows from the fact that, $\bar{s}(\bar{w}) \le s(\bar{w})+\eps^\prime n \le s(w)+\eps^\prime n$ and $\bar{s}(\bar{z}) \ge \bar{s}(z) \ge s(z) - \eps^\prime n$. The third inequality follows from \Cref{lem:scr}. Similarly, we bound $\MOV[\alpha] - \bar{M}$ as follows.
  
 \begin{align*}
  \MOV[\alpha] - \bar{M} &= \MOV[\alpha] - \frac{ \bar{s}(w) - \bar{s}(z) }{1.5\alpha_1}\\
  &\le \MOV[\alpha] - \frac{ s(w) - s(z) }{1.5\alpha_1} + \frac{2\eps^{\prime}n}{1.5}\\
  &\le \frac{1}{3}\MOV[\alpha] + \eps n
 \end{align*}
 
  This proves the claim. Now we bound the success probability of the algorithm as follows. Let $A$ be the event that $\forall x\in C, |\bar{s}(x) - s(x)| \le \epsilon^{\prime} n$.
  
 \begin{align*}
  & \Pr\left[ |\bar{M}-\MOV[\alpha]| \le \frac{1}{3}\MOV[\alpha] + \eps n \right] \\
  &\ge \Pr\left[ |\bar{M}-\MOV[\alpha]| \le \frac{1}{3}\MOV[\alpha] + \eps n \middle| A \right]  \Pr[ A ]\\
  &= \Pr[ A ] \\
  &\ge 1 - 2m\exp\left(-\epsilon^{\prime2} \ell/3\right)
 \end{align*}
 
  The third equality follows from \Cref{lem:scr} and the fourth inequality follows from \cref{eqn:scr}. Now by choosing $ \ell = (\nfrac{12}{\epsilon^2})\ln(\nfrac{2m}{\delta})$, we get a \textsc{$(\nfrac{1}{3}, \epsilon, \delta)$--MoV} algorithm for the scoring rules.
\end{proof}

Now we show an algorithm for the \textsc{$(0, \epsilon, \delta)$--MoV} problem for the $k$-approval voting rule which not only provides more accurate estimate of the margin of victory, but also has a lower sample complexity. The following structural result will be used subsequently.

\begin{lemma}\label{lem:kapp_mov}
 Let $\mathcal{E}=(V,C)$ be an arbitrary instance of a $k$-approval election. If $w$ and $z$ are the candidates that receive highest and second highest score respectively in $\mathcal{E}$ and $\MOV[k-approval]$ is the margin of victory of \EE, then,
 $$ 2(\MOV[k-approval] - 1) < s(w) - s(z) \le 2\MOV[k-approval] $$
\end{lemma}

\begin{proof}
We call a vote $v\in V$ {\em favorable} if $w$ appears within the top $k$ positions and $z$ does not appear within top the $k$ positions in $v$. We claim that the number of favorable votes must be at least $\MOV[k-approval]$. Indeed, otherwise, we swap the positions of $w$ and $z$ in all the favorable votes while keeping the other candidates fixed. This makes the score of $z$ at least as much as the score of $w$ which contradicts the fact that the margin of victory is $\MOV[k-approval]$. Now notice that the score of $z$ must remain less than the score of $w$ even if we swap the positions of $w$ and $z$ in $\MOV[k-approval]-1$ many favorable votes, since the margin of victory is $\MOV[k-approval]$. Each such vote change increases the score of $z$ by one and reduces the score of $w$ by one. Hence, $2(\MOV[k-approval]-1) < s(w) - s(z)$. 
Again, since the margin of victory is $\MOV[k-approval]$, there exists a candidate $x$ other than $w$ and $\MOV[k-approval]$ many votes in $V$ which can be modified such that $x$ becomes a winner of the modified election. Now each vote change can reduce the score of $w$ by at most one and increase the score of $x$ by at most one. Hence, $ s(w) - s(x) \le 2\MOV[k-approval] $ and thus $ s(w) - s(z) \le 2\MOV[k-approval] $ since $s(z) \ge s(x)$.
\end{proof}

With \Cref{lem:kapp_mov,lem:funmax} at hand, we now describe our margin of victory estimation algorithm for the $k$-approval voting rule.

\begin{theorem}\label{thm:kapp_mov}
 There is a polynomial time \textsc{$(0, \epsilon, \delta)$--MoV} algorithm for the $k$-approval voting rule with sample complexity at most $(\nfrac{12}{\epsilon^2})\ln(\nfrac{2k}{\delta})$.
\end{theorem}

\begin{proof}
Let $\mathcal{E} = (V, C)$ be an arbitrary $k$-approval election. We sample $\ell$ votes uniformly at random from $V$ with
replacement. For a candidate $x$, define a random variable $X_i(x)$ which takes value $1$ if $x$ appears among the top $k$ candidates in the $i^{th}$
sample vote, and $0$ otherwise. Define $\bar{s}(x) \eqdef \frac{n}{\ell}\sum_{i=1}^{\ell}X_i(x)$ the estimate of the score of the candidate $x$,
and let $s(x)$ be the actual score of $x$. Also define $\eps^{\prime} = \frac{\eps}{2}$. Then by the Chernoff bound
(\Cref{thm:chernoff}), we have:

 $$ \Pr\left[ |\bar{s}(x) - s(x)| > \epsilon^{\prime} n \right] \le
 2\exp\left(-\frac{\epsilon^{\prime2} \ell n}{3 s(x)}\right)$$
 
Now we apply the union bound to get the following.

 \begin{align*}
  & \Pr[ \exists x\in C, |\bar{s}(x) - s(x)| > \epsilon^{\prime} n ]\\
  &\le \sum_{x\in C} 2\exp\left(-\frac{\epsilon^{\prime2} \ell n}{3s(x)}\right) \\
  &\le 2k\exp\left(-\epsilon^{\prime2} \ell/3\right)\numberthis \label[ineq]{eqn:kapp}
 \end{align*}
 
 The second inequality follows from \Cref{lem:funmax} : The expression $\sum_{x\in C} 2\exp\left(-\frac{\epsilon^{\prime2} \ell n}{3
 s(x)}\right)$ is maximized subject to the constraints that $ 0\le s(x)\le n, \forall x\in C$ and $ \sum_{x\in C} s(x) = kn $, when $ s(x) = n \forall x\in C^{\prime} $ for any subset of candidates $C^{\prime}\subseteq C$ with $|C^{\prime}|=k$ and $s(x)=0 \forall x\in C\setminus C^{\prime}$. 
 
 Now to estimate the margin of victory of the given election $\mathcal{E}$, let $\bar{w}$ and $\bar{z}$ be candidates with maximum and second maximum estimated score respectively. That is, $\bar{w} \in \argmax_{x\in C} \{ \bar{s}(x) \} \text{ and } \bar{z} \in \argmax_{x\in C\setminus \{\bar{w}\}} \{ \bar{s}(x) \}$. We define $ \bar{M} \eqdef \nfrac{(\bar{s}(\bar{w}) - \bar{s}(\bar{z}))}{2} $ the estimate of the margin of victory of the election $\mathcal{E}$ (and thus the output of the algorithm). Let $A$ be the event that $\forall x\in C, |\bar{s}(x) - s(x)| \le \epsilon^{\prime} n$. 
 We bound the success probability of the algorithm as follows.
 
 \begin{eqnarray*}
  && \Pr\left[ |\bar{M}-\MOV[k-approval]| \le \eps n \right] \\
  &\ge& \Pr\left[ |\bar{M}-\MOV[k-approval]| \le \eps n \middle| A \right] \Pr[ A ]\\
  &=& \Pr[ A ] \\
  &\ge& 1 - 2k\exp\left(-\epsilon^{\prime2} \ell/3\right)
 \end{eqnarray*} 
 
 The second equality follows from \Cref{lem:kapp_mov} and an argument analogous to the proof of \Cref{thm:scr_mov}. The third inequality follows from \cref{eqn:kapp}. Now by choosing $ \ell = (\nfrac{12}{\epsilon^2})\ln(\nfrac{2k}{\delta})$, we get a \textsc{$(0, \epsilon, \delta)$--MoV} algorithm.
\end{proof}

Note that, the sample complexity upper bound matches with the lower bound proved in \Cref{cor:lb_mov} for the $k$-approval voting rule when $k$ is a constant, irrespective of the number of candidates. 
Next, we estimate the margin of victory of an approval election.

\begin{theorem}\label{thm:app_mov}
 There is a polynomial time \textsc{$(0, \epsilon, \delta)$--MoV} algorithm for the approval rule with sample complexity at most $(\nfrac{12}{\epsilon^2})\ln(\nfrac{2m}{\delta})$.
\end{theorem}

\begin{proof}
 We estimate the approval score of every candidate within an additive factor of $\nfrac{\eps n}{2}$ by sampling $(\nfrac{12}{\epsilon^2})\ln(\nfrac{2m}{\delta})$ many votes uniformly at random with replacement and the result follows from an argument analogous to the proofs of \Cref{lem:kapp_mov} and \Cref{thm:kapp_mov}.
\end{proof}

\subsubsection{Bucklin Voting Rule}

Now we consider the Bucklin voting rule. Given an election $\mathcal{E}=(V,C)$, a candidate $x\in C$, and an integer $\ell \in [m]$, we denote the number of votes in $V$ in which $x$ appears within the top $\ell$ positions by $n_\ell(x)$. We prove useful bounds on the margin of victory of any Bucklin election in \Cref{lem:bucklin_mov}.

\begin{lemma}\label{lem:bucklin_mov}
 Let $\mathcal{E}=(V,C)$ be an arbitrary instance of a Bucklin election, $w$ the winner of $\mathcal{E}$, and $\MOV[Bucklin]$ the margin of victory of \EE. Let us define a quantity $\Delta(\mathcal{E})$ as follows.
 $$ \Delta(\mathcal{E}) \eqdef \min_{\substack{\ell \in [m-1] : n_\ell(w) > n/2,\\ x\in C\setminus\{w\} : n_\ell(x) \le n/2}} \{ n_\ell(w) - n_\ell(x) +1 \} $$
 Then,
 $$ \frac{\Delta(\mathcal{E})}{2} \le \MOV[Bucklin] \le \Delta(\mathcal{E}) $$
\end{lemma}

\begin{proof}
 Pick any $\ell\in[m-1]$ and $x\in C\setminus\{w\}$ such that, $n_\ell(w)>n/2$ and $n_\ell(x)\le n/2$. Now by changing $n_\ell(w)-\floor{n/2}$ many votes, we can ensure that $w$ is not placed within the top $\ell$ positions in more than $n/2$ votes: choose $n_\ell(w)-\floor{n/2}$ many votes where $w$ appears within top $\ell$ positions and swap $w$ with candidates placed at the last position in those votes. Similarly, by changing $\floor{n/2}+1-n_\ell(x)$ many votes, we can ensure that $x$ is placed within top $\ell$ positions in more than $n/2$ votes. Hence, by changing at most $n_\ell(w)-\floor{n/2}+\floor{n/2}+1-n_\ell(x) = n_\ell(w) - n_\ell(x) +1$ many votes, we can make $w$ not win the election. Hence, $\MOV[Bucklin]\le n_\ell(w) - n_\ell(x) +1$. Now since we have picked an arbitrary $\ell$ and an arbitrary candidate $x$, we have $\MOV[Bucklin] \le \Delta(\mathcal{E})$.
 
 For the other inequality, since the margin of victory is $\MOV[Bucklin]$, there exists an $\ell^\prime\in[m-1]$, a candidate $x\in C\setminus\{w\}$, and $\MOV[Bucklin]$ many votes in $V$ such that, we can change those votes in such a way that in the modified election, $w$ is not placed within top $\ell^\prime$ positions in more than $n/2$ votes and $x$ is placed within top $\ell^\prime$ positions in more than $n/2$ votes. Hence, we have the following.
 
 $ \MOV[Bucklin] \ge n_\ell^\prime(w)-\floor*{\frac{n}{2}}, \MOV[Bucklin]\ge \floor*{\frac{n}{2}}+1-n_\ell^\prime(x) $

 \begin{align*}
  \Rightarrow \MOV[Bucklin] &\ge \max\{n_{\ell^\prime}(w)-\floor*{\frac{n}{2}}, \floor*{\frac{n}{2}}+1-n_{\ell^\prime}(x)\}\\ 
  \Rightarrow  \MOV[Bucklin] &\ge\frac{n_{\ell^\prime}(w)-\floor*{\frac{n}{2}}+\floor*{\frac{n}{2}}+1-n_{\ell^\prime}(x)}{2} \\
  &\ge \frac{\Delta(\mathcal{E})}{2}
 \end{align*}
 
\end{proof}
 Notice that, given an election $\mathcal{E}$, $\Delta(\mathcal{E})$ can be computed in a polynomial amount of time. \Cref{lem:bucklin} leads us to the following result for the Bucklin voting rule.
 
\begin{theorem}\label{thm:bucklin_mov}
 There is a polynomial time \textsc{$(\nfrac{1}{3}, \epsilon, \delta)$--MoV} algorithm for the Bucklin rule with sample complexity $(\nfrac{12}{\epsilon^2})\ln(\nfrac{2m}{\delta})$.
\end{theorem}

\begin{proof}
 Similar to the proof of \Cref{thm:kapp_mov}, we estimate, for every candidate $x\in C$ and for every integer $\ell\in[m]$, the number of votes where $x$ appears within top $\ell$ positions within an approximation factor of $(0, \nfrac{\eps}{2})$. Next, we compute an estimate of $\bar{\Delta}(\mathcal{E})$ from the sampled votes and output the estimate for the margin of victory as $\bar{\Delta}(\mathcal{E})/1.5$. Using \Cref{lem:bucklin_mov}, we can argue the rest of the proof in a way that is analogous to the proofs of \Cref{thm:kapp,thm:scr}.
\end{proof}

\subsubsection{Maximin Voting Rule}

Next, we present our \textsc{$(\nfrac{1}{3}, \epsilon, \delta)$--MoV} algorithm for the maximin voting rule.

\begin{theorem}\label{thm:maximin_mov}
 There is a polynomial time \textsc{$(\nfrac{1}{3}, \epsilon, \delta)$--MoV} algorithm for the maximin rule with sample complexity $(\nfrac{24}{\epsilon^2})\ln(\nfrac{2m}{\delta})$.
\end{theorem}

\begin{proof}
 Let $\mathcal{E} = (V, C)$ be an instance of maximin election. Let $x$ and $y$ be any two candidates. We sample $\ell$
 votes uniformly at random from the set of all votes with replacement. 
 Let $X_i(x,y)$ be a random variable  defined as follows.
 $$ X_i(x,y) = \begin{cases}
           1,& \text{if } x\succ y \text{ in the } i^{th} \text{ sample vote}\\
           -1,& \text{else}
          \end{cases}
 $$
 Define $\bar{D_{\mathcal{E}}}(x,y) = \frac{n}{\ell}\sum_{i=1}^{\ell}X_i(x,y)$. By using the Chernoff bound and union bound, we have the following.
 \begin{equation*}
 \Pr\left[ \exists x,y\in C, |\bar{D_{\mathcal{E}}}(x,y) - D_{\mathcal{E}}(x,y)| > \epsilon n \right] \le 2m^2\exp\left(-\frac{\epsilon^2  \ell}{3}\right)
 \end{equation*}
 We define $\bar{M} \eqdef \nfrac{(\bar{s}(\bar{w})-\bar{s}(\bar{z}))}{3}$, the estimate of the margin of victory of $\mathcal{E}$, where 
 $\bar{w}\in \argmax_{x\in C}\{\bar{s}(x)\}$ and $\bar{z}\in \argmax_{x\in C\setminus\{\bar{w}\}}\{\bar{s}(x)\}$. 
 Now using \Cref{lem:maximin}, we can complete the rest of the proof in a way that is analogous to the proof of \Cref{thm:scr_mov}.
\end{proof}

\subsubsection{Copeland$^\alpha$ Voting Rule}

Now we present our result for the Copeland$^\alpha$ voting rule. The following lemma is immediate from Theorem 11 in \cite{xia2012computing}.

\begin{lemma}\label{lem:copeland_mov}
 $\Gamma(\mathcal{E}) \le \MOV[Copeland^\alpha] \le 2(\ceil*{\ln m} +1)\Gamma(\mathcal{E}).$
\end{lemma}

\begin{proof}
 Follows from Theorem 11 in \cite{xia2012computing}.
\end{proof}

\begin{theorem}\label{thm:copeland_mov}
 For the Copeland$^\alpha$ voting rule, there is a polynomial time \textsc{$\left(1-O\left(\nfrac{1}{\ln m}\right), \epsilon, \delta\right)$--MoV} algorithm whose sample complexity is $(\nfrac{96}{\epsilon^2})\ln(\nfrac{2m}{\delta})$.
\end{theorem}

\begin{proof}
 Let $\mathcal{E} = (V, C)$ be an instance of a Copeland$^\alpha$ election. For every $x, y\in C$, we compute $\bar{D_{\mathcal{E}}}(x, y)$, which is an estimate of $D_{\mathcal{E}}(x, y)$, within an approximation factor of $(0, \eps^\prime)$, where $\eps^\prime = \nfrac{\eps}{4}$. This can be achieved with an error probability at most $\delta$ by sampling $(\nfrac{96}{\epsilon^2})\ln(\nfrac{2m}{\delta})$ many votes uniformly at random with replacement (the argument is same as the proof of \Cref{thm:scr}). We define $\bar{s}^\prime_t(V, x) = |\{ y\in C: y\ne x, D_{\mathcal{E}}(y,x)<2t \}| + \alpha|\{ y\in C: y\ne x, D_{\mathcal{E}}(y,x)=2t \}|$. 
 We also define $\overline{RM}(x,y)$ between $x$ and $y$ to be the minimum integer $t$ such that, $\bar{s}^\prime_{-t}(V, x) \le s^\prime_t(V, y)$. 
 Let $\bar{w}$ be the winner of the sampled election, $\bar{z} = \argmin_{x\in C\setminus\{\bar{w}\}} \{\overline{RM}(w,x)\}$, $w$ the winner of $\mathcal{E}$, and $z = \argmin_{x\in C\setminus\{w\}} \{RM(w,x)\}$. Since, $\bar{D_{\mathcal{E}}}(x, y)$ is an approximation of $D_{\mathcal{E}}(x, y)$ within a factor of $(0, \eps^\prime)$, we have the following for every candidate $x, y\in C$.
 $$ s^\prime_t(V, x) -\eps^\prime n \le \bar{s}^\prime_t(V, x) \le s^\prime_t(V, x) +\eps^\prime n$$ 
 \begin{align*}
 RM(x,y) - 2\eps^\prime n \le \overline{RM}(x,y) \le RM(x,y) + 2\eps^\prime n \numberthis \label[ineq]{eqn:copeland}
 \end{align*}
 Define $\bar{\Gamma}(\mathcal{E}) = \overline{RM}(\bar{w}, \bar{z})$ to be the estimate of $\Gamma(\mathcal{E})$. We show the following claim.
 \begin{claim}\label{claim:copeland_mov}
  With the above definitions of $w, z, \bar{w},$ and $\bar{z}$, we have the following.
  $$ \Gamma(\mathcal{E}) - 4\eps^\prime n \le \bar{\Gamma}(\mathcal{E}) \le \Gamma(\mathcal{E}) + 4\eps^\prime n $$
 \end{claim}
 \begin{proof}
  Below, we show the upper bound for $\bar{\Gamma}(\mathcal{E})$.
  \begin{align*}
   \bar{\Gamma}(\mathcal{E}) = \overline{RM}(\bar{w}, \bar{z}) &\le \overline{RM}(w, \bar{z}) + 2\eps^\prime n\\
   &\le \overline{RM}(w, z) + 2\eps^\prime n\\
   &\le RM(w, z) + 4\eps^\prime n\\
   &= \Gamma(\mathcal{E}) + 4\eps^\prime n
  \end{align*}
  The second inequality follows from the fact that $\bar{D_{\mathcal{E}}}(x,y)$ is an approximation of $D_{\mathcal{E}}(x,y)$ by a factor of $(0, \eps^\prime)$. The third inequality follows from the definition of $\bar{z}$, and the fourth inequality uses \cref{eqn:copeland}. Now we show the lower bound for $\bar{\Gamma}(\mathcal{E})$.
  \begin{align*}
   \bar{\Gamma}(\mathcal{E}) = \overline{RM}(\bar{w}, \bar{z}) &\ge \overline{RM}(w, \bar{z}) - 2\eps^\prime n\\
   &\ge RM(w,\bar{z}) - 4\eps^\prime n\\
   &\ge RM(w,z) - 4\eps^\prime n\\
   &= \Gamma(\mathcal{E}) - 4\eps^\prime n
  \end{align*}
  The third inequality follows from \cref{eqn:copeland} and the fourth inequality follows from the definition of $z$.
 \end{proof}
  We define $\bar{M}$, the estimate of $\MOV[Copeland^\alpha]$, to be $\frac{4(\ln m + 1)}{2\ln m + 3}\bar{\Gamma}(\mathcal{E})$. The following argument shows that $\bar{M}$ is a $\left(1-O\left(\frac{1}{\ln m}\right), \epsilon, \delta\right)$--estimate of $\MOV[Copeland^\alpha]$.
  
 \begin{align*}
  &\bar{M} - \MOV[Copeland^\alpha]\\ 
  &= \frac{4(\ln m + 1)}{2\ln m + 3}\bar{\Gamma}(\mathcal{E}) - \MOV[Copeland^\alpha]\\
  &\le \frac{4(\ln m + 1)}{2\ln m + 3}\Gamma(\mathcal{E}) - \MOV[Copeland^\alpha] + \frac{16(\ln m + 1)}{2\ln m + 3}\eps^\prime n\\
  &\le \frac{4(\ln m + 1)}{2\ln m + 3}\MOV[Copeland^\alpha] - \MOV[Copeland^\alpha] + \eps n\\
  &\le \frac{2\ln m +1}{2\ln m +3}\MOV[Copeland^\alpha] + \eps n\\
  &\le \left(1-O\left(\frac{1}{\ln m}\right) \right)\MOV[Copeland^\alpha] + \eps n
 \end{align*}
 
 The second inequality follows from \Cref{claim:copeland_mov} and the  third inequality follows from \Cref{lem:copeland_mov}. Analogously, we have:

 \begin{align*} 
  &\MOV[Copeland^\alpha]-\bar{M}\\
  &= \MOV[Copeland^\alpha] - \frac{4(\ln m + 1)}{2\ln m + 3}\bar{\Gamma}(\mathcal{E})\\
  &\le \MOV[Copeland^\alpha] - \frac{4(\ln m + 1)}{2\ln m + 3}\Gamma(\mathcal{E}) + \frac{16(\ln m + 1)}{2\ln m + 3}\eps^\prime n\\
  &\le \MOV[Copeland^\alpha] - \frac{2(\ln m + 1)}{2\ln m + 3}\MOV[Copeland^\alpha] + \eps n\\
  &\le \frac{2\ln m +1}{2\ln m +3}\MOV[Copeland^\alpha] + \eps n\\
  &\le \left(1-O\left(\frac{1}{\ln m}\right) \right)\MOV[Copeland^\alpha] + \eps n
 \end{align*}
 
 The second line follows \Cref{claim:copeland_mov} and the  third line follows from \Cref{lem:copeland_mov}.
\end{proof}

 The approximation factor in \Cref{thm:copeland_mov} is weak when we have a large number of candidates. The main difficulty for showing a better approximation factor for the Copeland$^\alpha$ voting rule is to find a polynomial time computable quantity (for example, $\Gamma(\mathcal{E})$ in \Cref{lem:copeland_mov}) that exhibits tight bounds with margin of victory. We remark that, existence of such a quantity will not only imply a better estimation algorithm, but also, a better approximation algorithm (the best known approximation factor for finding the margin of victory for the Copeland$^\alpha$ voting rule is $O(\ln m)$ and it uses the quantity $\Gamma(\mathcal{E})$). 
However, we remark that \Cref{thm:copeland_mov} will be useful in applications, for example, post election audit and polling, where the number of candidates is often small.

\section{Conclusion}\label{sec:con}

In this work, we introduced the $(\eps,\delta)$-\WD
problem and showed (often tight) bounds for the sample complexity for
many common voting rules. We have also introduced the $(c, \epsilon, \delta)$--\MV problem and presented efficient sampling based algorithms for solving it for many commonly used voting rules which are also often observes an optimal number of sample votes. We observe that predicting the winner of an elections needs least number of queries, whereas more involved voting rules like Borda and maximin need significantly more queries. 

In the next chapter, we study the problem of finding a winner of an election when votes are arriving one by one in a sequential manner.
\chapter{Streaming Algorithms for Winner Determination}
\label{chap:winner_stream}

\blfootnote{A preliminary version of the work in this chapter was published as \cite{deystream}: Arnab Bhattacharyya, Palash Dey, and David P. Woodruff. An optimal algorithm for l1-heavy hitters in insertion streams and related problems. In Proc. 35th ACM SIGMOD-SIGACT-SIGAI Symposium on Principles of Database Systems, PODS ’16, pages 385-400, New York, NY, USA, 2016. ACM.}

\begin{quotation}
 {\small We give the first optimal bounds for returning the $\ell_1$-heavy hitters in a data stream of insertions, together with their approximate frequencies, closing a long line of work on this problem. For a stream of $m$ items in $\{1, 2, \ldots, n\}$ and parameters $0 < \epsilon < \phi \leq 1$, let $f_i$ denote the frequency of item $i$, i.e., the number of times item $i$ occurs in the stream. With arbitrarily large constant probability, our algorithm returns all items $i$ for which $f_i \geq \phi m$, returns no items $j$ for which $f_j \leq (\phi -\epsilon)m$, and returns approximations $\tilde{f}_i$ with $|\tilde{f}_i - f_i| \leq \epsilon m$ for each item $i$ that it returns. Our algorithm uses $O(\epsilon^{-1} \log\phi^{-1} + \phi^{-1} \log n + \log \log m)$ bits of space, processes each stream update in $O(1)$ worst-case time, and can report its output in time linear in the output size. We also prove a lower bound, which implies that our algorithm is optimal up to a constant factor in its space complexity. A modification of our algorithm can be used to estimate the maximum frequency up to an additive $\epsilon m$ error in the above amount of space, resolving Question 3 in the IITK 2006 Workshop on Algorithms for Data Streams for the case of $\ell_1$-heavy hitters. We also introduce several variants of the heavy hitters and maximum frequency problems, inspired by rank aggregation and voting schemes, and show how our techniques can be applied in such settings. Unlike the traditional heavy hitters problem, some of these variants look at comparisons between items rather than numerical values to determine the frequency of an item.}
\end{quotation}

\section{Introduction}
The data stream model has emerged as a standard model for
processing massive data sets. Because of the sheer size
of the data, traditional algorithms are no longer
feasible, e.g., it may be hard or impossible to store the
entire input, and algorithms need to run in linear or even
sublinear time. Such algorithms typically need to be both
randomized and approximate. Moreover, the data may not
physically reside on any device, e.g., if it is internet
traffic, and so if the data is not stored by the algorithm,
it may be impossible to recover it. Hence, many algorithms
must work 
given only a single pass over the data. Applications
of data streams include data warehousing \cite{HSST05,br99,fsgmu98,HPDW01},
network measurements \cite{ABW03,GKMS08,demaine2002frequency,ev03}, 
sensor networks \cite{BGS01,SBAS04},
and compressed sensing \cite{GSTV07,CRT05}. We refer the reader
to recent surveys on the data stream model 
\cite{muthukrishnan2005data,nelson2012sketching,Cormode2012}.

One of the oldest and most fundamental problems in the area of data
streams is the problem of finding the $\ell_1$-heavy hitters (or simply, ``heavy hitters''), also known as the
top-$k$, most popular items, frequent items, elephants, or iceberg queries.
Such algorithms can be used as subroutines in network
flow identification at IP routers \cite{ev03}, association
rules and frequent itemsets \cite{as94,son95,toi96,hid99,hpy00}, 
iceberg queries 
and iceberg datacubes \cite{fsgmu98,br99,HPDW01}. The survey \cite{cormode2008finding} presents 
an overview of the state-of-the-art for this problem, from both theoretical and practical standpoints.

We now formally define the heavy hitters problem that we
focus on in this work: 
\begin{definition}{\bf ($(\epsilon, \phi)$-Heavy Hitters Problem)}\label{def:hh}
In the $(\epsilon, \phi)$-Heavy Hitters Problem, we are given parameters
$0 < \epsilon < \phi \leq 1$ and 
a stream $a_1, \ldots, a_m$ of items $a_j \in \{1, 2, \ldots, n\}.$
Let $f_i$ denote the number of occurrences of item $i$, i.e., its 
frequency.  
The algorithm should make one pass over the stream and at the
end of the stream output a set $S \subseteq \{1, 2, \ldots, n\}$
for which if $f_i \geq \phi m$, then $i \in S$, while if 
$f_i \leq (\phi - \epsilon)m$, then $i \notin S$. Further, for
each item $i \in S$, the algorithm should output an estimate
$\tilde{f}_i$ of the frequency $f_i$ which satisfies
$|f_i - \tilde{f}_i| \leq \epsilon m$. 
\end{definition}
Note that other natural definitions of heavy hitters are possible and sometimes used. For example, {\em $\ell_2$-heavy hitters} are those items $i$ for which $f_i^2 \geq \phi^2 \sum_{j=1}^n f_j^2$, and more generally, {\em $\ell_p$-heavy hitters} are those items $i$ for which $f_i^p \geq \phi^p \sum_{j=1}^n f_j^p$. It is in this sense that Definition \ref{def:hh} corresponds to { $\ell_1$-heavy hitters}. While $\ell_p$-heavy hitters for $p>1$ relax $\ell_1$-heavy hitters and algorithms for them have many interesting applications, we focus on the most direct and common formulation of the heavy hitters notion.

We are
interested in algorithms which use as little space in bits
as possible to solve the {\bf $(\epsilon, \phi)$-Heavy Hitters Problem}.
Further, we are also interested in minimizing the {\it update time}
and {\it reporting time} of such algorithms. Here, the update time
is defined to be the time the algorithm needs to update its data
structure when processing a stream insertion. The reporting time is the
time the algorithm needs to report the answer after having 
processed the stream. We allow the algorithm to be randomized and
to succeed with probability at least $1-\delta$ for $0<\delta<1$.
We do not make any assumption on the ordering of the stream
$a_1, \ldots, a_m$. This
is desirable as often in applications one cannot assume a best-case or even
a random order. We are also interested in the case when the length $m$ of the stream is
not known in advance, and give algorithms in this more general 
setting. 

The first algorithm for the {\bf $(\epsilon, \phi)$-Heavy Hitters Problem}
was given by Misra and Gries \cite{misra82}, who achieved $O(\epsilon^{-1} (\log n + \log m))$
bits of space for any $\phi > \epsilon$. This algorithm was rediscovered
by Demaine et al. \cite{demaine2002frequency}, and again by Karp et al. \cite{karp2003simple}. 
Other than these 
algorithms, which are deterministic, there are also a number of randomized
algorithms, such as the CountSketch \cite{charikar2004finding}, Count-Min sketch \cite{cormode2005improved}, sticky sampling \cite{mm02},
lossy counting \cite{mm02}, space-saving \cite{MetwallyAA05}, sample and hold \cite{ev03}, multi-stage bloom filters
\cite{cfm09}, and sketch-guided sampling \cite{kx06}. 
Berinde
et al. \cite{bics10} show that using $O(k \epsilon^{-1} \log(mn))$ bits of space, 
one can achieve the stronger guarantee
of reporting, for each item $i \in S$, $\tilde{f}_i$ with
$|\tilde{f}_i - f_i| \leq \nfrac{\epsilon}{k} F^{res(k)}_1$, 
where $F^{res(k)}_1 < m$ denotes
the sum of frequencies of items in $\{1, 2, \ldots, n\}$ excluding the frequencies
of the $k$ most frequent items. 

We emphasize that prior to our work the
best known algorithms for the {\bf $(\epsilon, \phi)$-Heavy Hitters Problem}
used $O(\epsilon^{-1} (\log n + \log m))$ bits of space. 
Two previous lower bounds were known. The first is a lower bound of $\log({n \choose 1/\phi})
= \Omega(\phi^{-1} \log(\phi n))$ bits, which  comes from the fact that the output set $S$ can contain $\phi^{-1}$
items and it takes this many bits to encode them. 
The second lower bound is
$\Omega(\epsilon^{-1})$ 
which follows from a folklore 
reduction from the randomized communication complexity of the {\sf Index} problem. 
In this problem, there are two players, Alice and Bob. Alice has a bit string $x$ of length $(2\epsilon)^{-1}$,
while Bob has an index $i$. Alice creates a stream of length $(2\eps)^{-1}$ consisting of one copy of each $j$ for which $x_j = 1$ and copies of a dummy item to fill the rest of the stream. She runs the heavy hitters streaming algorithm on her stream and sends the state of the
algorithm to Bob. Bob appends $(2\epsilon)^{-1}$ 
copies of the item $i$ to the stream 
and continues the execution
of the algorithm. For $\phi = \nfrac{1}{2}$, it holds that $i \in S$. Moreover, $f_i$ differs
by an additive $\epsilon m$ factor depending on whether $x_i = 1$ or $x_i = 0$. 
Therefore by the randomized communication complexity of the {\sf Index} problem 
\cite{kremer1999randomized}, the 
$(\epsilon, \nfrac{1}{2})$-heavy hitters problem requires $\Omega(\epsilon^{-1})$ bits of space.
Although this proof was for $\phi = \nfrac{1}{2}$, no better lower bound is known for any
$\phi > \epsilon$. 

Thus, while the upper bound for the {\bf $(\epsilon, \phi)$-Heavy Hitters Problem} 
is $O(\epsilon^{-1} (\log n + \log m))$ bits, the best known 
lower bound is only $\Omega(\phi^{-1} \log n + \epsilon^{-1})$ bits. For constant
$\phi$, and $\log n \approx \epsilon^{-1}$, this represents a nearly quadratic gap
in upper and lower bounds. Given the limited resources of devices which typically
run heavy hitters algorithms, such as internet routers, this quadratic gap can 
be critical in applications. 

A problem related to the {\bf $(\epsilon, \phi)$-Heavy Hitters Problem} is estimating
the {\it maximum frequency} in a data stream, also known as the $\ell_{\infty}$-norm. In
the IITK 2006 Workshop on Algorithms for Data Streams, Open Question 3 asks for an algorithm
to estimate the maximum frequency of any item up to an additive $\epsilon m$ error using
as little space as possible. The best known space bound is still $O(\epsilon^{-1} \log n)$
bits, as stated in the original formulation of the question (note that the ``$m$'' 
in the question there corresponds to the ``$n$'' here). Note that, if one can find an 
item whose frequency is the largest, up to an additive $\epsilon m$ error, then one can solve
this problem. The latter problem is independently interesting and corresponds to finding approximate
plurality election winners in voting streams \cite{DeyB15}. 
We refer to this problem as the \textbf{$\epsilon$-Maximum} problem. 

Finally, we note that there are many other variants of the 
{\bf $(\epsilon, \phi)$-Heavy Hitters Problem} that one can consider. One simple variant
of the above is to output an item of frequency within $\epsilon m$ of the {\it minimum frequency} 
of any item in the universe. We refer to this as the 
\textbf{$\epsilon$-Minimum} problem. 
This only makes sense for small universes, as otherwise outputting
a random item typically works. This is useful when one wants to count
the ``number of dislikes'', or in anomaly detection; see more motivation
below. In other settings, 
one may not have numerical scores associated with the items, but rather, 
each stream update consists of a ``ranking'' or ``total ordering'' of all stream items. 
This may be the case in ranking aggregation on the web (see, e.g., \cite{mbg04,MYCC07}) 
or in voting streams (see, e.g., \cite{conitzer2005communication, caragiannis2011voting, DeyB15, xia2012computing}). 
One may consider a variety of aggregation measures, such as the {Borda score} of 
an item $i$, which asks for the sum, over rankings, of the number of items $j \neq i$ for
which $i$ is ranked ahead of $j$ in the ranking. Alternatively, 
one may consider the { Maximin score}
of an item $i$, which asks for the minimum, over items $j \neq i$, of the number
of rankings for which $i$ is ranked ahead of $j$. For these aggregation measures, one may
be interested in finding an item whose score is an approximate maximum. 
This is the analogue of the 
\textbf{$\epsilon$-Maximum} problem above. Or, one may be interested in listing
all items whose score is above a threshold, which is the analogue of the 
{\bf $(\epsilon, \phi)$-Heavy Hitters Problem}. 

We give more motivation of these variants of heavy hitters in this section below, and more precise definitions in Section \ref{sec:prob_def_heavy_hitters}. 

\begin{table*}[t]
 \begin{center}
  \resizebox{\textwidth}{!}{
   \begin{tabular}{|c|c|c|}\hline
   
      \multirow{2}{*}{\textbf{Problem}}	& \multicolumn{2}{c|}{\textbf{Space complexity}} \\\cline{2-3}
       & Upper bound & Lower bound \\\hline\hline
      
      &&\\[-10pt]
      
      $(\epsilon, \phi)$-Heavy Hitters & \makecell{$O\left( \epsilon^{-1} \log \phi^{-1} + \phi^{-1} \log n + \log \log m \right )$ \\~[\Cref{thm:heavy_hitters2,thm:UbUnknownMax}]}
      & 
      \makecell{$\Omega \left (\epsilon^{-1} \log \phi^{-1} + \phi^{-1} \log n + \log \log m \right )$\\~[\Cref{thm:eps_eps,thm:loglogn}]}
      
      \\\hline
      
      &&\\[-10pt]
      
      $\epsilon$-Maximum and $\ell_{\infty}$-approximation& \makecell{$O\left(\epsilon^{-1} \log \epsilon^{-1} + \log n + \log \log m \right )$\\~[\Cref{thm:heavy_hitters,thm:UbUnknownMax}]} &
      \makecell{$\Omega \left (\epsilon^{-1} \log \epsilon^{-1} + \log n + \log \log m \right )$\\~[\Cref{thm:eps_maximum,thm:loglogn}]}
      \\\hline
      
      &&\\[-10pt]
      
      $\epsilon$-Minimum &  \makecell{$O\left(\epsilon^{-1} \log \log \epsilon^{-1} + \log \log m \right)$ \\~[\Cref{thm:rare,thm:UbUnknownMin}]}
      & \makecell{$\Omega \left (\epsilon^{-1} + \log \log m \right )$\\~[\Cref{thm:veto_lb,thm:loglogn}]}\\ \hline
      &&\\[-10pt]
      
  $\epsilon$-Borda &  \makecell{$O\left(n(\log \epsilon^{-1} + \log n) + \log \log m \right)$ \\~[\Cref{thm:borda,thm:UbUnknownMin}]}
      & \makecell{$\Omega \left (n (\log\epsilon^{-1} + \log n) + \log \log m \right )$\\~[\Cref{thm:lwb_borda,thm:loglogn}]}\\ \hline
      
  $\epsilon$-Maximin &  \makecell{$O\left(n \epsilon^{-2}\log^2 n + \log \log m \right)$ \\~[\Cref{thm:maximin,thm:UbUnknownMin}]}
      & \makecell{$\Omega \left (n (\epsilon^{-2} + \log n) + \log \log m \right )$\\~[\Cref{thm:mmlb}]}\\ \hline
   \end{tabular}
  }
  \caption{\small 
  The bounds hold for constant success probability algorithms and for $n$ sufficiently large in terms of $\eps$. For the {\bf $(\epsilon, \phi)$-Heavy Hitters}
problem and the {\bf $\epsilon$-Maximum} problem, we also
achieve $O(1)$ update time and reporting time which is linear in the size
of the output. The upper bound for
{\bf $\epsilon$-Borda} (resp. {\bf $\epsilon$-Maximin})
is for returning every item's Borda score (resp. Maximin score)
up to an additive $\epsilon mn$ (resp. additive $\epsilon m$), 
while the lower bound for {\bf $\epsilon$-Borda} 
(resp. {\bf $\epsilon$-Maximin}) is for
returning only the approximate Borda score (resp. Maximin score) 
of an approximate maximum.}\label{tbl:summary}
 \end{center}
\end{table*}

\subsection{Our Contribution}
Our results are summarized in Table \ref{tbl:summary}. We note that independently of this work and nearly parallelly, there have been improvements to the space complexity of the $\ell_2$-heavy hitters problem in insertion streams \cite{BCIW16} and to the time complexity of the $\ell_1$-heavy hitters problem in turnstile\footnote{In a turnstile stream, updates modify an underlying $n$-dimensional vector $x$ initialized at the zero vector; each update is of the form $x \leftarrow x+ e_i$ or $x \leftarrow x-e_i$ where $e_i$ is the $i$'th standard unit vector. In an insertion stream, only updates of the form $x \leftarrow x+ e_i$  are allowed.} streams \cite{LNNT16}. These works use very different techniques.

Our first contribution is an optimal algorithm and lower bound 
for the {\bf $(\epsilon, \phi)$-Heavy Hitters Problem}. 
Namely, we show that there
is a randomized algorithm with constant probability of success which solves this problem using 
$$O(\epsilon^{-1} \log \phi^{-1} + \phi^{-1} \log n + \log \log m)$$ 
bits of space, and we prove a lower bound matching up to constant factors. 
In the unit-cost RAM model with $O(\log n)$ bit words, 
our algorithm has $O(1)$ update time
and reporting time linear in the output size, under the
standard assumptions that the length of the
stream and universe size are at least $\text{poly}(\eps^{-1} \log(1/\phi))$. Furthermore, we can achieve nearly the optimal space complexity even when the
length $m$ of the stream is {\it not known in advance}.
Although the results of \cite{bics10} achieve stronger error bounds in terms of the tail, which are useful for skewed streams, here we focus on the original formulation of the problem.

Next, we turn to the problem of estimating the maximum frequency in a data
stream up to an additive $\epsilon m$. We give an algorithm using
$$O(\epsilon^{-1} \log \epsilon^{-1} + \log n + \log \log m)$$
bits of space, improving the previous best algorithms which required
space at least $\Omega(\epsilon^{-1} \log n)$ bits, and show that our bound is tight. As an example setting of parameters, if $\epsilon^{-1} = \Theta(\log n)$ and $\log \log m = O(\log n)$, our space
complexity is $O(\log n \log \log n)$ bits, improving the previous 
$\Omega(\log^2 n)$ bits of space algorithm. We also prove a lower bound
showing our algorithm is optimal up to constant factors. 
This resolves Open Question 3 from the IITK 2006 Workshop on Algorithms for Data Streams in the case of insertion streams, for the case of ``$\ell_1$-heavy hitters''. Our algorithm also returns the identity of the item with the 
approximate maximum frequency, solving the \textbf{$\epsilon$-Maximum} 
problem.

We then focus on a number of variants of these problems. 
We first give nearly tight bounds for finding an item whose frequency
is within $\epsilon m$ of the minimum possible frequency. While this can be solved
using our new algorithm for the {\bf $(\epsilon, \epsilon)$-Heavy Hitters Problem},
this would incur $\Omega(\epsilon^{-1} \log \epsilon^{-1} + \log\log m)$ bits of space, whereas we give
an algorithm using only $O(\epsilon^{-1} \log \log (\epsilon^{-1}) + \log\log m)$ bits of space.  We
also show a nearly matching $\Omega(\epsilon^{-1} + \log\log m)$ bits of space lower bound. We note
that for this problem, a dependence on $n$ is not necessary since if the number of
possible items is sufficiently large, then outputting the identity
of a random item among the first say, $10\epsilon^{-1}$ items, is a correct solution
with large constant probability. 

Finally, we study variants of heavy hitter problems that are ranking-based. In this setting, each
stream update consists of a total ordering of the $n$ universe items. For the $\epsilon$-{\bf Borda}
problem, we give an algorithm using $O(n (\log \epsilon^{-1} + \log \log n) + \log \log m)$
bits of space to report the Borda score of every item up to an additive $\epsilon m n$. 
We also show this is nearly optimal by proving an $\Omega(n \log \epsilon^{-1} + \log\log m)$ bit lower bound for the problem, even in the case when one is only interested in outputting an item maximum Borda score up to an additive $\epsilon m n$. For the $\epsilon$-{\bf Maximin} problem, we give an algorithm using $O(n \epsilon^{-2} \log^2n + \log \log m)$ bits of space to report the maximin score of every item up to an additive $\epsilon m$, and prove an $\Omega(n \epsilon^{-2} + \log\log m)$ bits of space lower bound even in the case when one is only interested in outputting the maximum maximin
score up to an additive $\epsilon m$. 
This shows that finding heavy hitters
with respect to the maximin score is significantly more expensive than with respect to
the Borda score.

\subsection{Motivation for Variants of Heavy Hitters}
While the {\bf $(\epsilon, \phi)$-Heavy Hitters} and
{\bf $\epsilon$-Maximum} problem are very well-studied in the data stream
literature, the other variants introduced are not. We provide additional
motivation for them here. 

For the {\bf $\epsilon$-Minimum} problem, in our formulation, an item with
frequency zero, i.e., one that does not occur in the stream, is a valid solution
to the problem. In certain scenarios, this might not make sense, e.g., if a stream containing only a small fraction of IP addresses. However, in other scenarios we argue this is a natural problem.
For instance, consider an online portal where users register complaints
about products. Here, minimum frequency items correspond to the ``best'' items.
That is, such frequencies arise in the context of voting or more generally
making a choice: in cases for which one does not have a strong preference
for an item, but definitely does not like certain items, this problem
applies, since the frequencies correspond to ``number of dislikes''. 

The {\bf $\epsilon$-Minimum} problem may also be useful for anomaly detection.
Suppose one has a known set of sensors broadcasting information and one
observes the ``From:'' field in the broadcasted packets. Sensors which send
a small number of packets may be down or defective, and an algorithm
for the {\bf $\epsilon$-Minimum} problem could find such sensors. 
%

Finding items with maximum and minimum frequencies in a stream correspond to finding winners under plurality and veto voting rules respectively in the context of voting\footnote{In fact, the first work \cite{Moore81} to formally pose the heavy hitters problem couched it in the context of voting.} \cite{brandt2015handbook}. The streaming aspect of voting could be crucial in applications like online polling~\citep{kellner2011polling}, recommender systems~\citep{resnick1997recommender,herlocker2004evaluating,adomavicius2005toward} where the voters are providing their votes in a streaming fashion and at every point in time, we would like to know the popular items.
   While in some elections, such as for political positions, the scale of
the election may not be large enough to require a streaming algorithm, one
key aspect of these latter voting-based problems is that they are rank-based
which is useful when numerical scores are not available. Orderings naturally
arise in several applications - for instance, if a website has multiple
parts, the order in which a user visits the parts given by its clickstream
defines a voting, and for data mining and recommendation purposes the
website owner may be interested in aggregating the orderings across users. 
 Motivated by this connection, we define similar problems for two other important voting rules, namely Borda and maximin. 
 The Borda scoring method finds its applications in a wide range of areas of artificial intelligence, for example, machine learning~\citep{ho1994decision,carterette2006learning,volkovs2014new,prasad2015distributional}, image processing~\citep{lumini2006detector,monwar2009multimodal}, information retrieval~\citep{li2014learning,aslam2001models,nuray2006automatic}, etc. The Maximin score is often used when the spread between the best and worst outcome is very large (see, e.g., p. 373 of \cite{mr91}). The maximin scoring method also has been used frequently in machine learning~\citep{wang2004feature,jiang2014diverse}, human computation~\citep{mao_hcomp12_2,Mao_aaai13}, etc.

\section{Problem Definitions}\label{sec:prob_def_heavy_hitters}

We now formally define the problems we study here. Suppose we have $0 < \eps < \varphi <1$.
\begin{definition}\textsc{$(\eps, \varphi)$-List heavy hitters}\\
 Given an insertion-only stream of length $m$ over a universe $\mathcal{U}$ of size $n$, find all items in $\mathcal{U}$ with frequency more than $\varphi m$, along with their frequencies up to an additive error of $\eps m$, and report no items with frequency less than $(\varphi-\eps)m$.
\end{definition}
 
\begin{definition}\textsc{$\eps$-Maximum}\\
 Given an insertion-only stream of length $m$ over a universe $\mathcal{U}$ of size $n$, find the maximum frequency up to an additive error of $\eps m$.
\end{definition}

Next we define the {\em minimum} problem for $0 < \eps <1$.

\begin{definition}\textsc{$\eps$-Minimum}\\
 Given an insertion-only stream of length $m$ over a universe $\mathcal{U}$ of size $n$, find the minimum frequency up to an additive error of $\eps m$.
\end{definition}

Next we define related heavy hitters problems in the context of rank aggregation. The input is a stream of rankings (permutations) over an item set $\UU$ for the problems below. The {\em Borda score} of 
an item $i$ is the sum, over all rankings, of the number of items $j \neq i$ for
which $i$ is ranked ahead of $j$ in the ranking.

\begin{definition}{\sc $(\eps, \varphi)$-List borda}\label{def:listborda}\\
 Given an insertion-only stream over a universe $\LL(\UU)$ where $|\UU|=n$, find all items with Borda score more than $\varphi mn$, along with their Borda score up to an additive error of $\eps mn$, and report no items with Borda score less than $(\varphi-\eps)mn$.
\end{definition}

\begin{definition}{\sc $\eps$-Borda}\label{def:borda}\\
 Given an insertion-only stream over a universe $\LL(\UU)$ where $|\UU|=n$, find the maximum Borda score up to an additive error of $\eps mn$.
\end{definition}

The {\em maximin score} of an item $i$ is the minimum, over all items $j \neq i$, of the number
of rankings for which $i$ is ranked ahead of $j$.

\begin{definition}{\sc $(\eps, \varphi)$-List maximin}\label{def:listmaximin}\\
 Given an insertion-only stream over a universe $\LL(\UU)$ where $|\UU|=n$, find all items with maximin score more than $\varphi m$ along with their maximin score up to an additive error of $\eps m$, and report no items with maximin score less than $(\varphi-\eps)m$.
\end{definition}

\begin{definition}{\sc $\eps$-maximin}\label{def:maximin}\\
 Given an insertion-only stream over a universe $\LL(\UU)$ where $|\UU|=n$, find the maximum maximin score up to an additive error of $\eps m$.
\end{definition}

Notice that the maximum possible Borda score of an item is $m(n-1) = \Theta(mn)$ and the maximum possible maximin score of an item is $m$. This justifies the approximation factors in \Cref{def:listborda,def:borda,def:listmaximin,def:maximin}. We note that finding an item with maximum Borda score within additive $\eps m n$ or maximum maximin score within additive $\eps m$ corresponds to finding an approximate winner of an election (more precisely, what is known as an $\eps$-winner)~\citep{DeyB15}.

\section{Our Algorithms}

In this section, we present our upper bound results. All omitted proofs are in Appendix B. 
Before describing specific algorithms, we record some claims for later use. We begin with the following space efficient algorithm for picking an item uniformly at random from a universe of size $n$ below.

\begin{lemma}\label{lem:sampling_ub}
 Suppose $m$ is a power of two\footnote{In all our algorithms, whenever we pick an item with probability $p>0$, we can assume, without loss of generality, that $\nfrac{1}{p}$ is a power of two. If not, then we replace $p$ with $p'$ where $\nfrac{1}{p'}$ is the largest power of two less than $\nfrac{1}{p}$. This does not affect correctness and performance of our algorithms.}. Then there is an algorithm $\mathcal{A}$ for choosing an item with probability $\nfrac{1}{m}$ that has space complexity of  $O(\log\log m)$ bits and time complexity of $O(1)$ in the unit-cost RAM model.
\end{lemma}

\begin{proof}
 We generate a $(\log_2 m)$-bit integer $C$ uniformly at random and record the sum of the digits in $C$. Choose the item only if the sum of the digits is $0$, i.e. if $C=0$.
\end{proof}
We remark that the algorithm in \Cref{lem:sampling_ub} has optimal space complexity as shown in \Cref{lem:sampling_lb} in Appendix B.

We remark that the algorithm in \Cref{lem:sampling_ub} has optimal space complexity as shown in \Cref{lem:sampling_lb} below which may be of independent interest. We also note that every algorithm needs to toss a fair coin at least $\Omega(\log m)$ times to perform any task with probability at least $\nfrac{1}{m}$.

\begin{proposition}\label{lem:sampling_lb}
 Any algorithm that chooses an item from a set of size $n$ with probability $p$ for $0< p \le \frac{1}{n}$, in unit cost RAM model must use $\Omega(\log\log m)$ bits of memory. 
\end{proposition}

\begin{proof}
 The algorithm generates $t$ bits uniformly at random (the number of bits it generates uniformly at random may also depend on the outcome of the previous random bits) and finally picks an item from the say $x$. Consider a run $\mathcal{R}$ of the algorithm where it chooses the item $x$ with {\em smallest number of random bits getting generated}; say it generates $t$ random bits in this run $\mathcal{R}$. This means that in any other run of the algorithm where the item $x$ is chosen, the algorithm must generate at least $t$ many random bits. Let the random bits generated in $\mathcal{R}$ be $r_1, \cdots, r_t$. Let $s_i$ be the memory content of the algorithm immediately after it generates $i^{th}$ random bit, for $i\in [t]$, in the run $\mathcal{R}$. First notice that if $t < \log_2 n$, then the probability with which the item $x$ is chosen is more than $\frac{1}{n}$, which would be a contradiction. Hence, $t \ge \log_2 n$. Now we claim that all the $s_i$'s must be different. Indeed otherwise, let us assume $s_i = s_j$ for some $i<j$. Then the algorithm chooses the item $x$ after generating $t-(j-i)$ many random bits (which is strictly less than $t$) when the random bits being generated are $r_1, \cdots, r_i, r_{j+1}, \cdots, r_t$. This contradicts the assumption that the run $\mathcal{R}$ we started with chooses the item $x$ with smallest number of random bits generated.
\end{proof}

Our second claim is a standard result for universal families of hash functions.
\begin{lemma}\label{lem:hash}
 For $S\subseteq A$, $\delta \in (0,1)$, and universal family of hash functions $\mathcal{H} = \{ h | h: A \rightarrow [\lceil\nfrac{ |S|^2 }{\delta}\rceil]\}$:
 \[ \Pr_{{h \in_U \mathcal{H}}} [\exists i\ne j\in S, h(i) = h(j) ] \le \delta \]
\end{lemma}

\begin{proof}
 For every $i\ne j \in S$, since $\mathcal{H}$ is a universal family of hash functions, we have
$\Pr_{h \in_\text{U} \mathcal{H}} [h(i) = h(j)] \le \frac{1}{\lceil{ |S|^2 }/{\delta}\rceil}$.
 Now we apply the union bound to get
$ \Pr_{h \in_{{U}}\mathcal{H}} [\exists i\ne j\in S, h(i) = h(j)] \le \frac{|S|^2}{\lceil{ |S|^2 }/{\delta}\rceil} \le \delta$
 \end{proof}

Our third claim is folklore and also follows from the celebrated DKW inequality \cite{DKW}. 
We provide a simple proof here that works for constant $\delta$.
\begin{restatable}{lemma}{LemDKW}\label{lem:dkw}
 Let $f_i$ and $\hat{f}_i$ be the frequencies of an item $i$ in a stream $\SS$ and in a random sample $\TT$  of size $r$ from $\SS$, respectively. Then for $r \geq  2\eps^{-2}\log(2\delta^{-1})$, with probability $1-\delta$, 
for every universe item $i$ simultaneously,  
$$\left |\frac{\hat{f}_i}{r} - \frac{f_i}{m} \right | \leq \epsilon.$$
\end{restatable}

\begin{proof}[Proof for constant $\delta$]
 This follows by Chebyshev's inequality and a union bound. Indeed, consider a given $i \in [n]$ with frequency $f_i$ and suppose
we sample each of its occurrences pairwise-independently with probability $r/m$, for a parameter $r$.  
Then the expected number ${\bf E}[\hat{f_i}]$ of sampled occurrences is 
$f_i \cdot r/m$ and the variance ${\bf Var}[\hat{f_i}]$ is $f_i \cdot r/m (1-r/m) \leq f_ir/m$. Applying Chebyshev's inequality, 
$$
\Pr \left [\left | \hat{f_i} - {\bf E}[\hat{f_i}] \right | \geq \frac{r \epsilon}{2} \right ] \leq \frac{{\bf Var}[\hat{f_i}]}{(r \epsilon/2)^2}
\leq \frac{4f_ir}{m r^2 \epsilon^2}.$$
Setting $r = \frac{C}{\epsilon^2}$ for a constant $C > 0$ makes this probability at most $\frac{4f_i}{C m}$. By the union bound, if we
sample each element in the stream independently with probability $\frac{r}{m}$, then the probability there exists an $i$ for which
$|\hat{f_i} - {\bf E}[\hat{f_i}]| \geq \frac{r \epsilon}{2}$ is at most $\sum_{i = 1}^n \frac{4 f_i}{C m} \leq \frac{4}{C}$, which for
$C \geq 400$ is at most $\frac{1}{100}$, as desired. 
\end{proof}

For now, assume that the length of the stream is known in advance; we show in \cref{sec:unknown} how to remove this assumption.

\subsection{List Heavy Hitters Problem}\label{sec:lhh}
For the {\sc List heavy hitters} problem, we present two algorithms. The first is slightly suboptimal, but simple conceptually and already constitutes a very large improvement in the space complexity over known algorithms. We expect that this algorithm could be useful in practice as well. The second algorithm is more complicated, building on ideas from the first algorithm, and achieves the optimal space complexity upto constant factors.

We note that both algorithms proceed by sampling $O(\epsilon^{-2} \ln (1/\delta))$ stream items and updating a data structure as the stream progresses. In
both cases, the time to update the data structure is bounded
by $O(1/\epsilon)$, and so, under the standard assumption that the
length of the stream is at least
$\textrm{poly}(\ln(1/\delta) \epsilon)$, 
the time to perform this update can be
spread out across the next $O(1/\epsilon)$ stream updates, 
since with large probability there will be no items sampled among these 
next $O(1/\epsilon)$ stream updates. Therefore, we achieve worst-case\footnote{We emphasize that this is stronger than an amortized guarantee, as on every insertion, the cost will be $O(1)$.}
update time of $O(1)$.

\subsubsection{A Simpler, Near-Optimal Algorithm}

\begin{restatable}{theorem}{ThmMisra}\label{thm:heavy_hitters}
 Assume the stream length is known beforehand. Then there is a randomized one-pass algorithm $\mathcal{A}$ for the \textsc{$(\eps, \varphi)$-List heavy hitters} problem which succeeds with probability at least $1-\delta$ using $O\left(\eps^{-1}(\log\eps^{-1} + \log\log\delta^{-1}) + \varphi^{-1}\log n + \log\log m \right)$ bits of space. Moreover, $\mathcal{A}$ has an update time of $O(1)$ and reporting time linear in its output size. 
\end{restatable}

\paragraph{Overview}

The overall idea is as follows. We sample $\el=O(\eps^{-2})$ many items from the stream uniformly at random as well as hash the id's (the word ``id'' is short for identifier) of the sampled elements into a space of size $O(\eps^{-4})$. Now, both the stream length as well as the universe size are $\text{poly}(\eps^{-1})$. From \cref{lem:dkw}, it suffices to solve the heavy hitters problem on the sampled stream. From \cref{lem:hash}, because the hash function is chosen from a universal family, the sampled elements have distinct hashed id's. We can then feed these elements into a standard Misra-Gries data structure with $\eps^{-1}$ counters, incurring space $O(\eps^{-1} \log \eps^{-1})$. Because we want to return the unhashed element id's for the heavy hitters, we also use $\log n$ space for recording the $\phi^{-1}$ top items according to the Misra-Gries data structure and output these when asked to report.

\begin{algorithm}[!htbp]
  \caption{for \textsc{$(\eps, \varphi)$-List heavy hitters}
  \label{alg:heavy_hitters}}
  \begin{algorithmic}[1]
    \Require{A stream $\SS$ of length $m$ over $\UU = [n]$; let $f(x)$ be the frequency of $x\in \UU$ in $\SS$}    
    \Ensure{A set $X\subseteq \UU$ and a function $\hat{f}:X \rightarrow \NB$ such that if $f(x)\ge \varphi m$, then $x\in X$ and $f(x) -\eps m \le \hat{f}(x) \le f(x) + \eps m$ and if $f(y) \le (\phi-\eps)m$, then $y\notin X$ for every $x, y \in \UU$}
    \Initialize{\vspace{-2ex}
     \Let{$\ell$}{$\nfrac{6\log(\nfrac{6}{\delta})}{\eps^2}$}\label{alg:ell}
     \State Hash function $h$ uniformly at random from a universal family $\HH \subseteq \{h:[n]\rightarrow \lceil \nfrac{4\ell^2}{\delta} \rceil\}$.
     \State An empty table $\TT_1$ of (key, value) pairs of length ${\eps^{-1}}$. Each key entry of $\TT_1$ can store an integer in $[0, \lceil \nfrac{400\ell^2}{\delta} \rceil]$ and each value entry can store an integer in $[0, 11\ell]$.\label{alg:m_table} \Comment{The table $\TT_1$ will be in sorted order 	by value throughout.}
     \State An empty table $\TT_2$ of length $\nfrac{1}{\varphi}$. Each entry of $\TT_2$ can store an integer in $[0, n]$. \Comment{The entries of $\TT_2$ will correspond to ids of the keys in $\TT_1$ of the highest $\nfrac{1}{\varphi}$ values}
    }
~\\
\Procedure{Insert}{x}
            \State With probability $p = \nfrac{6\ell}{m}$, continue. Otherwise, \Return.
       \State Perform Misra-Gries update using $h(x)$ maintaining $\TT_1$ sorted by values.\label{alg:u1}
       \If{The value of $h(x)$ is among the highest $\nfrac{1}{\varphi}$ valued items in $\TT_1$}
        \If{$x_i$ is not in $\TT_2$}
         \If{$\TT_2$ currently contains $\nfrac{1}{\varphi}$ many items}
         \State For $y$ in $\TT_2$  such that $h(y)$ is not among the highest $\nfrac{1}{\varphi}$ valued items in $\TT_1$, replace $y$ with $x$. 
         \Else
         \State We put $x$ in $\TT_2$.
         \EndIf
        \EndIf
        \State	Ensure that elements in $\TT_2$ are ordered according to corresponding values in $\TT_1$.
       \EndIf\label{alg:u2}
\EndProcedure
~\\
\Procedure{Report}{~}
    \State \Return items in $\TT_2$ along with their corresponding values in $\TT_1$ 
    \EndProcedure
  \end{algorithmic}
\end{algorithm}
\clearpage

\begin{proof}[Proof of \cref{thm:heavy_hitters}]
 The pseudocode of our \textsc{$(\eps, \varphi)$-List heavy hitters} algorithm is in \Cref{alg:heavy_hitters}. By \cref{lem:dkw}, if we select a subset $\mathcal{S}$ of size at least $\ell = 6\eps^{-2}{\log({6}\delta^{-1})}$ uniformly at random from the stream, then $\Pr[\forall i\in \mathcal{U}, |(\nfrac{\hat{f_i}}{|\SS|}) - (\nfrac{f_i}{n})| \le \eps] \ge 1 - \nfrac{\delta}{3}$, where $f_i$ and $\hat{f_i}$ are the frequencies of item $i$ in the input stream and $\mathcal{S}$ respectively. First we show that with the choice of $p$ in line \ref{alg:p} in \Cref{alg:heavy_hitters}, the number of items sampled is at least $\ell$ and at most $11\ell$ with probability at least $(1-\nfrac{\delta}{3})$. Let $X_i$ be the indicator random variable of the event that the item $x_i$ is sampled for $i\in[m]$. Then the total number of items sampled $X = \sum_{i=1}^m X_i$. We have $\EB[X] = 6\ell$ since $p = \nfrac{6\ell}{m}$. Now we have the following.
 \[ \Pr[X \le \ell \text{ or } X\ge 11\ell] \le \Pr[|X-\EB[X]|\ge 5\ell] \le \nfrac{\delta}{3} \]
{ The inequality follows from the Chernoff bound and the value of $\ell$. }From here onwards we assume that the number of items sampled is in $[\ell, 11\ell]$.
 
 We use (a modified version of) the Misra-Gries algorithm~\citep{misra82} to estimate the frequencies of items in $\mathcal{S}$. The length of the table in the Misra-Gries algorithm is $\eps^{-1}$. We pick a hash function $h$ uniformly at random from a universal family $\mathcal{H} = \{h | h:[n]\rightarrow \lceil\nfrac{4\ell^2}{\delta}\rceil\}$  of hash functions of size $|\HH| = O(n^2)$. Note that picking a hash function $h$ uniformly at random from $\mathcal{H}$ can be done using $O(\log n)$ bits of space.  
\Cref{lem:hash} shows that there are no collisions in $\mathcal{S}$ under this hash function $h$ with probability at least $1-{\delta}/{3}$. From here onwards we assume that there is no collision among the ids of the sampled items under the hash function $h$.
 
 We modify the Misra-Gries algorithm as follows. Instead of storing the id of any item $x$ in the Misra-Gries table (table $\TT_1$ in line \ref{alg:m_table} in \Cref{alg:heavy_hitters}) we only store the hash $h(x)$ of the id $x$. We also store the ids (not the hash of the id) of the items with highest $\nfrac{1}{\varphi}$ values in $\TT_1$ in another table $\TT_2$. Moreover, we always maintain the table $\TT_2$ consistent with the table $\TT_1$ in the sense that the $i^{th}$ highest valued key in $\TT_1$ is the hash of the $i^{th}$ id in $\TT_2$. 
 
 Upon picking an item $x$ with probability $p$, we create an entry corresponding to $h(x)$ in $\TT_1$ and make its value one if there is space available in $\TT_1$; decrement the value of every item in $\TT_1$ by one if the table is already full; increment the entry in the table corresponding to $h(x)$ if $h(x)$ is already present in the table. When we decrement the value of every item in $\TT_1$, the table $\TT_2$ remains consistent and we do not need to do anything else. Otherwise there are three cases to consider. Case $1$: $h(x)$ is not among the $\nfrac{1}{\varphi}$ highest valued items in $\TT_1$. In this case, we do not need to do anything else. Case $2$: $h(x)$ was not among the $\nfrac{1}{\varphi}$ highest valued items in $\TT_1$ but now it is among the $\nfrac{1}{\varphi}$ highest valued items in $\TT_1$. In this case the last item $y$ in $\TT_2$ is no longer among the $\nfrac{1}{\varphi}$ highest valued items in $\TT_1$. We replace $y$ with $x$ in $\TT_2$. Case 3: $h(x)$ was among the $\nfrac{1}{\varphi}$ highest valued items in $\TT_1$. 
 When the stream finishes, we output the ids of all the items in table $\TT_2$ along with the values corresponding to them in table $\TT_1$. Correctness follows from the correctness of the Misra-Gries algorithm and the fact that there is no collision among the ids of the sampled items.
\end{proof}

\subsubsection{An Optimal Algorithm}

\begin{restatable}{theorem}{ThmOpt}\label{thm:heavy_hitters2}
 Assume the stream length is known beforehand. Then there is a randomized one-pass algorithm $\mathcal{A}$ for the \textsc{$(\eps, \varphi)$-List heavy hitters} problem which succeeds with constant probability using $O\left(\eps^{-1}\log \phi^{-1}  + \phi^{-1}\log n + \log\log m \right)$ bits of space. Moreover,  $\mathcal{A}$ has an update time of $O(1)$ and reporting time linear in its output size.
\end{restatable}

Note that in this section, for the sake of simplicity, we ignore floors and ceilings and state the results for a constant error probability, omitting the explicit dependence on $\delta$.

\begin{algorithm}[!htbp]
  \caption{for \textsc{$(\eps, \varphi)$-List heavy hitters}
  \label{alg:heavy_hitters2}}
  \begin{algorithmic}[1]
    \Require{A stream $\SS$ of length $m$ over universe $\UU = [n]$; let $f(x)$ be the frequency of $x\in \UU$ in $\SS$}    
    \Ensure{A set $X\subseteq \UU$ and a function $\hat{f}:X \rightarrow \NB$ such that if $f(x)\ge \varphi m$, then $x\in X$ and $f(x) -\eps m \le \hat{f}(x) \le f(x) + \eps m$ and if $f(y) \le (\phi-\eps)m$, then $y\notin X$ for every $x, y \in \UU$}
    \Initialize{\vspace{-2ex}
     \Let{$\ell$}{$10^5  \eps^{-2}$}\label{alg:ell}
     \Let{$s$}{0}
     \State Hash functions $h_1, \dots, h_{200\log(12\phi^{-1})}$ uniformly at random from a universal family $\HH \subseteq \{h:[n]\to [\nfrac{100}{\eps}]\}$.
     \State An empty table $\TT_1$ of (key, value) pairs of length ${2}{\phi^{-1}}$. Each key entry of $\TT_1$ can store an element of $[n]$ and each value entry can store an integer in $[0, 10\ell]$.\label{alg:m_table}
     \State An empty table $\TT_2$ with $100 \eps^{-1}$ rows and $200 \log (12\phi^{-1})$ columns. Each entry of $\TT_2$ can store an integer in $[0, 100\eps \ell]$.
     \State An empty $3$-dimensional table $\TT_3$ of size at most $100 \eps^{-1} \times 200 \log(12\phi^{-1}) \times 4 \log(\eps^{-1})$. Each entry of $\TT_3$ can store an  integer in $[0, 10 \ell]$. \Comment These are upper bounds; not all the allowed cells will actually be used. 
    }
    ~\\
    \Procedure{Insert}{$x$}
      \State With probability $\nfrac{\ell}{m}$, increment $s$ and continue. Else, \Return\label{alg:p}
      \State Perform Misra-Gries update on $\TT_1$ with $x$.\label{alg:u1}
      \For{$j \gets 1 \text{ to } 200\log(12\phi^{-1})$}
      	        \State $i \gets h_j(x)$
	   \State With probability $\eps$, increment $\TT_2[i,j]$
      \State $t \gets \lfloor \log(10^{-6} \TT_2[i,j]^{2})\rfloor$ and $p \gets \min(\eps \cdot 2^{t}, 1)$\label{alg:tpcomp} 
	      \If{$t \geq 0$}
	      \State With probability $p$, increment $\TT_3[i, j, t]$ \label{alg:t3inc}
	      \EndIf 
      \EndFor
      \EndProcedure
 ~\\   
    \Procedure{Report}{~}
    \State $X \gets \emptyset$
    \For{{each key} $x \text{ with nonzero value in } \TT_1$}
    \For{$ j \gets 1 \text{ to } 200 \log(12\phi^{-1})$}
    \State $\hat{f}_j(x) \gets \sum_{t=0}^{4 \log(\eps^{-1})} \TT_3[h(x), j, t]/\min(\eps 2^{t}, 1)$ \label{lin:fj}
    \EndFor
    \State $\hat{f}(x) \gets \text{median}(\hat{f}_1, \dots, \hat{f}_{10 \log \phi^{-1}})$ \label{lin:fhat}
    \If{$\hat{f}(x) \geq (\phi -\eps/2)s$}
    \State $X \gets X \cup \{x\}$    
   \EndIf 
    \EndFor
    \State \Return $X, \hat{f}$
    \EndProcedure
  \end{algorithmic}
\end{algorithm}

\paragraph{Overview}
As in the simpler algorithm, we sample $\ell = O(\eps^{-2})$ many stream elements and solve the $(\eps/2,\phi)$-\textsc{List heavy hitters} problem on this sampled stream. Also, the Misra-Gries algorithm for $(\phi/2, \phi)$-\textsc{List heavy hitters} returns a candidate set of $O(\phi^{-1})$ items containing all items of frequency at least $\phi \ell$. It remains to count the frequencies of these $O(\phi^{-1})$ items with upto $\eps \ell/2 = O(\eps^{-1})$ {\em additive} error, so that we can remove those whose frequency is less than $(\phi - \eps/2)\ell$.

Fix some item $i \in [n]$, and let $f_i$ be $i$'s count in the sampled stream. A natural approach to count $f_i$ approximately is to increment a 
 counter probabilistically, instead of deterministically, at every occurrence of $i$. Suppose that we increment a counter with probability $0\leq p_i \leq 1$ whenever item $i$ arrives in the stream. Let the value of the counter be $\hat{c}_i$, and let $\hat{f}_i = \hat{c}_i/p_i$. We see that $\E{\hat{f}_i} = f_i$ and $\mathsf{Var}[\hat{f}_i]  \leq f_i/p_i$. It follows that if $p_i = \Theta(\eps^2 f_i)$, then $\mathsf{Var}[\hat{f}_i] = O(\eps^{-2})$, and hence, $\hat{f}_i$ is an unbiased estimator of $f_i$ with additive error $O(\eps^{-1})$  with constant probability. We call such a counter an {\em accelerated counter} as the probability of incrementing accelerates with increasing counts. For each $i$, we can maintain $O(\log \phi^{-1})$ accelerated counters independently and take their median to drive the probability of deviating by more than $O(\eps^{-1})$ down to $O(\phi)$. So, with constant probability, the frequency for each of the $O(\phi^{-1})$ items in the Misra-Gries data structure is estimated within $O(\eps^{-1})$ error, as desired.

However, there are two immediate issues with this approach. The first problem is that we may need to keep counts for $\Omega(\ell) = \Omega(\eps^{-2})$ distinct items, which is too costly for our purposes. To get around this, we use a hash function from a universal family to hash the universe to a space of size $u = \Theta(\eps^{-1})$, and we work throughout with the hashed id's. We can then show that the space complexity for each iteration is $O(\eps^{-1})$. Also, the accelerated counters now estimate frequencies of hashed id's instead of actual items, but because of universality, the expected frequency of any  hashed id is $\ell/u = O(\eps^{-1})$, our desired error bound.

The second issue is that we need a constant factor approximation of $f_i$, so that we can set $p_i$ to $\Theta	(\eps^2 f_i)$. But because  the algorithm needs to be one-pass, we cannot first compute $p_i$ in one pass and then run the accelerated counter in another. So, we divide the stream into {\em epochs} in which  $f_i$ stays within a factor of $2$, and use a different $p_i$ for each epoch. In particular, set $p_i^t = \eps \cdot 2^{t}$ for $0 \leq t \leq \log(p_i/\eps)$.  We want to keep a running estimate of $i$'s count to within a factor of $2$ to know if the current epoch should be incremented.
For this, we subsample each element of the stream with probability $\eps$ independently  and maintain exact counts for the observed hashed id's. It is easy to see that this requires only $O(\eps^{-1})$ bits in expectation. Consider any $i \in [u]$ and the prefix of the stream upto $b \leq \ell$, and let $f_i(b)$ be $i$'s frequency in the prefix, let $\bar{c}_i(b)$ be $i$'s frequency among the samples in the prefix, and $\bar{f}_i(b) = \frac{\bar{c}_i(b)}{\eps}$. We see that $\E{\bar{f}_i(b)} = f_i(b)$, 
and $\mathsf{Var}[\bar{f}_i(b)] \leq \frac{f_i(b)}{\eps}$.  By Chebyshev, for any fixed $b$, $\Pr[|\bar{f}_i(b) - f_i(b)| > f_i(b)/\sqrt{2}] \leq \frac{2}{f_i(b) \eps}$, and hence, 
can show that $\bar{f}_i(b)$ is a $\sqrt{2}$-factor approximation of $f_i(b)$ with probability $1 - O((f_i(b) \eps)^{-1})$. Now, let $p_i(b) = \Theta(\eps^2 f_i(b))$, and for any epoch $t$, set $b_{i,t} = \min\{b: p_i(b) > p_i^{t-1}\}$. The last makes sense because $p_i(b)$ is non-decreasing with $b$. Also, note that $f_i(b_{i,t}) = \Omega(2^{t/2}/\eps)$. So, by the union bound, the probability that there exists $t$ for which $\bar{f}_i(b_{i,t})$ is not a $\sqrt2$-factor approximation of $f_i(b_{i,t})$ is at most $\sum_t \frac{1}{\Omega(f_i(b_{i,t}) \eps)} = \sum_t \frac{1}{\Omega(2^{t/2})}$, a small constant. In fact, it follows then that with constant probability, for all $b \in [\ell]$,  $\bar{f}_i(b)$ is a $2$-factor approximation of $f_i(b)$.
Moreover, we show that for any $b \in [\ell]$, $\bar{f}_i(b)$ is a $4$-factor approximation of $f_i(b)$ with constant probability. By repeating $O(\log \phi^{-1})$ times independently and taking the median, the error probability can be driven down to $O(\phi)$. 

Now, for every hashed id $i \in [u]$, we need not one accelerated counter but $O(\log(\eps f_i))$ many, one corresponding to each epoch $t$. When an element with hash id $i$ arrives at position $b$, we decide, based on $\bar{f}_i(b)$, the epoch $t$ it belongs to and then increment the $t$'th accelerated counter with probability $p_i^t$.  The storage cost over all $i$ is still $O(1/\eps)$. Also, we iterate the whole set of accelerated counters $O(\log \phi^{-1})$ times, making the total storage cost $O(\eps^{-1}\log \phi^{-1})$.

Let $\hat{c}_{i,t}$ be the count in the accelerated counter for hash id $i$ and epoch $t$. Then, let $\hat{f}_i = \sum_t {\hat{c}_{i,t}}/{p_i^t}$. Clearly, $\E{\hat{f}_i} = f_i$. The variance is $O(\eps^{-2})$ in each epoch, and so, $\mathsf{Var}[\hat{f}_i]=O(\eps^{-2} \log \eps^{-1})$, not $O(\eps^{-2})$ which we wanted. This issue is fixed by a change in how the sampling probabilities are defined. We now go on to the formal proof.

\begin{proof}[Proof of \cref{thm:heavy_hitters2}]
Pseudocode appears in \cref{alg:heavy_hitters2}. Note that the numerical constants are chosen for convenience of analysis and have not been optimized. Also, for the sake of simplicity, the pseudocode does not have the optimal reporting time, but it can be modified to achieve this; see the end of this proof for details.

By standard Chernoff bounds, with probability at least $99/100$, the length of the sampled stream $\ell/10 \leq s \leq 10\ell$. For $x \in [n]$, let $f_{\text{samp}}(x)$ be the frequency of $x$ in the sampled stream. By \cref{lem:dkw}, with probability at least $9/10$, for all $x \in [n]$:
$$\left|\frac{f_{\text{samp}}(x)}{s} - \frac{f(x)}{m}\right| \leq \frac{\eps}{4}$$
Now, fix $j \in [10 \log \phi^{-1}]$ and $x \in [n]$. Let $i = h_j(x)$ and $f_i = \sum_{y: h_j(y) = h_j(x)} f_{\text{samp}}(y)$. Then, for a random $h_j \in \mathcal{H}$, the expected value of $\frac{f_i}{s} - \frac{f_{\text{samp}}(x)}{s}$ is $\frac{\eps}{100}$, since $\mathcal{H}$ is a universal mapping to a space of size $100\eps^{-1}$. Hence, using Markov's inequality and the above:
\begin{equation}\Pr\left[\left| \frac{f(x)}{m} - \frac{f_i}{s}\right| \geq \frac{\eps}{2}\right] 
\leq \Pr\left[\left| \frac{f(x)}{m} - \frac{f_{\text{samp}}}{s}\right| \geq \frac{\eps}{4}\right] + \Pr\left[\left| \frac{f_{\text{samp}}(x)}{m} - \frac{f_i}{s}\right| \geq \frac{\eps}{4}\right]
< \frac{1}{10}  + \frac{1}{25} < \frac{3}{20}\end{equation}
In Lemma \ref{lem:pmain} below, we show that for each $j \in [200 \log(12\phi^{-1})]$, with error probability at most $3/10$, $\hat{f}_j(x)$ (in line \ref{lin:fj}) estimates $f_i$ with additive error at most $5000\eps^{-1}$, hence estimating $\frac{f_i}{s}$ with additive error at most $\frac{\eps}{2}$. Taking the median over $200 \log( 12\phi^{-1})$ repetitions (line \ref{lin:fhat}) makes the error probability go down to $\frac{\phi}{6}$ using standard Chernoff bounds. Hence, by the union bound, with probability at least $2/3$, for each of the $2/\phi$ keys $x$ with nonzero values in $\TT_1$, we have an estimate of $\frac{f(x)}{m}$ within additive error $\eps$, thus showing correctness.

\begin{lemma}\label{lem:pmain}
Fix $x \in [n]$ and $j \in [200 \log 12\phi^{-1}]$, and let $i = h_j(x)$. Then, $\Pr[|\hat{f}_j(x) - f_i| > 5000 \eps^{-1}] \leq 3/10$, where $\hat{f}_j$ is the quantity computed in line \ref{lin:fj}.
\end{lemma}
\begin{proof}
Index the sampled stream elements $1, 2, \dots, s$, and for $b \in [s]$, let $f_i(b)$ be the frequency of items with hash id $i$ restricted to the first $b$ elements of the sampled stream. Let $\bar{f}_i(b)$ denote the value of $\TT_2[i,j]\cdot \eps^{-1}$ after the procedure \textsc{Insert} has been called for the first $b$ items of the sampled stream.
\begin{claim}\label{lem:approx}
With probability at least $9/10$, for all $b \in [s]$ such that $f_i(b) \geq 100 \eps^{-1}$, $\bar{f}_i(b)$ is within a factor of 4 of $f_i(b)$.
\end{claim}
\begin{proof}
Fix $b \in [s]$. Note that $\E{\bar{f}_i(b)} = f_i(b)$ as $\TT_2$ is incremented with rate $\eps$. $\mathsf{Var}[\bar{f}_i(b)] \leq f_i/\eps$, and so by Chebyshev's inequality:
$$\Pr[|\bar{f}_i(b) - f_i(b)| > f_i(b)/2] < \frac{4}{f_i(b) \eps}$$
We now break the stream into chunks, apply this inequality to each chunk and then take a union bound to conclude. Namely, for any integer $t\geq 0$, define $b_t$ to be the first	 $b$ such that $100 \eps^{-1} 2^{t} \leq f_i(b) < 100 \eps^{-1} 2^{t+1}$ if such a $b$ exists. Then:
\begin{align*}
\Pr[\exists t \geq 0: |\bar{f}_i(b_t) - f_i(b_t)| > f_i(b_t)/2] &< \sum_t \frac{4}{100 \cdot 2^{t-1}}\\ &< \frac{1}{10}
 \end{align*} 
 So, with probability at least $9/10$, every $\bar{f}_i(b_t)$ and $f_i(b_t)$ are within a factor of $2$ of each other. Since for every $b\geq b_0$, $f_i(b)$ is within a factor of ${2}$ from some $f_i(b_t)$, the claim follows.
\end{proof}
Assume the event in \cref{lem:approx} henceforth. Now, we are ready to analyze $\TT_3$ and in particular, $\hat{f}_j(x)$. First of all, observe that if $t<0$ in line \ref{alg:tpcomp}, at some position $b$ in the stream, then $\TT_2[i,j]$ at that time must be at most 1000, and so by standard Markov and Chernoff bounds,  with probability at least $0.85$, 
\begin{equation}\label{eqn:ass}
f_i(b) 
\begin{cases}
< 4000 \eps^{-1}, & \text{ if } t<0\\
> 100 \eps^{-1}, & \text { if } t \geq 0
\end{cases}
\end{equation}
Assume this event. Then, $f_i - 4000 \eps^{-1} \leq \E{\hat{f}_j(x)}  \leq f_i$.
\begin{claim}
$$\mathsf{Var}(\hat{f}_j(x)) \leq {20000}{\eps^{-2}}$$
\end{claim}
\begin{proof}
If the stream element at position $b$ causes an increment in $\TT_3$ with probability $\eps2^t$ (in line \ref{alg:t3inc}), then $1000 \cdot 2^{t/2} \leq \TT_2[i,j] \leq 1000\cdot 2^{(t+1)/2}$, and so, $\bar{f}_i(b) \leq 1000\eps^{-1} 2^{(t+1)/2}$.  This must be the case for the highest $b = \bar{b}_t$ at which the count for $i$  in $\TT_3$ increments at the $t$'th slot. The number of such occurrences of $i$ is at most $f_i(\bar{b}_t) \leq 4 \bar{f_i}(\bar{b}_t)\leq 4000 \eps^{-1}2^{(t+1)/2}$ by \cref{lem:approx} (which can be applied since  $f_i(b) > 100\eps^{-1}$ by Equation \ref{eqn:ass}). So:
\begin{align*}
\mathsf{Var}[\hat{f}_j(x)] \leq \sum_{t \geq 0} \frac{f_i(\bar{b}_t)}{\eps 2^t} \leq \sum_{t \geq 0} \frac{4000}{\eps^2} 2^{-t/3}\leq {20000}{\eps^{-2}}
\end{align*}
Elements inserted with probability $1$ obviously do not contribute to the variance.
\end{proof}
So, conditioning on the events mentioned, the probability that $\hat{f}_j(x)$ deviates from $f_i$ by more than $5000\eps^{-1}$ is at most $1/50$. Removing all the conditioning yields what we wanted:
$$\Pr[|\hat{f}_j(x) - f_i| > 5000\eps^{-1}] \leq \frac{1}{50} + \frac{3}{20} + \frac{1}{10} \leq 0.3$$ 
\end{proof}

We next bound the space complexity.
\begin{claim}\label{lem:spacehh}
With probability at least $2/3$, \cref{alg:heavy_hitters2} uses $O(\eps^{-1} \log \phi^{-1} + \phi^{-1} \log n + \log \log m)$ bits of storage, if $n = \omega(\eps^{-1})$.
\end{claim}
\begin{proof}
The expected length of the sampled stream is $\ell = O(\eps^{-2})$. So, the number of bits stored in $\TT_1$ is $O(\phi^{-1} \log n)$. For $\TT_2$, note that in lines 13-15, for any given $j$, $\TT_2$ is storing a total of $\eps \ell = O(\eps^{-1})$ elements in expectation. So, for $k \geq 0$, there can be at most $O((\eps 2^k)^{-1})$ hashed id's with counts between $2^k$ and $2^{k+1}$. Summing over all $k$'s and accounting for the empty cells gives $O(\eps^{-1})$ bits of storage, and so the total space requirement of $\TT_2$ is $O(\eps^{-1} \log \phi^{-1})$. .

The probability that a hashed id $i$ gets counted in table $\TT_3$ is at most $10^{-6}\eps^3 \bar{f}_i^2(s)$ from line \ref{alg:tpcomp} and our definition of $\bar{f}_i$ above. Moreover, from \cref{lem:approx}, we have that this is at most $16 \cdot 10^{-6} \eps^3 {f}_i^2(s)$  if $f_i > 100 \eps^{-1}$. Therefore, if $f_i = 2^k\cdot 100 \eps^{-1}$ with $k \geq 0$, then the expected value of a cell in $\TT_3$ with first coordinate $i$ is at most $1600\cdot 2^{2k} \eps = 2^{O(k)}$. Taking into account that there are at most $O((\eps 2^k)^{-1})$ many such id's $i$ and that the number of epochs $t$ associated with such an $i$ is at most $\log(16 \cdot 10^{-6} \eps^2 {f}_i^2) = O(\log(\eps f_i)) = O(k)$ (from line \ref{alg:tpcomp}), we get that the total space required for $\TT_3$ is:
\begin{align*}&\sum_{j = 1}^{O(\log \phi^{-1})} \left(O(\eps^{-1}) + \sum_{k=0}^\infty O((\eps 2^k)^{-1}) \cdot O(k) \cdot O(k)\right) \\ &= O(\eps^{-1} \log \phi^{-1})\end{align*}
where the first $O(\eps^{-1})$ term inside the summation is for the $i$'s with $f_i < 100 \eps^{-1}$. Since we have an expected space bound, we obtain a worst-case space bound with error probability $1/3$ by a Markov bound. 

The space required for sampling is an additional $O(\log \log m)$, using \cref{lem:sampling_ub}.
\end{proof}
We note that the space bound can be made worst case by aborting the algorithm if it tries to use more space.

The only remaining aspect of \cref{thm:heavy_hitters2} is the time complexity. As observed in \cref{sec:lhh}, the update time can be made $O(1)$ per insertion under the standard assumption of the stream being sufficiently long. The reporting time can also be made linear in the output by changing the bookkeeping a bit. Instead of computing $\hat{f}_j$ and $\hat{f}$ at reporting time, we can maintain them after every insertion. Although this apparently makes INSERT costlier, this is not true in fact because we can spread the cost over future stream insertions. The space complexity grows by a constant factor.
\end{proof}

\subsection{$\eps$-Maximum Problem}

By tweaking \Cref{alg:heavy_hitters} slightly, we get the following result for the {\sc $\eps$-Maximum} problem. 

\begin{restatable}{theorem}{ThmHeavyHittersWeak}\label{thm:heavy_hitters_weak}
 Assume the length of the stream is known beforehand. Then there is a randomized one-pass algorithm $\mathcal{A}$ for the \textsc{$\eps$-Maximum} problem which succeeds with probability at least $1-\delta$ using $O\left(\min\{\nfrac{1}{\eps}, n\}(\log\nfrac{1}{\eps} + \log\log\nfrac{1}{\delta}) + \log n + \log\log m \right)$ bits of space. Moreover, the algorithm $\mathcal{A}$ has an update time of $O(1)$.
\end{restatable}

\begin{proof}
 Instead of maintaining the table $\TT_2$ in \Cref{alg:heavy_hitters}, we just store the actual id of the item with maximum frequency in the sampled items. 
\end{proof}

\subsection{$\eps$-Minimum Problem}
\begin{restatable}{theorem}{ThmRare}\label{thm:rare}
 Assume the length of the stream is known beforehand. Then there is a randomized one-pass algorithm $\mathcal{A}$ for the \textsc{$\eps$-Minimum} problem which succeeds with probability at least $1-\delta$ using $O\left((\nfrac{1}{\eps})\log\log(\nfrac{1}{\eps\delta}) + \log\log m \right)$ bits of space. Moreover, the algorithm $\mathcal{A}$ has an update time of $O(1)$.
\end{restatable}

\paragraph{Overview}

Pseudocode is provided in \cref{alg:rare}.
The idea behind our $\epsilon$-Minimum problem is as follows. It is most easily explained by looking at the REPORT(x) procedure starting in line 13. In lines 14-15 we ask, is the universe size $|U|$ significantly larger than $1/\epsilon$? Note that if it is, then outputting a random item from $|U|$ is likely to be a solution. Otherwise $|U|$ is $O(1/\epsilon)$.

The next point is that if the number of distinct elements in the stream were smaller than $1/(\epsilon \log(1/\epsilon))$, then we could just store all the items together with their frequencies with $O(1/\epsilon)$ bits of space. Indeed, we can first sample $O(1/\epsilon^2)$ stream elements so that all relative frequencies are preserved up to additive $\epsilon$, thereby ensuring each frequency can be stored with $O(\log(1/\epsilon)$ bits. Also, since the universe size is $O(1/\epsilon)$, the item identifiers can also be stored with $O(\log(1/\epsilon)$ bits. So if this part of the algorithm starts taking up too much space, we stop, and we know the number of distinct elements is at least $1/(\epsilon \log(1/\epsilon))$, which means that the minimum frequency is at most $O(m \epsilon \log(1/\epsilon))$. This is what is being implemented in steps 9-10 and 18-19 in the algorithm.

We can also ensure the minimum frequency is at least $\Omega(m \epsilon / \log(1/\epsilon))$. Indeed, by randomly sampling $O((\log(1/\epsilon)/\epsilon)$ stream elements, and maintaining a bit vector for whether or not each item in the universe occurs - which we can with $O(1/\epsilon)$ bits of space since $|U| = O(1/\epsilon)$ - any item with frequency at least $\Omega(\epsilon m/ \log(1/\epsilon))$ will be sampled and so if there is an entry in the bit vector which is empty, then we can just output that as our solution. This is what is being implemented in steps 8 and 16-17 of the algorithm.

Finally, we now know that the minimum frequency is at least $\Omega(m \epsilon / \log(1/\epsilon))$ and at most $O(m \epsilon \log(1/\epsilon))$. At this point if we randomly sample $O((\log^6 1/\epsilon)/\epsilon)$ stream elements, then by Chernoff bounds all item frequencies are preserved up to a relative error factor of $(1 \pm 1/\log^2 (1/\epsilon))$, and in particular the relative minimum frequency is guaranteed to be preserved up to an additive $\epsilon$. At this point we just maintain the exact counts in the sampled stream but truncate them once they exceed $\textrm{poly}(\log(1/\epsilon)))$ bits, since we know such counts do not correspond to the minimum. Thus we only need $O(\log \log (1/\epsilon))$ bits to represent their counts. This is implemented in step 11 and step 20 of the algorithm. 

\begin{algorithm}[!htbp]
  \caption{for \textsc{$\eps$-Minimum}
    \label{alg:rare}}
  \begin{algorithmic}[1]
    \Require{A stream $\SS = (x_i)_{i\in[m]}\in \UU^m$ of length $m$ over $\UU$; let $f(x)$ be the frequency of $x\in \UU$ in $\SS$}    
    \Ensure{An item $x\in\UU$ such that $f(x) \le f(y) + \eps m$ for every $y\in\UU$}
    \Initialize{
    \Let{$\ell_1$}{$\nfrac{\log(\nfrac{6}{\eps\delta})}{\eps}$}, $\ell_2 \leftarrow \nfrac{\log(\nfrac{6}{\delta})}{\eps^2}$, $\ell_3 \leftarrow \nfrac{\log^6 (\nfrac{6}{\delta\eps})}{\eps} $
    \Let{$p_1$}{$\nfrac{6\ell_1}{m}$}, $p_2 \leftarrow \nfrac{6\ell_2}{m}$, $p_3 \leftarrow \nfrac{6\ell_3}{m}$
    \Let{$\SS_1, \SS_2, \SS_3$}{$\emptyset$}
    \Let{$\BB_1$}{the bit vector for $\SS_1$}}
    ~\\
    \Procedure{Insert}{x}
      \State Put $x$ in $\SS_1$ with probability $p_1$ by updating the bit vector $\BB_1$
      \If{the number of distinct items in the stream so far is at most $\nfrac{1}{(\eps\log(\nfrac{1}{\eps}))}$}
	\State Pick $x$ with probability $p_2$ and put the id of $x$ in $\SS_2$ and initialize the corresponding counter to $1$ if $x\notin\SS_2$ and increment the counter corresponding to $x$ by $1$.
      \EndIf
      \State Pick $x$ with probability $p_3$, put the id of $x$ in $\SS_3$ and initialize the corresponding counter to $1$ if $x_i\notin\SS_3$ and increment the counter corresponding to $x_i$ by $1$. Truncate counters of $\SS_3$ at $2\log^7(\nfrac{2}{\eps\delta})$.
    \EndProcedure
    ~\\
    
    \Procedure{Report}{~}
    \If{$|\UU|\ge \nfrac{1}{((1-\delta)\eps)}$}
      \State \Return an item $x$ from the first $\nfrac{1}{((1-\delta)\eps)}$ items in $\UU$ (ordered arbitrarily) uniformly at random
      \EndIf
   \If{$\SS_1 \ne \UU$}
      \State \Return{any item from $\UU\setminus\SS_1$}\label{alg:S1_out}
    \EndIf
    \If{the number of distinct items in the stream is at most $\nfrac{1}{(\eps\log(\nfrac{1}{\eps}))}$}
      \State \Return an item in $\SS_2$ with minimum counter value in $\SS_2$\label{alg:S2_out}
    \EndIf
    \State \Return the item with minimum frequency in $\SS_3$
    \EndProcedure
  \end{algorithmic}
\end{algorithm}

\begin{proof}[Proof of \cref{thm:rare}]
 The pseudocode of our \textsc{$\eps$-Minimum} algorithm is in \Cref{alg:rare}. If the size of the universe $|\UU|$ is at least $\nfrac{1}{((1-\delta)\eps)}$, then we return an item $x$ chosen from $\UU$ uniformly at random. Note that there can be at most $\nfrac{1}{\eps}$ many items with frequency at least $\eps m$. Hence every item $x$ among other remaining $\nfrac{\delta}{((1-\delta)\eps)}$ many items has frequency less than $\eps m$ and thus is a correct output of the instance. Thus the probability that we answer correctly is at least $(1-\delta)$. From here on, let us assume $|\UU|<\nfrac{1}{((1-\delta)\eps)}$.
 
 Now, by the value of $p_j$, it follows from the proof of \Cref{thm:heavy_hitters} that we can assume $\ell_j < |\SS_j| < 11\ell_j$ for $j = 1, 2, 3$ which happens with probability at least $(1-(\nfrac{\delta}{3}))$. We first show that every item in $\UU$ with frequency at least $\eps m$ is sampled in $\SS_1$ with probability at least $(1-(\nfrac{\delta}{6}))$. For that, let $X_i^j$ be the indicator random variable for the event that the $j^{th}$ sample in $\SS_1$ is item $i$ where $i\in\UU$ is an item with frequency at least $\eps m$. Let $\HH\subset\UU$ be the set of items with frequencies at least $\eps m$. Then we have the following.
 
 \[ \Pr[X_i^j = 0] = 1-\eps \Rightarrow \Pr[X_i^j = 0 ~\forall j\in\SS_1] \le (1-\eps)^{\ell_1} \le \exp\{-\eps\ell_1\} = \nfrac{\eps\delta}{6} \]
 
 Now applying union bound we get the following.
 \[ \Pr[ \exists i\in \HH, X_i^j = 0 ~\forall j\in\SS_1] \le (\nfrac{1}{\eps})\nfrac{\eps\delta}{6} \le \nfrac{\delta}{6} \]
 Hence with probability at least $(1-(\nfrac{\delta}{3})-(\nfrac{\delta}{6})) \ge (1-\delta)$, the output at line \ref{alg:S1_out} is correct. Now we show below that if the frequency of any item $x\in\UU$ is at most $\nfrac{\eps\ln(\nfrac{6}{\delta})}{\ln(\nfrac{6}{\eps \delta})}$, then $x\in \SS_1$ with probability at least $(1-(\nfrac{\delta}{6}))$.
 \[ \Pr[ x\notin \SS_1 ] = (1-\nfrac{\eps\ln (\nfrac{6}{\delta})}{\ln (\nfrac{6}\eps\delta)})^{\nfrac{\ln(\nfrac{6}{\eps\delta})}{\eps}} \le \nfrac{\delta}{6} \]
 Hence from here onwards we assume that the frequency of every item in $\UU$ is at least $\nfrac{\eps m\ln (\nfrac{6}{\delta})}{\ln (\nfrac{6}\eps\delta)}$. 
 
 If the number of distinct elements is at most $\nfrac{1}{(\eps\ln(\nfrac{1}{\eps}))}$, then line \ref{alg:S2_out} outputs the minimum frequency item up to an additive factor of $\eps m$ due to Chernoff bound. Note that we need only $O(\ln(\nfrac{1}{((1-\delta)\eps)}))$ bits of space for storing ids. Hence $\SS_2$ can be stored in space $O((\nfrac{1}{\eps\ln(\nfrac{1}{\eps})}) \ln(\nfrac{1}{((1-\delta)\eps)} \ln\ln(\nfrac{1}{\delta})) = O(\nfrac{1}{\eps} \ln\ln(\nfrac{1}{\delta}))$.
 
 Now we can assume that the number of distinct elements is at least $\nfrac{1}{(\eps\ln(\nfrac{1}{\eps}))}$. Hence if $f(t)$ is the frequency of the item $t$ with minimum frequency, then we have $ m\nfrac{\eps}{\ln (\nfrac{1}{\eps})} \le f(t) \le m\eps \ln(\nfrac{1}{\eps})$. 
 
 Let $f_i$ be the frequency of item $i\in \UU$, $e_i$ be the counter value of $i$ in $\SS_3$, and $\hat{f}_i = \nfrac{e_i m}{\ell_3}$. Now again by applying Chernoff bound we have the following for any fixed $i\in\UU$.
\begin{eqnarray*}
\Pr[ |f_i - \hat{f}_i| > \nfrac{f_i}{\ln^2(\nfrac{1}{\eps})} ] 
&\le & 2\exp\{ -\nfrac{\ell_3 f_i}{(m \ln^4(\nfrac{1}{\eps}))} \}\\
& \le & 2 \exp\{ -\nfrac{f_i\ln^2 (\nfrac{6}{\eps\delta})}{(\eps m)} \}\\
& \le & \nfrac{\eps\delta}{6}.
\end{eqnarray*}
 Now applying union bound we get the following using the fact that $|\UU|\le \nfrac{1}{\eps (1-\delta)}$.
 \[ \Pr[ \forall i\in\UU, |f_i - \hat{f}_i| \le \nfrac{f_i}{\ln^2(\nfrac{1}{\eps})} ] > 1-\nfrac{\delta}{6} \]
 Again by applying Chernoff bound and union bound we get the following.
 \[ \Pr[ \forall i\in\UU \text{ with } f_i > 2m\eps\ln(\nfrac{1}{\eps}), |f_i - \hat{f}_i| \le \nfrac{f_i}{2} ] > 1-\nfrac{\delta}{6} \]
 Hence the items with frequency more than $2m\eps\ln(\nfrac{1}{\eps})$ are approximated up to a multiplicative factor of $\nfrac{1}{2}$ from below in $\SS_3$. The counters of these items may be truncated. The other items with frequency at most $2m\eps\ln(\nfrac{1}{\eps})$ are be approximated up to $(1 \pm \nfrac{1}{\ln^2(\nfrac{1}{\eps})})$ relative error and thus up to an additive error of $\nfrac{\eps m}{3}$. The counters of these items would not get truncated. Hence the item with minimum counter value in $\SS_3$ is the item with minimum frequency up to an additive $\eps m$.
 
 We need $O(\ln(\nfrac{1}{\eps\delta}))$ bits of space for the bit vector $\BB_1$ for the set $\SS_1$. We need $O(\ln^2(\nfrac{1}{\eps\delta}))$ bits of space for the set $\SS_2$ and $O((\nfrac{1}{\eps})\ln\ln(\nfrac{1}{\eps\delta}))$ bits of space for the set $\SS_3$ (by the choice of truncation threshold). We need an additional $O\left( \ln\ln m \right)$ bits of space for sampling using \Cref{lem:sampling_ub}. Moreover, using the data structure of Section 3.3 of \citep{demaine2002frequency} \Cref{alg:rare} can be performed in $O(1)$ time. Alternatively, we may also use the strategy described in \cref{sec:lhh} of spreading update operations over several insertions to make the cost per insertion be $O(1)$.
\end{proof}

\subsection{Problems for the Borda and Maximin Voting Rules}

\begin{restatable}{theorem}{ThmBorda}\label{thm:borda}
 Assume the length of the stream is known beforehand. Then there is a randomized one-pass algorithm $\mathcal{A}$ for \textsc{$(\eps, \varphi)$-List Borda} problem which succeeds with probability at least $1-\delta$ using $O\left( n\left( \log n + \log\frac{1}{\eps} + \log\log\frac{1}{\delta} \right) + \log\log m \right)$ bits of space.
\end{restatable} 

\begin{proof}
Let $\ell = 6\eps^{-2} \log(6n\delta^{-1})$ and $p = \nfrac{6\ell}{m}$.  On each insertion of a vote $v$, select $v$ with probability $p$ and store for every $i \in [n]$, the number of candidates that candidate $i$ beats in the vote $v$. Keep these exact counts in a counter of length $n$.

Then it follows from the proof of \Cref{thm:heavy_hitters} that $\ell \le |\SS| \le 11\ell$ with probability at least $(1-{\delta}/{3})$. Moreover, from a straightforward application of the Chernoff bound (see \cite{DeyB15}),  it follows that if $\hat{s}(i)$ denotes the Borda score of candidate $i$ restricted to the sampled votes, then:
$$\Pr\left[\forall i \in [n], \left|\frac{m}{|\SS|} \hat{s}(i) - s(i)\right| < \eps m n\right] > 1-\delta$$

The space complexity for exactly storing the counts is $O(n \log (n \ell)) = O(n (\log n + \log \eps^{-1} + \log \log \delta^{-1}))$ and the space for sampling the votes is $O(\log \log m)$ by \cref{lem:sampling_ub}.
\end{proof}

\begin{restatable}{theorem}{ThmMaximinUB}\label{thm:maximin}
 Assume the length of the stream is known beforehand. Then there is a randomized one-pass algorithm $\mathcal{A}$ for \textsc{$(\eps, \varphi)$-List maximin} problem which succeeds with probability at least $1-\delta$ using $O\left(n \eps^{-2} \log^2 n + n\eps^{-2} \log n \log \delta^{-1} + \log\log m \right)$ bits of space.
\end{restatable}
\begin{proof}
  Let $\ell=(\nfrac{8}{\eps^2})\ln(\nfrac{6n}{\delta})$ and $p = \nfrac{6\ell}{m}$. We put the current vote in a set $\SS$ with probability $p$. Then it follows from the proof of \Cref{thm:heavy_hitters} that $\ell \le |\SS| \le 11\ell$ with probability at least $(1-\nfrac{\delta}{3})$. Suppose $|\SS| = \ell_1$; let $\mathcal{S} = \{ v_i : i\in[\ell_1] \}$ be the set of votes sampled. Let $D_\EE(x,y)$ be the total number of votes in which $x$ beats $y$ and $D_\SS(x, y)$) be the number of such votes in $\SS$. Then by the choice of $\ell$ and the Chernoff bound (see \cite{DeyB15}), it follows that $|D_\SS(x,y)\nfrac{m}{\ell_1} - D_\EE(x,y)| \le \nfrac{\eps m}{2}$ for every pair of candidates $x, y \in \UU$. Note that each vote can be stored in $O(n\log n)$ bits of space. Hence simply finding $D_\SS(x,y)$ for every $x, y\in \UU$ by storing $\SS$ and returning all the items with maximin score at least $(\phi - \eps/2) \ell_1$ in $\SS$ requires $O\left( n\eps^{-2} \log n (\log n + \log \delta^{-1})  + \log \log m \right)$ bits of memory, with the additive $O(\log \log m)$ due to \cref{lem:sampling_ub}.
\end{proof}

\subsection{Algorithms with Unknown Stream Length}\label{sec:unknown}

Now we consider the case when the length of the stream is not known beforehand. We present below an algorithm for {\sc $(\eps, \varphi)$-List heavy hitters} and {\sc $\eps$-Maximum} problems in the setting where the length of the stream is not known beforehand.

\begin{theorem}\label{thm:UbUnknownMax}
 There is a randomized one-pass algorithm for {\sc $(\eps, \varphi)$-List heavy hitters} and {\sc $\eps$-Maximum} problems with space complexity $O\left({\eps^{-1}}\log \eps^{-1} + \varphi^{-1}\log n + \log\log m \right)$ bits and update time $O(1)$ even when the length of the stream is not known beforehand.
\end{theorem}

\begin{proof}
 We describe below a randomized one-pass algorithm for the {\sc $(8\eps, \varphi)$-List heavy hitters} problem. We may assume that the length of the stream is at least $\nfrac{1}{\eps^2}$; otherwise, we use the algorithm in \Cref{thm:heavy_hitters} and get the result. Now we guess the length of the stream to be $\nfrac{1}{\eps^2}$, but run an instance $\II_1$ of \Cref{alg:heavy_hitters} with $\ell=\nfrac{\log(\nfrac{6}{\delta})}{\eps^3}$ at line \ref{alg:ell}. By the choice of the size of the sample (which is $\Theta(\nfrac{\log(\nfrac{1}{\delta})}{\eps^3})$), $\II_1$ outputs correctly with probability at least $(1-\delta)$, if the length of the stream is in $[\nfrac{1}{\eps^2},\nfrac{1}{\eps^3}]$. If the length of the stream exceeds $\nfrac{1}{\eps^2}$, we run another instance $\II_2$ of \Cref{alg:heavy_hitters} with $\ell=\nfrac{\log(\nfrac{6}{\delta})}{\eps^3}$ at line \ref{alg:ell}. Again by the choice of the size of the sample, $\II_2$ outputs correctly with probability at least $(1-\delta)$, if the length of the stream is in $[\nfrac{1}{\eps^3},\nfrac{1}{\eps^4}]$. If the stream length exceeds $\nfrac{1}{\eps^3}$, we discard $\II_1$, free the space it uses, and run an instance $\II_3$ of \Cref{alg:heavy_hitters} with $\ell=\nfrac{\log(\nfrac{6}{\delta})}{\eps^3}$ at line \ref{alg:ell} and so on. At any point of time, we have at most two instances of \Cref{alg:heavy_hitters} running. When the stream ends, we return the output of the older of the instances we are currently running. We use the approximate counting method of Morris \citep{morris1978counting} to approximately count the length of the stream. We know that the Morris counter outputs correctly with probability $(1-2^{-\nfrac{k}{2}})$ using $O(\log\log m + k)$ bits of space at any point in time \citep{flajolet1985approximate}. Also, since the Morris counter increases only when an item is read, it outputs correctly up to a factor of four at every position if it outputs correctly at positions $1, 2, 4, \ldots, 2^{\lfloor \log_2 m \rfloor}$; call this event $E$. Then we have $\Pr(E) \ge 1-\delta$ by choosing $k=2\log_2(\nfrac{\log_2 m}{\delta})$ and applying union bound over the positions $1, 2, 4, \ldots, 2^{\lfloor \log_2 m \rfloor}$. The correctness of the algorithm follows from the correctness of \Cref{alg:heavy_hitters} and the fact that we are discarding at most $\eps m$ many items in the stream (by discarding a run of an instance of \Cref{alg:heavy_hitters}). The space complexity and the $O(1)$ update time of the algorithm follow from \Cref{thm:heavy_hitters},  the choice of $k$ above, and the fact that we have at most two instances of \Cref{alg:heavy_hitters} currently running at any point of time.
 
 The algorithm for the {\sc $\eps$-Maximum} problem is same as the algorithm above except we use the algorithm in \Cref{thm:heavy_hitters_weak} instead of \Cref{alg:heavy_hitters}.
\end{proof}

Note that this proof technique does not seem to apply to our optimal \cref{alg:heavy_hitters2}. Similarly to \Cref{thm:UbUnknownMax}, we get the following result for the {\sc $\eps$-Minimum, $(\eps,\phi)$-Borda,} and {\sc $(\eps,\phi)$-Maximin} problems.

\begin{theorem}\label{thm:UbUnknownMin}
 There are randomized one-pass algorithms for {\sc $\eps$-Minimum, $(\eps,\phi)$-Borda,} and {\sc $(\eps,\phi)$-Maximin} problems with space complexity $O\left((\nfrac{1}{\eps})\log\log(\nfrac{1}{\eps\delta}) + \log\log m \right)$, $O\left( n\left( \log n + \log\frac{1}{\eps} + \log\log\frac{1}{\delta} \right) + \log\log m \right)$, and $O\left(n\eps^{-2} \log^2 n + n \eps^{-2} \log n \log(\nfrac{1}{\delta}) + \log\log m \right)$ bits respectively even when the length of the stream is not known beforehand. Moreover, the update time for {\sc $\eps$-Minimum} is $O(1)$.
\end{theorem}

\section{Results on Space Complexity Lower Bounds}\label{subsec:lwb}

In this section, we prove space complexity lower bounds for the {\sc $\eps$-Heavy hitters}, {\sc $\eps$-Minimum}, {\sc $\eps$-Borda}, and {\sc $\eps$-maximin} problems. We present reductions from certain communication problems for proving space complexity lower bounds. Let us first introduce those communication problems with necessary results.

\subsection{Communication Complexity}

\begin{definition}(\textsc{Indexing}$_{m,t}$)\\
 Let $t$ and $m$ be positive integers. Alice is given a string $x = (x_1, \cdots, x_t)\in [m]^t$. Bob is given an index $i\in [t]$. Bob has to output $x_i$.
\end{definition}

The following is a well known result~\citep{Kushilevitz}.

\begin{lemma}\label{lem:index}
$\mathcal{R}_\delta^{\text{1-way}}(\textsc{Indexing}_{m,t}) = \Omega(t \log m)$ for constant $\delta\in(0,1)$.
\end{lemma}

\begin{definition}(\textsc{Augmented-indexing}$_{m,t}$)\\
 Let $t$ and $m$ be positive integers. Alice is given a string $x = (x_1, \cdots, x_t)\in [m]^t$. Bob is given an integer $i\in [t]$ and $(x_1, \cdots, x_{i-1})$. Bob has to output $x_i$.
\end{definition}

The following communication complexity lower bound result is due to~\citep{ergun2010periodicity} by a simple extension of the arguments of Bar-Yossef et al \cite{bar2002information}.

\begin{lemma}\label{lem:aug}
$\mathcal{R}_\delta^{\text{1-way}}(\textsc{Augmented-indexing}_{m,t}) = \Omega((1-\delta)t \log m)$ for any $\delta < 1 - \frac{3}{2m}$.
\end{lemma}

\citep{SunW15} defines a communication problem called {\sc Perm}, which we generalize to {\sc $\eps$-Perm} as follows.

\begin{definition}({\sc $\eps$-Perm})\\
 Alice is given a permutation $\sigma$ over $[n]$ which is partitioned into $\nfrac{1}{\eps}$ many contiguous blocks. Bob is given an index $i\in[n]$ and has to output the block in $\sigma$ where $i$ belongs.
\end{definition}

Our lower bound for {\sc $\eps$-Perm} matches the lower bound for {\sc Perm} in Lemma $1$ in \citep{SunW15} when $\eps = \nfrac{1}{n}$. For the proof, the reader may find useful some information theory facts described in Appendix A.

\begin{restatable}{lemma}{LemPerm}\label{lem:perm}
 $\RR_\delta^{\text{1-way}} (\eps-\textsc{Perm}) = \Omega(n\log(\nfrac{1}{\eps}))$, for any constant $\delta < \nfrac{1}{10}$.
\end{restatable}

\begin{proof}
 Let us assume $\sigma$, the permutation Alice has, is uniformly distributed over the set of all permutations. Let $\tau_j$ denotes the block the item $j$ is in for $j\in[n]$, $\tau = (\tau_1, \ldots, \tau_n)$, and $\tau_{<j} = (\tau_1, \ldots, \tau_{j-1})$. Let $M(\tau)$ be Alice's message to Bob, which is a random variable depending on the randomness of $\sigma$ and the private coin tosses of Alice. Then we have $\RR^{1-way}(\eps-\textsc{Perm}) \ge H(M(\tau)) \ge I(M(\tau); \tau)$. Hence it is enough to lower bound $I(M(\tau); \tau)$. Then we have the following by chain rule.
 \begin{align*}
 I(M(\tau); \tau) &= \sum_{j=1}^n I(M(\tau); \tau_j | \tau_{<j})\\
 &= \sum_{j=1}^n H(\tau_j | \tau_ < j) - H(\tau_j | M(\tau), \tau_<j)\\
 &\ge \sum_{j=1}^n H(\tau_j | \tau_ < j) - \sum_{j=1}^n H(\tau_j | M(\tau))\\
 &= H(\tau) - \sum_{j=1}^n H(\tau_j | M(\tau))
 \end{align*}
 The number of ways to partition $n$ items into $\nfrac{1}{\eps}$ blocks is  $\nfrac{n!}{((\eps n)!)^{(\nfrac{1}{\eps})}}$ which is $\Omega(\nfrac{(\nfrac{n}{e})^n}{(\nfrac{\eps n}{e})^n})$. Hence we have $H(\tau) = n\log (\nfrac{1}{\eps})$. Now we consider $H(\tau_j | M(\tau))$. By the correctness of the algorithm, Fano's inequality, we have $H(\tau_j | M(\tau)) \le H(\delta) + (\nfrac{1}{10})\log_2((\nfrac{1}{\eps})-1) \le (\nfrac{1}{2}) \log (\nfrac{1}{\eps})$. Hence we have the following.
 \[ I(M(\tau); \tau) \ge (\nfrac{n}{2}) \log (\nfrac{1}{\eps}) \]
\end{proof}

Finally, we consider the \textsc{Greater-than} problem.
\begin{definition}(\textsc{Greater-than}$_{n}$)\\
 Alice is given an integer $x\in [n]$ and Bob is given an integer $y\in [n], y\ne x$. Bob has to output $1$ if $x>y$ and $0$ otherwise.
\end{definition}

The following result is due to \citep{smirnov88, MiltersenNSW98}. We provide a simple proof of it that seems to be missing\footnote{A similar proof appears in \cite{kremer1999randomized} but theirs gives a weaker lower bound.} in the literature.

\begin{restatable}{lemma}{LemGT}\label{lem:gt}
$\mathcal{R}_\delta^{\text{1-way}}(\textsc{Greater-than}_{n}) = \Omega(\log n)$, for every $\delta < \nfrac{1}{4}$. 
\end{restatable}
\begin{proof}
 We reduce the \textsc{Augmented-indexing}$_{2,\lceil\log n\rceil + 1}$ problem to the \textsc{Greater-than}$_{n}$ problem thereby proving the result. Alice runs the \textsc{Greater-than}$_{n}$ protocol with its input number whose representation in binary is $a=(x_1x_2\cdots x_{\lceil\log n\rceil}1)_2$. Bob participates in the \textsc{Greater-than}$_{n}$ protocol with its input number whose representation in binary is $b=(x_1x_2\cdots x_{i-1}1\underbrace{0 \cdots 0}_{(\lceil\log n\rceil-i+1)~ 0's})_2$. Now $x_i=1$ if and only if $a>b.$
\end{proof}

\subsection{Reductions from Problems in Communication Complexity}

We observe that a trivial $\Omega((\nfrac{1}{\varphi})\log n)$ bits lower bound for {\sc $(\eps, \varphi)$-List heavy hitters, $(\eps, \varphi)$-List borda, $(\eps, \varphi)$-List maximin} follows from the fact that any algorithm may need to output $\nfrac{1}{\phi}$ many items from the universe. Also, there is a trivial $\Omega(n \log n)$ lower bound for \textsc{$(\eps, \varphi)$-List borda} and \textsc{$(\eps, \varphi)$-List maximin} because each stream item is a permutation on $[n]$, hence requiring $\Omega(n \log n)$ bits to read.

We show now a space complexity lower bound of $\Omega(\frac{1}{\eps}\log \frac{1}{\phi})$ bits for the \textsc{$\eps$-Heavy hitters} problem.

\begin{restatable}{theorem}{ThmHeavyLbEps}\label{thm:eps_eps}
 Suppose the size of universe $n$ is at least $\nfrac{1}{(\eps\phi^{\mu})}$ for any constant $\mu>0$ and that $\phi > 2 \eps$. Any randomized one pass {\sc $(\eps,\phi)$-Heavy hitters} algorithm with success probability at least $(1-\delta)$ must use $\Omega((\nfrac{1}{\eps})\log \nfrac{1}{\phi})$ bits of space, for constant $\delta\in(0,1)$.
\end{restatable}
\begin{proof}
 Let $\mu>0$ be any constant. Without loss of generality, we can assume $\mu\le 1$. We will show that, when $n\ge\nfrac{1}{(\eps\phi^{\mu})}$, any \textsc{$\eps$-Heavy hitters} algorithm must use $\Omega((\nfrac{1}{\eps})\log\nfrac{1}{\phi})$ bits of memory, thereby proving the result. Consider the \textsc{Indexing}$_{\nfrac{1}{\phi^\mu}, \nfrac{1}{\eps}}$ problem where Alice is given a string $x=(x_1, x_2, \cdots, x_{\nfrac{1}{\eps}})\in [\nfrac{1}{\phi^\mu}]^{\nfrac{1}{\eps}}$ and Bob is given an index $i\in [\nfrac{1}{\eps}]$. The stream we generate is over $[\nfrac{1}{\phi^\mu}]\times[\nfrac{1}{\eps}]\subseteq \UU$ (this is possible since $|\UU|\ge \nfrac{1}{(\eps\phi^{\mu})}$). Alice generates a stream of length $\nfrac{m}{2}$ in such a way that the frequency of every item in $\{(x_j,j) : j\in [\nfrac{1}{\eps}]\}$ is at least $\lfloor{\eps m}/{2}\rfloor$ and the frequency of any other item is $0$. Alice now sends the memory content of the algorithm to Bob. Bob resumes the run of the algorithm by generating another stream of length ${m}/{2}$ in such a way that the frequency of every item in $\{(j,i) : j\in [\nfrac{1}{\phi^\mu}]\}$ is at least $\lfloor{\phi^\mu m}/{2}\rfloor$ and the frequency of any other item is $0$. The frequency of the item $(x_i, i)$ is at least $\lfloor \nfrac{\eps m}{2} + \nfrac{\phi^\mu m}{2}\rfloor$ whereas the frequency of every other item is at most $\lfloor{\phi^\mu m}/{2}\rfloor$. Hence from the output of the {\sc $(\nfrac{\eps}{5},\nfrac{\phi}{2})$-Heavy hitters} algorithm Bob knows $i$ with probability at least $(1-\delta)$. Now the result follows from \Cref{lem:index}.
\end{proof}

We now use the same idea as in the proof of \Cref{thm:eps_eps} to prove an $\Omega(\frac{1}{\eps}\log \frac{1}{\eps})$ space complexity lower bound for the $\eps$-Maximum problem.

\begin{theorem}\label{thm:eps_maximum}
 Suppose the size of universe $n$ is at least $\frac{1}{\eps^{1+\mu}}$ for any constant $\mu>0$. Any randomized one pass {\sc $\eps$-Maximum} algorithm with success probability at least $(1-\delta)$ must use $\Omega(\frac{1}{\eps}\log \frac{1}{\eps})$ bits of space, for constant $\delta\in(0,1)$.
\end{theorem}
\begin{proof}
 Let $\mu>0$ be any constant. Without loss of generality, we can assume $\mu\le 1$. We will show that, when $n\ge\frac{1}{\eps^{1+\mu}}$, any \textsc{$\eps$-Maximum} algorithm must use $\Omega(\frac{1}{\eps}\log\frac{1}{\eps})$ bits of memory, thereby proving the result. Consider the \textsc{Indexing}$_{\nfrac{1}{\eps^\mu}, \nfrac{1}{\eps}}$ problem where Alice is given a string $x=(x_1, x_2, \cdots, x_{\nfrac{1}{\eps}})\in [\nfrac{1}{\eps^\mu}]^{\nfrac{1}{\eps}}$ and Bob is given an index $i\in [\nfrac{1}{\eps}]$. The stream we generate is over $[\nfrac{1}{\eps^\mu}]\times[\nfrac{1}{\eps}]\subseteq \UU$ (this is possible since $|\UU|\ge \frac{1}{\eps^{1+\mu}}$). Alice generates a stream of length $\nfrac{m}{2}$ in such a way that the frequency of every item in $\{(x_j,j) : j\in [\nfrac{1}{\eps}]\}$ is at least $\lfloor{\eps m}/{2}\rfloor$ and the frequency of any other item is $0$. Alice now sends the memory content of the algorithm to Bob. Bob resumes the run of the algorithm by generating another stream of length ${m}/{2}$ in such a way that the frequency of every item in $\{(j,i) : j\in [\nfrac{1}{\eps^\mu}]\}$ is at least $\lfloor{\eps^\mu m}/{2}\rfloor$ and the frequency of any other item is $0$. The frequency of the item $(x_i, i)$ is at least $\lfloor \nfrac{\eps m}{2} + \nfrac{\eps^\mu m}{2}\rfloor$ where as the frequency of every other item is at most $\lfloor{\eps^\mu m}/{2}\rfloor$. Hence the {\sc $\nfrac{\eps}{5}$-Maximum} algorithm must output $(x_i, i)$ with probability at least $(1-\delta)$. Now the result follows from \Cref{lem:index}.
\end{proof}

For {\sc $\eps$-Minimum}, we prove a space complexity lower bound of $\Omega(\nfrac{1}{\eps})$  bits.

\begin{restatable}{theorem}{ThmVetoLB}\label{thm:veto_lb}
 Suppose the universe size $n$ is at least $\nfrac{1}{\eps}$. Then any randomized one pass \textsc{$\eps$-Minimum} algorithm must use $\Omega(\nfrac{1}{\eps})$  bits of space.
\end{restatable}
\begin{proof}
 We reduce from \textsc{Indexing}$_{2,\nfrac{5}{\eps}}$ to \textsc{$\eps$-Minimum} thereby proving the result. Let the inputs to Alice and Bob in \textsc{Indexing}$_{2,\nfrac{5}{\eps}}$ be $(x_1, \ldots, x_{\nfrac{5}{\eps}}) \in \{0,1\}^{\nfrac{5}{\eps}}$ and an index $i\in[\nfrac{5}{\eps}]$ respectively. Alice and Bob generate a stream $\SS$ over the universe $[(\nfrac{5}{\eps})+1]$. Alice puts two copies of item $j$ in $\SS$ for every $j\in\UU$ with $x_j=1$ and runs the \textsc{$\eps$-Minimum} algorithm. Alice now sends the memory content of the algorithm to Bob. Bob resumes the run of the algorithm by putting two copies of every item in $\UU\setminus\{i,(\nfrac{5}{\eps})+1\}$ in the stream $\SS$. Bob also puts one copy of $(\nfrac{5}{\eps})+1$ in $\SS$. Suppose the size of the support of $(x_1, \ldots, x_{\nfrac{5}{\eps}})$ be $\ell$. Since $\nfrac{1}{(2\ell+(\nfrac{2}{\eps})-1)} > \nfrac{\eps}{5}$, we have the following. If $x_i=0$, then the \textsc{$\eps$-Minimum} algorithm must output $i$ with probability at least $(1-\delta)$. If $x_i=1$, then the \textsc{$\eps$-Minimum} algorithm must output $(\nfrac{5}{\eps})+1$ with probability at least $(1-\delta)$. Now the result follows from \Cref{lem:index}.
\end{proof}

We show next a $\Omega(n\log (\nfrac{1}{\eps}))$ bits space complexity lower bound for {\sc $\eps$-Borda}.

\begin{theorem}\label{thm:lwb_borda}
 Any one pass algorithm for {\sc $\eps$-Borda} must use $\Omega(n\log (\nfrac{1}{\eps}))$ bits of space.
\end{theorem}

\begin{proof}
 We reduce {\sc $\eps$-Perm} to {\sc $\eps$-Borda}. Suppose Alice has a permutation $\sigma$ over $[n]$ and Bob has an index $i\in[n]$. The item set of our reduced election is $\UU = [n]\sqcup\DD$, where $\DD = \{d_1, d_2, \ldots, d_{2n}\}$. Alice generates a vote $\vvv$ over the item set $\UU$ from $\sigma$ as follows. The vote $\vvv$ is $\BB_1 \succ \BB_2 \succ \cdots \succ \BB_{\nfrac{1}{\eps}}$ where $\BB_j$ for $j=1, \ldots, \nfrac{1}{\eps}$ is defined as follows.
 \begin{align*}
 \BB_j &= d_{(j-1)2\eps n + 1} \succ d_{(j-1)2\eps n + 2} \succ \cdots \succ d_{(2j-1)\eps n} \\
 &\succ \sigma_{j\eps n + 1} \succ \cdots \succ \sigma_{(j+1)\eps n} \succ d_{(2j-1)\eps + 1} \succ \cdots \succ d_{2j\eps n}
 \end{align*}

 Alice runs the {\sc $\eps$-Borda} algorithm with the vote $\vvv$ and sends the memory content to Bob. Let $\DD_{-i} = \DD \setminus \{i\}$, $\overrightarrow{\DD_{-i}}$ be an arbitrary but fixed ordering of the items in $\DD_{-i}$, and $\overleftarrow{\DD_{-i}}$ be the reverse ordering of $\overrightarrow{\DD_{-i}}$. Bob resumes the algorithm by generating two votes each of the form $i\succ \overrightarrow{\DD_{-i}}$ and $i\succ \overleftarrow{\DD_{-i}}$. Let us call the resulting election $\EE$. The number of votes $m$ in $\EE$ is $5$. The Borda score of the item $i$ is at least $12n$. The Borda score of every item $x\in\UU$ is at most $9n$. Hence for $\eps < \nfrac{1}{15}$, the {\sc $\eps$-Borda} algorithm must output the item $i$. Moreover, it follows from the construction of $\vvv$ that an $\eps mn$ additive approximation of the Borda score of the item $i$ reveals the block where $i$ belongs in the {\sc $\eps$-Perm} instance.
\end{proof}

We next give a nearly-tight lower bound for the {\sc $\eps$-maximin} problem.
\begin{theorem}\label{thm:mmlb}
Any one-pass algorithm for {\sc $\eps$-maximin} requires $\Omega(n/\eps^2)$ memory bits of storage.
\end{theorem}

\begin{proof}
We reduce from {\sc Indexing}. Let $\gamma = 1/\eps^2$. Suppose Alice has a string $y$ of length $(n-\gamma)\cdot \gamma$, partitioned into $n-\gamma$ blocks of length $\gamma$ each. Bob has an index $\ell = i + (j-\gamma-1)\cdot \gamma$ where $i \in [\gamma], j \in \{\gamma+1, \dots, n\}$. The {\sc Indexing} problem is to return $y_\ell$ for which there is a $\Omega(|y|) = \Omega(n/\eps^2)$ lower bound (\Cref{lem:index}).

The initial part of the reduction follows the construction in the proof of Theorem 6 in \cite{VWWZ15}, which we encapsulate in the following lemma. 

\begin{lemma}[Theorem 6 in \cite{VWWZ15}]\label{lem:vangucht}
Given $y$, Alice can construct a matrix $P \in \{0,1\}^{n \times \gamma}$ using public randomness, such that if $P^i$ and $P^j$ are the $i$'th and $j$'th rows of $P$ respectively, then with probability at least $2/3$, $\Delta(P^i,P^j) \geq \frac{\gamma}{2} + \sqrt{\gamma}$ if $y_\ell = 1$ and $\Delta(a,b) \leq \frac{\gamma}{2}-\sqrt{\gamma}$ if $y_\ell = 0$.
\end{lemma}

Let Alice construct $P$ according to \Cref{lem:vangucht} and then adjoin the bitwise complement of the matrix $P$ below $P$ to form the matrix $P' \in \{0,1\}^{2n \times \gamma}$; note that each column of $P'$ has exactly $n$ 1's and $n$ 0's. Now, we interpret each row of $P$ as a candidate and each column of $P$ as a vote in the following way: 
for each $v \in [\gamma]$, vote $v$ has the candidates in $\{c : P'_{c,v}=1\}$ in ascending order in the top $n$ positions and the rest of the candidates in ascending order in the bottom $n$ positions. Alice inserts these $\gamma$ votes into the stream and sends the state of the {\sc $\eps$-Maximin} algorithm to Bob as well as the Hamming weight of each row in $P'$. Bob inserts $\gamma$ more votes, in each of which candidate $i$ comes first, candidate $j$ comes second, and the rest of the $2n-2$ candidates are in arbitrary order.

Note that because of Bob's votes, the maximin score of $j$ is the number of votes among the ones casted by Alice in which $j$ defeats  $i$. Since $i < j$, in those columns $v$ where $P_{i,v} = P_{j,v}$, candidate $i$ beats candidate $j$. Thus, the set of votes in which $j$ defeats $i$ is $\{v \mid P_{i,v} = 0, P_{j,v}=1\}$. The size of this set is $\frac{1}{2}\left(\Delta(P^i, P^j) + |P^j| - |P^i|\right)$. Therefore, if Bob can estimate the maximin score of $j$ upto $\sqrt{\gamma}/4$ additive error, he can find $\Delta(P^i,P^j)$ upto $\sqrt{\gamma}/2$ additive error as Bob knows $|P^i|$ and $|P^j|$. This is enough, by \Cref{lem:vangucht}, to solve the {\sc Indexing} problem with probability at least $2/3$.
\end{proof}

Finally, we show a space complexity lower bound that depends on the length of the stream $m$.

\begin{restatable}{theorem}{ThmLogLogn}\label{thm:loglogn}
 Any one pass algorithm for {\sc $\eps$-Heavy hitters}, {\sc $\eps$-Minimum}, {\sc $\eps$-Borda}, and {\sc $\eps$-maximin} must use $\Omega(\log \log m)$ memory bits, even if the stream is over a universe of size $2$, for every $\eps < \frac{1}{4}$.
\end{restatable}

\begin{proof}
 It is enough to prove the result only for {\sc $\eps$-Heavy hitters} since the other three problems reduce to {\sc $\eps$-Heavy hitters} for a universe of size $2$. Suppose we have a randomized one pass {\sc $\eps$-Heavy hitters} algorithm which uses $s(m)$ bits of space. Using this algorithm, we will show a communication protocol for the \textsc{Greater-than}$_{m}$ problem whose communication complexity is $s(2^m)$ thereby proving the statement. The universal set is $\UU=\{0,1\}$. Alice generates a stream of $2^x$ many copies of the item $1$. Alice now sends the memory content of the algorithm. Bob resumes the run of the algorithm by generating a stream of $2^y$ many copies of the item $0$. If $x>y$, then the item $1$ is the only $\eps$-winner; whereas if $x<y$, then the item $0$ is the only $\eps$-winner.
\end{proof}

\section{Conclusion}

In this work, we not only resolve a long standing fundamental open problem in the data streaming literature namely heavy hitters but also provide an optimal algorithm for a substantial generalization of heavy hitters by introducing another parameter $\phi$. We also initiate a promising direction of research on finding a winner of a stream of votes.

In the next chapter, we study the scenario when voters are allowed to have incomplete preferences in the form of a partial order.
\chapter{Kernelization for Possible Winner and Coalitional Manipulation}
\label{chap:kernel}

\blfootnote{A preliminary version of the work in this chapter was published as \cite{DeyMN15}: Palash Dey, Neeldhara Misra, and Y. Narahari. Kernelization complexity of possible winner and coalitional manipulation problems in voting. In Proc. 2015 International
Conference on Autonomous Agents and Multiagent Systems, AAMAS 2015, Istanbul,
Turkey, May 4-8, 2015, pages 87-96, 2015.

The full version of the work in this chapter was published as \cite{journalsDeyMN16}: Palash Dey, Neeldhara Misra, and Y. Narahari. Kernelization complexity of possible
winner and coalitional manipulation problems in voting. Theor. Comput. Sci., 616:111-
125, 2016.}

\begin{quotation}
{\small In the \textsc{Possible winner} problem in computational social
choice theory, we are given a set of partial preferences and the question
is whether a distinguished candidate could be made winner by extending the
partial preferences to linear preferences. Previous work has provided, for many
common voting rules, fixed parameter tractable algorithms 
for the \textsc{Possible winner} problem, with
number of candidates as the parameter. 
However, the corresponding kernelization question is still
open and in fact, has been mentioned as a key research challenge~\citep{ninechallenges}. 
In this work, we settle this open question for many common voting rules. 

We show that the \textsc{Possible winner} problem for
maximin, Copeland, Bucklin, ranked pairs, and a class of scoring rules that
includes the Borda voting rule do not admit a polynomial kernel with the number
of candidates as the parameter. We show however that the \textsc{Coalitional manipulation}
problem which is an important special case of the \textsc{Possible winner} problem
does admit a polynomial kernel for maximin, Copeland, ranked pairs, and a class of
scoring rules that includes the Borda voting rule, when the number of
manipulators is polynomial in the number of candidates. 
A significant conclusion of this work is that the \textsc{Possible winner} problem 
is harder than the \textsc{Coalitional manipulation} problem since the \textsc{Coalitional manipulation} 
problem admits a polynomial kernel whereas the \textsc{Possible winner} problem does not admit 
a polynomial kernel.}
\end{quotation}

\section{Introduction}

Usually, in a voting setting, it is assumed that the votes are complete orders over the
candidates. However, due to many reasons, for example, lack of knowledge of voters 
about some candidates, a voter may be indifferent between some pairs of candidates. 
Hence, it is both natural and important to 
consider scenarios where votes are partial orders over the candidates. When votes are only partial orders over the candidates, the winner cannot be determined with certainty since it depends on 
how these partial orders are extended to linear orders. 
This leads to a natural computational problem called the \textsc{Possible winner}~\cite{konczak2005voting} problem:  given a set of partial votes $P$ and a distinguished candidate $c$, is there a way to extend the partial votes to linear ones to make $c$ win?
The \textsc{Possible winner} problem has been studied extensively in the literature 
\cite{lang2007winner,pini2007incompleteness,walsh2007uncertainty,xia2008determining,betzler2009multivariate,betzler2009towards,chevaleyre2010possible,betzler2010partial,baumeister2011computational,lang2012winner,faliszewski2014complexity} following its definition in \cite{konczak2005voting}. Betzler et al.~\cite{betzler2009towards} and Baumeister et al.~\cite{BaumeisterR12} show that the \PW winner problem is \NPC for all scoring rules except for the plurality and veto voting rules; the \PW winner problem is in \Pb for the plurality and veto voting rules. Moreover the \PW problem is known to be \NPC for many common voting rules, for example, scoring rules, maximin, Copeland, Bucklin etc. even when the maximum number of undetermined pairs of candidates in every vote is bounded above by small constants~\cite{xia2008determining}. Walsh showed that the \PW problem can be solved in polynomial time for all the voting rules mentioned above when we have a constant number of candidates \cite{walsh2007uncertainty}.
An important special case of the \textsc{Possible winner} problem is the 
\textsc{Coalitional manipulation} problem \cite{bartholdi1989computational} 
where only two kinds of partial votes 
are allowed - complete preference and empty preference. The set of empty votes 
is called the manipulators' vote and is denoted by $M$. The \textsc{Coalitional manipulation} problem 
is $\NPCb$ for maximin, Copeland, and ranked pairs voting rules even when $|M|\ge 2$ 
\cite{faliszewski2008copeland,faliszewski2010manipulation,xia2009complexity}. 
The \textsc{Coalitional manipulation} problem is in $\Pb$ for the Bucklin voting rule 
\cite{xia2009complexity}. We refer to \cite{xia2008determining,walsh2007uncertainty,xia2009complexity} 
for detailed overviews.

\subsection{Our Contribution} Discovering kernelization algorithms is currently an active and interesting area of 
research in computational social choice theory~\cite{betzler2010partial,betzler2010problem,betzler2009fixed,bredereck2012multivariate,froese2013parameterized,bredereck2014prices,Betzlerxyz,dorn2012multivariate}. 
Betzler et al. \cite{betzler2009multivariate} showed that the 
\textsc{Possible winner} problem admits fixed parameter tractable algorithm when parameterized 
by the total number of candidates for scoring rules, maximin, Copeland, Bucklin, and ranked 
pairs voting rules. Yang et al.~\cite{Yang:2013:EAW:2484920.2485207,DBLP:conf_ecai_Yang14} provides efficient fixed parameter tractable algorithms for the \textsc{Coalitional manipulation} problem for the Borda, maximin, and Copeland voting rules. A natural and practical follow-up question is whether the \textsc{Possible winner} and \textsc{Coalitional manipulation} problems admit a polynomial kernel when parameterized by the number of candidates. This question has been open ever since the work of Betzler et al. and in fact, has been mentioned as a key research challenge in parameterized algorithms for computational social choice theory~\cite{ninechallenges}. Betzler et al. showed non-existence of polynomial kernel for the \textsc{Possible winner} problem for the $k$-approval voting rule when parameterized by $(t,k)$, where $t$ is the number of partial votes~\cite{betzler2010problem}. 
The $\NPCb$ reductions for the \textsc{Possible winner} 
problem for scoring rules, maximin, Copeland, Bucklin, and ranked 
pairs voting rules given by Xia et al.~\cite{xia2009complexity} are from the \textsc{Exact 3 set cover} 
problem. Their results do not throw any light on the existence of a polynomial kernel since \textsc{Exact 3 set cover} has a trivial $O(m^3)$ kernel where $m$ is the size of the universe. In this work, we show that there is no polynomial kernel (unless~\caveat{}) for the \textsc{Possible winner} problem, when parameterized by the total number of candidates, with respect to maximin, Copeland, Bucklin, and ranked pairs voting rules, and a class of scoring rules that includes the Borda voting rule. These hardness 
results are shown by a parameter-preserving many-to-one reduction from the 
\textsc{Small universe set cover} problem for which there does not exist any polynomial kernel 
parameterized by universe size unless~\caveat{} \cite{DomLokSau2009}. 

On the other hand, we show that 
the \textsc{Coalitional manipulation} problem admits a polynomial kernel for maximin, 
Copeland, and ranked pairs voting rules, and a class of scoring rules that includes the Borda voting rule when we have $poly(m)$ number of manipulators -- specifically, we exhibit  an $O(m^2|M|)$ kernel for maximin and Copeland voting rules, and an $O(m^4|M|)$ kernel for the ranked pairs voting rule, where $m$ is the number of candidates and $M$ is the set of manipulators. The \textsc{Coalitional manipulation} problem for the Bucklin voting rule is in $\Pb$ \cite{xia2009complexity} and thus the kernelization question does not arise. 

A significant conclusion of our work is that, although the \textsc{Possible winner} and 
\textsc{Coalitional manipulation} problems are both $\NPCb$ for many voting rules, the \textsc{Possible winner} problem is harder than the \textsc{Coalitional manipulation} problem since the \textsc{Coalitional manipulation} 
problem admits a polynomial kernel whereas the \textsc{Possible winner} problem does not admit 
a polynomial kernel. 

\section{Problem Definitions}

The {\sc Possible winner} problem with respect to $r$ is the following:

\defparproblem{Possible Winner [$r$]}
{
A set $\CC$ of candidates and votes $\VV$, where each vote is a partial order over $\CC$, and a candidate $c \in \CC$.
}
{$m$}
{
Is there a linear extension of $\VV$ for which $c$ is the unique winner with respect to $r$?}

An important special case of the {\sc Possible winner} problem is the {\sc Coalitional manipulation} problem where every partial vote is either an empty order or a complete order. We call the complete votes the votes of the non-manipulators and the empty votes the votes of the manipulators. Formally the {\sc Coalitional manipulation} problem is defined as follows.

\defparproblem{Coalitional Manipulation [$r$]}
{
A set $\CC$ of candidates, a set $\VV$ of complete votes, an integer $t$ corresponding to the number of manipulators, and a candidate $c \in \CC$.
}
{$m$}
{
Does there exist a set of votes $\VV^\prime$ of size $t$ such that $c$ is the unique winner with respect to $r$ for the voting profile $(\VV,\VV^\prime)$?
}

When the voting rule $r$ is clear from the context, we often refer to the problems as {\sc Possible winner} and {\sc Coalitional manipulation} without any further qualification. We note that one might also consider the variant of the problem where the designated candidate $c$ is only required to be a co-winner, instead of being the unique winner. All the results in this work can be easily adapted to this variant as well. 

We now briefly describe the framework in which we analyze the computational complexity of \textsc{Possible winner} and {\sc Coalitional manipulation} problems.

\section{Kernelization Hardness of the Possible Winner Problem}\label{sec:hard}

In this section, we show non-existence of polynomial kernels for the \textsc{Possible 
Winner} problem for the maximin, Copeland, Bucklin, and ranked pairs voting rules, and a class of scoring rules that includes the Borda voting rule. 
We do this by demonstrating polynomial parameter transformations from the \textsc{Small universe set cover} problem, which 
is the classic \textsc{Set Cover} problem, but now parameterized by the size of the universe and the budget.
\defparproblem{Small Universe Set Cover}
{
A set $\UU = \{u_1,\ldots,u_m\}$ and a family $\FF = \{S_1, \ldots, S_t\}$. 
}
{$m+k$}
{
Is there a subfamily $\HH \subseteq \FF$ of size at most $k$ such that every element of the universe belongs to 
at least one $H \in \HH$?}

It is well-known~\cite{DomLokSau2009} that \textsc{Red-blue dominating set} parameterized by $k$ and the number 
of non-terminals does not admit a polynomial kernel unless~\caveat{}. It follows, by the  
duality between dominating set and set cover, that \textsc{Set cover} when parameterized by the solution size 
and the size of the universe (in other words, the \textsc{Small universe set cover} problem defined above) does 
not admit a polynomial kernel unless~\caveat{}. 

We now consider the \textsc{Possible winner} problem parameterized by the number of candidates for the maximin, 
Copeland, Bucklin, and ranked pairs voting rules, and a class of scoring rules that includes the Borda rule,
and establish that they do not admit a polynomial kernel 
unless~\caveat{}, by polynomial parameter transformations from \textsc{Small universe set cover}. 

\subsection{Results for the Scoring Rules}

We begin with proving the hardness of kernelization for the \textsc{Possible winner} problem for a class of scoring rules that includes the Borda voting rule. For that, we use the following lemma which has been used before~\cite{baumeister2011computational}.

\begin{lemma}\label{score_gen}
Let $\mathcal{C} = \{c_1, \ldots, c_m\} \uplus D, (|D|>0)$ be a set of candidates, and $\vec{\alpha}$ a normalized score vector of length $|\mathcal{C}|$. Then, for any given $\mathbf{X} = (X_1, \ldots, X_m) \in \mathbb{Z}^m$, there exists $\lambda\in \mathbb{R}$ and a voting profile such that the $\vec{\alpha}$-score of $c_i$ is $\lambda + X_i$ for all $1\le i\le m$,  and the score of candidates $d\in D$ is less than $\lambda$. Moreover, the number of votes is $O(poly(|\mathcal{C}|\cdot \sum_{i=1}^m |X_i|))$.
\end{lemma}

With the above lemma at hand, we now show hardness of polynomial kernel for the class of strict scoring rules.

\begin{theorem}
\label{thm:npk_borda}
The \textsc{Possible winner} problem for any strict scoring rule, when parameterized by the number of candidates, does not admit a polynomial kernel unless~\caveat{}.
\end{theorem}

\begin{proof}
 Let $(\UU,\FF,k)$ be an instance of \textsc{Small universe set cover}, where $\UU = \{u_1,\ldots,u_m\}$ and $\FF = \{S_1, \ldots, S_t\}$. We use  $T_i$ to denote $\UU \setminus S_i$ for $i\in[t]$. We let $\vec{\alpha}= (\alpha_1, \alpha_2, \ldots, \alpha_{2m+3})$ denote the score vector of length $t$, and let $\delta_i$ denote the difference $(\alpha_{i} - \alpha_{i+1})$ for $i\in[2m+2]$. Note that for a strict scoring rule, all the $\delta_i$'s will be strictly positive. We now construct an instance $(\CC,V,c)$ of \textsc{Possible winner} as follows.  
 
\paragraph*{Candidates} $\CC = \UU \uplus \VV \uplus \{w,c,d\}$, where $\VV := \{v_1,\dots ,v_m\}$. 

\paragraph*{Partial Votes, $P$}~The first part of the voting profile comprises $t$ partial votes, and will be denoted by $P$. Let $V_j$ denote the set $\{v_1, \ldots, v_j\}$. For each $i\in[t]$, we first consider a profile built on a total order $\eta_i$:
$$\eta_i := d \succ S_i \succ V_j \succ w \succ \text{others}, \mbox{where } j = m - |S_i|.$$ 
Now we obtain a partial order $\lambda_i$ based on $\eta_i$ for every $i\in[t]$ as follows:
$$\lambda_i := \eta_i \setminus \left( \{w\} \times \left( \{d\} \uplus S_i \uplus V_j\right)\right)$$
That is, $\lambda_i$ is the partial vote where the order between the candidates $x$ and $y$ for $x\ne w$ or $y\notin \{d\} \uplus S_i \uplus V_j$, is same as the order between $x$ and $y$ in $\eta_i$. Whereas, if $x=w$ and $y\in \{d\} \uplus S_i \uplus V_j$, then the order between $x$ and $y$ is unspecified. Let $P^\prime$ be the set of votes $\{\eta_i ~|~ i\in[t]\}$ and $P$ be the set of votes $\{\lambda_i ~|~ i\in[t]\}$. 

\paragraph*{Complete Votes, $Q$}~We now add complete votes, which we denote by $Q$, such that $s(c) = s(u_i), s(d) - s(c) = (k -1)\delta_1, s(c) - s(w) = k(\delta_2 + \delta_3 + \cdots + \delta_{m+1}) + \delta_1, s(c) > s(v_i) + 1$ for all $i\in[t]$, where $s(a)$ is the score of candidate $a$ from the combined voting profile $P^\prime \uplus Q$. From the proof of \Cref{score_gen}, we can see that such votes can always be constructed. In particular, also note that the voting profile $Q$ consists of complete votes. Note that the number of candidates is $2m+3$, which is polynomial in the size of the universe, as desired. 

We now claim that the candidate $c$ is a possible winner for the voting profile $P \uplus Q$ with respect to the strict scoring rule $\vec{\alpha}$ if and only if $(\UU,\FF)$ is a YES instance of Set Cover. 

In the forward direction, suppose, without loss of generality, that $S_1, \ldots, S_k$ form a set cover. Then we propose the following extension for the partial votes $\lambda_1, \ldots, \lambda_k$:
$$w > d > S_i > V_j > \text{others} ,$$

and the following extension for the partial votes $\lambda_{k+1}, \ldots, \lambda_t$:
$$d > S_i > V_j > w > \text{others} $$

For $i\in[k]$, the position of $d$ in the extension of $\lambda_i$ proposed above is one lower than its original position in $\eta_i$. Therefore, the score of $d$ decreases by $k\delta_1$ making the final score of $d$ less than the score of $c$. Similarly, since $S_1, \ldots, S_k$ form a set cover, the score of $u_i$ decreases by at least $\min_{i=2}^m \{\delta_i\}$ for every $i\in[m]$, which is strictly positive as the scoring rule is strict. Finally, the score of $w$ increase by at most $k(\delta_2 + \delta_3 + \cdots + \delta_{m+1})$, since there are at most $k$ votes where the position of $w$ in the extension of $\lambda_i$ improved from it's original position in $\eta_i$ for $i\in[t]$. Therefore, the score of $c$ is greater than any other candidate, implying that $c$ is a possible winner. 

For the reverse direction, notice that there must be at least $k$ extensions where $d$ is in the second position, since the score of $d$ is $(k-1)\delta_1$ more than the score of $c$. In these extensions, observe that $w$ will be at the first position. On the other hand, placing $w$ in the first position causes its score to increase by $(\delta_2 + \delta_3 + \cdots + \delta_{m+1})$, therefore, if $w$ is in the first position in $\ell$ extensions, its score increases by $\ell(\delta_2 + \delta_3 + \cdots + \delta_{m+1})$. Since the score difference between $w$ and $c$ is only $k(\delta_2 + \delta_3 + \cdots + \delta_{m+1}) + 1$, we can afford to have $w$ in the first position in \emph{at most} $k$ votes. Therefore, apart from the extensions where $d$ is in the second position, in all remaining extensions, $w$ appears after $V_j$, and therefore the candidates from $S_i$ continue to be in their original positions. Moreover, there must be exactly $k$ votes where $d$ is at the second position. We now claim that the sets corresponding to the $k$ votes where $d$ is at the second position form a set cover. Indeed, if not, suppose the element $u_i$ is not covered. It is easily checked that the score of such a $u_i$ remains unchanged in this extension, and therefore its score is equal to $c$, contradicting our assumption that we started with an extension for which $c$ was a winner. 
\end{proof}

The proof of \Cref{thm:npk_borda} can be generalized to a wider class of scoring rules as stated in the following corollary.
\begin{corollary}
 Let $r$ be a positional scoring rule such that there exists a polynomial function $f:\mathbb{N}\rightarrow \mathbb{N}$, such that for every $m\in \mathbb{N}$, there exists an index $l$ in the $f(m)$ length score vector $\vec{\alpha}$ satisfying following,
 $$ \alpha_i - \alpha_{i+1} > 0 ~\forall l\le i\le l+m $$
 Then the \textsc{Possible winner} problem for $r$, when parameterized by the number of candidates, does not admit a polynomial kernel unless~\caveat{}.
\end{corollary}

\subsection{Results for the Maximin Voting Rule}

We will need the following lemma in subsequent proofs. The lemma has been used before \cite{mcgarvey1953theorem,xia2008determining}.

\begin{lemma}\label{thm:mcgarvey}
 Let $f:\mathcal{C} \times \mathcal{C} \longrightarrow \mathbb{Z}$ be a function such that
 \begin{enumerate}
  \item $\forall a,b \in \mathcal{C}, f(a,b) = -f(b,a)$.
  \item $\forall a, b, c, d \in \mathcal{C}, f(a,b) + f(c,d)$ is even.
 \end{enumerate}
 Then we can construct in time $O\left(|\mathcal{C}|\sum_{\{a,b\}\in \mathcal{C}\times\mathcal{C}} |f(a,b)|\right)$ an election $E$ with $n$ votes over the candidate set $\mathcal{C}$ such that for all $a,b \in \mathcal{C}$, $D_E(a, b) = f(a,b)$. 
\end{lemma}

We now describe the reduction for the \textsc{Possible winner} problem for the maximin voting rule parameterized by the number of candidates.

\begin{theorem}
\label{thm:npk_maximin}
The \textsc{Possible winner} problem for the maximin voting rule, when parameterized by the number of candidates, does not admit a polynomial kernel unless~\caveat{}.
\end{theorem}

\begin{proof}
 Let $(\UU,\FF,k)$ be an instance of \textsc{Small universe set cover}, where $\UU = \{u_1,\ldots,u_m\}$ and $\FF = \{S_1, \ldots, S_t\}$. We use  $T_i$ to denote $\UU \setminus S_i$. We now construct an instance $(\CC,V,c)$ of the \textsc{Possible winner} as follows. 

\paragraph*{Candidates}~$C := \UU \uplus W \uplus \{c,d,x\} \uplus L$, where $W := \{w_1, w_2, \ldots, w_m, w_x\}, L:= \{l_1,l_2,l_3\}$.

\paragraph*{Partial Votes, $P$}~The first part of the voting profile comprises $t$ partial votes, and will be denoted by $P$. For each $i\in[t]$, we first consider a profile built on a total order $\eta_i$. We denote the order $w_1 \succ \cdots \succ w_m \succ w_x$ 
by $\vec{W}$. From this point onwards, whenever we place a set of candidates in some position of a partial order, we mean that the candidates in the set can be ordered arbitrarily. For example, the candidates in $S_i$ can be ordered arbitrarily among themselves in the total order $\eta_i$ below for every $i\in[t]$.
\[\eta_i := L \succ \vec{W} \succ x \succ S_i \succ d \succ c \succ T_i\]
Now we obtain a partial order $\lambda_i$ based on $\eta_i$ for every $i\in[t]$ as follows:
\[\lambda_i := \eta_i \setminus \left(W \times \left(\{c,d,x\} \uplus \UU\right)\right) \]
The profile $P$ consists of $\{ \lambda_i ~|~ 1\in[t] \}$.

\paragraph*{Complete Votes, $Q$}~We now describe the remaining votes in the profile, which are linear orders designed to achieve specific pairwise difference scores among the candidates. This profile, denoted by $Q$, is defined according to~\Cref{thm:mcgarvey} to contain votes such that the pairwise score differences of $P \cup Q$ satisfy the following.

\begin{itemize}
\item $D(c,w_1) = -2k$.
\item $D(c,l_1) = -t$.
\item $D(d,w_1) = -2k - 2$.
\item $D(x,w_x) = -2k - 2$.
\item $D(w_i,u_i) = -2t$ $\forall$ $i\in[m]$.
\item $D(a_i,l_1) = D(w_x,l_1) = -4t$.
\item $D(l_1,l_2) = D(l_2,l_3) = D(l_3,l_1) = -4t$.
\item $D(l,r) \leq 1$ for all other pairs $(l,r) \in C \times C$.
\end{itemize}
We note that the for all $c,c' \in \CC$, the difference $|D(c,c') - D_P(c,c')|$ is always even, as long as $t$ is even and the number of sets in $\FF$ that contain any element $u \in \UU$ is always even. Note that the latter can always be ensured without loss of generality: indeed, if $u \in \UU$ occurs in an odd number of sets, then we can always add the set $\{u\}$ if it is missing and remove it if it is present, flipping the parity in the process. In case $\{u\}$ is the only set containing the element $u$, then we remove the set from both $\FF$ and $\UU$ and decrease $k$ by one. The number of sets $t$ can be assumed to be even by adding a dummy element and a dummy pair of sets that contains the said element. It is easy to see that these modifications always preserve the instance. 
Thus, the constructed instance of \textsc{Possible winner} is $(\CC,V,c)$, where $V := P \cup Q$. We now turn to the proof of correctness. 

In the forward direction, let $\HH \subseteq \FF$ be a set cover of size at most $k$. Without loss of generality, let $|\HH| = k$ (since a smaller set cover can always be extended artificially) and let $\HH = \{S_1, \ldots, S_k\}$ (by renaming). 

If $i \leq k$, let: 
\[\lambda_i^* := L \succ x \succ S_i \succ d \succ c \succ \vec{W} \succ T_i\]

If $k < i \leq t$, let: 
\[\lambda_i^* := L \succ \vec{W} \succ x \succ S_i \succ d \succ c \succ T_i\]
Clearly $\lambda_i^*$ extends $\lambda_i$ for every $i\in[t]$. Let $V^*$ denote the extended profile consisting of the votes $\{ \lambda_i^* ~|~ i\in[t]\} \cup Q$. We now claim that $c$ is the unique winner with respect to the maximin voting rule in $V^*$. 

Since there are $k$ votes in $V^*$ where $c$ is preferred over $w_1$ and $(t-k)$ votes where $w_1$ is preferred to $c$, we have:
\begin{eqnarray*}
D_{V^*}(c,w_1) &=& D_{V}(c,w_1) + k - (t-k)
\\ &=& -2k + k - (t-k) = -t
\end{eqnarray*}
It is easy to check that maximin score of $c$ is $-t$. Also, it is straightforward to verify the following 
table.

\begin{table}[!htbp]
\begin{minipage}{\textwidth}
  \begin{center}
  {\renewcommand{\arraystretch}{2}
 \begin{tabular}{|c|c| }\hline
  Candidate	& maximin score	\\\hline
  $w_i, \forall i\in \{1, 2, \dots, m\}$		& $ < -t $	\\\hline
  $u_i, \forall i\in \{1, 2, \dots, m\}$		& $\leq -4t$	\\\hline 
  $w_x$		& $\leq -4t$	\\\hline 
  $l_1, l_2, l_3$		& $\leq -4t$	\\\hline
  $x$		& $\leq -t-2$	\\\hline 
  $d$		& $\leq -t-2$\\\hline
 \end{tabular}
 }
 \end{center}
\end{minipage}
\label{summary}
\end{table}

Therefore, $c$ is the unique winner for the profile $V^*$.

We now turn to the reverse direction. Let $P^*$ be an extension of $P$ such that $V^* := P^* \cup Q$ admits $c$ as a unique winner with respect to the maximin voting rule. We first argue that $P^*$ must admit a certain structure, which will lead us to an almost self-evident set cover for $\UU$.

Let us denote by $P^*_C$ the set of votes in $P^*$ which are consistent with $c \succ w_1$, and let $P^*_W$ be the set of votes in $P^*$ which are consistent with $w_1 \succ c$. We first argue that $P^*_C$ has at most $k$ votes.
\begin{claim} Let $P^*_C$ be as defined above. Then $|P^*_C| \leq k$.
\end{claim}

\begin{proof}
Suppose, for the sake of contradiction, that more than $k$ extensions are consistent with $c \succ w_1$. Then we have:
\begin{eqnarray*}
D_{V^*}(c,w_1) & \geq & D_{V}(c,w_1) + k+1 - (t-k-1)
\\ &=&-2k + 2k - t + 2 = -t+2
\end{eqnarray*}
Since $D_{V^*}(c,l_1) = -t$, the maximin score of $c$ is $-t$. On the other hand, we also have that the maximin score of $d$ is given by $D_{V^*}(d,w_1)$, which is now at least $(-t)$:
\begin{eqnarray*}
D_{V^*}(d,w_1) & \geq & D_{V}(d,w_1) + k+1 - (t-k-1)
\\ &=&-2k-2 + 2k - t + 2 = -t
\end{eqnarray*}
Therefore, $c$ is no longer the unique winner in $V^*$ with respect to the maximin voting rule, 
which is the desired contradiction. 
\end{proof}

We next propose that a vote that is consistent with $w_1 \succ c$ must be consistent with $w_x \succ x$.
\begin{claim} 
Let $P^*_W$ be as defined above. Then any vote in $P^*_W$ must respect $w_x \succ x$.
\end{claim}

\begin{proof}
Suppose there are $r$ votes in $P^*_C$, and suppose that in at least one vote in $P^*_W$ where $x \succ w_x$. Notice that any vote in $P^*_C$ is consistent with $x \succ w_x$. Now we have:
\begin{eqnarray*}
D_{V^*}(c,w_1) & = & D_{V}(c,w_1) + r - (t-r)
\\ &=&-2k + 2r - t 
\\ &=&-t - 2(k-r)
\end{eqnarray*}
And further:
\begin{eqnarray*}
D_{V^*}(x,w_x) & \geq & D_{V}(x,w_x) + (r+1) - (t-r-1)
\\ &=&-2k-2 + 2r - t + 2 
\\ &=&-t - 2(k-r)
\end{eqnarray*}
It is easy to check that the maximin score of $c$ in $V^*$ is at most $-t - 2(k-r)$, witnessed by  $D_{V^*}(c,w_1)$, and the maximin score of $x$ is at least $-t - 2(k-r)$, witnessed by  $D_{V^*}(x,w_x)$. Therefore, $c$ is no longer the unique winner in $V^*$ with respect to the maximin voting rule, and we have a contradiction.
\end{proof}

We are now ready to describe a set cover of size at most $k$ for $\UU$ based on $V^*$. Define $J \subseteq [t]$ as being the set of all indices $i$ for which the extension of $\lambda_i$ in $V^*$ belongs to $P_C^*$. Consider: \[\HH := \{S_i ~|~ i \in J\}.\] The set $\HH$ is our proposed set cover. Clearly, $|\HH| \leq k$. It remains to show that $\HH$ is a set cover. 

We assume, for the sake of contradiction, that there is an element $u_i \in \UU$ that is not covered by $\HH$. This means that we have $u_i \in T_i$ for all $i \in J$, and thus $w_i \succ u_i$ in the corresponding extensions of $\lambda_i$ in $V^*$.  Further, for all $i \notin J$, we have that the extension of $\lambda_i$ in $V^*$ is consistent with: 
\[w_1 \succ \cdots \succ w_i \succ \cdots \succ w_x \succ x \succ S_i \succ c \succ T_i,\]
implying again that $w_i \succ u_i$ in these votes. Therefore, we have:
\[D_{V^*}(w_i,u_i) = D_V(w_i,u_i) + k + (t-k) = -2t + t = -t.\]
We know that the maximin score of $c$ is less than or equal to $-t$, since $D_{V^*}(c,l_1) = -t$, and we now have that the maximin score of $w_i$ is $-t$. This overrules $c$ as the unique winner in $V^*$, contradicting our assumption to that effect. This completes the proof. 
\end{proof}

\subsection{Results for the Copeland Voting Rule}

We now describe the result for the \textsc{Possible winner} problem for the Copeland voting rule parameterized by the number of candidates.

\begin{theorem}
\label{thm:npk_copeland}
The \textsc{Possible winner} problem for the Copeland voting rule, when parameterized by the number of candidates, does not admit a polynomial kernel unless~\caveat{}.
\end{theorem}

\begin{proof}
Let $(\UU,\FF,k)$ be an instance of \textsc{Small universe set cover}, where $\UU = \{u_1,\ldots,u_m\}$ and $\FF = \{S_1, \ldots, S_t\}$. For the purpose of this proof, we assume (without loss of generality) that $m \geq 6$. We now construct an instance $(\CC,V,c)$ of \textsc{Possible winner} as follows. 

\paragraph*{Candidates}~$\CC := \UU \uplus \{z, c, d, w\}$.

\paragraph*{Partial Votes, $P$}~The first part of the voting profile comprises of $m$ partial votes, and will be denoted by $P$. For each $i\in[t]$, we first consider a profile built on a total order:
\[\eta_i := \UU \setminus S_i \succ z \succ c \succ d \succ S_i \succ w\]
Now we obtain a partial order $\lambda_i$ based on $\eta_i$ as follows for each $i\in[t]$:
\[\lambda_i := \eta_i \setminus \left( \{z,c\} \times \left( S_i \uplus \{d,w\} \right)\right)\]
The profile $P$ consists of $\{ \lambda_i ~|~ i\in[t] \}$.

\paragraph*{Complete Votes, $Q$}~We now describe the remaining votes in the profile, which are linear orders designed to achieve specific pairwise difference scores among the candidates. This profile, denoted by $Q$, is defined according to~\Cref{thm:mcgarvey} to contain votes such that the pairwise score differences of $P \cup Q$ satisfy the following.

\begin{itemize}
\item $D(c,d) = t - 2k + 1$
\item $D(z,w) = t - 2k - 1$
\item $D(c,u_i) = t - 1$
\item $D(c,z) = t+1$ 
\item $D(c,w) = -t-1$ 
\item $D(u_i,d) = D(z,u_i) = t+1$ $\forall$ $i\in[m]$
\item $D(z,d) = t+1$
\item $D(u_i,u_j) = t+1$ $\forall$ $j \in [i+1\pmod* m,i+\lfloor m/2 \rfloor\pmod* m]$
\end{itemize}

We note that the difference $|D(c,c') - D_P(c,c')|$ is always even for all $c,c' \in \CC$, as long as $t$ is odd and the number of sets in $\FF$ that contain any element $a \in \UU$ is always odd. Note that the latter can always be ensured without loss of generality: indeed, if $a \in \UU$ occurs in an even number of sets, then we can always add the set $\{a\}$ if it is missing and remove it if it is present, flipping the parity in the process. In case $\{a\}$ is the only set containing the element $a$, then we remove the set from both $\FF$ and $\UU$ and decrease $k$ by one. The number of sets $t$ can be assumed to be odd by adding a dummy element in $\UU$, adding a dummy set that contains the said element in $\FF$, and incrementng $k$ by one. It is easy to see that these modifications always preserve the instance. 

Thus the constructed instance of \textsc{Possible winner} is $(\CC,V,c)$, where $V := P \cup Q$. We now turn to the proof of correctness. 

In the forward direction, let $\HH \subseteq \FF$ be a set cover of size at most $k$. Without loss of generality, let $|\HH| = k$ (since a smaller set cover can always be extended artificially) and let $\HH = \{S_1, \ldots, S_k\}$ (by renaming).  

If $i \leq k$, let: 
\[\lambda_i^* := \UU \setminus S_i \succ z \succ c \succ d \succ S_i \succ w\]

If $k < i \leq t$, let: 
\[\lambda_i^* := \UU \setminus S_i \succ d \succ S_i \succ w \succ z \succ c\]
Clearly $\lambda_i^*$ extends $\lambda_i$ for every $i\in[t]$. Let $V^*$ denote the extended profile consisting of the votes $\{ \lambda_i^* ~|~ i\in[t]\} \cup Q$. We now claim that $c$ is the unique winner with respect to the Copeland voting rule in $V^*$. 

First, consider the candidate $z$. For every $i\in[m]$, between $z$ and $u_i$, even if $z$ loses to $u_i$ in $\lambda_j^*$, for every $j\in[t]$, because $D(z,u_i) = t+1$, $z$ wins the pairwise election between $z$ and $u_i$. The same argument holds between $z$ and $d$. Therefore, the Copeland score of $z$, no matter how the partial votes were extended, is at least $(m+1)$. 

Further, note that all other candidates (apart from $c$) have a Copeland score of less than $(m+1)$, because they are guaranteed to lose to at least three candidates (assuming $m \geq 6$). In particular, observe that $u_i$ loses to at least $\lfloor m/2 \rfloor$ candidates, and $d$ loses to $u_i$ (merely by its position in the extended votes), and $w$ loses to $u_i$ (because of way the scores were designed) for every $i\in[m]$. Therefore, the Copeland score of all candidates in $\CC\setminus \{z,c\}$ is strictly less than the Copeland score of $z$, and therefore they cannot be possible (co-)winners.

Now we restrict our attention to the contest between $z$ and $c$. First note that $c$ beats $u_i$ for every $i\in[m]$: since the sets of $\HH$ form a set cover, $u_i$ is placed in a position after $c$ in some $\lambda^*_j$ for $j\in[k]$. Since the difference of score between $c$ and $u_i$ was $(t-1)$, even if $c$ suffered defeat in every other extension, we have the pairwise score of $c$ and $u_i$ being at least $t-1 - (t-1) + 1 = 1$, which implies that $c$ defeats every $u_i$ in their pairwise election. Note that $c$ also defeats $d$ by getting ahead of $d$ in $k$ votes, making its final score $t - 2k + 1 + k - (t-k) = 1$. Finally, $c$ is defeated by $w$, simply by the preset difference score. Therefore, the Copeland score of $c$ is $(m+2)$.

Now all that remains to be done is to rule $z$ out of the running. Note that $z$ is defeated by $w$ in their pairwise election: this is because $z$ defeats $w$ in $k$ of the extended votes, and is defeated by $w$ in the remaining. This implies that its final pairwise score with respect to $w$ is at most $t - 2k - 1 + k - (t-k) = -1$. Also note that $z$ loses to $c$ because of its predefined difference score. Thus, the Copeland score of $z$ in the extended vote is exactly $(m+1)$, and thus $c$ is the unique winner of the extended vote. 

We now turn to the reverse direction. Let $P^*$ be an extension of $P$ such that $V^* := P^* \cup Q$ admits $c$ as a unique winner with respect to the Copeland voting rule. As with the argument for the maximin voting rule, we first argue that $P^*$ must admit a certain structure, which will lead us to an almost self-evident set cover for $\UU$.

Let us denote by $P^*_C$ the set of votes in $P^*$ which are consistent with $c \succ d$, and let $P^*_W$ be the set of votes in $P^*$ which are consistent with $w \succ z$. Note that the votes in $P^*_C$ necessarily have the form:
\[\lambda_i^* := \UU \setminus S_i \succ z \succ c \succ d \succ S_i \succ w\]
and those in $P^*_W$ have the form: 
\[\lambda_i^* := \UU \setminus S_i \succ d \succ S_i \succ w \succ z \succ c\]
It is easy to check that this structure is directly imposed by the relative orderings that are fixed by the partial orders.

Before we argue the details of the scores, let us recall that in any extension of $P$, $z$ loses to $c$ and $z$ wins over $d$ and all candidates in $\UU$. Thus the Copeland score of $z$ is at least $(m+1)$. On the other hand, in any extension of $P$, $c$ loses to $w$, and therefore the Copeland score of $c$ is at most $(m+2)$. (These facts follow from the analysis in the forward direction.) 

Thus, we have the following situation. If $z$ wins over $w$, then $c$ cannot be the unique winner in the extended vote, because the score of $z$ goes up to $(m+2)$. Similarly, $c$ cannot afford to lose to any of $\UU \cup \{d\}$, because that will cause its score to drop below $(m+2)$, resulting in either a tie with $z$, or defeat. These facts will successively lead us to the correctness of the reverse direction. 

Now let us return to the sets $P^*_C$ and $P^*_W$. If $P^*_C$ has more than $k$ votes, then $z$ wins over $w$: the final score of $z$ is at least $t - 2k - 1 + (k+ 1) - (t-k-1) = 1$, and we have a contradiction. If $P^*_C$ has fewer than $k$ votes, then $c$ loses to $d$, with a score of at most $t - 2k + 1 + (k-1) - (t-k+1) = -1$, and we have a contradiction. Hence, $P^*_C$ must have exactly $k$ votes.

Finally, suppose the sets corresponding to the votes of $P^*_C$ do not form a set cover. Consider an element $u_i\in\UU$ not covered by the union of these sets. Observe that $c$ now loses the pairwise election between itself and $u_i$ and is no longer in the running for being the unique winner in the extended vote. Therefore, the sets corresponding to the votes of $P^*_C$ form a set cover of size exactly $k$, as desired.
\end{proof}

\subsection{Results for the Bucklin Voting Rule}

We now describe the result for the \textsc{Possible winner} problem for the Bucklin voting rule parameterized by the number of candidates.

\begin{theorem}
\label{thm:npk_bucklin}
The \textsc{Possible winner} problem for the Bucklin voting rule, when parameterized by the number 
of candidates, does not admit a polynomial kernel unless~\caveat{}.
\end{theorem}

\begin{proof}
Let $(\UU,\FF,k)$ be an instance of \textsc{Small universe set cover}, where $\UU = \{u_1,\ldots,u_m\}$ and 
$\FF = \{S_1, \ldots, S_t\}$. Without loss of generality, we assume that $t>k+1$, and that every set in $\FF$ has at 
least two elements. We now construct an instance $(\CC,V,c)$ of \textsc{Possible winner} as follows. 

\paragraph*{Candidates}~$\CC := \UU \uplus \{z, c, a\} \uplus W \uplus D_1 \uplus D_2 \uplus D_3$, where $D_1$, $D_2$, and $D_3$ are sets of \textit{``dummy candidates''} such that $|D_1|=m$, $|D_2|=2m$, and $|D_3|= 2m$. $W:=\{w_1, w_2, \dots, w_{2m} \}$.

\paragraph*{Partial Votes, $P$}~The first part of the voting profile comprises of $t$ partial votes, and will be denoted by $P$. For each $i\in[t]$, we first consider a profile built on a total order:
\begin{equation*}
\eta_i := \UU \setminus S_i \succ S_i \succ w_{i\pmod* m} \succ w_{i+1\pmod* m} \succ z \succ c \succ D_3 \succ \text{others}
\end{equation*}
Now we obtain a partial order $\lambda_i$ based on $\eta_i$ for every $i\in[t]$ as follows:
\[\lambda_i := \eta_i \setminus \left(\left( \{w_{i\pmod* m}, w_{i+1\pmod* m}, z,c\} \uplus D_3\right) \times S_i \right)\]
The profile $P$ consists of $\{ \lambda_i ~|~ i\in[t] \}$.

\paragraph*{Complete Votes, $Q$} 

\begin{eqnarray*}
 t-k-1 &:& D_1 \succ z \succ c \succ \text{others} \\
 1 &:& D_1 \succ c \succ a \succ z \succ \text{others} \\
 k-1 &:& D_2 \succ \text{others}
\end{eqnarray*}

We now show that $(\UU,\FF,k)$ is a \YES{} instance if and only if $(\CC,V,c)$ is a \YES{} instance. 
Suppose $\{S_j : j\in J\}$ forms a set cover. Then consider the following extension of $P$: 
\begin{equation*}
 (\UU \setminus S_j) \succ w_{j\pmod* m} \succ w_{j+1\pmod* m} \succ z \succ c \succ D_3 \succ S_j \succ \text{others}, \text{ for } j\in J
\end{equation*}
\begin{equation*}
 (\UU \setminus S_j) \succ S_j \succ w_{j\pmod* m} \succ w_{j+1\pmod* m} \succ z \succ c \succ D_3 \succ \text{others}, \text{ for } j\notin J
\end{equation*}
We claim that in this extension, $c$ is the unique winner with Bucklin score $(m+2)$. First, let us establish the score of $c$. The candidate $c$ is already within the top $(m+1)$ choices in $(t-k)$ of the complete votes. In all the sets that form the set cover, $c$ is ranked within the first $(m+2)$ votes in the proposed extension of the corresponding vote (recall that every set has at least two elements). Therefore, there are a total of $t$ votes where $c$ is ranked within the top $(m+2)$ preferences. Further, consider a candidate $v \in \UU$. Such a candidate is not within the top $(m+2)$ choices of any of the complete votes. Let $S_i$ be the set that covers the element $v$. Note that in the extension of the vote $\lambda_i$, $v$ is not ranked among the top $(m+2)$ spots, since there are at least $m$ candidates from $D_3$ getting in the way. Therefore, $v$ has strictly fewer than $t$ votes where it is ranked among the top $(m+2)$ spots, and thus has a Bucklin score more than $c$.

Now the candidate $z$ is within the top $(m+2)$ ranks of at most $(t-k-1)$ votes among the complete votes. In the votes corresponding to the sets \emph{not} in the set cover, $z$ is placed beyond the first $(m+2)$ spots. Therefore, the number of votes where $z$ is among the top $(m+2)$ candidates is at most $(t-1)$, which makes its Bucklin score strictly more than $(m+2)$. 

The candidates from $W$ are within the top $(m+2)$ positions only in a constant number of votes. The candidates $D_1 \cup \{a\}$ have $(t-k)$ votes (among the complete ones) in which they are ranked among the top $(m+2)$ preferences, but in all extensions, these candidates have ranks below $(m+2)$. Finally, the candidates in $D_3$ do not feature in the top $(m+2)$ positions of any of the complete votes, and similarly, the candidates in $D_2$ do not feature in the top $(m+2)$ positions of any of the extended votes. Therefore, the Bucklin scores of all these candidates is easily seen to be strictly more than $(m+2)$, concluding the argument in the forward direction. 

Now consider the reverse direction. Suppose $(\CC,V,c)$ is a \YES{} instance. For the same reasons described in the forward direction, observe that only the following candidates can win depending upon how the partial 
preferences get extended - either one of the candidates in $\UU$, or one of $z$ or $c$. Note that the Bucklin score of $z$ in any extension is at most $(m+3)$. Therefore, the Bucklin score of $c$ has to be $(m+2)$ or less. Among the complete votes $Q$, there are $(t-k)$ votes where the candidate $c$ appears in the top $(m+2)$ positions. To get majority within top $(m+2)$ positions, $c$ should be within top $(m+2)$ positions for at least $k$ of the extended votes in $P$.
Let us call these set of votes $P^{\prime}$. Now notice that whenever $c$ comes within top 
$(m+2)$ positions in a valid extension of $P$, the candidate $z$ also comes within top $(m+2)$ positions in the same vote. However, the candidate $z$ is already ranked among the top $(m+2)$ candidates in $(t-k-1)$ complete votes. Therefore, $z$ can appear within top $(m+2)$ positions in \emph{at most} $k$ extensions (since $c$ is the unique winner), implying that $|P^{\prime}|=k$. Further, note that the Bucklin score of $c$ cannot be strictly smaller than $(m+2)$ in any extension. Indeed, candidate $c$ features in only one of the complete votes within the top $(m+1)$ positions, and it would have to be within the top $(m+1)$ positions in at least $(t-1)$ extensions. However, as discussed earlier, this would give $z$ exactly the same mileage, and therefore its Bucklin score would be $(m-1)$ or even less; contradicting our assumption that $c$ is the unique winner. 

Now we claim that the $S_i$'s corresponding to the votes in $P^{\prime}$ form a set cover for $\UU$. If not, there is an element $x\in \UU$ that is uncovered. Observe that $x$ appears within top $m$ positions in all the extensions of the votes in $P^\prime$, by assumption. Further, in all the remaining extensions, since $z$ is not present among the top $(m+2)$ positions, we only have room for two candidates from $W$. The remaining positions must be filled by all the candidates corresponding to elements of $\UU$. Therefore, $x$ appears within the top $(m+2)$ positions of all the extended votes. Since these constitute half the total number of votes, we have that $x$ ties with $c$ in this situation, a contradiction. 
\end{proof}

\subsection{Results for the Ranked Pairs Voting Rule}

We now describe the reduction for  \textsc{Possible winner} parameterized by the number of candidates, for the ranked pairs voting rule.

\begin{theorem}
\label{thm:npk_rankedpairs}
The \textsc{Possible winner} problem for the ranked pairs voting rule, when parameterized by the 
number of candidates, does not admit a polynomial kernel unless~\caveat{}.
\end{theorem}

\begin{proof}
Let $(\UU,\FF,k)$ be an instance of \textsc{Small universe set cover}, where $\UU = \{u_1,\ldots,u_m\}$ and 
$\FF = \{S_1, \ldots, S_t\}$. Without loss of generality, we assume that $t$ is even. We now construct an 
instance $(\CC,V,c)$ of \textsc{Possible winner} as follows. 

\paragraph*{Candidates}~$\CC := \UU \uplus \{a, b, c, w \}$.

\paragraph*{Partial Votes, $P$}~The first part of the voting profile comprises of $t$ partial votes, and will 
be denoted by $P$. For each $i\in[t]$, we first consider a profile built on a total order:
\[\eta_i := \UU \setminus S_i \succ S_i \succ b \succ a \succ c \succ \text{others}\]
Now we obtain a partial order $\lambda_i$ based on $\eta_i$ for every $i\in[t]$ as follows:
\[\lambda_i := \eta_i \setminus \left(\{ a, c\} \times \left( S_i \uplus \{ b \} \right)\right)\]
The profile $P$ consists of $\{ \lambda_i ~|~ i\in[t] \}$.

\paragraph*{Complete Votes, $Q$} We add complete votes such that along with the already determined 
pairs from the partial votes $P$, we have the following.

\begin{itemize}
\item $D(u_i, c) = 2$ $\forall$ $i\in[m] $
\item $D(c,b) = 4t$
\item $D(c,w) = t+2$
\item $D(b,a) = 2k + 4$ 
\item $D(w,a) = 4t$ 
\item $D(a,c) = t+2$
\item $D(w, u_i) = 4t$ $\forall$ $i\in[m]$ 
\end{itemize}

We now show that $(\UU,\FF,k)$ is a \YES{} instance if and only if $(\CC,V,c)$ is a \YES{} instance. 
Suppose $\{S_j : j\in J\}$ forms a set cover. Then consider the following extension of $P$ : 
\[ \UU \setminus S_j \succ a \succ c \succ S_j \succ b \succ \text{others}~ \forall j\in J\]
\[ \UU \setminus S_j \succ S_j \succ b \succ a \succ c \succ \text{others}~ \forall j\notin J\]
We claim that the candidate $c$ is the unique winner in this extension. Note that the pairs $(w \succ a)$ and  $(w \succ u_i)$ for every $i\in[t]$ get locked first (since these differences are clearly the highest and unchanged). The pair $(c,b)$ gets locked next, with a difference score of $(3t + 2k)$. Now since the votes in which $c \succ b$ are based on a set cover of size at most $k$, the pairwise difference between $b$ and $a$ becomes at least $2k + 4 - k + (t - k) = t + 4$. Therefore, the next pair to get locked is $b \succ a$. Finally, for every element $u_i \in \UU$, the difference $D(u_i,c)$ is at most $2 + (t-1) = t+1$, since there is at least one vote where $c \succ u_i$ (given that we used a set cover in the extension). It is now easy to see that the next highest pairwise difference is between $c$ and $w$, so the ordering $c \succ w$ gets locked, and at this point, by transitivity, $c$ is superior to $w, b, a$ and all $u_i$. It follows that $c$ wins the election irrespective the sequence in which pairs are considered subsequently.

Now suppose $(\CC,V,c)$ is a \YES{} instance. Notice that, irrespective of the extension of the votes in $P$, $c \succ b, w \succ a, w \succ u_i ~\forall i\in[m]$ are locked first. Now if $b \succ c$ in all the extended votes, then it is easy to check that $b \succ a$ gets locked next, with a difference score of $2k+4+t$; leaving us with $D(u_i,c) = t+2 = D(c,w)$, where $u_i \succ c$ could be a potential lock-in. This implies the possibility of a $u_i$ being a winner in some choice of tie-breaking, a contradiction to the assumption that $c$ is the unique winner. Therefore, there are at least some votes in the extended profile where $c \succ b$. We now claim that there are at most $k$ such votes. Indeed, if there are more, then $D(b,a) = 2k + 4 - (k+1) + (t-k-1) = t + 2$. Therefore, after the forced lock-ins above, we have $D(b,a) = D(c,w) = D(a,c) = t + 2$. Here, again, it is possible for $a \succ c$ to be locked in before the other choices, and we again have a contradiction.  

Finally, we have that $c \succ b$ in at most $k$ many extensions in $P$. Call the set of indices of these extensions $J$. We claim that $\{ S_j : j\in J \}$ forms a set cover. If not, then suppose an element $u_i\in \UU$ is not covered by $\{ S_j : j\in J \}$. Then the candidate $u_i$ comes before $c$ in all the extensions which makes $ D(u_i,c) $ become $(t+2)$, which in turn ties with $D(c,w)$. This again contradicts the fact that $c$ is the unique winner. Therefore, if there is an extension that makes $c$ the unique winner, then we have the desired set cover. 
\end{proof}

\section{Polynomial Kernels for the Coalitional Manipulation Problem}\label{sec:easy}

We now describe a kernelization algorithm for every scoring rule which satisfies certain properties mentioned in \Cref{thm:pk_borda} below. Note that the Borda voting rule satisfies these properties.

\begin{theorem}
\label{thm:pk_borda} 
 For $m\in \mathbb{N}$, let $(\alpha_1, \ldots, \alpha_m)$ and $(\alpha_1^{\prime}, \ldots, \alpha_{m+1}^{\prime})$ be the normalized score vectors for a scoring rule $r$ for an election with $m$ and $(m+1)$ candidates respectively. Let $\alpha_1^{\prime} = poly(m)$ and $\alpha_i = \alpha_{i+1}^{\prime}$ for every $i\in[m]$. Then the \textsc{Coalitional manipulation} problem for $r$ admits a polynomial kernel 
 when the number of manipulators is $poly(m)$.
\end{theorem}

\setcounter{reductionrule}{0}


\begin{proof}
 Let $c$ be the candidate whom the manipulators aim to make winner. Let $M$ be the set of manipulators and $\mathcal{C}$ the set of candidates. Let $s_{NM}(x)$ be the score of candidate $x$ from the votes of the non-manipulators. Without loss of generality, we assume that, all the manipulators place $c$ at top position in their votes. Hence, the final score of $c$ is $s_{NM}(c) + |M|\alpha_1$, which we denote by $s(c)$.
 Now if $s_{NM}(x) \ge s(c)$ for any $x \ne c$, then $c$ cannot win and we output \textit{no}. Hence, we assume that $s_{NM}(x) < s(c)$ for all $x \ne c$. Now let us define $s_{NM}^*(x)$ as follows.
 $$ s_{NM}^*(x) := \max \{ s_{NM}(x), s_{NM}(c) \} $$
 Also define $s_{NM}^*(c)$ as follows.
 $$ s_{NM}^*(c) := s_{NM}(c) - |M|(\alpha_1^\prime - \alpha_1)$$
 We define a \textsc{Coalitional manipulation} instance with $(m+1)$ candidates as $(\mathcal{C}^{\prime},NM,M,c)$, where $\mathcal{C}^{\prime} = \mathcal{C} \uplus \{d\}$ is the set of candidates, $M$ is the set of manipulators, $c$ is the distinguished candidate, and $NM$ is the non-manipulators' vote such that it generates score of $x\in \mathcal{C}$ to be $ K + (s_{NM}^*(x) - s_{NM}(c))$, where $K \in \mathbb{N}$ is same for $x\in \mathcal{C}$, and the score of $d$ is less than $ K - \alpha_1^\prime|M| $. The existence of such a voting profile $NM$ of size $poly(m)$ is due to~\Cref{score_gen} and the fact that $\alpha_1^\prime = poly(m)$. Hence, once we show the equivalence of these two instances, we have a kernel whose size is polynomial in $m$. 
 The equivalence of the two instances is due to the following facts: (1) The new instance has $(m+1)$ candidates and $c$ is always placed at the top position without loss of generality. The candidate $c$ recieves $|M|(\alpha_1^\prime - \alpha_1)$ score more than the initial instance and this is compensated in $s_{NM}^*(c)$. (2) The relative score difference from the final score of $c$ to the current score of every $x\in \mathcal{C}\setminus \{c\}$ is same in both the instances. (3) In the new instance, we can assume without loss of generality that the candidate $d$ will be placed in the second position in all the manipulators' votes.
\end{proof}

We now move on to the voting rules that are based on the weighted majority graph.
The reduction rules modify the weighted majority graph maintaining the property that there exists a set of votes that can realize the modified weighted majority graph. In particular, the final weighted majority graph is realizable with a set of votes.

\setcounter{reductionrule}{0}
\begin{theorem}
\label{thm:pk_maximin} 
The \textsc{Coalitional manipulation} problem for the maximin voting rule admits a polynomial kernel 
when the number of manipulators is $poly(m)$.
\end{theorem}

\begin{proof}
 Let $c$ be the distinguished candidate of the manipulators. Let $M$ be the set of all manipulators. 
 We can assume that $ |M| \ge 2$ since for $|M| = 1$, the problem is in $\Pb$ \cite{bartholdi1989computational}.
 Define $s$ to be $\min_{x\in C\setminus \{c\}} D_{(\VV\setminus M)}(c,x)$. 
 So, $s$ is the maximin score of the candidate $c$ from the votes except from $M$. Since the maximin voting 
 rule is monotone, we can assume that the voters in $M$ put the candidate $c$ at top position of their 
 preferences. Hence, $c$'s final maximin score will be $s+|M|$. This provides the following reduction 
 rule.
 
 \begin{reductionrule}\label{rr:condorcet_winner}
  If $ s+|M| \ge 0 $, then output \YES{}.
 \end{reductionrule}
 
 In the following, we will assume $s+|M|$ is negative. Now we propose the following reduction rules on the weighted majority graph.
 
 \begin{reductionrule}\label{rr:weight_up}
  If $D_{(\VV\setminus M)}(c_i, c_j) < 0$ and $D_{(\VV\setminus M)}(c_i, c_j) > 2|M| + s$, then make $D_{(\VV\setminus M)}(c_i, c_j)$ either $2|M|+s+1$ or $2|M|+s+2$ whichever keeps the parity of 
  $D_{(\VV\setminus M)}(c_i, c_j)$ unchanged.
 \end{reductionrule}
 
 If $D_{(\VV\setminus M)}(c_i, c_j) > 2|M| + s$, then $D_{\VV}(c_i, c_j) > |M| + s$ irrespective of the way the manipulators vote. Hence, given any votes of the manipulators, whether or not the maximin score of $c_i$ and $c_j$ will exceed the maximin score of $c$ does not 
 gets affected by this reduction rule. Hence, \Cref{rr:weight_up} is sound.
 
 \begin{reductionrule}\label{rr:weight_lw}
  If $D_{(\VV\setminus M)}(c_i, c_j) < s $, then make $D_{(\VV\setminus M)}(c_i, c_j)$ either $s-1$ or
  $s-2$ whichever keeps the parity of $D_{(\VV\setminus M)}(c_i, c_j)$ unchanged.
 \end{reductionrule}
 
 The argument for the correctness of \Cref{rr:weight_lw} is similar to the argument for \Cref{rr:weight_up}. 
 Here onward, we may assume that whenever $D_{(\VV\setminus M)}(c_i, c_j) < 0$, 
 $ s-2 \le D_{(\VV\setminus M)}(c_i, c_j) \le 2|M|+s+2 $
 
 \begin{reductionrule}\label{rr:weight_closer}
  If $ s < -4|M| $ then subtract $ s+5|M| $ from $D_{(\VV\setminus M)}(x, y)$ for every $x, y\in \mathcal{C}, x\ne y$.
 \end{reductionrule}
 
The correctness of \Cref{rr:weight_closer} follows from the fact that it adds linear fixed offsets to all the 
 edges of the weighted majority graph. Hence, if there a voting profile of the voters in $M$ that makes the candidate 
 $c$ win in the original instance, the same voting profile will make $c$ win the election in the reduced 
 instance and vice versa.
 
 Now we have a weighted majority graph with $O(|M|)$ weights for every edge. Also, all the weights have uniform parity and thus the result follows from \Cref{thm:mcgarvey}.
\end{proof}
\setcounter{reductionrule}{0}

We next present a polynomial kernel for the \textsc{Coalitional manipulation} problem for the Copeland voting rule.

\begin{theorem}
\label{thm:pk_copeland} 
The \textsc{Coalitional manipulation} problem for the Copeland voting rule admits a polynomial kernel 
when the number of manipulators is $poly(m)$.
\end{theorem}

\begin{proof}
We apply the following reduction rule.
 \begin{reductionrule}\label{rr:weight_up}
  If $D_{(\VV\setminus M)}(x, y) > |M| $ for $x, y\in\mathcal{C}$, then make $D_{(\VV\setminus M)}(x, y)$ 
  either $|M|+1$ or $|M|+2$ whichever keeps the parity of 
  $D_{(\VV\setminus M)}(x, y)$ unchanged.
 \end{reductionrule}
 Given any votes of $M$, we have $D_{\VV}(x, y) > 0 $ in the original instance if and only if 
 $D_{\VV}(x, y) > 0 $ in the reduced instance for every $x, y\in\mathcal{C}$. Hence, each candidate has the same Copeland score and thus the reduction rule is correct.
 
 Now we have a weighted majority graph with $O(|M|)$ weights for every edges. Also, all the weights have uniform parity. 
 From \Cref{thm:mcgarvey}, we can realize the weighted majority graph using $O(m^2|M|)$ votes. 
\end{proof}
\setcounter{reductionrule}{0}

Now we move on to the ranked pairs voting rule.

\begin{theorem}
\label{thm:pk_rankedpair} 
The \textsc{Coalitional manipulation} problem for the ranked pairs voting rule admits a polynomial kernel 
when the number of manipulators is $poly(m)$.
\end{theorem}

\begin{proof}
 Consider all non-negative $D_{(\VV\setminus M)}(c_i, c_j)$ and arrange them in non-decreasing order. Let 
 the ordering be $x_1, x_2, \dots, x_l$ where $ l = {m \choose 2} $. Now keep applying following reduction 
 rule till possible. Define $x_0 = 0$.
 \begin{reductionrule}
  If there exist any $i$ such that, $x_i - x_{i-1} > |M|+2$, subtract an even offset to all 
  $x_i, x_{i+1}, \dots, x_l$ such that $x_i$ becomes either $(x_{i-1} + |M| + 1)$ or 
  $(x_{i-1} + |M| + 2)$.
 \end{reductionrule} 
 The reduction rule is correct since for any set of votes by $M$, for any four candidates $a, b, x, y \in \CC$, 
 $ D(a,b) > D(x,y) $ in the original instance if and only if $ D(a,b) > D(x,y) $ in the reduced instance. Now 
 we have a weighted majority graph with $O(m^2|M|)$ weights for every edges. Also, all the weights have uniform 
 parity and hence can be realized with $O(m^4|M|)$ votes by~\Cref{thm:mcgarvey}. 
\end{proof}

\section{Conclusion}
Here we showed that the \textsc{Possible winner} problem does not admit a polynomial 
kernel for many common voting rules under the complexity theoretic assumption that 
\caveat is not true. We also showed the existence of polynomial kernels for the 
\textsc{Coalitional manipulation} problem for many common voting rules. This shows that the \textsc{Possible winner} 
problem is a significantly harder problem than the \textsc{Coalitional manipulation} problem, 
although both the problems are $\NPCb{}$.

With this, we conclude the winner determination part of the thesis. We now move on to the last part of the thesis which study computational complexity of various form of election control.

\part{Election Control}
\vspace{5ex}

In this part of the thesis, we present our results on the computational complexity of various problems in the context of election control. We have the following chapters in this part.

\begin{itemize}
 \item In \Cref{chap:partial} -- \nameref{chap:partial} -- we show that manipulating an election becomes a much harder problem when the manipulators only have a partial knowledge about the votes of the other voters. Hence, manipulating an election, although feasible in theory, may often be harder in practice.
 \item In \Cref{chap:detection} -- \nameref{chap:detection} -- we initiate the work on detecting possible instances of manipulation behavior in elections. Our work shows that detecting possible instances of manipulation may often be a much easier problem than manipulating the election itself.
 \item In \Cref{chap:frugal_bribery} -- \nameref{chap:frugal_bribery} -- we show that the computational problem of bribery in an election remains an intractable problem even with a much weaker notion of bribery which we call frugal bribery. Hence, our results strengthen the intractability results from the literature on bribery.
\end{itemize}

\blankpage
\chapter{Manipulation with Partial Votes}
\label{chap:partial}

\blfootnote{A preliminary version of the work in this chapter was published as \cite{deypartial}: Palash Dey, Neeldhara Misra, and Y. Narahari. Complexity of manipulation with partial
information in voting. In Proc. Twenty-Fifth International Joint Conference on Artificial
Intelligence, IJCAI 2016, New York, NY, USA, 9-15 July 2016, pages 229-235, 2016.}

\begin{quotation}
{\small The Coalitional Manipulation problem has been
studied extensively in the literature for many voting rules. 
However, most studies have focused on the complete information setting, wherein
the manipulators know the votes of the non-manipulators. While this assumption
is reasonable for purposes of showing intractability, it is unrealistic for algorithmic
considerations. In most real-world scenarios, it is impractical for the manipulators 
to have accurate knowledge of all the other votes. In this work, we investigate
manipulation with incomplete information. In our framework,
the manipulators know a partial order for each voter that is consistent with
the true preference of that voter. In this setting, we formulate three natural
computational notions of manipulation, namely weak, opportunistic, and strong manipulation. 
We say that an extension of a partial order is
\textit{viable} if there exists a manipulative vote for that extension. We propose the following notions of manipulation when manipulators have incomplete information about the votes of other voters.

\begin{enumerate}
\item \textsc{Weak Manipulation}: the manipulators seek to
vote in a way that makes their preferred candidate win in \textit{at least one
extension} of the partial votes of the non-manipulators.
\item \textsc{Opportunistic Manipulation}: the manipulators seek to
vote in a way that makes their preferred candidate win\textit{ in every
viable extension} of the partial votes of the non-manipulators.
\item \textsc{Strong Manipulation}: the manipulators seek to
vote in a way that makes their preferred candidate win \textit{in every
extension} of the partial votes of the non-manipulators.
\end{enumerate}

We consider several scenarios for which the traditional manipulation problems are easy (for instance, Borda with a single manipulator). For many of them, the corresponding manipulative questions that we propose turn out to be computationally intractable. Our hardness results often hold even when very little information is missing, or in other words, even when the instances are very close to the complete information setting. Our results show that the impact of paucity of information on the computational complexity of manipulation crucially depends on the notion of manipulation under consideration. Our overall conclusion is that computational hardness continues to be a valid obstruction to manipulation, in the context of a more realistic model.}
\end{quotation}

\section{Introduction}

A central issue in voting is the possibility of \emph{manipulation}. For many voting rules, it turns out that even a single vote, if cast differently, can alter the outcome. In particular, a voter manipulates an election if, by misrepresenting her preference, she obtains an outcome that she prefers over the ``honest'' outcome. In a cornerstone impossibility result, Gibbard and Satterthwaite~\cite{gibbard1973manipulation,satterthwaite1975strategy} show that every unanimous and non-dictatorial voting rule with three candidates or more is manipulable. We refer to~\cite{brandt2015handbook} for an excellent introduction to various strategic issues in computational social choice theory.

Considering that voting rules are indeed susceptible to manipulation, it is natural to seek ways by which elections can be protected from manipulations. The works of Bartholdi et al.~\cite{bartholdi1989computational,bartholdi1991single} approach the problem from the perspective of computational intractability. They exploit the possibility that voting rules, despite being vulnerable to manipulation in theory, may be hard to manipulate in practice. Indeed, a manipulator is faced with the following decision problem: given a collection of votes $\mathcal{P}$ and a distinguished candidate $c$, does there exist a vote $v$ that, when tallied with $\mathcal{P}$, makes $c$ win for a (fixed) voting rule $r$? The manipulation problem has subsequently been generalized to the problem of \textsc{Coalitional manipulation} by Conitzer et al.~\cite{conitzer2007elections}, where one or more manipulators collude together and try to make a distinguished candidate win the election. The manipulation problem, fortunately, turns out to be \NP{}-hard in several settings. This established the success of the approach of demonstrating a computational barrier to manipulation.

However, despite having set out to demonstrate the hardness of manipulation, the initial results in~\cite{bartholdi1989computational} were to the contrary, indicating that many voting rules are in fact easy to manipulate. Moreover, even with multiple manipulators involved, popular voting rules like plurality, veto, $k$-approval, Bucklin, and Fallback continue to be easy to manipulate~\cite{xia2009complexity}. While we know that the computational intractability may not provide a strong barrier \cite{procaccia2006junta,ProcacciaR07,friedgut2008elections,xia2008generalized,xia2008sufficient,faliszewski2010using,walsh2010empirical,walsh2011hard,isaksson2012geometry,dey2015computational,DeyMN15,journalsDeyMN16,DeyMN15a,dey2014asymptotic,dey2015asymptoticjournal} even for rules for which the coalitional manipulation problem turns out to be \NP{}-hard, in all other cases the possibility of manipulation is a much more serious concern.

\subsection{Motivation and Problem Formulation}

In our work, we propose to extend the argument of computational intractability to address the cases where the approach appears to fail. We note that most incarnations of the manipulation problem studied so far are in the complete information setting, where the manipulators have complete knowledge of the preferences of the truthful voters. While these assumptions are indeed the best possible for the computationally negative results, we note that they are not reflective of typical real-world scenarios. Indeed, concerns regarding privacy of information, and in other cases, the sheer volume of information, would be significant hurdles for manipulators to obtain complete information. Motivated by this, we consider the manipulation problem in a natural \emph{partial information} setting. In particular, we model the partial information of the manipulators about the votes of the non-manipulators as partial orders over the set of candidates. A partial order over the set of candidates will be called a partial vote. Our results show that several of the voting rules that are easy to manipulate in the complete information setting become intractable when the manipulators know only partial votes. Indeed, for many voting rules, we show that even if the ordering of a small number of pairs of candidates is missing from the profile, manipulation becomes an~intractable problem. Our results therefore strengthen the view that manipulation may not be practical if we limit the information the manipulators have at their disposal about the votes of other voters \cite{conitzer2011dominating}.

We introduce three new computational problems that, in a natural way, extend the question of manipulation to the partial information setting. In these problems, the input is a set of partial votes $\mathcal{P}$ corresponding to the votes of the non-manipulators, a non-empty set of manipulators $M$, and a preferred candidate $c$. The task in the \textsc{Weak Manipulation (WM)} problem is to determine if there is a way to cast the manipulators' votes such that $c$ wins the election for at least one extension of the partial votes in $\mathcal{P}$. On the other hand, in the \textsc{Strong Manipulation (SM)} problem, we would like to know if there is a way of casting the manipulators' votes such that $c$ wins the election in \emph{every extension} of the partial votes in $\mathcal{P}$. 

We also introduce the problem of \textsc{Opportunistic Manipulation (OM)}, which is an ``intermediate'' notion of manipulation. Let us call an extension of a partial profile \textit{viable} if it is possible for the manipulators to vote in such a way that the manipulators' desired candidate wins in that extension. In other words, a viable extension is a \YES-instance of the standard \CM problem. We have an opportunistic manipulation when it is possible for the manipulators to cast a vote which makes $c$ win the election in \emph{ all} viable extensions. Note that any \YES-instance of \SM is also a \YES-instance of \OM, but this may not be true in the reverse direction. As a particularly extreme example, consider a partial profile where there are no viable extensions: this would be a \NO-instance for \SM, but a (vacuous) \YES-instance of \OM. The \OM problem allows us to explore a more relaxed notion of manipulation: one where the manipulators are obliged to be successful only in extensions where it is possible to be successful. Note that the goal with \SM is to be successful in all extensions, and therefore the only interesting instances are the ones where all extensions are viable. 

It is easy to see that \YES{} instance of \SM is also a \YES{} instance of \OM and \WM. Beyond this, we remark that all the three problems are questions with different goals, and neither of them render the other redundant. We refer the reader to~\Cref{Fig:EG} for a simple example distinguishing these scenarios.

All the problems above generalize \CM, and hence any computational intractability result for \CM immediately yields a corresponding intractability result for \WM, \SM, and \OM under the same setting. For example, it is known that the \CM problem is intractable for the maximin voting rule when we have at least two manipulators~\cite{xia2009complexity}. Hence, the \WM, \SM, and \OM problems are intractable for the maximin voting rule when we have at least two manipulators.

\begin{figure}[t]
\centering
\includegraphics[scale=0.3]{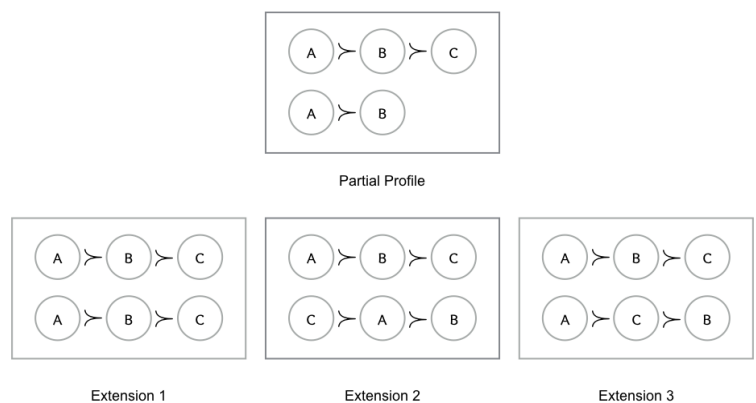}
\caption{An example of a partial profile. Consider the plurality voting rule with one manipulator. If the favorite candidate is A, then the manipulator simply has to place A on the top his vote to make A win in \textit{any} extension. If the favorite candidate is B, there is \textit{no vote} that makes B win in any extension. Finally, if the favorite candidate is C, then with a vote that places C on top, the manipulator can make C win in the only viable extension (Extension 2).}
\label{Fig:EG}	
\end{figure}

\subsection{Related Work}

A notion of manipulation under partial information has been considered by Conitzer et al.~\cite{conitzer2011dominating}. They focus on whether or not there exists a dominating manipulation and show that this problem is \NP{}-hard for many common voting rules. Given some partial votes, a dominating manipulation is a non-truthful vote that the manipulator can cast which makes the winner at least as preferable (and sometimes more preferable) as the winner when the manipulator votes truthfully. The dominating manipulation problem and the \WM, \OM, and \SM problems do not seem to have any apparent complexity-theoretic connection. For example, the dominating manipulation problem is \NP{}-hard for all the common voting rules except plurality and veto, whereas, the \SM problem is easy for most of the cases (see \Cref{table:partial_summary}). However, the results in \cite{conitzer2011dominating} establish the fact that it is indeed possible to make manipulation intractable by restricting the amount of information the manipulators possess about the votes of the other voters. Elkind and Erd{\'e}lyi~\cite{elkind2012manipulation} study manipulation under voting rule uncertainty. However, in our work, the voting rule is fixed and known to the manipulators.

Two closely related problems that have been extensively studied in the context of incomplete votes are \textsc{Possible Winner} and \textsc{Necessary Winner}~\cite{konczak2005voting}. In the \PW problem, we are given a set of partial votes $\mathcal{P}$ and a candidate $c$, and the question is whether there exists an extension of $\mathcal{P}$ where $c$ wins, while in the \NW problem, the question is whether $c$ is a winner in every extension of $\mathcal{P}$. Following the work in~\cite{konczak2005voting}, a number of  special cases and variants of the \PW problem have been studied in the literature~\cite{chevaleyre2010possible,bachrach2010probabilistic,baumeister2011computational,baumeister2012possible,gaspers2014possible,xia2008determining,ding2013voting,NarodytskaW14,BaumeisterFLR12,MenonL15}. The flavor of the \WM problem is clearly similar to \PW. However, we emphasize that there are subtle distinctions between the two problems. A more elaborate comparison is made in the next section.

\subsection{Our Contribution}

Our primary contribution in this work is to propose and study three natural and realistic generalizations of the computational problem of manipulation in the incomplete information setting. We summarize the complexity results in this work in \Cref{table:partial_summary}. Our results provides the following interesting insights on the impact of lack of information on the computational difficulty of manipulation. We note that the number of undetermined pairs of candidates per vote are small constants in all our hardness results.

\begin{itemize}
 \item We observe that the computational problem of manipulation for the plurality and veto voting rules remains polynomial time solvable even with lack of information, irrespective of the notion of manipulation under consideration [\Cref{pv_easy,thm:vetoOM,sm_k_easy,obs:plurality_fallback}]. We note that the plurality and veto voting rule also remain vulnerable under the notion of dominating manipulation~\cite{conitzer2011dominating}.
 
 \item The impact of absence of information on the computational complexity of manipulation is more dynamic for the $k$-approval, $k$-veto, Bucklin, Borda, and maximin voting rules. Only the \WM and \OM problems are computationally intractable for the $k$-approval [\Cref{wm_k_hard,thm:kappOM}], $k$-veto [\Cref{k_veto_wm,thm:kvetoOM}], Bucklin [\Cref{wm_bucklin_hard,thm:bucklinOM}], Borda [\Cref{cm_hard_one,thm:bordaOM}], and maximin [\Cref{cm_hard_one,thm:maximinOM}] voting rules, whereas the \SM problem remains computationally tractable [\Cref{sm_k_easy,sm_sr_easy,sm_bucklin_easy,sm_maximin_easy}].
 
 \item \Cref{table:partial_summary} shows an interesting behavior of the fallback voting rule. The Fallback voting rule is the only voting rule among the voting rules we study here for which the \WM problem is NP-hard [\Cref{wm_bucklin_hard}] but both the \OM and \SM problems are polynomial time solvable [\Cref{sm_bucklin_easy,obs:plurality_fallback}]. This is because the \OM problem can be solved for the fallback voting rule by simply making manipulators vote for their desired candidate.
 
 \item Our results show that absence of information makes all the three notions of manipulations intractable for the Copeland$^\alpha$ voting rule for every rational $\alpha\in[0,1]\setminus\{0.5\}$ for the \WM problem [\Cref{cm_hard_one}] and for every $\alpha\in[0,1]$ for the \OM and \SM problems [\Cref{sm_copeland_hard,thm:copelandOM}].
\end{itemize}

Our results (see \Cref{table:partial_summary}) show that whether lack of information makes the manipulation problems harder, crucially depends on the notion of manipulation applicable to the situation under consideration. All the three notions of manipulations are, in our view, natural extension of manipulation to the incomplete information setting and tries to capture different behaviors of manipulators. For example, the \WM problem may be applicable to an optimistic manipulator whereas for an pessimistic manipulator, the \SM problem may make more sense.

\begin{table}
\centering
\resizebox{\linewidth}{!}{
\renewcommand*{\arraystretch}{1.3}
\begin{tabular}{|c|c|c|c|c|c|c|}
\hline
                  & WM, $\ell = 1$                                                          & WM                                                         & OM, $\ell = 1$                            & OM                           & SM, $\ell = 1$                              & SM                                                                                                          \\ \hline
Plurality         & \multicolumn{2}{c}{}                                                                                        & \multicolumn{2}{c}{\cellcolor[HTML]{67FD9A}}                            & \multicolumn{2}{c|}{\cellcolor[HTML]{67FD9A}}                                                                                                             \\ \cline{1-1}
Veto              & \multicolumn{2}{c}{\multirow{-2}{*}{P }}                                                     & \multicolumn{2}{c}{\multirow{-2}{*}{\cellcolor[HTML]{67FD9A}P}}         & \multicolumn{2}{c|}{\cellcolor[HTML]{67FD9A}}                                                                                                             \\ \cline{1-1}

$k$-Approval      & \multicolumn{2}{c}{\cellcolor[HTML]{FE996B}}                                                                                        & \multicolumn{2}{c}{\cellcolor[HTML]{FFCB2F}}                            & \multicolumn{2}{c|}{\cellcolor[HTML]{67FD9A}}                                                                                                             \\ \cline{1-1}
$k$-Veto          & \multicolumn{2}{c}{\cellcolor[HTML]{FE996B}}                                                                                        & \multicolumn{2}{c}{\cellcolor[HTML]{FFCB2F}}                            & \multicolumn{2}{c|}{\cellcolor[HTML]{67FD9A}}                                                                                                             \\ \cline{1-1}
Bucklin           & \multicolumn{2}{c}{\cellcolor[HTML]{FE996B}}                                                                                        & \multicolumn{2}{c}{\multirow{-3}{*}{\cellcolor[HTML]{FFCB2F}coNP-hard}} & \multicolumn{2}{c|}{\cellcolor[HTML]{67FD9A}}                                                                                                             \\ \cline{1-1}
Fallback          & \multicolumn{2}{c}{\multirow{-4}{*}{\cellcolor[HTML]{FE996B}NP-complete}}                                                           & \multicolumn{2}{c}{\cellcolor[HTML]{67FD9A}P}                           & \multicolumn{2}{c|}{\multirow{-6}{*}{\cellcolor[HTML]{67FD9A}P}}                                                                                          \\ \cline{1-1}
Borda             & \multicolumn{2}{c}{}                                                                                        & \multicolumn{2}{c}{\cellcolor[HTML]{FFCB2F}}                            & \multicolumn{1}{c}{\cellcolor[HTML]{67FD9A}}                    &                                                                                 \\ \cline{1-1}
maximin           & \multicolumn{2}{c}{}                                                                                        & \multicolumn{2}{c}{\cellcolor[HTML]{FFCB2F}}                            & \multicolumn{1}{c}{\multirow{-2}{*}{\cellcolor[HTML]{67FD9A}P}} &                                                                                     \\ \cline{1-1}
Copeland$^\alpha$ & \multicolumn{2}{c}{\multirow{-3}{*}{
\begin{tabular}[c]{@{}c@{}}NP-complete\\ \end{tabular}}} & \multicolumn{2}{c}{\multirow{-3}{*}{\cellcolor[HTML]{FFCB2F}coNP-hard}} & \multicolumn{1}{c}{\cellcolor[HTML]{FFCB2F}coNP-hard}         & \multirow{-3}{*}{\begin{tabular}[c]{@{}c@{}}NP-hard\\ \end{tabular}} \\ \hline
\end{tabular}}
\caption{Summary of Results ($\el$ denotes the number of manipulators). The results in white follow immediately from the literature (\Cref{pw_hard,cm_hard,cm_hard_one}). Our results for the Copeland$^\alpha$ voting rule hold for every rational $\alpha\in[0,1]\setminus\{0.5\}$ for the \WM problem and for every $\alpha\in[0,1]$ for the \OM and \SM problems.}
\label{table:partial_summary}
\end{table}

\subsection{Problem Definitions}\label{sec:probdef}

We now formally define the three problems that we consider in this work, namely \textsc{Weak Manipulation},  \textsc{Opportunistic Manipulation}, and \textsc{Strong Manipulation}. Let $r$ be a fixed voting rule. We first introduce the \textsc{Weak Manipulation} problem.

\begin{definition}\textbf{\textsc{$r$-Weak Manipulation}}\\
 Given a set of partial votes $\mathcal{P}$ over a set of candidates $\mathcal{C}$, a positive integer $\el~(>0)$ denoting the number of manipulators, and a candidate $c$, do there exist votes $\succ_1, \ldots, \succ_\el\, \in \mathcal{L(\mathcal{C})}$ such that there exists an extension $\succ\, \in \mathcal{\mathcal{L(\mathcal{C})}^{|\mathcal{P}|}}$ of $\mathcal{P}$ with $r(\succ, \succ_1, \ldots, \succ_\el) = c$?
\end{definition}

To define the \textsc{Opportunistic Manipulation} problem, we first introduce the notion of an $(r,c)$-opportunistic voting profile, where $r$ is a voting rule and $c$ is any particular candidate.

\begin{definition}\textbf{$(r,c)$-Opportunistic Voting Profile}\\
 Let $\el$ be the number of manipulators and $\PP$ a set of partial votes. An $\el$-voter profile $(\succ_i)_{i\in[\el]}\in\LL(\CC)^\el$ is called an $(r,c)$-opportunistic voting profile if for each extension $\overline{\PP}$ of $\PP$ for which there exists an $\el$-vote profile $(\succ^\prime_i)_{i\in[\el]}\in\LL(\CC)^\el$ with $r\left(\overline{\PP}\cup\left(\succ^\prime_i\right)_{i\in[\el]}\right) = c$, we have $r\left(\overline{\PP}\cup\left(\succ_i\right)_{i\in[\el]}\right) = c$.
\end{definition}

In other words, an $\ell$-vote profile is $(r,c)$-opportunistic with respect to a partial profile if, when put together with the truthful votes of any extension,  $c$ wins if the extension is viable to begin with. We are now ready to define the \textsc{Opportunistic Manipulation} problem.

\begin{definition}\textbf{\textsc{$r$-Opportunistic Manipulation}}\\
 Given a set of partial votes $\mathcal{P}$ over a set of candidates $\mathcal{C}$, a positive integer $\el~(>0)$ denoting the number of manipulators, and a candidate $c$, does there exist an $(r,c)$-opportunistic $\el$-vote profile?
\end{definition}

We finally define the \textsc{Strong Manipulation} problem.

\begin{definition}\textbf{\textsc{$r$-Strong Manipulation}}\\
 Given a set of partial votes $\mathcal{P}$ over a set of candidates $\mathcal{C}$, a positive integer $\el~(>0)$ denoting the number of manipulators, and a candidate $c$, do there exist votes $(\succ_i)_{i\in\el} \in \mathcal{L(\mathcal{C})}^\el$ such that for every extension $\succ\, \in \mathcal{\mathcal{L(\mathcal{C})}^{|\mathcal{P}|}}$ of $\mathcal{P}$, we have $r(\succ, (\succ_i)_{i\in[\el]}) = c$?
\end{definition}

We use $(\PP, \el, c)$ to denote instances of \WM, \OM, and \SM, where $\PP$ denotes a profile of partial votes, $\el$ denotes the number of manipulators, and $c$ denotes the desired winner.

For the sake of completeness, we provide the definitions of the \textsc{Coalitional Manipulation} and \textsc{Possible Winner} problems below.
\begin{definition}\textbf{\textsc{$r$-Coalitional Manipulation}}\\
 Given a set of complete votes $\succ$ over a set of candidates $\mathcal{C}$, a positive integer $\el~(>0)$ denoting the number of manipulators, and a candidate $c$, do there exist votes $(\succ_i)_{i\in\el} \in \mathcal{L(\mathcal{C})}^\el$ such that $r\left(\succ, \left(\succ_i\right)_{i\in[\el]}\right) = c$?
\end{definition}

\begin{definition}\textbf{\textsc{$r$-Possible Winner}}\\
 Given a set of partial votes $\mathcal{P}$ and a candidate $c$, does there exist an extension $\succ$ of the partial votes in $\mathcal{P}$ to linear votes such that $r(\succ)=c$?
\end{definition}

\subsection{Comparison with Possible Winner and Coalitional Manipulation Problems} ~For any fixed voting rule, the \WM problem with $\ell$ manipulators reduces to the \PW problem. This is achieved by simply using the same set as truthful votes and introducing $\ell$ empty votes. We summarize this in the observation below.

\begin{observation}\label{pw_hard}
 The \WM problem many-to-one reduces to the \PW problem for every voting rule.
\end{observation}

\begin{proof}
 Let $(\PP, \el, c)$ be an instance of \WM. Let $\QQ$ be the set consisting of \el many copies of partial votes $\{\emptyset\}$. Clearly the \WM instance $(\PP, \el, c)$ is equivalent to the \PW instance $(\PP\cup\QQ)$.
\end{proof}

However, whether the \PW problem reduces to the \WM problem or not is not clear since in any \WM problem instance, there must exist at least one manipulator and a \PW instance may have no empty vote. From a technical point of view, the difference between the \WM and \PW problems may look marginal; however we believe that the \WM problem is a very natural generalization of the \CM problem in the partial information setting and thus worth studying. Similarly, it is easy to show, that the \CM problem with $\ell$ manipulators reduces to \WM, \OM, and \SM problems with $\ell$ manipulators, since the former is a special case of the latter ones.

\begin{observation}\label{cm_hard}
 The \CM problem with \el manipulators many-to-one reduces to \WM, \OM, and \SM problems with \el manipulators for all voting rules and for all positive integers \el.
\end{observation}

\begin{proof}
 Follows from the fact that every instance of the \CM problem is also an equivalent instance of the \WM, \OM, and \SM problems.
\end{proof}

Finally, we note that the \CM problem with $\ell$ manipulators can be reduced to the \WM problem with just one manipulator, by introducing $\ell-1$ empty votes. These votes can be used to witness a good extension in the forward direction. In the reverse direction, given an extension where the manipulator is successful, the extension can be used as the manipulator's votes. This argument leads to the following observation.

\begin{observation}\label{cm_hard_one}
 The \CM problem with \el manipulators many-to-one reduces to the \WM problem with one manipulator for every voting rule and for every positive integer \el.
\end{observation}

\begin{proof}
 Let $(\PP, \el, c)$ be an instance of \CM. Let $\QQ$ be the set of consisting of $\el-1$ many copies of partial vote $\{c\succ \text{others}\}$. Clearly the \WM instance $(\PP\cup\QQ, 1, c)$ is equivalent to the \CM instance $(\PP, \el, 1)$.
\end{proof}

This observation can be used to derive the hardness of \WM even for one manipulator whenever the hardness for \CM is known for any fixed number of manipulators (for instance, this is the case for the voting rules such as Borda, maximin and Copeland). However, determining the complexity of \WM with one manipulator requires further work for voting rules where \CM is polynomially solvable for any number of manipulators (such as $k$-approval, Plurality, Bucklin, and so on).

\section{Hardness Results for \WM, \OM, and \SM Problems}\label{sec:hard}

In this section, we present our hardness results. While some of our reductions are from the \textsc{Possible Winner} problem, the other reductions in this section are from the \textsc{Exact Cover by 3-Sets} problem, also referred to as X3C. This is a well-known \NPC{}~\cite{garey1979computers} problem, and is defined as follows.

\begin{definition}[Exact Cover by 3-Sets (X3C)] Given a set $\UU$ and a collection $\SS = \{S_1,S_2, \dots, S_t\}$ of $t$ subsets of $ \UU$ with $|S_i|=3 ~\forall i=1, \dots, t,$ does there exist a $\TT\subset\SS$ with $|\TT|=\frac{|\UU|}{3}$ such that $\cup_{X\in \TT} X = \UU$?
\end{definition}

We use $\overline{\text{X3C}}$ to refer to the complement of X3C, which is to say that an instance of $\overline{\text{X3C}}$ is a \YES instance if and only if it is a \NO instance of X3C. The rest of this section is organized according to the problems being addressed.

\subsection{Weak Manipulation Problem} 

To begin with, recall that the \CM problem is \NPC for the Borda~\cite{davies2011complexity,betzler2011unweighted}, maximin~\cite{xia2009complexity}, and Copeland$^\alpha$~\cite{faliszewski2008copeland,FaliszewskiHHR09,faliszewski2010manipulation} voting rules for every rational $\alpha\in[0,1]\setminus\{0.5\}$, when we have two manipulators. Therefore, it follows from~\Cref{cm_hard_one} that the \WM problem is \NPC for the Borda, maximin, and Copeland$^\alpha$ voting rules for every rational $\alpha\in[0,1]\setminus\{0.5\}$, even with one manipulator. 

For the $k$-approval and $k$-veto voting rules, we reduce from the corresponding \PW problems. While it is natural to start from the same voting profile, the main challenge is in undoing the advantage that the favorite candidate receives from the manipulator's vote, in the reverse direction.

\subsubsection{Result for the $k$-Approval Voting Rule}

We begin with proving that the \WM problem is \NP{}-complete for the $k$-approval voting rule even with one manipulator and at most $4$ undetermined pairs per vote.

\begin{theorem}\label{wm_k_hard}
The \WM problem is \NP{}-complete for the $k$-approval voting rule even with one manipulator for any constant $k>1$, even when the number of undetermined pairs in each vote is no more than $4$.
\end{theorem}

\begin{proof}
For simplicity of presentation, we prove the theorem for $2$-approval.
We reduce from the \PW problem for $2$-approval which is \NP{}-complete~\cite{xia2008determining}, even when the number of undetermined pairs in each vote is no more than $4$. 
Let $\mathcal{P}$ be the set of partial votes in a \PW instance, and 
let ${\mathcal C} = \{c_1, \ldots, c_m, c\}$ be the set of candidates, 
where the goal is to check if there is an extension of $\mathcal{P}$ that makes $c$ win. For developing the instance of \WM, we need to ``reverse'' any advantage that the candidate $c$ obtains from the vote of the manipulator. Notice that the most that the manipulator can do is to increase the score of $c$ by one. Therefore, in our construction, we \textit{``artificially''} increase the score of all the other candidates by one, so that despite of the manipulator's vote, $c$ will win the new election if and only if it was a possible winner in the \PW instance. To this end, we introduce $(m+1)$ many \textit{dummy} candidates $d_1, \ldots, d_{m+1}$ and the complete votes:
$$w_i = c_i \succ d_i \succ \text{others},\text{ for every } i\in \{1, \dots, m\}$$
Further, we extend the given partial votes of the \PW instance to force the dummy candidates 
to be preferred least over the rest - by defining, for every $v_i \in \mathcal{P}$, the corresponding partial vote $v_i^{\prime}$ as follows.
$$v_i^\prime = v_i \cup \{{\mathcal C} \succ \{d_1, \ldots, d_{m+1}\}\}.$$ 
This ensures that all the dummy candidates do not receive any score 
from the modified partial votes corresponding to the partial votes of the \PW instance. Notice that since the number of undetermined pairs in $v_i$ is no more than $4$, the number of undetermined pairs in $v_i^\prime$ is also no more than $4$. Let $({\mathcal C^\prime},\mathcal{Q},c)$ denote this constructed \WM instance. We claim that the two instances are equivalent.

In the forward direction, suppose $c$ is a possible winner with respect to $\mathcal{P}$, and let $\overline{\mathcal{P}}$ be an extension where $c$ wins. Then it is easy to see that the manipulator can make $c$ win in some extension by placing $c$ and $d_{m+1}$ in the first two positions of her vote (note that the partial score of $d_{m+1}$ is zero in $\mathcal{Q}$). Indeed, consider the extension of $\mathcal{Q}$ obtained by mimicking the extension $\overline{\mathcal{P}}$ on the ``common'' partial votes, $\{v_i^\prime ~|~ v_i \in \mathcal{P}\}$. Notice that this is well-defined since $v_i$ and $v_i^\prime$ have exactly the same set of incomparable pairs. In this extension, the score of $c$ is strictly greater than the scores of all the other candidates, since the scores of all candidates in $\mathcal{C}$ is exactly one more than their scores in $\mathcal{P}$, and all the dummy candidates have a score of at most one. 

In the reverse direction, notice that the manipulator puts the candidates $c$ and $d_{m+1}$ in the top two positions without loss of generality. Now suppose the manipulator's vote $c\succ d_{m+1}\succ \text{others}$ makes $c$ win the election for an extension $\overline{\mathcal{Q}}$ of $\mathcal{Q}$. Then consider the extension $\overline{\mathcal{P}}$ obtained by restricting $\overline{\mathcal{Q}}$ to $\mathcal{C}$. Notice that the score of each candidate in $\mathcal{C}$ in this extension is one less than their scores in $\mathcal{Q}$. Therefore, the candidate $c$ wins this election as well, concluding the proof. 

The above proof can be imitated for any other constant values of $k$ by reducing it from the \PW problem for $k$-approval and introducing $(m+1)(k-1)$ dummy candidates.
\end{proof}

\subsubsection{Result for the $k$-Veto Voting Rule}

\begin{theorem}\label{k_veto_wm}
 The \WM problem for the $k$-veto voting rule is \NP{}-complete even with one manipulator for any constant $k>1$.
\end{theorem}

\begin{proof}
We reduce from the \PW problem for the $k$-veto voting rule which is known to be \NPC~\cite{betzler2009towards}. Let $\mathcal{P}$ be the set of partial votes in a \PW problem instance, and 
let ${\mathcal C} = \{c_1, \ldots, c_m, c\}$ be the set of candidates, 
where the goal is to check if there is an extension that makes $c$ win with respect to $k$-veto. We assume without loss of generality that $c$'s position is fixed in all the partial votes (if not, then we fix the position of $c$ as high as possible in every vote).

We introduce $k+1$ many \textit{dummy} candidates $d_1, \ldots, d_k, d$. The role of the first $k$ dummy candidates is to ensure that the manipulator is forced to place them at the ``bottom $k$'' positions of her vote, so that all the original candidates get the same score from the additional vote of the manipulator. The most natural way of achieving this is to ensure that the dummy candidates have the same score as $c$ in any extension (note that we know the score of $c$ since $c$'s position is fixed in all the partial votes). This would force the manipulator to place these $k$ candidates in the last $k$ positions. Indeed, doing anything else will cause these candidates to tie with $c$, even when there is an extension of $\mathcal{P}$ that makes $c$ win.

To this end, we begin by placing the dummy candidates in the top $k$ positions in all the partial votes. Formally, we modify every partial vote as follows:
$$w = d_i \succ \text{others},\text{ for every } i\in \{1, \dots, k\}$$
At this point, we know the scores of $c$ and $d_i,\text{ for every } i\in \{1, \dots, k\}$. Using \Cref{score_gen}, we add complete votes such that the final score of $c$ is the same with the score of every $d_i$ and the score of $c$ is strictly more than the score of $d$. The relative score of every other candidate remains the same. This completes the description of the construction. We denote the augmented set of partial votes by $\overline{\mathcal{P}}$. 

We now argue the correctness. In the forward direction, if there is an extension of the votes that makes $c$ win, then we repeat this extension, and the vote of the manipulator puts the candidate $d_i$ at the position $m+i+2$; and all the other candidates in an arbitrary fashion. Formally, we let the manipulator's vote be:
$$\mathfrak{v} = c \succ c_1 \succ \cdots \succ c_m \succ d \succ d_1 \succ \cdots \succ d_k.$$
By construction $c$ wins the election in this particular setup. In the reverse direction, consider a vote of the manipulator and an extension $\overline{\mathcal{Q}}$ of $\overline{\mathcal{P}}$ in which $c$ wins. Note that the manipulator's vote necessarily places the candidates $d_i$ in the bottom $k$ positions --- indeed, if not, then $c$ cannot win the election by construction. We extend a partial vote $w \in \mathcal{P}$ by mimicking the extension of the corresponding partial vote $w^\prime \in \overline{\mathcal{P}}$, that is, we simply project the extension of $w^\prime$ on the original set of candidates $\mathcal{C}$. Let $\mathcal{Q}$ denote this proposed extension of $\mathcal{P}$. We claim that $c$ wins the election given by $\mathcal{Q}$. Indeed, suppose not. Let $c_i$ be a candidate whose score is at least the score of $c$ in the extension $\mathcal{Q}$. Note that the scores of $c_i$ and $c$ in the extension  $\overline{\mathcal{Q}}$ are exactly the same as their scores in $\mathcal{Q}$, except for a constant offset --- importantly, their scores are offset by the same amount. This implies that the score of $c_i$ is at least the score of $c$ in $\overline{\mathcal{Q}}$ as well, which is a contradiction. Hence, the two instances are equivalent.
\end{proof}

\subsubsection{Result for the Bucklin Voting Rule}

We next prove, by a reduction from X3C, that the \textsc{Weak Manipulation} problem for the Bucklin and simplified Bucklin voting rules is \NPC{} even with one manipulator and at most $16$ undetermined pairs per vote.

\begin{theorem}\label{wm_bucklin_hard}
The \textsc{Weak Manipulation} problem is \NPC{} for Bucklin, simplified Bucklin, Fallback, and simplified Fallback voting rules, even when we have only one manipulator and the number of undetermined pairs in each vote is no more than $16$. 
\end{theorem}

\begin{proof}
We reduce the X3C problem to \textsc{Weak Manipulation} for simplified Bucklin. Let $(\UU = \{u_1, \ldots, u_m\}, \SS:= \{S_1,S_2, \dots, S_t\})$ be an instance of X3C, where each $S_i$ is a subset of $\UU$ of size three. We construct a \WM instance based on $(\UU,\SS)$ as follows.
$$ \text{Candidate set: } \mathcal{C} = \WW\cup \XX \cup \DD\cup \UU\cup \{c,w, a, b\},\text{ where } |\WW|=m-3, |\XX|=4, |\DD|=m+1$$ 
We first introduce the following partial votes $\PP$ in correspondence with the sets in the family as follows.
$$ \WW\succ \XX\succ S_i\succ c\succ (\UU\setminus S_i)\succ \DD \setminus \left(\{\XX \times (\{c\}\cup S_i)\}\right), \forall i\le t$$
Notice that the number of undetermined pairs in every vote in $\PP$ is $16$. We introduce the following additional complete votes $\QQ$:
\begin{itemize}
 \item $t$ copies of $\UU\succ c\succ \text{others}$
 \item $\nfrac{m}{3}-1$ copies of $\UU\succ a\succ c\succ \text{others}$
 \item $\nfrac{m}{3}+1$ copies of $\DD\succ b\succ \text{others}$
\end{itemize}
The total number of voters, including the manipulator, is $2t+\nfrac{2m}{3}+1$. Now we show equivalence of the two instances. 

In the forward direction, suppose we have an exact set cover $\TT\subset \SS$. Let the vote of the manipulator $\vvv$ be $c\succ D\succ \text{others}$. We consider the following extension $\overline{\PP}$ of $\PP$. 

$$\WW \succ S_i\succ c\succ \XX\succ (\UU\setminus S_i)\succ \DD$$

On the other hand, if $S_i\in\SS\setminus\TT$, then we have:
$$\WW \succ \XX\succ S_i\succ c\succ (\UU\setminus S_i)\succ \DD$$

We claim that $c$ is the unique simplified Bucklin winner in the profile $(\overline{\PP},\WW,\vvv)$. Notice that the simplified Bucklin score of $c$ is $m+1$ in this extension, since it appears in the top $m+1$ positions in the $m/3$ votes corresponding to the set cover, $t$ votes from the complete profile $\QQ$ and one vote $\vvv$ of the manipulator. For any other candidate $u_i\in \UU$, $u_i$ appears in the top $m+1$ positions once in $\overline{\PP}$ and $t+\frac{m}{3}-1$ times in $\QQ$. Thus, $u_i$ does not get majority in top $m+1$ positions making its simplified Bucklin score at least $m+2$. Hence, $c$ is the unique simplified Bucklin winner in the profile $(\overline{\PP},\WW,\vvv)$.
Similarly, the candidate $w_1$ appears only $t$ times in the top $m+1$ positions. 
The same can be argued for the remaining candidates in $\DD, \WW,$ and $w$. 

In the reverse direction, suppose the \WM is a \YES instance. We may assume without loss of generality that the manipulator's vote $\vvv$ is $c\succ \DD\succ \text{others}$, since the simplified Bucklin score of the candidates in $\DD$ is at least $2m$. Let $\overline{\PP}$ be the extension of $\PP$ such that $c$ is the unique winner in the profile $(\overline{\PP},\QQ,\vvv)$. As $w_1$ is ranked within top $m+2$ positions in $t+\frac{m}{3}+1$ votes in $\QQ$, for $c$ to win, $c\succ w_{m-2}$ must hold in at least $\frac{m}{3}$ votes in $\overline{\PP}$. In those votes, all the candidates in $S_i$ are also within top $m+2$ positions. Now if any candidate in $\UU$ is within top $m+1$ positions in $\overline{\PP}$ more than once, then $c$ will not be the unique winner. Hence, the $S_i$'s corresponding to the votes where $c\succ w_{m-2}$ in $\overline{\PP}$ form an exact set cover. 

The reduction above also works for the Bucklin voting rule. Specifically, the argument for the forward direction is exactly the same as the simplified Bucklin above and the argument for the reverse direction is as follows. The candidate $w_1$ is ranked within top $m+2$ positions in $t+\frac{m}{3}+1$ votes in $\QQ$ and $c$ is never placed within top $m+2$ positions in any vote in $\QQ$. Hence, for $c$ to win, $c\succ w_{m-2}$ must hold in at least $\frac{m}{3}$ votes in $\overline{\PP}$. In those votes, all the candidates in $S_i$ are also within top $m$ positions. Notice that $c$ never gets placed within top $m$ positions in any vote in $(\overline{\PP},\QQ)$. Now if any candidate $x\in\UU$ is within top $m$ positions in $\overline{\PP}$ more than once, then $x$ gets majority within top $m$ positions and thus $c$ cannot win.

The result for the Fallback and simplified Fallback voting rules follow from the corresponding results for the Bucklin and simplified Bucklin voting rules respectively since every Bucklin and simplified Bucklin election is also a Fallback and simplified Fallback election respectively.
\end{proof}

\subsection{Strong Manipulation Problem} 

We know that the \CM problem is \NPC for the Borda, maximin, and Copeland$^\alpha$ voting rules for every rational $\alpha\in[0,1]\setminus\{0.5\}$, when we have two manipulators. Thus, it follows from~\Cref{cm_hard} that \SM is \NPH{} for Borda, maximin, and Copeland$^\alpha$ voting rules for every rational $\alpha\in[0,1]\setminus\{0.5\}$ for at least two manipulators. 

For the case of one manipulator, \SM turns out to be polynomial-time solvable for most other voting rules. For Copeland$^\alpha$, however, we show that the problem is \coNPH for every $\alpha\in[0,1]$ for a single manipulator, even when the number of undetermined pairs in each vote is bounded by a constant. This is achieved by a careful reduction from $\overline{\text{X3C}}$.

We have following intractability result for the \textsc{Strong Manipulation} problem for the Copeland$^\alpha$ rule with one manipulator and at most $10$ undetermined pairs per vote.

\begin{theorem}\label{sm_copeland_hard}
\textsc{Strong Manipulation} is \coNPH for Copeland$^\alpha$ voting rule for every $\alpha\in[0,1]$ even when we have only one manipulator and the number of undetermined pairs in each vote is no more than $10$. 
\end{theorem}

\begin{proof}
 We reduce X3C to \SM for Copeland$^\alpha$ rule. Let $(\mathcal{U} = \{u_1, \ldots, u_m\}, \mathcal{S}=\{S_1,S_2, \dots, S_t\})$ is an X3C instance. We assume, without loss of generality, $t$ to be an even integer (if not, replicate any set from $\mathcal{S}$). We construct a corresponding \WM instance for Copeland$^\alpha$ as follows.
 $$ \text{Candidate set } \mathcal{C} = \mathcal{U} \cup \{c, w, z, d\} $$
 Partial votes $\mathcal{P}$:
 $$ \forall i\le t, (\mathcal{U}\setminus S_i) \succ c\succ z\succ d\succ S_i\succ w \setminus \{ \{z, c\} \times (S_i\cup \{d,w\}) \}$$
 Notice that the number of undetermined pairs in every vote in $\PP$ is $10$. Now we add a set $\QQ$ of complete votes with $|\QQ|$ even and $|\QQ|=poly(m,t)$ using \Cref{thm:mcgarvey} to achieve the following margin of victories in pairwise elections.
 \begin{itemize}
  \item $D_\QQ(d,z) = D_\QQ(z,c) = D_\QQ(c,d) = D_\QQ(w,z) = 4t$
  \item $D_\QQ(u_i,d) = D_\QQ(c,u_i) = 4t ~\forall u_i\in \mathcal{U}$
  \item $D_\QQ(z,u_i) = t ~\forall u_i\in \mathcal{U}$
  \item $D_\QQ(c,w) = t-\frac{2q}{3}-2$
  \item $D_\QQ(u_i, u_{i+1 \pmod* q}) = 4t ~\forall u_i\in \mathcal{U}$
  \item $D_\QQ(a,b)=0$ for every $a,b\in\CC, a\ne b,$ not mentioned above
 \end{itemize}
 We have only one manipulator who tries to make $c$ winner. Notice that the number of votes in the \SM instance $(\PP\cup\QQ,1,c)$ including the manipulator's vote is odd (since $|\PP|$ and $|\QQ|$ are even integers). Therefore, $D_{\PP^*\cup\QQ\cup\{v^*\}}(a,b)$ is never zero for every $a,b\in\CC, a\ne b$ in every extension $\PP^*$ of $\PP$ and manipulators vote $v^*$ and consequently the particular value of $\alpha$ does not play any role in this reduction. Hence, we assume, without loss of generality, $\alpha$ to be zero from here on and simply use the term Copeland instead of Copeland$^\alpha$.
 
 Now we show that the X3C instance $(\mathcal{U},\mathcal{S})$ is a \YES instance if and only if the \SM instance $(\PP\cup\QQ,1,c)$ is a \NO instance (a \SM instance is a \NO instance if there does not exist a vote of the manipulator which makes $c$ the unique winner in every extension of the partial votes). We can assume without loss of generality that manipulator puts $c$ at first position and $z$ at last position in her vote $\vvv$.
  
 Assume that the X3C instance is a \YES instance. Suppose (by renaming) that $S_1, \dots, S_{\frac{m}{3}}$ forms an exact set cover. We claim that the following extension $\overline{\PP}$ of $\PP$ makes both $z$ and $c$ Copeland co-winners.
 
 Extension $\overline{\PP}$ of $\mathcal{P}$:
 $$ i\le \frac{m}{3}, (\mathcal{U}\setminus S_i) \succ c\succ z\succ d\succ S_i\succ w $$
 $$ i\ge \frac{m}{3} + 1, (\mathcal{U}\setminus S_i) \succ d\succ S_i\succ w\succ c\succ z $$
 
 We have summarize the pairwise margins between $z$ and $c$ and the rest of the candidates from the profile  $(\overline{\PP}\cup\QQ\cup\vvv)$ in \Cref{tbl:copeland}. The candidates $z$ and $c$ are the co-winners with Copeland score $(m+1)$.
 
 \begin{table}[htbp!]
  \centering
  {
   \renewcommand*{\arraystretch}{1.5}
  \begin{tabular}{|c|c|c|c|c|}\cline{1-2}\cline{4-5}
   $\CC\setminus\{z\}$ & $D_{\overline{\PP}\cup\QQ\cup\vvv}(z, \cdot)$ && $\CC\setminus\{c\}$ & $D_{\overline{\PP}\cup\QQ\cup\vvv}(c, \cdot)$ \\\cline{1-2}\cline{4-5}
   
   $c$ & $\ge 3t$ && $z, u_i\in\UU$ & $\ge 3t$ \\\cline{1-2}\cline{4-5}
   $w, d$ & $\le -3t$ && $w$ & $-1$ \\\cline{1-2}\cline{4-5}
   $u_i\in\UU$ & $1$ && $d$ & $\le -3t$ \\\cline{1-2}\cline{4-5}
  \end{tabular}}
  \caption{$D_{\overline{\PP}\cup\QQ\cup\vvv}(z, \cdot)$ and $D_{\overline{\PP}\cup\QQ\cup\vvv}(c, \cdot)$}\label{tbl:copeland}
 \end{table}
 
 For the other direction, notice that Copeland score of $c$ is at least $m+1$ since $c$ defeats $d$ and every candidate in $\mathcal{U}$ in every extension of $\PP$. Also notice that the Copeland score of $z$ can be at most $m+1$ since $z$ loses to $w$ and $d$ in every extension of $\PP$. Hence the only way $c$ cannot be the unique winner is that $z$ defeats all candidates in $\mathcal{U}$ and $w$ defeats $c$. 
 
 This requires $w\succ c$ in at least $t-\frac{m}{3}$ extensions of $\mathcal{P}$. We claim that the sets $S_i$ in the remaining of the extensions where $c\succ w$ forms an exact set cover for $(\UU,\SS)$. Indeed, otherwise some candidate $u_i\in \mathcal{U}$ is not covered. Then, notice that $u_i \succ z$ in all $t$ votes, making $D(z,u_i) = -1$.
\end{proof}

\subsection{Opportunistic Manipulation Problem}

All our reductions for the \coNPH{}ness for \OM start from $\overline{\text{X3C}}$. We note that all our hardness results  hold even when there is only one manipulator.  Our overall approach is the following. We engineer a set of partial votes in such a way that the manipulator is forced to vote in a limited number of ways to have any hope of making her favorite candidate win. For each such vote, we demonstrate a viable extension where the vote fails to make the candidate a winner, leading to a \NO{} instance of \OM. These extensions rely on the existence of an exact cover. On the other hand, we show that if there is no exact set cover, then there is no viable extension, thereby leading to an instance that is vacuously a \YES{} instance of \OM. 

\subsubsection{Result for the $k$-Approval Voting Rule}

Our first result on \OM shows that the \OM problem is \coNPH for the $k$-approval voting rule for constant $k\ge 3$ even when the number of manipulators is one and the number of undetermined pairs in each vote is no more than $15$.

\begin{theorem}\label{thm:kappOM}
 The \OM problem is \coNPH for the $k$-approval voting rule for constant $k\ge 3$ even when the number of manipulators is one and the number of undetermined pairs in each vote is no more than $15$.
\end{theorem}

\begin{proof}
 We reduce $\overline{\text{X3C}}$ to \OM for $k$-approval rule. Let $(\mathcal{U} = \{u_1, \ldots, u_m\}, \mathcal{S}=\{S_1,S_2, \dots, S_t\})$ is an $\overline{\text{X3C}}$ instance. We construct a corresponding \OM instance for $k$-approval voting rule as follows. We begin by introducing a candidate for every element of the universe, along with $k-3$ dummy candidates (denoted by $\WW$), and special candidates $\{c,z_1,z_2, d,x,y\}$. Formally, we have:
 $$ \text{Candidate set } \mathcal{C} = \mathcal{U} \cup \{c, z_1, z_2, d, x, y\} \cup \WW.$$ 
Now, for every set $S_i$ in the universe, we define the following total order on the candidate set, which we denote by $\PP^\prime_i$: 
 $$\WW\succ S_i \succ y \succ z_1 \succ z_2 \succ x \succ (\mathcal{U}\setminus S_i) \succ c \succ d$$ 
 Using $\PP^\prime_i$, we define the partial vote $\PP_i$ as follows: $$\PP_i = \PP^\prime_i \setminus (\{ \{y, x, z_1, z_2\} \times S_i \} \cup \{(z_1, z_2), (x, z_1), (x,z_2)\}).$$ 
 
 We denote the set of partial votes $\{\PP_i: i\in[t]\}$ by $\PP$ and $\{\PP^\prime_i: i\in[t]\}$ by $\PP^\prime$. We remark that the number of undetermined pairs in each partial vote $\PP_i$ is $15$. 
 
 We now invoke Lemma 1 from~\cite{journalsDeyMN16}, which allows to achieve any pre-defined scores on the candidates using only polynomially many additional votes. Using this, we add a set $\QQ$ of complete votes with $|\QQ|=\text{poly}(m,t)$ to ensure the following scores, where we denote the $k$-approval score of a candidate from a set of votes $\VV$ by $s_\VV(\cdot)$: $s_\QQ(z_1) = s_\QQ(z_2) =  s_\QQ(y) = s_\QQ(c) - \nfrac{m}{3}; s_\QQ(d), s_\QQ(w) \le s_\QQ(c) - 2t ~\forall w\in\WW; s_\QQ(x) = s_\QQ(c) -1; s_{\PP^\prime\cup\QQ} (u_j) = s_\QQ(c) + 1 ~\forall j\in[m]$.
 
 Our reduced instance is $(\PP\cup\QQ,1,c)$. The reasoning for this score configuration will be apparent as we argue the equivalence. We first argue that if we had a \YES{} instance of $\overline{\text{X3C}}$ (in other words, there is no exact cover), then we have a \YES{} instance of \OM. It turns out that this will follow from the fact that there are no viable extensions, because, as we will show next, a viable extension implies the existence of an exact set cover.
 
 To this end, first observe that the partial votes are constructed in such a way that $c$ gets no additional score from \textit{any} extension. Assuming that the manipulator approves $c$ (without loss of generality), the final score of $c$ in any extension is going to be  $s_\QQ(c) + 1$. Now, in any viable extension, every candidate $u_j$ has to be ``pushed out'' of the top $k$ positions at least once. Observe that whenever this happens, $y$ is forced into the top $k$ positions. Since $y$ is behind the score of $c$ by only $m/3$ votes, $S_i$'s can be pushed out of place in only $m/3$ votes. For every $u_j$ to lose one point, these votes must correspond to an exact cover. Therefore, if there is no exact cover, then there is no viable extension, showing one direction of the reduction. 
 
 On the other hand, suppose we have a \NO{} instance of $\overline{\text{X3C}}$ -- that is, there is an exact cover. We will now use the exact cover to come up with two viable extensions, both of which require the manipulator to vote in different ways to make $c$ win. Therefore, there is no single manipulative vote that accounts for both extensions, leading us to a \NO{} instance of \OM. 

 First, consider this completion of the partial votes: 
 $$ i=1, \WW\succ y \succ x \succ z_1 \succ z_2 \succ S_i \succ (\mathcal{U}\setminus S_i) \succ c \succ d$$
 $$ 2\le i\le \nfrac{m}{3}, \WW\succ y \succ z_1 \succ z_2 \succ x \succ S_i \succ (\mathcal{U}\setminus S_i) \succ c \succ d$$
 $$ \nfrac{m}{3} + 1\le i\le t, \WW\succ S_i \succ y \succ z_1 \succ z_2 \succ x \succ (\mathcal{U}\setminus S_i) \succ c \succ d$$
Notice that in this completion, once accounted for along with the votes in $\QQ$, the score of $c$ is tied with the scores of all $u_j$'s, $z_1, x$ and $y$, while the score of $z_2$ is one less than the score of $c$. Therefore, the only $k$ candidates that the manipulator can afford to approve are $\WW$, the candidates $c,d$ and $z_2$. However, consider the extension that is identical to the above except with the first vote changed to:
 $$ \WW\succ y \succ x \succ z_2 \succ z_1 \succ S_i \succ (\mathcal{U}\setminus S_i) \succ c \succ d$$
Here, on the other hand, the only way for $c$ to be an unique winner is if the manipulator approves $\WW, c,d$ and $z_1$. Therefore, it is clear that there is no way for the manipulator to provide a consolidated vote for both these profiles. Therefore, we have a \NO{} instance of \OM.
\end{proof}

\subsubsection{Result for the $k$-Veto Voting Rule}

We next move on to the $k$-veto voting rule and show that the \OM problem for the $k$-veto is \coNPH for every constant $k\ge 4$ even when the number of manipulators is one and the number of undetermined pairs in each vote is no more than $9$.

\begin{theorem}\label{thm:kvetoOM}
 The \OM problem is \coNPH for the $k$-veto voting rule for every constant $k\ge 4$ even when the number of manipulators is one and the number of undetermined pairs in each vote is no more than $9$.
\end{theorem}

\begin{proof}
 We reduce X3C to \OM for $k$-veto rule. Let $(\mathcal{U} = \{u_1, \ldots, u_m\}, \mathcal{S}=\{S_1,S_2, \dots, S_t\})$ is an X3C instance. We construct a corresponding \OM instance for $k$-veto voting rule as follows.
 $$ \text{Candidate set } \mathcal{C} = \mathcal{U} \cup \{c, z_1, z_2, d, x, y\} \cup \AA \cup \WW, \text{ where } \AA = \{a_1, a_2, a_3\}, |\WW|=k-4 $$
 For every $i\in[t]$, we define $\mathcal{P}^\prime_i$ as follows:
 $$ \forall i\le t, c \succ \AA \succ (\mathcal{U}\setminus S_i) \succ d \succ S_i \succ y \succ x \succ z_1 \succ z_2\succ \WW$$
 Using $\PP^\prime_i$, we define partial vote $\PP_i = \PP^\prime_i \setminus (\{ \{y, x, z_1, z_2\} \times S_i \} \cup \{(z_1, z_2), (x, z_1), (x,z_2)\})$ for every $i\in[t]$. We denote the set of partial votes $\{\PP_i: i\in[t]\}$ by $\PP$ and $\{\PP^\prime_i: i\in[t]\}$ by $\PP^\prime$. We note that the number of undetermined pairs in each partial vote $\PP_i$ is $9$. Using \Cref{score_gen}, we add a set $\QQ$ of complete votes with $|\QQ|=\text{poly}(m,t)$ to ensure the following. We denote the $k$-veto score of a candidate from a set of votes $\WW$ by $s_\WW(\cdot)$.
 \begin{itemize}
  \item $s_{\PP^\prime\cup\QQ} (z_1) = s_{\PP^\prime\cup\QQ} (z_2) = s_{\PP^\prime\cup\QQ} (c) - \nfrac{m}{3}$
  \item $s_{\PP^\prime\cup\QQ} (a_i) = s_{\PP^\prime\cup\QQ} (u_j) = s_{\PP^\prime\cup\QQ} (w) = s_{\PP^\prime\cup\QQ} (c) ~\forall a_i\in \AA, u_j\in\UU, w\in\WW$
  \item $s_{\PP^\prime\cup\QQ} (y) = s_{\PP^\prime\cup\QQ} (c) - \nfrac{m}{3} - 1$
  \item $s_{\PP^\prime\cup\QQ} (x) = s_{\PP^\prime\cup\QQ} (c) - 2$
 \end{itemize}
 
 We have only one manipulator who tries to make $c$ winner. Now we show that the X3C instance $(\mathcal{U},\mathcal{S})$ is a \YES instance if and only if the \OM instance $(\PP\cup\QQ,1,c)$ is a \NO instance.
 
 In the forward direction, let us now assume that the X3C instance is a \YES instance. Suppose (by renaming) that $S_1, \dots, S_{\nfrac{m}{3}}$ forms an exact set cover. Let us assume that the manipulator's vote $\vvv$ disapproves every candidate in $\WW\cup\AA$ since otherwise $c$ can never win uniquely. We now show that if $\vvv$ does not disapprove $z_1$ then, $\vvv$ is not a $c$-optimal vote. Suppose $\vvv$ does not disapprove $z_1$. Then we consider the following extension $\overline{\PP}$ of $\PP$.
 
 $$ i=1, c \succ \AA \succ (\mathcal{U}\setminus S_i) \succ d \succ y \succ z_1 \succ x \succ z_2 \succ S_i\succ \WW$$
 $$ 2\le i\le \nfrac{m}{3}, c \succ \AA \succ (\mathcal{U}\setminus S_i) \succ d \succ y \succ z_1 \succ z_2 \succ x \succ S_i\succ \WW$$
 $$ \nfrac{m}{3} + 1\le i\le t, c \succ \AA \succ (\mathcal{U}\setminus S_i) \succ d \succ S_i \succ y \succ x \succ z_1 \succ z_2\succ \WW$$
 
 We have the following scores $s_{\overline{\PP}\cup\QQ}(c) = s_{\overline{\PP}\cup\QQ}(z_1) = s_{\overline{\PP}\cup\QQ}(z_2) + 1 = s_{\overline{\PP}\cup\QQ}(x) + 1 = s_{\overline{\PP}\cup\QQ}(u_j) + 1 ~\forall u_j\in\UU$. Hence, both $c$ and $z_1$ win for the votes $\overline{\PP}\cup\QQ\cup\{\vvv\}$. However, the vote $\vvv^\prime$ which disapproves $a_1, a_2, a_3, z_1$ makes $c$ a unique winner for the votes $\overline{\PP}\cup\QQ\cup\{\vvv^\prime\}$. Hence, $\vvv$ is not a $c$-optimal vote. Similarly, we can show that if the manipulator's vote does not disapprove $z_2$ then, the vote is not $c$-optimal. Hence, there does not exist any $c$-optimal vote and the \OM instance is a \NO instance.
 
 In the reverse direction, we show that if the X3C instance is a \NO instance, then there does not exist a vote $\vvv$ of the manipulator and an extension $\overline{\PP}$ of $\PP$ such that $c$ is the unique winner for the votes $\overline{\PP}\cup\QQ\cup\{\vvv^\prime\}$ thereby proving that the \OM instance is vacuously \YES (and thus every vote is $c$-optimal). Notice that, there must be at least $\nfrac{m}{3}$ votes $\PP_1$ in $\overline{\PP}$ where the corresponding $S_i$ gets pushed in bottom $k$ positions since $s_{\PP^\prime\cup\QQ} (u_j) = s_{\PP^\prime\cup\QQ} (c) ~\forall a_i\in \AA, u_j\in\UU$. However, in each vote in $\PP_1$, $y$ is placed within top $m-k$ many position and thus we have $|\PP_1|$ is exactly $\nfrac{m}{3}$ since $s_{\PP^\prime\cup\QQ} (y) = s_{\PP^\prime\cup\QQ} (c) - \nfrac{m}{3} - 1$. Now notice that there must be at least one candidate $u\in\UU$ which is not covered by the sets $S_i$s corresponding to the votes $\PP_1$ because the X3C instance is a \NO instance. Hence, $c$ cannot win the election uniquely irrespective of the manipulator's vote. Thus every vote is $c$-optimal and the \OM instance is a \YES instance.
\end{proof}

\subsubsection{Result for the Borda Voting Rule}

We show next similar intractability result for the Borda voting rule too with only at most $7$ undetermined pairs per vote.

\begin{theorem}\label{thm:bordaOM}
 The \OM problem is \coNPH for the Borda voting rule even when the number of manipulators is one and the number of undetermined pairs in every vote is no more than $7$.
\end{theorem}

\begin{proof}
 We reduce X3C to \OM for the Borda rule. Let $(\mathcal{U} = \{u_1, \ldots, u_m\}, \mathcal{S}=\{S_1,S_2, \dots, S_t\})$ is an X3C instance. Without loss of generality we assume that $m$ is not divisible by $6$ (if not, then we add three new elements $b_1, b_2, b_3$ to $\UU$ and a set $\{b_1, b_2, b_3\}$ to $\SS$). We construct a corresponding \OM instance for the Borda voting rule as follows.
 $$ \text{Candidate set } \mathcal{C} = \mathcal{U} \cup \{c, z_1, z_2, d, y\} $$
 For every $i\in[t]$, we define $\mathcal{P}^\prime_i$ as follows:
 
 $$ \forall i\le t, y \succ S_i \succ z_1 \succ z_2 \succ (\mathcal{U}\setminus S_i) \succ d \succ c$$
 
 Using $\PP^\prime_i$, we define partial vote $\PP_i = \PP^\prime_i \setminus (\{ (\{y\} \cup S_i) \times \{z_1, z_2\} \} \cup \{(z_1, z_2)\})$ for every $i\in[t]$. We denote the set of partial votes $\{\PP_i: i\in[t]\}$ by $\PP$ and $\{\PP^\prime_i: i\in[t]\}$ by $\PP^\prime$. We note that the number of undetermined pairs in each partial vote $\PP_i$ is $7$. Using \Cref{score_gen}, we add a set $\QQ$ of complete votes with $|\QQ|=\text{poly}(m,t)$ to ensure the following. We denote the Borda score of a candidate from a set of votes $\WW$ by $s_\WW(\cdot)$.
 
 \begin{itemize}
  \item $s_{\PP^\prime\cup\QQ} (y) = s_{\PP^\prime\cup\QQ} (c) + m + \nfrac{m}{3} + 3$
  \item $s_{\PP^\prime\cup\QQ} (z_1) = s_{\PP^\prime\cup\QQ} (c) - 3\lfloor\nfrac{m}{6}\rfloor - 2$
  \item $s_{\PP^\prime\cup\QQ} (z_2) = s_{\PP^\prime\cup\QQ} (c) - 5\lfloor\nfrac{m}{6}\rfloor - 3$
  \item $s_{\PP^\prime\cup\QQ} (u_i) = s_{\PP^\prime\cup\QQ} (c) + m + 5 - i ~\forall i\in[m]$
  \item $s_{\PP^\prime\cup\QQ} (d) \le s_{\PP^\prime\cup\QQ} (c) - 5m$
 \end{itemize}
 
 We have only one manipulator who tries to make $c$ winner. Now we show that the X3C instance $(\mathcal{U},\mathcal{S})$ is a \YES instance if and only if the \OM instance $(\PP\cup\QQ,1,c)$ is a \NO instance. Notice that we can assume without loss of generality that the manipulator places $c$ at the first position, $d$ at the second position, the candidate $u_i$ at $(m+5-i)^{th}$ position for every $i\in[m]$, and $y$ at the last position, since otherwise $c$ can never win uniquely irrespective of the extension of $\PP$ (that it, the manipulator's vote looks like $c\succ d \succ \{z_1, z_2\} \succ u_m \succ u_{m-1} \succ \cdots \succ u_1 \succ y$). 
 
 In the forward direction, let us now assume that the X3C instance is a \YES instance. Suppose (by renaming) that $S_1, \dots, S_{\nfrac{m}{3}}$ forms an exact set cover. Let the manipulator's vote $\vvv$ be $c \succ d \succ z_1 \succ z_2 \succ u_m \succ \cdots \succ u_1 \succ y$. We now argue that $\vvv$ is not a $c$-optimal vote. The other case where the manipulator's vote $\vvv^\prime$ be $c \succ d \succ z_2 \succ z_1 \succ u_m \succ \cdots \succ u_1 \succ y$ can be argued similarly. We consider the following extension $\overline{\PP}$ of $\PP$.
 
 $$ 1\le i\le \lfloor\nfrac{m}{6}\rfloor, z_2 \succ y \succ S_i \succ z_1 \succ (\mathcal{U}\setminus S_i) \succ d \succ c $$
 $$ \lceil\nfrac{m}{6}\rceil\le i\le \nfrac{m}{3}, z_1 \succ y \succ S_i \succ z_2 \succ (\mathcal{U}\setminus S_i) \succ d \succ c $$
 $$ \nfrac{m}{3} + 1\le i\le t, y \succ S_i \succ z_1 \succ z_2 \succ (\mathcal{U}\setminus S_i) \succ d \succ c$$
 
 We have the following Borda scores $s_{\overline{\PP}\cup\QQ\cup\{\vvv\}}(c) = s_{\overline{\PP}\cup\QQ\cup\{\vvv\}}(y) + 1 = s_{\overline{\PP}\cup\QQ\cup\{\vvv\}}(z_2) + 6 = s_{\overline{\PP}\cup\QQ\cup\{\vvv\}}(z_1) = s_{\overline{\PP}\cup\QQ\cup\{\vvv\}}(u_i) + 1 ~\forall i\in[m]$. Hence, $c$ does not win uniquely for the votes $\overline{\PP}\cup\QQ\cup\{\vvv\}$. However, $c$ is the unique winner for the votes $\overline{\PP}\cup\QQ\cup\{\vvv^\prime\}$. Hence, there does not exist any $c$-optimal vote and the \OM instance is a \NO instance.

 In the reverse direction, we show that if the X3C instance is a \NO instance, then there does not exist a vote $\vvv$ of the manipulator and an extension $\overline{\PP}$ of $\PP$ such that $c$ is the unique winner for the votes $\overline{\PP}\cup\QQ\cup\{\vvv^\prime\}$ thereby proving that the \OM instance is vacuously \YES (and thus every vote is $c$-optimal). Notice that the score of $y$ must decrease by at least $\nfrac{m}{3}$ for $c$ to win uniquely. However, in every vote $v$ where the score of $y$ decreases by at least one in any extension $\overline{\PP}$ of $\PP$, at least one of $z_1$ or $z_2$ must be placed at top position of the vote $v$. However, the candidates $z_1$ and $z_2$ can be placed at top positions of the votes in $\overline{\PP}$ at most $\nfrac{m}{3}$ many times while ensuring $c$ does not lose the election. Also, even after manipulator places the candidate $u_i$ at $(m+5-i)^{th}$ position for every $i\in[m]$, for $c$ to win uniquely, the score of every $u_i$ must decrease by at least one. Hence, altogether, there will be exactly $\nfrac{m}{3}$ votes (denoted by the set $\PP_1$) in any extension of $\PP$ where $y$ is placed at the second position. However, since the X3C instance is a \NO instance, the $S_i$s corresponding to the votes in $\PP_1$ does not form a set cover. Let $u\in\UU$ be an element not covered by the $S_i$s corresponding to the votes in $\PP_1$. Notice that the score of $u$ does not decrease in the extension $\overline{\PP}$ and thus $c$ does not win uniquely irrespective of the manipulator's vote. Thus every vote is $c$-optimal and thus the \OM instance is a \YES instance. Thus every vote is $c$-optimal and the \OM instance is a \YES instance.
\end{proof}

\subsubsection{Result for the Maximin Voting Rule}

For the maximin voting rule, we show intractability of \OM with one manipulator even when the number of undetermined pairs in every vote is at most $8$.

\begin{theorem}\label{thm:maximinOM}
 The \OM problem is \coNPH for the maximin voting rule even when the number of manipulators is one and the number of undetermined pairs in every vote is no more than $8$.
\end{theorem}

\begin{proof}
 We reduce X3C to \OM for the maximin rule. Let $(\mathcal{U} = \{u_1, \ldots, u_m\}, \mathcal{S}=\{S_1,S_2, \dots, S_t\})$ is an X3C instance. We construct a corresponding \OM instance for the maximin voting rule as follows.
 $$ \text{Candidate set } \mathcal{C} = \mathcal{U} \cup \{c, z_1, z_2, z_3, d, y\} $$
 For every $i\in[t]$, we define $\mathcal{P}^\prime_i$ as follows:
 
 $$ \forall i\le t, S_i \succ x \succ d \succ y \succ (\mathcal{U}\setminus S_i) \succ z_1 \succ z_2 \succ z_3$$
 
 Using $\PP^\prime_i$, we define partial vote $\PP_i = \PP^\prime_i \setminus (\{ (\{x\} \cup S_i) \times \{d, y\} \})$ for every $i\in[t]$. We denote the set of partial votes $\{\PP_i: i\in[t]\}$ by $\PP$ and $\{\PP^\prime_i: i\in[t]\}$ by $\PP^\prime$. We note that the number of undetermined pairs in each partial vote $\PP_i$ is $8$. We define another partial vote $\ppp$ as follows.
 
 $$ \ppp = (z_1 \succ z_2 \succ z_3 \succ \text{ others }) \setminus \{(z_1, z_2), (z_2, z_3), (z_1, z_3)\} $$
 
 Using \Cref{thm:mcgarvey}, we add a set $\QQ$ of complete votes with $|\QQ|=\text{poly}(m,t)$ to ensure the following pairwise margins (notice that the pairwise margins among $z_1, z_2,$ and $z_3$ does not include the partial vote $\ppp$).
 
 \begin{itemize}
  \item $D_{\PP^\prime\cup\QQ\cup\{\ppp\}} (d, c) = 4t + 1$
  \item $D_{\PP^\prime\cup\QQ\cup\{\ppp\}} (x, d) = 4t + \nfrac{2m}{3} + 1$
  \item $D_{\PP^\prime\cup\QQ\cup\{\ppp\}} (y, x) = 4t - \nfrac{2m}{3} + 1$
  \item $D_{\PP^\prime\cup\QQ\cup\{\ppp\}} (d, u_j) = 4t - 1 ~\forall u_j\in\UU$
  \item $D_{\PP^\prime\cup\QQ} (z_1, z_2) = D_{\PP^\prime\cup\QQ} (z_2, z_3) = D_{\PP^\prime\cup\QQ} (z_3, z_1) = 4t + 2$
  \item $|D_{\PP^\prime\cup\QQ\cup\{\ppp\}} (a, b)| \le 1$ for every $a, b\in \CC$ not defined above.
 \end{itemize}
 
 We have only one manipulator who tries to make $c$ winner. Now we show that the X3C instance $(\mathcal{U},\mathcal{S})$ is a \YES instance if and only if the \OM instance $(\PP\cup\QQ\cup\{\ppp\},1,c)$ is a \NO instance. Notice that we can assume without loss of generality that the manipulator's vote prefers $c$ to every other candidate, $y$ to $x$, $x$ to $d$, and $d$ to $u_j$ for every $u_j\in\UU$.
 
 In the forward direction, let us now assume that the X3C instance is a \YES instance. Suppose (by renaming) that $S_1, \dots, S_{\nfrac{m}{3}}$ forms an exact set cover. Notice that the manipulator's vote  must prefer either $z_2$ to $z_1$ or $z_1$ to $z_3$ or $z_3$ to $z_2$. We show that if the manipulator's vote $\vvv$ prefers $z_2$ to $z_1$, then $\vvv$ is not a $c$-optimal vote. The other two cases are symmetrical. Consider the following extension $\overline{\PP}$ of $\PP$ and $\overline{\ppp}$ of $\ppp$.
 
 $$ 1\le i\le \nfrac{m}{3}, d \succ y \succ S_i \succ x \succ (\mathcal{U}\setminus S_i) \succ z_1 \succ z_2 \succ z_3 $$
 $$ \nfrac{m}{3} + 1 \le i\le t, S_i \succ x \succ d \succ y \succ (\mathcal{U}\setminus S_i) \succ z_1 \succ z_2 \succ z_3 $$
 $$ \overline{\ppp} = z_2 \succ z_3 \succ z_1 \succ \text{ others } $$
 
 From the votes in $\overline{\PP}\cup\QQ\cup\{\vvv, \overline{\ppp}\}$, the maximin score of $c$ is $-4t$, of $d, x, u_j ~\forall u_j\in\UU$ are $-4t-2$, of $z_1, z_3$ are at most than $-4t-2$, and of $z_2$ is $-4t$. Hence, $c$ is not the unique maximn winner. However, the manipulator's vote $c \succ z_1 \succ z_2 \succ z_3 \succ \text{ other }$ makes $c$ the unique maximin winner. Hence, $\vvv$ is not a $c$-optimal vote.
 
 For the reverse direction, we show that if the X3C instance is a \NO instance, then there does not exist a vote $\vvv$ of the manipulator and an extension $\overline{\PP}$ of $\PP$ such that $c$ is the unique winner for the votes $\overline{\PP}\cup\QQ\cup\{\vvv^\prime\}$ thereby proving that the \OM instance is vacuously \YES (and thus every vote is $c$-optimal). Consider any extension $\overline{\PP}$ of $\PP$. Notice that, for $c$ to win uniquely, $y \succ x$ must be at least $\nfrac{m}{3}$ of the votes in $\overline{\PP}$; call these set of votes $\PP_1$. However, $d\succ x$ in every vote in $\PP_1$ and $d\succ x$ can be in at most $\nfrac{m}{3}$ votes in $\overline{\PP}$ for $c$ to win uniquely. Hence, we have $|\PP_1| = \nfrac{m}{3}$. Also for $c$ to win, each $d\succ u_j$ must be at least one vote of $\overline{\PP}$ and $d\succ u_j$ is possible only in the votes in $\PP_1$. However, the sets $S_i$s corresponding to the votes in $\PP_1$ does not form a set cover since the X3C instance is a \NO instance. Hence, there must exist a $u_j\in \UU$ for which $u_j\succ d$ in every vote in $\overline{\PP}$ and thus $c$ cannot win uniquely irrespective of the vote of the manipulator. Thus every vote is $c$-optimal and the \OM instance is a \YES instance.
\end{proof}

\subsubsection{Result for the Copeland$^\alpha$ Voting Rule}

Our next result proves that the \OM problem is \coNPH for the Copeland$^\alpha$ voting rule too for every $\alpha\in[0,1]$ even with one manipulator and at most $8$ undetermined pairs per vote.

\begin{theorem}\label{thm:copelandOM}
 The \OM problem is \coNPH for the Copeland$^\alpha$ voting rule for every $\alpha\in[0,1]$ even when the number of manipulators is one and the number of undetermined pairs in each vote is no more than $8$.
\end{theorem}

\begin{proof}
 We reduce X3C to \OM for the Copeland$^\alpha$ voting rule. Let $(\mathcal{U} = \{u_1, \ldots, u_m\}, \mathcal{S}=\{S_1,S_2, \dots, S_t\})$ is an X3C instance. We construct a corresponding \OM instance for the Copeland$^\alpha$ voting rule as follows.
 $$ \text{Candidate set } \mathcal{C} = \mathcal{U} \cup \{c, z_1, z_2, z_3, d_1, d_2, d_3, x, y\} $$
 For every $i\in[t]$, we define $\mathcal{P}^\prime_i$ as follows:
 
 $$ \forall i\le t, S_i \succ x \succ y \succ c \succ \text{ others}$$
 
 Using $\PP^\prime_i$, we define partial vote $\PP_i = \PP^\prime_i \setminus (\{ (\{x\} \cup S_i) \times \{c, y\} \})$ for every $i\in[t]$. We denote the set of partial votes $\{\PP_i: i\in[t]\}$ by $\PP$ and $\{\PP^\prime_i: i\in[t]\}$ by $\PP^\prime$. We note that the number of undetermined pairs in each partial vote $\PP_i$ is $8$. We define another partial vote $\ppp$ as follows.
 
 $$ \ppp = (z_1 \succ z_2 \succ z_3 \succ \text{ others }) \setminus \{(z_1, z_2), (z_2, z_3), (z_1, z_3)\} $$
 
 Using \Cref{thm:mcgarvey}, we add a set $\QQ$ of complete votes with $|\QQ|=\text{poly}(m,t)$ to ensure the following pairwise margins (notice that the pairwise margins among $z_1, z_2,$ and $z_3$ does not include the partial vote $\ppp$).
 
 \begin{itemize}
  \item $D_{\PP^\prime\cup\QQ\cup\{\ppp\}} (u_j, c) = 2 ~\forall u_j\in\UU$
  \item $D_{\PP^\prime\cup\QQ\cup\{\ppp\}} (x, y) = \nfrac{2m}{3}$
  \item $D_{\PP^\prime\cup\QQ\cup\{\ppp\}} (c, y) = D_{\PP^\prime\cup\QQ\cup\{\ppp\}} (x, c) = D_{\PP^\prime\cup\QQ\cup\{\ppp\}} (d_i, c) = D_{\PP^\prime\cup\QQ\cup\{\ppp\}} (z_k,c) = D_{\PP^\prime\cup\QQ\cup\{\ppp\}} (u_j, x) = D_{\PP^\prime\cup\QQ\cup\{\ppp\}} (x, z_k) = D_{\PP^\prime\cup\QQ\cup\{\ppp\}} (d_i, x) = D_{\PP^\prime\cup\QQ\cup\{\ppp\}} (y, u_j) = D_{\PP^\prime\cup\QQ\cup\{\ppp\}} (d_i, y) = D_{\PP^\prime\cup\QQ\cup\{\ppp\}} (y, z_k) = D_{\PP^\prime\cup\QQ\cup\{\ppp\}} (z_k, u_j) = D_{\PP^\prime\cup\QQ\cup\{\ppp\}} (u_j, d_i) = D_{\PP^\prime\cup\QQ\cup\{\ppp\}} (z_k, d_1) = D_{\PP^\prime\cup\QQ\cup\{\ppp\}} (z_k, d_2) = D_{\PP^\prime\cup\QQ\cup\{\ppp\}} (d_3, z_k) =  4t ~\forall i,k\in[3], j\in[m]$
  
  \item $D_{\PP^\prime\cup\QQ\cup\{\ppp\}} (u_j, u_\el) = -4t$ for at least $\nfrac{m}{3}$ many $u_\el\in\UU$  
  \item $D_{\PP^\prime\cup\QQ} (z_1, z_2) = D_{\PP^\prime\cup\QQ} (z_2, z_3) = D_{\PP^\prime\cup\QQ} (z_3, z_1) = 1$
  \item $|D_{\PP^\prime\cup\QQ\cup\{\ppp\}} (a, b)| \le 1$ for every $a, b\in \CC$ not defined above.
 \end{itemize}
 
 We have only one manipulator who tries to make $c$ winner. Now we show that the X3C instance $(\mathcal{U},\mathcal{S})$ is a \YES instance if and only if the \OM instance $(\PP\cup\QQ\cup\{\ppp\},1,c)$ is a \NO instance. Since the number of voters is odd, $\alpha$ does not play any role in the reduction and thus from here on we simply omit $\alpha$. Notice that we can assume without loss of generality that the manipulator's vote prefers $c$ to every other candidate and $x$ to $y$.
 
 In the forward direction, let us now assume that the X3C instance is a \YES instance. Suppose (by renaming) that $S_1, \dots, S_{\nfrac{m}{3}}$ forms an exact set cover. Suppose the manipulator's vote $\vvv$ order $z_1, z_2,$ and $z_3$ as $z_1 \succ z_2\succ z_3$. We will show that $\vvv$ is not a $c$-optimal vote. Symmetrically, we can show that the manipulator's vote ordering $z_1, z_2,$ and $z_3$ in any other order is not $c$-optimal. Consider the following extension $\overline{\PP}$ of $\PP$ and $\overline{\ppp}$ of $\ppp$.
 
 $$ 1\le i\le \nfrac{m}{3}, y \succ c \succ S_i \succ x \succ \text{others} $$
 $$ \nfrac{m}{3} + 1 \le i\le t, S_i \succ x \succ y \succ c \succ \text{others} $$
 $$ \overline{\ppp} = z_1 \succ z_2 \succ z_3 \succ \text{others } $$
 
 From the votes in $\overline{\PP}\cup\QQ\cup\{\vvv, \overline{\ppp}\}$, the Copeland score of $c$ is $m+4$ (defeating $y, z_k, u_j ~\forall k\in[3], j\in[m]$), of $y$ is $m+3$ (defeating $z_k, u_j ~\forall k\in[3], j\in[m]$), of $u_j$ is at most $\nfrac{2m}{3} + 4$ (defeating $x, d_i ~\forall i\in[3]$ and at most $\nfrac{2m}{3}$ many $u_\el\in\UU$), of $x$ is $5$ (defeating $c, y, z_k ~\forall l\in[3]$), of $d_1, d_2$ is $2$ (defeating $y$ and $c$), of $d_3$ is $5$ (defeating $y, c, z_k~\forall k\in[3]$). of $z_3$ is $m+3$ (defeating $d_i, u_j \forall i\in[3], j\in[m]$) for every $k\in[3]$, of $z_3$ is $m+2$ (defeating $d_1, d_2, u_j i\in[3], j\in[m]$), $z_2$ is $m+3$ (defeating $d_1, d_2, z_3, u_j i\in[3], j\in[m]$), $z_1$ is $m+4$ (defeating $d_1, d_2, z_2, z_3, u_j i\in[3], j\in[m]$). Hence, $c$ co-wins with $z_1$ with Copeland score $m+4$. However, the manipulator's vote $c\succ z_3\succ z_2 \succ z_1$ makes $c$ win uniquely. Hence, $\vvv$ is not a $c$-optimal vote and thus the \OM instance is a \NO instance.
 
 For the reverse direction, we show that if the X3C instance is a \NO instance, then there does not exist a vote $\vvv$ of the manipulator and an extension $\overline{\PP}$ of $\PP$ such that $c$ is the unique winner for the votes $\overline{\PP}\cup\QQ\cup\{\vvv^\prime\}$ thereby proving that the \OM instance is vacuously \YES (and thus every vote is $c$-optimal). Consider any extension $\overline{\PP}$ of $\PP$. Notice that, for $c$ to win uniquely, $c$ must defeat each $u_j\in\UU$ and thus $c$ is preferred over $u_j$ in at least one vote in $\overline{\PP}$; we call these votes $\PP_1$. However, in every vote in $\PP_1$, $y$ is preferred over $x$ and thus $|\PP_1|\le \nfrac{m}{3}$ because $x$ must defeat $y$ for $c$ to win uniquely. Since the X3C instance is a \NO instance, there must be a candidate $u\in\UU$ which is not covered by the sets corresponding to the votes in $\PP_1$ and thus $u$ is preferred over $c$ in every vote in $\PP$. Hence, $c$ cannot win uniquely irrespective of the vote of the manipulator. Thus every vote is $c$-optimal and the \OM instance is a \YES instance.
\end{proof}

\subsubsection{Result for the Bucklin Voting Rule}

For the Bucklin and simplified Bucklin voting rules, we show intractability of the \OM problem with at most $15$ undetermined pairs per vote and only one manipulator.

\begin{theorem}\label{thm:bucklinOM}
 The \OM problem is \coNPH for the Bucklin and simplified Bucklin voting rules even when the number of manipulators is one and the number of undetermined pairs in each vote is no more than $15$.
\end{theorem}

\begin{proof}
 We reduce X3C to \OM for the Bucklin and simplified Bucklin voting rules. Let $(\mathcal{U} = \{u_1, \ldots, u_m\}, \mathcal{S}=\{S_1,S_2, \dots, S_t\})$ is an X3C instance. We assume without loss of generality that $m$ is not divisible by $6$ (if not, we introduce three elements in $\UU$ and a set containing them in $\SS$) and $t$ is an even integer (if not, we duplicate any set in $\SS$). We construct a corresponding \OM instance for the Bucklin and simplified Bucklin voting rules as follows.
 $$ \text{Candidate set } \mathcal{C} = \mathcal{U} \cup \{c, z_1, z_2, x_1, x_2, d\} \cup W, \text{ where } |W|=m-3 $$
 For every $i\in[t]$, we define $\mathcal{P}^\prime_i$ as follows:
 
 $$ \forall i\le t, (\UU\setminus S_i) \succ S_i \succ d \succ x_1 \succ x_2 \succ z_1 \succ z_2 \succ \text{ others}$$
 
 Using $\PP^\prime_i$, we define partial vote $\PP_i = \PP^\prime_i \setminus (\{ (\{d\} \cup S_i) \times \{x_1, x_2, z_1, z_2\} \}\cup \{(z_1, z_2)\})$ for every $i\in[t]$. We denote the set of partial votes $\{\PP_i: i\in[t]\}$ by $\PP$ and $\{\PP^\prime_i: i\in[t]\}$ by $\PP^\prime$. We note that the number of undetermined pairs in each partial vote $\PP_i$ is $15$. We introduce the following additional complete votes $\QQ$:
 \begin{itemize}
  \item $\nfrac{t}{2}-\lfloor\nfrac{m}{6}\rfloor-1$ copies of $W\succ z_1\succ z_2\succ x_1\succ c\succ \text{ others}$
  \item $\nfrac{t}{2}-\lfloor\nfrac{m}{6}\rfloor-1$ copies of $W\succ z_1\succ z_2\succ x_2\succ c\succ \text{ others}$
  \item $2\lceil\nfrac{m}{6}\rceil$ copies of $W\succ z_1\succ z_2\succ d\succ c\succ \text{ others}$
  \item $\lfloor\nfrac{m}{6}\rfloor$ copies of $W\succ z_1\succ d\succ x_1\succ c\succ \text{ others}$
  \item $\lfloor\nfrac{m}{6}\rfloor$ copies of $W\succ z_1\succ d\succ x_2\succ c\succ \text{ others}$
  \item $2\lceil\nfrac{m}{6}\rceil-1$ copies of $\UU\succ x_1\succ \text{ others}$
  \item One $\UU\succ c\succ \text{ others}$
 \end{itemize}
 
 We have only one manipulator who tries to make $c$ winner. Now we show that the X3C instance $(\mathcal{U},\mathcal{S})$ is a \YES instance if and only if the \OM instance $(\PP\cup\QQ,1,c)$ is a \NO instance. The total number of voters in the \OM instance is $2t+\nfrac{2m}{3}+1$. We notice that within top $m+1$ positions of the votes in $\PP^\prime\cup\QQ$, $c$ appears $t+\nfrac{m}{3}$ times, $z_1$ and $z_2$ appear $t+\lfloor\nfrac{m}{6}\rfloor$ times, $x_1$ appears $\nfrac{t}{2}+\nfrac{m}{3}-1$ times, $x_2$ appears $\nfrac{t}{2}-1$ times, every candidate in $W$ appears $t+\nfrac{m}{3}-1$ times, every candidate in $\UU$ appears $t+\nfrac{m}{3}+1$ times. Also every candidate in $\UU$ appears $t+\nfrac{m}{3}+1$ times within top $m$ positions of the votes in $\PP\cup\QQ$. Hence, for both Bucklin and simplified Bucklin voting rules, we can assume without loss of generality that the manipulator puts $c$, every candidate in $W$, $x_1$, $x_2$, and exactly one of $z_1$ and $z_2$.
 
 In the forward direction, let us now assume that the X3C instance is a \YES instance. Suppose (by renaming) that $S_1, \dots, S_{\nfrac{m}{3}}$ forms an exact set cover. Suppose the manipulator's vote $\vvv$ puts $c$, every candidate in $W$, $x_1$, $x_2$, and $z_1$ within top $m+1$ positions. We will show that $\vvv$ is not $c$-optimal. The other case where the manipulator's vote $\vvv^\prime$ puts $c$, every candidate in $W$, $x_1$, $x_2$, and $z_2$ within top $m+1$ positions is symmetrical. Consider the following extension $\overline{\PP}$ of $\PP$:
 $$ 1\le i\le \lfloor\nfrac{m}{6}\rfloor, (\UU\setminus S_i) d \succ x_1 \succ x_2 \succ z_2 \succ S_i \succ \succ z_1 \succ \text{ others} $$
 $$ \lceil\nfrac{m}{6}\rceil \le i\le \nfrac{m}{3}, (\UU\setminus S_i) d \succ x_1 \succ x_2 \succ z_1 \succ S_i \succ \succ z_2 \succ \text{ others} $$
 $$ \nfrac{m}{3}+1 \le i\le t, (\UU\setminus S_i) \succ S_i \succ d \succ x_1 \succ x_2 \succ z_1 \succ z_2 \succ \text{ others}$$
 For both Bucklin and simplified Bucklin voting rules, $c$ co-wins with $z_1$ for the votes in $\PP\cup\QQ\cup\{\vvv\}$. However, $c$ wins uniquely for the votes in $\PP\cup\QQ\cup\{\vvv^\prime\}$. Hence, $\vvv$ is not a $c$-optimal vote and thus the \OM instance is a \NO instance. 
 
 For the reverse direction, we show that if the X3C instance is a \NO instance, then there does not exist a vote $\vvv$ of the manipulator and an extension $\overline{\PP}$ of $\PP$ such that $c$ is the unique winner for the votes $\overline{\PP}\cup\QQ\cup\{\vvv^\prime\}$ thereby proving that the \OM instance is vacuously \YES (and thus every vote is $c$-optimal). Consider any extension $\overline{\PP}$ of $\PP$. Notice that, for $c$ to win uniquely, every candidate must be pushed out of top $m+1$ positions in at least one vote in $\PP$; we call these set of votes $\PP_1$. Notice that, $|\PP_1|\ge \nfrac{m}{3}$. However, in every vote in $\PP_1$, at least one of $z_1$ and $z_2$ appears within top $m+1$ many positions. Since, the manipulator has to put at least one of $z_1$ and $z_2$ within its top $m+1$ positions and $z_1$ and $z_2$ appear $t+\lfloor\nfrac{m}{6}\rfloor$ times in the votes in $\PP^\prime\cup\QQ$, we must have $|\PP_1|\le \nfrac{m}{3}$ and thus $|\PP_1|=\nfrac{m}{3}$, for $c$ to win uniquely. However, there exists a candidate $u\in\UU$ not covered by the $S_i$s corresponding to the votes in $\PP_1$. Notice that $u$ gets majority within top $m$ positions of the votes and $c$ can never get majority within top $m+1$ positions of the votes. Hence, $c$ cannot win uniquely irrespective of the vote of the manipulator. Thus every vote is $c$-optimal and the \OM instance is a \YES instance.
\end{proof}

\section{Polynomial Time Algorithms for \WM, \SM, and \OM Problems}\label{sec:poly}

We now turn to the polynomial time cases depicted in~\Cref{table:partial_summary}. This section is organized in three parts, one for each problem considered. 

\subsection{Weak Manipulation Problem}
Since the \PW problem is in \Pb{} for the plurality and the veto voting rules~\cite{betzler2009towards}, it follows from \Cref{pw_hard} that the \WM problem is in \Pb{} for the plurality and veto voting rules for any number of manipulators. 

\begin{proposition}\label{pv_easy}
 The \WM problem is in \Pb{} for the plurality and veto voting rules for any number of manipulators.
\end{proposition}

\begin{proof}
 The \PW problem is in \Pb{} for the plurality and the veto voting rules~\cite{betzler2009towards}. Hence, the result follows from \Cref{pw_hard}.
\end{proof}

\subsection{Strong Manipulation Problem}

We now discuss our algorithms for the \SM{} problem. The common flavor in all our algorithms is the following: we try to devise an extension that is as adversarial as possible for the favorite candidate $c$, and if we can make $c$ win in such an extension, then roughly speaking, such a strategy should work for other extensions as well (where the situation only improves for $c$). However, it is challenging to come up with an extension that is globally dominant over all the others in the sense that we just described. So what we do instead is we consider every potential nemesis $w$ who might win instead of $c$, and we build profiles that are ``as good as possible'' for $w$ and ``as bad as possible'' for $c$. Each such profile leads us to constraints on how much the manipulators can afford to favor $w$ (in terms of which positions among the manipulative votes are \textit{safe} for $w$). We then typically show that we can determine whether there exists a set of votes that respects these constraints, either by using a greedy strategy or by an appropriate reduction to a flow problem. We note that the overall spirit here is similar to the approaches commonly used for solving the \NW problem, but as we will see, there are non-trivial differences in the details. We begin with the $k$-approval and $k$-veto voting rules.

\begin{theorem}\label{sm_k_easy}
The \SM problem is in \Pb{} for the $k$-approval and $k$-veto voting rules, for any $k$ and any number of manipulators.
\end{theorem}

\begin{proof}
 For the time being, we just concentrate on non-manipulators' votes. For each candidate $c^{\prime} \in \mathcal{C}\setminus\{c\}$, calculate the maximum possible value of $s^{max}_{NM}(c,c^{\prime}) = s_{NM}(c^{\prime}) - s_{NM}(c)$ from non-manipulators' votes, where $s_{NM}(a)$ is the score that candidate $a$  receives from the votes of the non-manipulators. This can be done by checking all possible $O(m^2)$ pairs of positions for $c$ and $c^{\prime}$ in each vote $v$ and choosing the one which maximizes $s_v(c^{\prime}) - s_v(c)$ from that vote. We now fix the position of $c$ at the top position for the manipulators' votes and we check if it is possible to place other candidates in the manipulators' votes such that the final value of $s^{max}_{NM}(c,c^{\prime}) + s_M(c^{\prime}) - s_M(c)$ is negative which can be solved easily by reducing it to the max flow problem which is polynomial time solvable.
\end{proof}

We now prove that the \SM problem for scoring rules is in \Pb{} for one manipulator.

\begin{theorem}\label{sm_sr_easy}
The \SM problem is in \Pb{} for any scoring rule when we have only one manipulator.
\end{theorem}

\begin{proof}
For each candidate $c^{\prime} \in \mathcal{C}\setminus\{c\}$, calculate $s^{max}_{NM}(c,c^{\prime})$ using same technique described in the proof of \Cref{sm_k_easy}. We now put $c$ at the top position of the manipulator's vote. For each candidate $c^{\prime} \in \mathcal{C}\setminus\{c\}$, $c^{\prime}$ can be placed at positions $i\in \{2,\ldots,m\}$ in the manipulator's vote which makes $s^{max}_{NM}(c,c^{\prime}) + \alpha_i - \alpha_1$ negative. Using this, construct 
a bipartite graph with $\mathcal{C}\setminus\{c\}$ on left and $\{2, \dots, m\}$ 
on right and there is an edge between $c^{\prime}$ and $i$ iff the candidate $c^{\prime}$ 
can be placed at $i$ in the manipulator's vote according to the above criteria. 
Now solve the problem by finding existence of perfect matching in this graph.
\end{proof}

Our next result proves that the \SM problem for the Bucklin, simplified Bucklin, fallback, and simplified fallback voting rules are in \Pb{}.

\begin{theorem}\label{sm_bucklin_easy}
The \SM problem is in \Pb{} for the Bucklin, simplified Bucklin, fallback, and simplified fallback voting rules, for any number of manipulators. 
\end{theorem}

\begin{proof}
Let $(\mathcal{C},\mathcal{P},M,c)$ be an instance of \SM for simplified Bucklin, and let $m$ denote the total number of candidates in this instance. Recall that the manipulators have to cast their votes so as to ensure that the candidate $c$ wins in every possible extension of $\mathcal{P}$. We use $\mathcal{Q}$ to denote the set of manipulating votes that we will construct. To begin with, without loss of generality, the manipulators place $c$ in the top position of all their votes. We now have to organize the positioning of the remaining candidates across the votes of the manipulators to ensure that $c$ is a necessary winner of the profile $(\mathcal{P}, \mathcal{Q})$.

To this end, we would like to develop a system of constraints indicating the overall number of times that we are free to place a candidate $w \in \mathcal{C} \setminus \{c\}$ among the top $\ell$ positions in the profile $\mathcal{Q}$. In particular, let us fix $w \in \mathcal{C} \setminus \{c\}$ and $2 \leq \ell \leq m$. Let $\eta_{w,\ell}$ be the maximum number of votes of $\mathcal{Q}$ in which $w$ can appear in the top $\ell$ positions. Our first step is to compute necessary conditions for $\eta_{w,\ell}$.

We use $\overline{\mathcal{P}}_{w,\ell}$ to denote a set of complete votes that we will construct based on the given partial votes. Intuitively, these votes will represent the ``worst'' possible extensions from the point of view of $c$ when pitted against $w$. These votes are engineered to ensure that the manipulators can make $c$ win the elections $\overline{\mathcal{P}}_{w,\ell}$ for all $w \in \mathcal{C} \setminus \{c\}$ and $\ell \in \{2,\ldots,m\}$, if, and only if, they can strongly manipulate in favor of $c$. More formally, there exists a voting profile $\mathcal{Q}$ of the manipulators so that $c$ wins the election $\overline{\mathcal{P}}_{w,\ell} \cup \mathcal{Q}$, for all $w \in \mathcal{C} \setminus \{c\}$ and $\ell \in \{2,\ldots,m\}$ if and only if $c$ wins in every extension of the profile $\mathcal{P} \cup \mathcal{Q}$. 

We now describe the profile $\overline{\mathcal{P}}_{w,\ell}$. The construction is based on the following case analysis, where our goal is to ensure that, to the extent possible, we position $c$ out of the top $\ell-1$ positions, and incorporate $w$ among the top $\ell$ positions. 
\begin{itemize}
 \item Let $v \in \mathcal{P}$ be such that either $c$ and $w$ are incomparable or $w \succ c$. We add the complete vote $v^\prime$ to  $\overline{\mathcal{P}}_{w,\ell}$, where $v^\prime$ is obtained from $v$ by placing $w$ at the highest possible position and $c$ at the lowest possible position, and extending the remaining vote arbitrarily. 
 \item Let $v \in \mathcal{P}$ be such that $c \succ w$, but there are at least $\ell$ candidates that are preferred over $w$ in $v$. We add the complete vote $v^\prime$ to  $\overline{\mathcal{P}}_{w,\ell}$, where $v^\prime$ is obtained from $v$ by placing $c$ at the lowest possible position, and extending the remaining vote arbitrarily.  
 \item Let $v \in \mathcal{P}$ be such that $c$ is forced to be within the top $\ell-1$ positions, then we add the complete vote $v^\prime$ to  $\overline{\mathcal{P}}_{w,\ell}$, where $v^\prime$ is obtained from $v$ by first placing $w$ at the highest possible position followed by placing $c$ at the lowest possible position, and extending the remaining vote arbitrarily.
 \item In the remaining votes, notice that whenever $w$ is in the top $\ell$ positions, $c$ is also in the top $\ell-1$ positions. Let $\mathcal{P}^*_{w,\ell}$ denote this set of votes, and let $t$ be the number of votes in $\mathcal{P}^*_{w,\ell}$.
\end{itemize}
We now consider two cases. Let $d_\ell(c)$ be the number of times $c$ is placed in the top $\ell-1$ positions in the profile $\overline{\mathcal{P}}_{w,\ell} \cup \mathcal{Q}$, and let $d_\ell(w)$ be the number of times $w$ is placed in the top $\ell$ positions in the profile $\overline{\mathcal{P}}_{w,\ell}$. Let us now formulate the requirement that in $\overline{\mathcal{P}}_{w,\ell} \cup \mathcal{Q}$, the candidate $c$ does \emph{not} have a majority in the top $\ell-1$ positions and $w$ \emph{does} have a majority in the top $\ell$ positions. Note that if this requirement holds for any $w$ and $\ell$, then strong manipulation is not possible. Therefore, to strongly manipulate in favor of $c$, we must ensure that for every choice of $w$ and $\ell$, we are able to negate the conditions that we derive. 

The first condition from above simply translates to $d_\ell(c) \le \nfrac{n}{2}$. The second condition amounts to requiring first, that there are at least $\nfrac{n}{2}$ votes where $w$ appears in the top $\ell$ positions, that is, $d_\ell(w) + \eta_{w,\ell} + t > \nfrac{n}{2}$. Further, note that the gap between $d_\ell(w)+\eta_{w,\ell}$ and majority will be filled by using votes from $\mathcal{P}^*_{w,\ell}$ to ``push'' $w$ forward. However, these votes  contribute equally to $w$ and $c$ being in the top $\ell$ and $\ell-1$ positions, respectively. Therefore, the difference between $d_\ell(w)+\eta_{w,\ell}$ and $\nfrac{n}{2}$ must be less than the difference between $d_\ell(c)$ and $\nfrac{n}{2}$. Summarizing, the following conditions, which we collectively denote by $(\star)$, are sufficient to defeat $c$ in some extension: $d_\ell(c) \le \nfrac{n}{2}, d_\ell(w) + \eta_{w,\ell} + t > \nfrac{n}{2}, \nfrac{n}{2} - d_\ell(w) + \eta_{w,\ell} < \nfrac{n}{2} - d_\ell(c)$.

From the manipulator's point of view, the above provides a set of constraints to be satisfied as they place the remaining candidates across their votes. Whenever $d_\ell(c) > \nfrac{n}{2}$, the manipulators place any of the other candidates among the top $\ell$ positions freely, because $c$ already has majority. On the other hand, if $d_\ell(c) \leq \nfrac{n}{2}$, then the manipulators must respect at least one of the following constraints: $\eta_{w,\ell} \leq \nfrac{n}{2} - t - d_\ell(w)$ and $\eta_{w,\ell} \leq d_\ell(c) - d_\ell(w)$.

Extending the votes of the manipulator while respecting these constraints (or concluding that this is impossible to do) can be achieved by a natural greedy strategy --- construct the manipulators' votes by moving positionally from left to right. For each position, consider each manipulator and populate her vote for that position with any available candidate. We output the profile if the process terminates by completing all the votes, otherwise, we say \textsc{No}.

We now argue the proof of correctness. Suppose the algorithm returns \textsc{No}. This implies that there exists a choice of $w \in \mathcal{C} \setminus \{c\}$ and $\ell \in \{2,\ldots,m\}$ such that for any voting profile $\mathcal{Q}$ of the manipulators, the conditions in $(\star)$ are satisfied. (Indeed, if there exists a voting profile that violated at least one of these conditions, then the greedy algorithm would have discovered it.) Therefore, no matter how the manipulators cast their vote, there exists an extension where $c$ is defeated. In particular, for the votes in $\mathcal{P} \setminus \mathcal{P}^*_{w,\ell}$, this extension is given by $\overline{\mathcal{P}}_{w,\ell}$. Further, we choose $\nfrac{n}{2} - \eta_{w,\ell} - d_\ell(w)$ votes among the votes in $\mathcal{P}^*_{w,\ell}$ and extend them by placing $w$ in the top $\ell$ positions (and extending the rest of the profile arbitrary). We extend the remaining votes in $\mathcal{P}^*_{w,\ell}$ by positioning $w$ outside the top $\ell$ positions. Clearly, in this extension, $c$ fails to achieve majority in the top $\ell-1$ positions while $w$ does achieve majority in the top $\ell$ positions. 

On the other hand, if the algorithm returns \textsc{Yes}, then consider the voting profile of the manipulators. We claim that $c$ wins in every extension of $\mathcal{P} \cup \mathcal{Q}$. Suppose, to the contrary, that there exists an extension $\mathcal{R}$ and a candidate $w$ such that the simplified Bucklin score of $c$ is no more than the simplified Bucklin score of $w$ in $\mathcal{R}$. In this extension, therefore, there exists $\ell \in \{2,\ldots,m\}$ for which $w$ attains majority in the top $\ell$ positions and $c$ fails to attain majority in the top $\ell-1$ positions. However, note that this is already impossible in any extension of the profile $\overline{\mathcal{P}}_{w,l} \cup \mathcal{P}^*_{w,\ell}$, because of the design of the constraints. By construction, the number of votes in which $c$ appears in the top $\ell-1$ positions in $\mathcal{R}$ is only greater than the number of times $c$ appears in the top $\ell-1$ positions in any extension of $\overline{\mathcal{P}}_{w,l} \cup \mathcal{P}^*_{w,\ell}$ (and similarly for $w$). This leads us to the desired contradiction. 

For the Bucklin voting rule, we do the following modifications to the algorithm. If $d_\el(c) > d_\el(w)$ for some $w\in\CC\setminus\{c\}$ and $\el<m$, then we make $\eta_{w,\el} = \infty$. The proof of correctness for the Bucklin voting rule is similar to the proof of correctness for the simplified Bucklin voting rule above. 

For Fallback and simplified Fallback voting rules, we consider the number of candidates each voter approves while computing $\eta_{w,\el}$. We output \YES if and only if $\eta_{w,\el}\ge 0$ for every $w\in\CC\setminus\{c\}$ and every $\el\le m$, since we can assume, without loss of generality, that the manipulator approves the candidate $c$ only. Again the proof of correctness is along similar lines to the proof of correctness for the simplified Bucklin voting rule.
\end{proof}

We next show that the \SM problem for the maximin voting rule is polynomial-time solvable when we have only one manipulator.

\begin{theorem}\label{sm_maximin_easy}
The \SM problem for the maximin voting rules are in \Pb{}, when we have only one manipulator.
\end{theorem}

\begin{proof}
 For the time being, just concentrate on non-manipulators' votes. Using the algorithm for NW for maximin in \cite{xia2008determining}, we compute for all pairs $w, w^{\prime} \in \mathcal{C}$, $N_{(w,w^{\prime})}(w,d)$ and $N_{(w,w^{\prime})}(c,w^{\prime})$ for all $d\in \mathcal{C}\setminus \{c\}$. This can be computed in polynomial time. Now we place $c$ at the top position in the manipulator's vote and increase all $N_{(w,w^{\prime})}(c,w^{\prime})$ by one. Now we place a candidate $w$ at the second position if for all $w^{\prime} \in \mathcal{C}$, $N_{(w,w^{\prime})}^{\prime}(w,d) < N_{(w,w^{\prime})}(c,w^{\prime})$ for all $d\in \mathcal{C}\setminus \{c\}$, where $N_{(w,w^{\prime})}^{\prime}(w,d) = N_{(w,w^{\prime})}(w,d)$ of the candidate $d$ has already been assigned some position in the manipulator's vote, and $N_{(w,w^{\prime})}^{\prime}(w,d) = N_{(w,w^{\prime})}(w,d)+1$ else. The correctness argument is in the similar lines of the classical greedy manipulation algorithm of~\cite{bartholdi1989computational}.
\end{proof}

\subsection{Opportunistic Manipulation Problem}

For the plurality, fallback, and simplified fallback voting rules, it turns out that the voting profile where all the manipulators approve only $c$ is a $c$-opportunistic voting profile, and therefore it is easy to devise a manipulative vote. 

\begin{observation}\label{obs:plurality_fallback}
 The \OM problem is in \Pb for the plurality and fallback voting rules for a any number of manipulators.
\end{observation}

For the veto voting rule, however, a more intricate argument is needed, that requires building a system of constraints and a reduction to a suitable instance of the maximum flow problem in a network, to show polynomial time tractability of \OM.

\begin{theorem}\label{thm:vetoOM}
 The \OM problem is in \Pb for the veto voting rule for a constant number of manipulators.
\end{theorem}

\begin{proof}
 Let $(\PP, \el, c)$ be an input instance of \OM. We may assume without loss of generality that the manipulators approve $c$. We view the voting profile of the manipulators as a tuple $(n_a)_{a\in\CC\setminus\{c\}}\in(\mathbb{N}\cup\{0\})^{m-1}$ with $\sum_{a\in\CC\setminus\{c\}} n_a = \el$, where the $n_a$ many manipulators disapprove $a$. We denote the set of such tuples as $\TT$ and we have $\TT=O((2m)^\el)$ which is polynomial in $m$ since \el is a constant. A tuple $(n_a)_{a\in\CC\setminus\{c\}}\in\TT$ is not $c$-optimal if there exists another tuple $(n_a^\prm)_{a\in\CC\setminus\{c\}}\in\TT$ and an extension $\overline{\PP}$ of $\PP$ with the following properties. We denote the veto score of a candidate from $\PP$ by $s(\cdot)$. For every candidate $a\in\CC\setminus\{c\}$, we define two quantities $w(a)$ and $d(a)$ as follows.
 \begin{itemize}
  \item $s(c) > s(a)$ for every $a\in\CC\setminus\{c\}$ with $n_a = n_a^\prm = 0$ and we define $w(a) = s(c) - 1, d(a)=0$
  \item $s(c) > s(a) - n_a^\prm  $ for every $a\in\CC\setminus\{c\}$ with $n_a \ge n_a^\prm$ and we define $w(a) = s(c) - n_a^\prm - 1, d(a)=0$
  \item $s(a) - n_a \ge s(c) > s(a) - n_a^\prm$ for every $a\in\CC\setminus\{c\}$ with $n_a < n_a^\prm$ and we define $w(a) = s(c) - n_a^\prm, d(a)=s(a) - n_a$
 \end{itemize}
 We guess the value of $s(c)$. Given a value of $s(c)$, we check the above two conditions by reducing this to a max flow problem instance as follows. We have a source vertex $s$ and a sink $t$. We have a vertex for every $a\in\CC$ (call this set of vertices $Y$) and a vertex for every vote $v\in\PP$ (call this set of vertices $X$). We add an edge from $s$ to each in $X$ of capacity one. We add an edge of capacity one from a vertex $x\in X$ to a vertex $y\in Y$ if the candidate corresponding to the vertex $y$ can be placed at the last position in an extension of the partial vote corresponding to the vertex $x$. We add an edge from a vertex $y$ to $t$ of capacity $w(a)$, where $a$ is the voter corresponding to the vertex $y$. We also set the demand of every vertex $y$ $d(a)$ (that is the total amount of flow coming into vertex $y$ must be at least $d(a)$), where $a$ is the voter corresponding to the vertex $y$. Clearly, the above three conditions are met if and only if there is a feasible $|\PP|$ amount of flow in the above flow graph. Since $s(c)$ can have only $|\PP|+1$ possible values (from $0$ to $\PP$) and $|\TT|=O((2m)^\el)$, we can iterate over all possible pairs of tuples in $\TT$ and all possible values of $s(c)$ and find a $c$-optimal voting profile if there exists a one.
\end{proof}

\section{Conclusion}\label{sec:con}

We revisited many settings where the complexity barrier for manipulation was non-existent, and studied the problem under an incomplete information setting. Our results present a fresh perspective on the use of computational complexity as a barrier to manipulation, particularly in cases that were thought to be dead-ends (because the traditional manipulation problem was polynomially solvable). To resurrect the argument of computational hardness, we have to relax the model of complete information, but we propose that the incomplete information setting is more realistic, and many of our hardness results work even with very limited amount of incompleteness in information.

In the next chapter, we will see how possible instances of manipulation can be detected for various voting rules.

\chapter{Manipulation Detection}
\label{chap:detection}

\blfootnote{A preliminary version of the work in this chapter was published as \cite{DeyMN15a}: Palash Dey, Neeldhara Misra, and Y. Narahari. Detecting possible manipulators in elections. In Proc. 2015 International Conference on Autonomous Agents and Multiagent Systems, AAMAS 2015, Istanbul, Turkey, May 4-8, 2015, pages 1441-1450, 2015.}

\begin{quotation}
{\small Manipulation is a problem of fundamental importance 
in voting theory in which the voters exercise their 
votes strategically instead of voting honestly to make
an alternative that is more preferred to her, win the election.
The classical Gibbard-Satterthwaite theorem 
shows that there is no strategy-proof voting rule that simultaneously satisfies 
certain combination of desirable properties. 
Researchers have attempted to get around
the impossibility result in several ways such as
domain restriction and computational hardness of manipulation. 
However, these approaches are known to have fundamental limitations. 
Since prevention of manipulation seems to be elusive even after substantial research effort, an interesting research direction therefore is detection of manipulation. Motivated by this, 
we initiate the study of detecting possible instances of manipulation 
in elections.

We formulate two pertinent computational 
problems in the context of manipulation detection - Coalitional Possible Manipulators (CPM) and 
Coalitional Possible Manipulators given Winner (CPMW), where a 
suspect group of voters is provided as input and we have to find whether 
they can be a potential coalition of manipulators.
In the absence of any suspect group, we formulate two more computational 
problems namely Coalitional Possible Manipulators Search (CPMS) and Coalitional Possible Manipulators Search 
given Winner (CPMSW). We provide polynomial time algorithms
for these problems, for several popular voting rules. For a few other voting rules,
we show that these problems are \NPshort{}-complete. We observe that 
detecting possible instances of manipulation may be easy even when the actual manipulation problem is computationally intractable, as seen
for example, in the case of the Borda voting rule.}
\end{quotation}

\section{Introduction}
A basic problem with voting rules is that the voters 
may vote strategically instead of voting honestly, leading to the 
selection of a candidate which is not the {\em actual winner}. 
We call a candidate actual winner if it wins the election when every voter votes truthfully. 
This phenomenon of strategic voting is called 
{\em manipulation} in the context of voting. The  
Gibbard-Satterthwaite (G-S) theorem \cite{gibbard1973manipulation,satterthwaite1975strategy} 
proves that manipulation is unavoidable for any \textit{unanimous}
and \textit{non-dictatorial} voting rule if we have at least three candidates. 
A voting rule is called \textit{unanimous} if 
whenever any candidate is most preferred by all the voters, such a candidate is the winner. 
A voting rule is called \textit{non-dictatorial} 
if there does not exist any voter whose most preferred candidate is always 
the winner irrespective of the votes of other voters.
The problem of manipulation is particularly relevant for multiagent systems
since agents have computational power to determine strategic votes. 
There have been several attempts to bypass the impossibility 
result of the G-S theorem.

Economists have proposed domain restriction as a way out of the impossibility 
implications of the G-S theorem. The G-S theorem assumes all possible 
preference profiles as the domain of voting rules. In a restricted domain, it has been
shown that we can have voting rules that are not vulnerable to manipulation. 
A prominent restricted domain is the domain of single peaked preferences, in 
which the median voting rule provides a satisfactory 
solution~\cite{mas1995microeconomic}. To know more about other domain restrictions, 
we refer to \cite{mas1995microeconomic,gaertner2001domain}. This approach of
restricting the domain, however, 
suffers from the requirement that the social planner needs to know the 
domain of preference profiles of the voters, which is often impractical. 

\subsection{Related Work}

Researchers in computational social choice theory have proposed invoking computational intractability of manipulation as a possible 
work around for the G-S theorem. 
Bartholdi et al.~\cite{bartholdi1991single,bartholdi1989computational} 
first proposed the idea of using computational hardness as a barrier 
against manipulation. Bartholdi et al. defined and studied the computational problem called manipulation 
where a set of manipulators have to compute their votes that make their preferred candidate win the election.
The manipulators know the votes of the truthful voters and the voting rule that will be used to compute the winner. Following this, a large body of research  
\cite{narodytska2011manipulation,davies2011complexity,xia2009complexity,xia2010scheduling,conitzer2007elections,obraztsova2011ties,elkind2005hybrid,faliszewski2013weighted,narodytska2013manipulating,gaspers2013coalitional,obraztsova2012optimal,faliszewski2010using,zuckerman2011algorithm,dey2014asymptotic,faliszewski2014complexity,elkind2012manipulation} 
shows that the manipulation problem is \NPshort{}-complete for many voting rules. However, 
Procaccia et al.~\cite{procaccia2006junta,ProcacciaR07} 
showed average case easiness of manipulation assuming \textit{junta} 
distribution over the voting profiles.  
Friedgut et al.~\cite{friedgut2008elections} showed that any \textit{neutral} voting rule which is sufficiently far from 
being dictatorial is manipulable with non-negligible probability at any uniformly random preference profile 
by a uniformly random preference. The above result holds for elections 
with three candidates only. 
A voting rule is called \textit{neutral} 
if the names of the candidates are immaterial. 
Isaksson et al.~\cite{isaksson2012geometry} 
generalize the above result to any number of candidates which has been further generalized to all voting rules which may not be neutral by Mossel and Racz in~\cite{mossel2015quantitative}. 
Walsh~\cite{walsh2010empirical} empirically shows ease of manipulating an STV 
(single transferable vote) election  
-- one of the very few voting rules 
where manipulation, even by one voter, is \NPCshort{}~\cite{bartholdi1991single}. 
In addition to the results mentioned above, there exist many other results in 
the literature that emphasize the weakness of considering computational complexity as a barrier against manipulation
\cite{conitzer2006nonexistence,xia2008sufficient,xia2008generalized,faliszewski2010ai,walsh2011computational}.
Hence, the barrier of computational hardness is ineffective against manipulation in many settings.

\subsection{Motivation}

In a situation where multiple attempts for prevention of manipulation fail to provide a fully satisfactory solution, detection of manipulation is a natural next step of research. There have been scenarios where establishing the occurrence of manipulation is straightforward, by observation or hindsight. For example, in sport, there have been occasions where the very structure of the rules of the game have encouraged teams to deliberately lose their matches. Observing such occurrences in, for example, football (the 1982 FIFA World Cup football match played between West Germany and Austria) and badminton (the quarter-final match between South Korea and China in the London 2012 Olympics), the relevant authorities have subsequently either changed the rules of the game (as with football) or disqualified the teams in question (as with the badminton example). The importance of detecting manipulation lies in the potential for implementing corrective measures in the future. For reasons that will be evident soon, it is not easy to formally define the notion of manipulation detection. Assume that we have the votes from an
election that has already happened. A voter is potentially a manipulator if
there exists a preference $\succ$, different from the voter's reported preference,
which is such that the voter had an ``incentive to deviate'' from the \suc. 
Specifically, suppose the candidate
who wins with respect to this voter's reported preference is preferred (in $\succ$) over the 
candidate who wins with respect to $\succ$. In such a situation, $\succ$ could 
potentially be the voter's truthful preference, and the voter could be 
refraining from being truthful because an untruthful vote leads to a
more favorable outcome with respect to $\succ$. Note that we do not (and indeed, cannot) conclusively suggest that a particular voter 
has manipulated an election. This is because the said voter can always claim that she voted truthfully; 
since her actual preference is only known to her, there is no way to prove or disprove 
such a claim. Therefore, we are inevitably limited to asking only 
whether or not a voter has \textit{possibly} manipulated an election.

Despite this technical caveat, it is clear that efficient detection of manipulation, even if it is only possible manipulation, is potentially of interest in practice. We believe that, the information whether a certain group of voters have possibly manipulated an election or not would be useful to social planners. For example, the organizers of an event, say London 2012 Olympics, may be interested to have this information. Also, in settings where data from many past elections (roughly over a fixed set of voters) is readily available, it is conceivable that possible manipulation could serve as suggestive evidence of real manipulation. Aggregate data about possible manipulations, although formally inconclusive, could serve as an important evidence of real manipulation, especially in situations where the instances of possible manipulation turn out to be statistically significant. Thus, efficient detection of possible manipulation would provide an useful input to a social planner for future elections. We remark that having a rich history is typically not a problem, 
particularly for AI related applications, 
since the data generated from an election is normally kept for future 
requirements (for instance, for data mining or learning). For example, several past affirmatives for possible manipulation is one possible way of formalizing the notion of erratic past behavior. Also, applications where benefit of doubt may be important, for example, elections in judiciary systems, possible manipulation detection may be useful. Thus the computational problem of detecting possible manipulation is of definite interest in many settings.

\subsection{Our Contribution} 

The novelty of this work is in initiating research on {\em detection
of possible manipulators} in elections. We formulate four pertinent computational
problems in this context:

\begin{itemize}
 \item CPM: In the {\em coalitional possible manipulators} problem, 
  we are interested in whether or not a given subset of voters is a possible coalition of manipulators [\Cref{cpmprobdef}]. 
 \item CPMW: The {\em coalitional possible manipulators given winner} is  
  the CPM problem with the additional information about who the winner would have been if the possible manipulators had all voted truthfully [\Cref{CPMWdef}].
 \item CPMS, CPMSW: In CPMS ({\em Coalitional Possible Manipulators Search}), we want to know, whether there exists any coalition of possible manipulators of a size at most $k$ [\Cref{cpmoprobdef}]. Similarly, we define CPMSW ({\em Coalitional Possible Manipulators Search given Winner}) [\Cref{cpmwoprobdef}].
\end{itemize}

Our specific findings are as follows.

\begin{itemize}
 \item We show that all the four problems above, for scoring rules and the maximin voting rule, 
 are in \Pshort{} when the coalition size is one [\Cref{PUMScoringRuleP,PUMMaximinP}].
 \item We prove that all the four problems, {\em for any coalition size}, are in \Pshort{} for a wide class of scoring rules which include the Borda voting rule [\Cref{CPMWScr}, \Cref{thm:CPMWOScr,bordakapp}]. 
 \item We show that, for the Bucklin voting rule [\Cref{UPCMBucklin}], both the CPM and CPMW problems 
 are in \Pshort{}, {\em for any coalition size}. The CPMS and CPMSW problems for the Bucklin voting rule are also in \Pshort{}, when we have maximum possible coalition size $k=O(1)$.
 \item We show that both the CPM and the CPMW problems are \NPCshort{} for the STV voting rule [\Cref{UPCMSTVNPC,CPMWSTVNPC}], {\em even for a coalition of size one}. We also prove that the CPMW problem is \NPCshort{} for maximin voting rule [\Cref{maxmincpmwnpc}], {\em for a coalition of size two}. 
\end{itemize}

We observe that all the four problems are computationally easy for many voting rules that we study in this work. This can be taken as a positive result. 
The results for the CPM and the CPMW problems are summarized in \Cref{tbl:detection_summary}. 

\begin{table}[htbp]
  \centering
{\renewcommand{\arraystretch}{1.17}
\begin{tabular}{|c|c|c|c|c|}\hline
  Voting Rule	& CPM, $k=1$	&CPM	  	&CPMW, $k=1$	&CPMW	 	\\\hline\hline
  Scoring Rules	& \Pshort{}	&?		&\Pshort{}	&?		\\\hline
  Borda		& \Pshort{}	& \Pshort{}	& \Pshort{}	&\Pshort{}	\\\hline 
  $k$-approval	& \Pshort{}	& \Pshort{}	& \Pshort{}	& \Pshort{}	\\\hline 
  Maximin	& \Pshort{}	&?		&\Pshort{}	&\NPCshort{}\\\hline
  Bucklin	& \Pshort{}	& \Pshort{}	& \Pshort{}	& \Pshort{}	\\\hline 
  STV		& \NPCshort{}	& \NPCshort{}	& \NPCshort{}	& \NPCshort{}	\\\hline
 \end{tabular}
 }
 \caption{\normalfont Results for CPM and CPMW ($k$ denotes coalition size). The `?' mark means that the problem is open.}\label{tbl:detection_summary}
\end{table} 

\section{Problem Formulation}

Consider an election that has already happened in which all the votes are known and thus the 
winner $x\in\mathcal{C}$ is also known. We call the candidate $x$ the {\em current winner} of the election. 
The authority may suspect that the voters belonging to a subset $M\subset \mathcal{V}$ of the set of voters have formed a coalition among themselves and manipulated the 
election by voting non-truthfully. The authority believes that other voters who do not belong to $M$, have voted truthfully. We denote $|M|$, the size of the coalition, by $k$.
Suppose the authority has auxiliary information, 
may be from some other sources, which says that the \textit{actual winner} should have been some candidate 
$y\in \mathcal{C}$ other than $x$. 
We call a candidate {\em actual winner} if it wins the election where all the voters vote truthfully. This means that the authority
thinks that, had the voters in $M$ voted truthfully, the candidate $y$ would have been the winner.  
We remark that there are practical situations, for example, 1982 FIFA World cup or 2012 London Olympics, where 
the authority knows the actual winner. 
This situation is formalized below.

\begin{definition}\label{cpmwdef}
Let $r$ be a voting rule, and $(\succ_i)_{i\in \mathcal{V}}$ be a voting profile of a set $\cal V$ of $n$ voters. Let $x$ be the winning candidate with respect to $r$ for this profile. For a candidate $y \neq x$, $M\subset \mathcal{V}$ is called a coalition of possible manipulators against $y$ with respect to $r$ if there exists a $|M|$-voters' profile 
 $(\succ_j^{\prime})_{j\in M} \in \mathcal{L(C)}^{|M|}$ such that $x \succ_j^\prime y, \forall j\in M$, and further,
$r((\succ_j)_{j\in \mathcal{V}\setminus M},(\succ_i^{\prime})_{i\in M}) = y$.  
\end{definition} 

Using the notion of coalition of possible manipulators, we formulate a computational problem called {\em Coalitional Possible Manipulators given Winner (CPMW)} as follows. Let $r$ be any voting rule.

\begin{definition}{($r$--CPMW Problem)}\label{CPMWdef}\\
Given a preference profile $(\succ_i)_{i\in \mathcal{V}}$ of a set of voters $\mathcal{V}$ over a set of candidates $\mathcal{C}$, a subset of voters $M\subset \mathcal{V}$, and a candidate $y$, determine if $M$ is a coalition of possible manipulators against $y$ with respect to $r$.
\end{definition}

In the CPMW problem, the actual winner is given in the input. However, 
it may very well happen that the authority does not have any other information to 
guess the \textit{actual winner} -- the candidate who would have won the election had the voters in $M$ 
voted truthfully. In this situation, the authority is interested in knowing whether there 
is a $|M|$-voter profile which along with the votes in $\mathcal{V}\setminus M$ makes some candidate 
$y\in \mathcal{C}$ the winner who is different from the current winner $x\in \mathcal{C}$ and all 
the preferences in the $|M|$-voters' profile prefer $x$ to $y$. If such a $|M|$-voter profile exists for the subset of voters $M$, then we call $M$ a \textit{coalition of possible manipulators} and the 
corresponding computational problem is called \textit{Coalitional Possible Manipulators (CPM)}. 
These notions are formalized below.

\begin{definition}
\label{cpmdef}
Let $r$ be a voting rule, and $(\succ_i)_{i\in \mathcal{V}}$ be a voting profile of a set $\cal V$ of $n$ voters. A subset of voters $M\subset \mathcal{V}$ is called a coalition of possible manipulators with respect to $r$ if $M$ is a coalition of possible manipulators against some candidate $y$ with respect to $r$.
\end{definition}

\begin{definition}{($r$--CPM Problem)}\label{cpmprobdef}\\
Given a preference profile $(\succ_i)_{i\in \mathcal{V}}$ of a set of voters $\mathcal{V}$ over a set of candidates $\mathcal{C}$, and a subset of voters $M\subset \mathcal{V}$, determine if $M$ is a coalition of possible manipulators with respect to $r$.
\end{definition}

In both the CPMW and CPM problems, a subset of voters which the authority suspect to be a coalition of manipulators, is given in the input. However, there can be situations where there is no specific subset of voters to suspect. In those scenarios, it may still be useful to know, what are the possible coalition of manipulators of size less than some number $k$. Towards that end, we extend the CPMW and CPM problems to search for a coalition of potential possible manipulators and call them {\em Coalitional  Possible Manipulators Search given Winner (CPMSW)} and {\em Coalitional Possible Manipulators Search (CPMS)} respective.

\begin{definition}{($r$--CPMSW Problem)}\label{cpmwoprobdef}\\
Given a preference profile $(\succ_i)_{i\in \mathcal{V}}$ of a set of voters $\mathcal{V}$ over a set of candidates $\mathcal{C}$, a candidate $y$, and an integer $k$, determine whether there exists any $M\subset \mathcal{V}$ with $|M|\le k$ such that $M$ is a coalition of possible manipulators against $y$.
\end{definition}

\begin{definition}{($r$--CPMS Problem)}\label{cpmoprobdef}\\
Given a preference profile $(\succ_i)_{i\in \mathcal{V}}$ of a set of voters $\mathcal{V}$ over a set of candidates $\mathcal{C}$, and an integer $k$, determine whether there exists any $M\subset \mathcal{V}$ with $|M|\le k$ such that $M$ is a coalition of possible manipulators.
\end{definition}

\subsection{Discussion} 

The CPMW problem may look very similar to the manipulation problem~\cite{bartholdi1989computational,conitzer2007elections} -- in both the problems a set of voters try to make a candidate winner. However, in the CPMW problem, the actual winner must be less preferred to the current winner in every manipulator's vote. Although it may look like a subtle difference, it changes the nature and complexity theoretic behavior of the problem completely. For example, we show that all the four problems have an efficient algorithm for a large class of voting rules that includes the Borda voting rule, for any coalition size. However, the manipulation problem for the Borda voting rule is \NPCshort{}, even when we have at least two manipulators~\cite{davies2011complexity,betzler2011unweighted}. Another important difference is that the manipulation problem, in contrast to the problems studied in this work, does not take care of manipulators' preferences. We believe that there does not exist any formal reduction between the CPMW problem and the manipulation problem. 

On the other hand, the CPMS problem is similar to the margin of victory problem defined by Xia~\cite{xia2012computing}, where also we are looking for changing the current winner by changing at most some $k$ number of votes, which in turn identical to the destructive bribery problem~\cite{faliszewski2009hard}. Whereas, in the CPMS problem, the vote changes can occur in a restricted fashion. An important difference between the two problems is that the margin of victory problem has the hereditary property which the CPMS problem does not possess (there is no coalition of possible manipulators of size $n$ in any election for all the common voting rules). These two problems do not seem to have any obvious complexity theoretic implications.

Now we explore the connection among the four problems that we study here. Notice that, a polynomial time algorithm for the CPM and the CPMW problems gives us a polynomial time algorithm for the CPMS and the CPMSW problems for any maximum possible coalition size $k=O(1)$. Also, observe that, a polynomial time algorithm for the CPMW (respectively CPMSW) problem implies a polynomial time algorithm for the CPM (respectively CPMS) problem. Hence, we have the following observations.

\begin{observation}\label{cpmcpmw}
 For every voting rule, if the maximum possible coalition size $k=O(1)$, then,
 \[CPMW\in\text{\Pshort{}} \Rightarrow CPM, CPMSW, CPMS\in\text{\Pshort{}}\]
\end{observation}

\begin{observation}\label{prop:cpmocpmwo}
 For every voting rule,
 \[CPMSW\in\text{\Pshort{}} \Rightarrow CPMS\in\text{\Pshort{}}\]
\end{observation}

\section{Results for the CPMW, CPM, CPMSW, and CPMS Problems}

In this section, we present our algorithmic results for the CPMW, CPM, CPMSW, and CPMS problems for various 
voting rules.

\subsection{Scoring Rules}

Below we have certain lemmas which form a crucial ingredient of our algorithms. To begin with, we define the notion of a {\em manipulated preference}. Let $r$ be a scoring rule and $\succ:=(\succ_i,\succ_{-i})$ be a voting profile of $n$ voters. Let $\succ_i^{\prime}$ be a preference such that 

$$r(\succ) >_i^{\prime}  r(\succ_i^{\prime},\succ_{-i})$$

Then we say that $\succ_i^{\prime}$ is a $(\succ,i)$-manipulated preference with respect to $r$. We omit the reference to $r$ if it is clear from the context.

\begin{lemma}\label{losersorted}
Let $r$ be a scoring rule and $\succ:=(\succ_i,\succ_{-i})$ be a voting profile of $n$ voters.  Let $a$ and $b$ be two candidates such that $score_{\succ_{-i}}(a) > score_{\succ_{-i}}(b)$, and let $\succ_i^{\prime}$ be $(\succ,i)$-manipulated preference where $a$ precedes $b$:
$$\succ_i^{\prime}:= \cdots > a > \cdots > b > \cdots$$
If $a$ and $b$ are not winners with respect to either $(\succ_i^{\prime},\succ_{-i})$ or  $\succ$, then the preference $\succ_i^{\prime\prime}$ obtained from $\succ_i^{\prime}$ by interchanging $a$ and $b$ is also $(\succ,i)$-manipulated. 
\end{lemma}

\begin{proof}
Let $x:= r(\succ_i^{\prime},\succ_{-i})$. If suffices to show that $x$ continues to win in the proposed profile $(\succ_i^{\prime\prime},\succ_{-i})$. To this end, it is enough to argue the scores of $a$ and $b$ with respect to $x$. First, consider the score of $b$ in the new profile:
 \begin{eqnarray*}
  score_{(\succ_i^{\prime\prime},\succ_{-i})}(b) 
  &=& score_{\succ_i^{\prime\prime}}(b) + score_{\succ_{-i}}(b) \nonumber \\
  &<& score_{\succ_i^{\prime}}(a) + score_{\succ_{-i}}(a) \nonumber \\
  &=& score_{(\succ_i^{\prime},\succ_{-i})}(a) \nonumber \\
  &\le& score_{(\succ_i^{\prime},\succ_{-i})}(x) \nonumber \\
  &=& score_{(\succ_i^{\prime\prime},\succ_{-i})}(x)
 \end{eqnarray*}
 The second line uses the fact that 
 $score_{\succ_i^{\prime\prime}}(b) = score_{\succ_i^{\prime}}(a)$ and 
 $score_{\succ_{-i}}(b) < score_{\succ_{-i}}(a)$. The fourth line comes from 
 the fact that $x$ is the winner and the last line follows from the fact that 
 the position of $x$ is same in both profiles. Similarly, we have the following argument for the score of $a$ in the new profile (the second line below simply follows 
 from the definition of scoring rules).
 \begin{eqnarray*}
  score_{(\succ_i^{\prime\prime},\succ_{-i})}(a) 
  &=& score_{\succ_i^{\prime\prime}}(a) + score_{\succ_{-i}}(a) \nonumber \\
  &\le& score_{\succ_i^{\prime}}(a) + score_{\succ_{-i}}(a) \nonumber \\
  &=& score_{(\succ_i^{\prime},\succ_{-i})}(a) \nonumber \\
  &\le& score_{(\succ_i^{\prime},\succ_{-i})}(x) \nonumber \\
  &=& score_{(\succ_i^{\prime\prime},\succ_{-i})}(x)
 \end{eqnarray*}
 Since the tie breaking rule is according to some predefined fixed order 
 $\succ_t  \in \mathcal{L(C)}$ and the candidates tied with winner in 
 $(\succ_i^{\prime\prime},\succ_{-i})$ also tied with winner in 
 $(\succ_i^{\prime},\succ_{-i})$, we have the following,
 \[ r(\succ) >_i^{\prime\prime} r(\succ_i^{\prime\prime},\succ_{-i})\]
\end{proof}
\setcounter{equation}{0}

We now show that, if there is some $(\succ,i)$-manipulated preference with respect to a scoring rule $r$, then there exists a $(\succ,i)$-manipulated preference with a specific structure.

\begin{lemma}\label{algotheorem}
 Let $r$ be a scoring rule and $\succ:=(\succ_i,\succ_{-i})$ be a voting profile of $n$ voters. If there is some $(\succ,i)$-manipulated preference with respect to $r$, then there also exists a $(\succ,i)$-manipulated preference $\succ_i^{\prime}$ where the actual winner $y$ immediately follows the current winner $x$:
  $$\succ_i^{\prime}:= \cdots > x > y > \cdots $$
 and the remaining candidates are in nondecreasing ordered of their scores from ${\succ_{-i}}$.
\end{lemma}

\begin{proof}
Let $\succ^{\prime\prime}$ be a $(\succ,i)$-manipulated preference with respect to $r$.  
 Let $x:= r(\succ), y:= r(\succ^{\prime\prime},\succ_{-i})$. From \Cref{losersorted}, 
 without loss of generality, we may assume that, all candidates except $x, y$ 
 are in nondecreasing order of $score_{\succ_{-i}}(.)$ in the preference $\succ^{\prime\prime}$. 
 If $\succ_i^{\prime\prime}:= \cdots \succ x \succ \cdots \succ y \succ \cdots \succ \cdots$, 
 we define $\succ_i^{\prime}:= \cdots \succ x \succ y \succ \cdots \succ \cdots$ from 
 $\succ_i^{\prime\prime}$ where 
 $y$ is moved to the position following $x$ and the position of the candidates in between 
 $x$ and $y$ in $\succ_i^{\prime\prime}$ is deteriorated by one position each. The position 
 of the rest of the candidates remain same in both $\succ_i^{\prime\prime}$ and $\succ_i^{\prime}$.
 Now we have following,
 \begin{eqnarray*}
  score_{(\succ_i^{\prime},\succ_{-i})}(y) 
  &=& score_{\succ_i^{\prime}}(y) + score_{\succ_{-i}}(y) \nonumber \\
  &\ge& score_{\succ_i^{\prime\prime}}(y) + score_{\succ_{-i}}(y) \nonumber \\
  &=& score_{(\succ_i^{\prime\prime},\succ_{-i})}(y)
 \end{eqnarray*}
 We also have,
 \begin{eqnarray*}
  score_{(\succ_i^{\prime},\succ_{-i})}(a) 
  \le score_{(\succ_i^{\prime\prime},\succ_{-i})}(a), 
  \forall a \in \mathcal{C} \setminus \{y\}
 \end{eqnarray*}
 Since the tie breaking rule is according to some predefined order $\succ_t  \in \mathcal{L(C)}$, 
 we have the following,
 \[ r(\succ) >_i^{\prime} r(\succ^{\prime},\succ_{-i}) \]
\end{proof}

Using \Cref{losersorted,algotheorem}, we now present our results for the scoring rules.

\begin{theorem}\label{PUMScoringRuleP}
 The CPMW, CPM, CPMSW, and CPMS problems for scoring rules are in \Pshort{} for a coalition of size $1$ (that is, the coalition size $k=1$).
\end{theorem}

\begin{proof}
From \Cref{cpmcpmw}, it is enough to give a polynomial time algorithm for the CPMW problem. So consider the CPMW problem. We are given the actual winner $y$ and we compute the current winner $x$ with respect to $r$. Let $\succ_{[j]}$ be a preference where $x$ and $y$ are in positions $j$ and $(j+1)$ respectively, and the rest of the candidates are in nondecreasing order of the score that they receive from $\succ_{-i}$. For $j \in \{1,2,\ldots,m-1\}$, we check if $y$ wins with the profile $(\succ_{-i},\succ_{[j]})$. If we are successful with at least one $j$ we report YES, otherwise we say NO. The correctness follows from \Cref{algotheorem}. 
Thus we have a polynomial time algorithm for CPMW when $k=1$.
\end{proof}

Now we present our results for the CPMW and the CPM problems when $k>1$. 
If $m=O(1)$, then both the CPMW and the CPM problems for any anonymous and efficient voting rule $r$ can be solved in polynomial time by iterating over all possible ${m!+k-1 \choose m!}$ ways the manipulators can have actual preferences. A voting rule is called efficient if winner determination under it is in \Pshort{}.

\begin{theorem}\label{CPMWScr}
 For scoring rules with $\alpha_1 - \alpha_2 \le \alpha_i - \alpha_{i+1}, \forall i$, 
 the CPMW and the CPM problems are in \Pshort{}, for any coalition size.
\end{theorem}

\begin{proof}
 We provide a polynomial time algorithm for the CPMW problem in this setting. 
 Let $x$ be the current winner and $y$ be the given actual winner. Let
 $M$ be the given subset of voters. Let $((\succ_i)_{i\in M}, (\succ_j)_{j\in V\setminus M})$ be the 
 reported preference profile. Without loss of generality, we assume that $x$ is the most preferred candidate 
 in every $\succ_i, i\in M$. Let us define $\succ_i^{\prime}, i\in M,$ by moving $y$ to the second position in the preference
 $\succ_i$. In the profile $((\succ_i^{\prime})_{i\in M}, (\succ_j)_{j\in V\setminus M})$, the winner is either $x$ or $y$ since 
 only $y$'s score has increased. We claim that $M$ is a coalition of possible manipulators with respect to $y$ 
 \textit{if and only if} $y$ is the winner in preference profile 
 $((\succ_i^{\prime})_{i\in M}, (\succ_j)_{j\in V\setminus M})$. This can be seen as follows. Suppose there exist 
 preferences $\succ_i^{\prime\prime}, $ with $x \succ_i^{\prime\prime} y, i\in M,$ for which $y$ wins in the profile 
 $((\succ_i^{\prime\prime})_{i\in M}, (\succ_j)_{j\in V\setminus M})$. Now without loss of generality, we can assume that 
 $y$ immediately follows $x$ in all $\succ_i^{\prime\prime}, i\in M,$ and  
 $\alpha_1 - \alpha_2 \le \alpha_i - \alpha_{i+1}, \forall i$ implies that we can also assume that $x$ and $y$ are in the 
 first and second positions respectively in all $\succ_i^{\prime\prime}, i\in M$. Now in both the profiles, 
 $((\succ_i^{\prime})_{i\in M}, (\succ_j)_{j\in V\setminus M})$ and $((\succ_i^{\prime\prime})_{i\in M}, (\succ_j)_{j\in V\setminus M})$, the score of $x$ and $y$ are same. But in the first profile $x$ wins and in the second profile $y$ wins, which is a contradiction.
\end{proof}

We now prove a similar result for the CPMSW and CPMS problems.

\begin{theorem}\label{thm:CPMWOScr}
 For scoring rules with $\alpha_1 - \alpha_2 \le \alpha_i - \alpha_{i+1}, \forall i$, 
 the CPMSW and the CPMS problems are in \Pshort{}, for any coalition size.
\end{theorem}

\begin{proof}
 From \Cref{prop:cpmocpmwo}, it is enough to prove that $CPMSW \in \mathcal{P}$. Let $x$ be the current winner, $y$ be the given actual winner and $s(x)$ and $s(y)$ be their current respective scores. For each vote $v\in \mathcal{V}$, we compute a number $\Delta(v) = \alpha_2 - \alpha_j - \alpha_1 + \alpha_i$, where $x$ and $y$ are receiving scores $\alpha_i$ and $\alpha_j$ respectively from the vote $v$. Now, we output yes iff there are $k$ votes $v_i, 1\le i\le k$ such that, $\sum_{i=1}^k \Delta(v_i) \ge s(x) - s(y)$, which can be checked easily by sorting the $\Delta(v)$'s in nonincreasing order and checking the condition for the first $k$ $\Delta(v)$'s, where $k$ is the maximum possible coalition size specified in the input. The proof of correctness follows by exactly in the same line of argument as the proof of \Cref{CPMWScr}.
\end{proof}

For the $k$-approval voting rule, we can solve all the problems easily using max flow. Hence, from \Cref{CPMWScr} and \Cref{thm:CPMWOScr}, we have the following result.

\begin{corollary}\label{bordakapp}
 The CPMW, CPM, CPMSW, and CPMS problems for the Borda and $k$-approval voting rules are in \Pshort{}, for any coalition size.
\end{corollary}

\subsection{Maximin Voting Rule}

For the maximin voting rule, we show that all the four problems are in \Pshort{}, when we have a coalition of size one.

\begin{theorem}\label{PUMMaximinP}
 The CPMW, CPM, CPMSW, and CPMS problems for maximin voting rule are in \Pshort{} for any coalition size $k=1$ (for CPMW and CPM) or maximum possible coalition size $k=1$ (for CMPWS and CPMS).
\end{theorem}

\begin{proof}
 Given a $n$-voters' profile $\succ \in \mathcal{L(C)}^n$ and a voter $v_i$, 
 let the \textit{current winner} be $x:= r(\succ)$ and the given \textit{actual winner} be $y$. We will construct 
 $\succ^{\prime} = (\succ_i^{\prime} , \succ_{-i})$, if it exists, such that 
 $r(\succ) >_i^{\prime} r(\succ^{\prime})=y$, thus deciding whether $v_i$ is a 
 possible manipulator or not. Now, the maximin score of $x$ and $y$ in the profile $\succ^{\prime}$ can take one of 
 values from the set $\{ score_{\succ_{-i}} (x) \pm 1\}$ and $\{ score_{\succ_{-i}} (y) \pm 1\}$. 
 The algorithm is as follows. We first guess the maximin score of $x$ and $y$ in the profile $\succ^{\prime}$. 
 There are only four possible guesses. Suppose, we guessed that $x$'s score will decrease by one and $y$'s score 
 will decrease by one assuming that this guess makes $y$ win. 
 Now notice that, without loss of generality, 
 we can assume that $y$ immediately follows $x$ in the preference $\succ_i^{\prime}$ since $y$ is the winner 
 in the profile $\succ^{\prime}$. This implies that there are only $O(m)$ many possible positions for $x$ and $y$ 
 in $\succ_i^{\prime}$. We guess the position of $x$ and thus the position of $y$ in $\succ_i^{\prime}$.  Let $B(x)$ 
 and $B(y)$ be the sets of candidates with whom $x$ and respectively $y$ performs worst. Now since, $x$'s score 
 will decrease and $y$'s score will decrease, we have the following constraint on $\succ_i^{\prime}$. There must be a candidate each from $B(y)$ and $B(x)$ that will precede $x$. We do not know a-priori if there is one candidate that will serve as a witness for both $B(x)$ and $B(y)$, or if there separate witnesses. In the latter situation, we also do not know what order they appear in. Therefore we guess if there is a common candidate, and if not, we guess the relative ordering of the distinct candidates from $B(x)$ and $B(y)$.
  Now we place any candidate at the top position of $\succ_i^{\prime}$ if this action does not make $y$ lose the election. If there are many choices, we prioritize in favor of candidates from $B(x)$ and $B(y)$ --- in particular, we focus on the candidates common to $B(x)$ and $B(y)$ if we expect to have a common witness, otherwise, we favor a candidate from one of the sets according to the guess we start with. If still there are multiple choices, we 
 pick arbitrarily.  After that we move on to the next position, and do the same thing (except we stop prioritizing explicitly for $B(x)$ and $B(y)$ once we have at least one witness from each set). The other situations can be handled similarly with minor modifications. In this way, if it is able to get a 
 complete preference, then it checks whether $v_i$ is a possible manipulator or not 
 using this preference. If \textit{yes}, then it returns YES. Otherwise, it tries 
 other positions for $x$ and $y$ and other possible scores of $x$ and $y$. After 
 trying all possible guesses, if it cannot find the desired preference, 
 then it outputs NO. Since there are only polynomial 
 many possible guesses, this algorithm runs in a polynomial amount of time. 
 The proof of correctness follows from the proof of Theorem 1 in~\cite{bartholdi1989computational}.
\end{proof}

We now show that the CPMW problem for the maximin voting rule is \NPCshort{} when 
we have $k>1$. Towards that, we use the fact that the \textit{unweighted 
coalitional manipulation (UCM)} problem for the maximin voting rule is \NPCshort{} 
\cite{xia2009complexity}, when we have $k>1$. The UCM problem is as follows. Let $r$ be any voting rule.

\begin{definition}{($r$--UCM Problem)}\\
 Given a set of manipulators $M\subset \mathcal{V}$, a profile of 
 non-manipulators' vote $(\succ_i)_{i\in \mathcal{V}\setminus M}$, 
 and a candidate $z\in \mathcal{C}$, we are asked whether there exists a 
 profile of manipulators' votes $(\succ_j^{\prime})_{j\in M}$ such that 
 $r((\succ_i)_{i\in \mathcal{V}\setminus M}, (\succ_j^{\prime})_{j\in M}) = z$. 
 Assume that ties are broken in favor of $z$.
\end{definition}

We define a restricted version of the UCM problem called R-UCM as follows. 

\begin{definition}{($r$--R-UCM Problem)}\\
 This problem is the same as the UCM problem with a given guarantee - let $k:= |M|$. 
 The candidate $z$ loses pairwise election with every other candidate by $4k$ 
 votes. For any two candidates $a,b \in \mathcal{C}$, either $a$ and $b$ ties or 
 one wins pairwise election against the other one by margin of either $2k+2$ or 
 of $4k$ or of $8k$. We denote the margin by which a candidate $a$ defeats $b$, by $d(a,b)$.
\end{definition}

The R-UCM problem for the maximin voting rule is \NPCshort{} \cite{xia2009complexity}, 
when we have $k>1$.

\begin{theorem}\label{maxmincpmwnpc}
 The CPMW problem for the maximin voting rule is \NPCshort{}, for a coalition of size at least $2$. 
\end{theorem}

\begin{proof}
Clearly the CPMW problem for maximin voting rule is \NPshort{}. We provide a many-one reduction from the 
 R-UCM problem for the maximin voting rule to it. Given a R-UCM problem 
 instance, we define a CPMW problem instance 
 $\Gamma = (\mathcal{C}^{\prime},(\succ_i^{\prime})_{i\in\mathcal{V^{\prime}}},M^{\prime})$ as follows.
 \[ \mathcal{C}^{\prime}:= \mathcal{C} \cup \{w,d_1,d_2,d_3\}\]
 We define $\mathcal{V^{\prime}}$ such that $d(a,b)$ is the same as the R-UCM instance, 
 for all $a,b\in \mathcal{C}$ and $d(d_1,w)=2k+2, d(d_1,d_2)=8k, d(d_2,d_3)=8k, d(d_3,d_1)=8k$. 
 The existence of such a $\mathcal{V^{\prime}}$ is guaranteed from \Cref{thm:mcgarvey}. Moreover, 
 \Cref{thm:mcgarvey} also ensures that $|\mathcal{V^{\prime}}|$ is $O(mc)$. 
 The votes of the voters in $M$ is $w\succ \dots$. Thus the \textit{current winner} is $w$. 
 The \textit{actual winner} is defined to be $z$. The tie breaking rule is $\succ_t = w\succ z\succ \dots$, 
 where $z$ is the candidate whom the manipulators in $M$ want to make winner in the 
 R-UCM problem instance. Clearly this reduction takes polynomial amount of time. 
 Now we show that, $M$ is a coalition of possible 
 manipulators \textit{iff} $z$ can be made a winner.
 
 The \textit{if} part is as follows. Let $\succ_i, i\in M$ be the votes that make $z$ win. 
 We can assume that $z$ is the most preferred candidate in all the preferences $\succ_i, i\in M$.
 Now consider the preferences for the voters in $M$ is follows.
 \[\succ_i^{\prime}:= d_1\succ d_2\succ d_3\succ w\succ_i, i\in M\]
 The score of every candidate in $\mathcal{C}$ is not more than $z$. The score 
 of $z$ is $-3k$. The score of $w$ is $-3k-2$ and the scores of $d_1, d_2,$ and $d_3$ 
 are less than $-3k$. Hence, $M$ is a coalition of possible manipulators with the actual preferences 
 $\succ_i^{\prime}:= d_1\succ d_2\succ d_3\succ w\succ_i, i\in M$.
 
 The \textit{only if} part is as follows. Suppose $M$ is a coalition of possible manipulators 
 with actual preferences $\succ_i^{\prime}, i\in M$. Consider the preferences $\succ_i^{\prime}, i\in M$, 
 but restricted to the set $\mathcal{C}$ only. Call them $\succ_i, i\in M$. We claim that 
 $\succ_i, i\in M$ with the votes from $\mathcal{V}$ makes $z$ win the election. 
 If not then, there exists a candidate, say $a\in \mathcal{C}$, whose score is strictly more than 
 the score of $z$ - this is so because the tie breaking rule is in favor of $z$. 
 But this contradicts the fact that $z$ wins the election when the voters 
 in $M$ vote $\succ_i^{\prime}, i\in M$ along with the votes from $\mathcal{V^{\prime}}$.
\end{proof}

\subsection{Bucklin Voting Rule}

In this subsection, we design polynomial time algorithms for both the CPMW and the CPM 
problems for the Bucklin voting rule. Again, we begin by showing that if there are profiles witnessing manipulation, then there exist profiles that do so with some additional structure, which we will exploit subsequently in our algorithm.

\begin{lemma}\label{lemmaBucklin}
Consider a preference profile $(\succ_i)_{i \in \cal{V}}$, where $x$ is the winner with respect to the Bucklin voting rule. Suppose a subset of voters $M \subset \mathcal{V}$ forms a coalition of possible manipulators. Let $y$ be the actual winner. Then there exist preferences $(\succ_i^\prime)_{i \in M}$ such that $y$ is a Bucklin winner in $((\succ_i)_{i \in \cal{V} \setminus M},(\succ_i^\prime)_{i \in M})$, and further:
 \begin{enumerate}
  \item $y$ immediately follows $x$ in each $\succ_i^{\prime}$.
  \item The rank of $x$ in each $\succ_i^{\prime}$ is in one of the 
  following - first, $b(y)-1$, $b(y)$, $b(y)+1$, where $b(y)$ be the Bucklin score of $y$ in $((\succ_i)_{i \in \cal{V} \setminus M},(\succ_i^\prime)_{i \in M})$.
 \end{enumerate}
\end{lemma}

\begin{proof}
 From \Cref{cpmdef}, $y$'s rank must be worse than $x$'s rank in each $\succ_i^{\prime}$. 
 We now exchange the position of $y$ with the candidate which immediately follows $x$ in 
 $\succ_i^{\prime}$. This process does not decrease Bucklin score of any candidate except 
 possibly $y$'s, and $x$'s score does not increase. Hence $y$ will continue to win and thus  
 $\succ_i^{\prime}$ satisfies the first condition.
 
 Now to begin with, we assume that $\succ_i^{\prime}$ satisfies the first condition. 
 If the position of $x$ in $\succ_i^{\prime}$ is $b(y)-1$ or $b(y)$, we do not change it. 
 If $x$ is above $b(y)-1$ in $\succ_i^{\prime}$, then move $x$ and $y$ at the first and 
 second positions respectively. Similarly if $x$ is below $b(y)+1$ in $\succ_i^{\prime}$, 
 then move $x$ and $y$ at the $b(y)+1$ and $b(y)+2$ positions respectively. This process 
 does not decrease score of any candidate except $y$ because the Bucklin score of $x$ is at least 
 $b(y)$. The transformation cannot increase the score $y$ since its position has only been 
 improved. Hence $y$ continues to win and thus $\succ_i^{\prime}$ satisfies the second condition. 
\end{proof}

\Cref{lemmaBucklin} leads us to the following theorem.

\begin{theorem}\label{UPCMBucklin}
 The CPMW problem and the CPM problems for Bucklin voting rule are in \Pshort{} for any coalition of size. Therefore, by \Cref{cpmcpmw}, the CPMSW and the CPMS problems are in \Pshort{} when the maximum coalition size $k=O(1)$.
\end{theorem}

\begin{proof}
 \Cref{cpmcpmw} says that it is enough to prove that the CPMW problem is in \Pshort{}. 
 Let $x$ be the \textit{current winner} and $y$ be the given \textit{actual winner}. 
 For any final Bucklin score $b(y)$ of $y$, there are polynomially many 
 possibilities for the positions of $x$ and $y$ in the profile of $\succ_i, i\in M$, 
 since Bucklin voting rule is anonymous. Once the positions of $x$ and $y$ is fixed, we try 
 to fill the top $b(y)$ positions of each $\succ_i^{\prime}$ - place a candidate in an empty 
 position above $b(y)$ in any $\succ_i^{\prime}$ if doing so does not make $y$ lose the election. 
 If we are able to successfully fill the top $b(y)$ positions of all $\succ_i^{\prime}$ 
 for all $i \in M$, then $M$ is a coalition of possible manipulators. 
 If the above process fails for all possible above mentioned 
 positions of $x$ and $y$ and all possible guesses of $b(y)$, then 
 $M$ is not a coalition of possible manipulators. Clearly the above algorithm runs in \textit{poly(m,n)} 
 time.
 
 The proof of correctness is as follows. If the algorithm outputs that $M$ is 
 a coalition of possible manipulators,  
 then it actually has constructed $\succ_i^{\prime}$ for all $i \in M$ with respect 
 to which they form a coalition of possible manipulators. 
 On the other hand, if they form a coalition of possible manipulators, then \Cref{lemmaBucklin}
 ensures that our algorithm explores all the sufficient positions of $x$ and $y$ in 
 $\succ_i^{\prime}$ for all $i \in M$. Now if $M$ is a possible coalition of manipulators, then 
 the corresponding positions for $x$ and $y$ have also been searched. Our greedy
 algorithm must find it since permuting the candidates except $x$ and $z$ which are 
 ranked above $b(y)$ in $\succ_i^{\prime}$ cannot stop $y$ to win the 
 election since the Bucklin score of other candidates except $y$ is at least $b(y)$.
\end{proof}

\subsection{STV Voting Rule}

Next we prove that the CPMW and the CPM problems for STV rule is \NPCshort{}. To this end, we reduce from the Exact Cover by 3-Sets Problem (X3C), which is known to be \NPCshort{} \cite{garey1979computers}. 
The X3C problem is as follows.

\begin{definition}(X3C Problem)\\
 Given a set $S$ of cardinality $n$ and $m$ subsets $S_1,S_2, \dots, S_m \subset S$ with $|S_i|=3, 
 \forall i=1, \dots, m,$ does there exist an index set $I\subseteq \{1,\dots,m\}$ 
 with $|I|=\frac{|S|}{3}$ such that $\cup_{i\in I} S_i = S$.
\end{definition}

\begin{theorem}\label{UPCMSTVNPC}
 The CPM problem for STV rule is \NPCshort{} even for a coalition of size $1$.
\end{theorem}

\begin{proof}
 Clearly the problem is \NPshort{}. To show \NPshort{} hardness, we show a many-one reduction from 
 the X3C problem to it. The reduction is analogous to the reduction given in \cite{bartholdi1991single}. 
 Given an X3C instance, we construct an election as follows. The unspecified 
 positions can be filled in any arbitrary way. The candidate set is as follows.
 $$
 \begin{array}{rclcl}
 \mathcal{C} = \{x,y\} &\cup& \{a_1, \dots, a_m\} \cup \{\overline{a}_1, \dots, \overline{a}_m\}\\
			&\cup& \{b_1, \dots, b_m\} \cup \{\overline{b}_1, \dots, \overline{b}_m\}\\
			&\cup& \{d_0, \dots, d_n\} \cup \{g_1, \dots, g_m\}
 \end{array}
 $$
 The votes are as follows.
 \begin{itemize}
  \item $12m$ votes for $y \succ x \succ \dots$
  \item $12m-1$ votes for $x \succ y \succ \dots$
  \item $10m+\frac{2n}{3}$ votes for $d_0 \succ x \succ y \succ \dots$
  \item $12m-2$ votes for $d_i \succ x \succ y \succ \dots, \forall i\in [n]$
  \item $12m$ votes for $g_i \succ x \succ y \succ \dots, \forall i\in [m]$
  \item $6m+4i-5$ votes for $b_i \succ \overline{b}_i \succ x \succ y \succ \dots, \forall i\in [m]$
  \item $2$ votes for $b_i \succ d_j \succ x \succ y \succ \dots, \forall i\in [m], \forall j\in S_i$
  \item $6m+4i-1$ votes for $\overline{b}_i \succ b_i \succ x \succ y \succ \dots, \forall i\in [m]$
  \item $2$ votes for $\overline{b}_i \succ d_0 \succ x \succ y \succ \dots, \forall i\in [m]$
  \item $6m+4i-3$ votes for $a_i \succ g_i \succ x \succ y \succ \dots, \forall i\in [m]$
  \item $1$ vote for $a_i \succ b_i \succ g_i \succ x \succ y \succ \dots, \forall i\in [m]$
  \item $2$ votes for $a_i \succ \overline{a}_i \succ g_i \succ x \succ y \succ \dots, \forall i\in [m]$
  \item $6m+4i-3$ votes for $\overline{a}_i \succ g_i \succ x \succ y \succ \dots, \forall i\in [m]$
  \item $1$ vote for $\overline{a}_i \succ \overline{b}_i \succ g_i \succ x \succ y \succ \dots, \forall i\in [m]$
  \item $2$ votes for $\overline{a}_i \succ a_i \succ g_i \succ x \succ y \succ \dots, \forall i\in [m]$
 \end{itemize}
 The tie breaking rule is $\succ_t = \cdots \succ x$. The vote of $v$ is $x\succ \cdots$. We claim that 
 $v$ is a possible manipulator \textit{iff} the X3C is a \textit{yes} instance. 
 Notice that, of the first $3m$ candidates to be eliminated, $2m$ of them are $a_1, \dots, a_m$ and 
 $\overline{a}_1, \dots, \overline{a}_m$. Also exactly one of $b_i$ and $\overline{b}_i$ will be 
 eliminated among the first $3m$ candidates to be eliminated because if one of $b_i$, $\overline{b}_i$ 
 then the other's score exceeds $12m$.
 We show that the winner is either $x$ or $y$ irrespective of the vote of one more candidate. Let 
 $J:=\{ j: b_j \text{ is eliminated before } \overline{b}_j \}$. If $J$ is an index of set cover then the winner is 
 $y$. This can be seen as follows. Consider the situation after the first $3m$ eliminations. Let $i\in S_j$ 
 for some $j\in J$. Then $b_j$ has been eliminated and thus the score of $d_i$ is at least $12m$. 
 Since $J$ is an index of a set cover, every $d_i$'s score is at least $12m$. Notice that $\overline{b}_j$ 
 has been eliminated for all $j \notin J$. Thus the revised score of $d_0$ is at least $12m$. 
 After the first $3m$ eliminations, the remaining candidates are $x, y, \{d_i:i\in[n]\}, \{g_i:i\in[m]\}, 
 \{b_j: j\notin J\}, \{\overline{b}_j: j\in J\}$. All the remaining candidates except $x$ has score 
 at least $12m$ and $x$'s score is $12m-1$. Hence $x$ will be eliminated next which makes $y$'s score 
 at least $24m-1$. Next $d_i$'s will get eliminated which will in turn make $y$'s score $(12n+36)m-1$. 
 At this point $g_i$'s score is at most $32m$. Also all the remaining $b_i$ and $\overline{b}_i$'s score 
 is at most $32m$. Since each of the remaining candidate's scores gets transferred to $y$ once they are 
 eliminated, $y$ is the winner.
 
 Now we show that, if $J$ is not an index of set cover then the winner is $x$. This can be seen as follows. 
 If $|J|>\frac{n}{3}$, then the number of $\overline{b}_j$ that gets eliminated in the first $3m$ iterations 
 is less than $m-\frac{n}{3}$ . This makes the score 
 of $d_0$ at most $12m-2$. Hence $d_0$ gets eliminated before $x$ and all its scores gets transferred to $x$. 
 This makes the elimination of $x$ impossible before $y$ and makes $x$ the winner of the election.
 
 If $|J|\le \frac{n}{3}$ and there exists an $i\in S$ that is not covered by the corresponding set cover, then $d_i$ gets eliminated before $x$ with a score of $12m-2$ and its score gets transferred to $x$. This makes $x$ win the election.
 
 Hence $y$ can win \textit{iff} X3C is a \textit{yes} 
 instance. Also notice that if $y$ can win the election, then it can do so 
 with the voter $v$ voting a preference like $\cdots \succ x \succ y \succ \cdots$. 
\end{proof}
 
From the proof of the above theorem, we have the following corollary by specifying $y$ as the \textit{actual winner} for the CPMW problem.

\begin{corollary}\label{CPMWSTVNPC}
 The CPMW problem for STV rule is \NPCshort{} even for a coalition of size $1$.
\end{corollary}

\section{Conclusion}

In this work, we have initiated a promising research direction for detecting possible instances of manipulation in elections. We have proposed the notion of \emph{possible} manipulation and explored several concrete computational problems, which we believe to be important in the context of voting theory. These problems involve identifying if a given set of voters are possible manipulators (with or without a specified candidate winner). We have also studied the search versions of these problems, where the goal is to simply detect the presence of possible manipulation with the maximum coalition size. 
We believe there is theoretical as well as practical interest in studying the proposed problems. We have
provided algorithms and hardness results for these problems for many common voting rules. 
It is our conviction that both the problems that we have studied here have initiated an interesting research direction with significant promise and potential for future work.

In the next chapter of the thesis, we study another interesting form of election control called bribery.
\blankpage
\chapter{Frugal Bribery}
\label{chap:frugal_bribery}

\blfootnote{A preliminary version of the work in this chapter was published as \cite{frugalDeyMN16}: Palash Dey, Neeldhara Misra, and Y. Narahari. Frugal bribery in voting. In Proc. Thirtieth AAAI Conference on Artificial Intelligence, February 12-17, 2016, Phoenix, Arizona, USA., pages 2466–2472, 2016.}

\begin{quotation}
 {\small Bribery in elections is an important problem in computational
social choice theory. We introduce and study two important special cases of the classical
\textsc{\$Bribery} problem, namely, \textsc{Frugal-bribery} and \textsc{Frugal-\$bribery} where the briber is frugal in nature. By this, we mean that the briber is 
only able to influence voters who benefit from the suggestion of the briber.
More formally, a voter is {\em vulnerable} if the outcome of the election  
improves according to her own preference when she accepts the suggestion of the 
briber. In the \textsc{Frugal-bribery} problem, the goal is to make a certain
candidate win the election by changing {\em only} the vulnerable votes. In the
\textsc{Frugal-\$bribery} problem, the vulnerable votes have prices and the
goal is to make a certain candidate win the election by changing only the 
vulnerable votes, subject to a budget constraint. We further
formulate two natural variants of the \textsc{Frugal-\$bribery} problem namely
\textsc{Uniform-frugal-\$bribery} and \textsc{Nonuniform-frugal-\$bribery}
where the prices of the vulnerable votes are, respectively, all the same or different.

We observe that, even if we have only a small number of
candidates, the problems are intractable for all voting rules studied here for
weighted elections, with the sole exception of the \textsc{Frugal-bribery}
problem for the plurality voting rule. In contrast, we have polynomial time
algorithms for the \textsc{Frugal-bribery} problem for plurality, veto,
$k$-approval, $k$-veto, and plurality with runoff voting rules for unweighted
elections. However, the \textsc{Frugal-\$bribery} problem is intractable for
all the voting rules studied here barring the plurality and the veto voting
rules for unweighted elections.
These intractability results demonstrate that bribery is a hard computational 
problem, in the sense that several special cases of this problem continue to be
computationally intractable. This strengthens the view that bribery, although a possible attack on an election in principle, may be infeasible in practice.}
\end{quotation}

\section{Introduction}

Activities that try to influence voter opinions, in favor of specific candidates, are very common during the time that an election is in progress. For example, in a political election, candidates often conduct elaborate campaigns to promote themselves among a general or targeted audience. Similarly, it is not uncommon for people to protest against, or rally for, a national committee or court that is in the process of approving a particular policy. An extreme illustration of this phenomenon is \emph{bribery} --- here, the candidates may create financial incentives to sway the voters. Of course, the process of influencing voters may involve costs even without the bribery aspect; for instance, a typical political campaign or rally entails considerable expenditure. 

All situations involving a systematic attempt to influence voters usually have the following aspects: an external agent, a candidate that the agent would like to be the winner, a budget constraint, a cost model for a change of vote, and knowledge of the existing election. The formal computational problem that arises from these inputs is the following: is it possible to make a distinguished candidate win the election in question by incurring a cost that is within the budget? This question, with origins in~\cite{faliszewski2006complexity,faliszewski2009hard,FaliszewskiHHR09}, has been subsequently studied intensely in computational social choice literature. 
In particular, bribery has been studied under various cost models, for example, uniform price per 
vote which is known as \textsc{\$Bribery}~\cite{faliszewski2006complexity}, nonuniform price per vote~\cite{faliszewski2008nonuniform},  nonuniform price per shift of the distinguished candidate per vote which is called \textsc{Shift bribery}, nonuniform price per swap of candidates 
per vote which is called \textsc{Swap bribery}~\cite{elkind2009swap}. A closely related problem known as campaigning has been studied for various vote models, for example, truncated ballots~\cite{BaumeisterFLR12}, 
soft constraints~\cite{pini2013bribery}, CP-nets~\cite{dorn2014hardness}, combinatorial domains~\cite{mattei2012bribery} and 
probabilistic lobbying~\cite{erdelyi2009complexity}. 
The bribery problem has also been studied under 
voting rule uncertainty~\cite{erdelyi2014bribery}. Faliszewski et al.~\cite{faliszewski2014complexity} 
study the complexity of bribery in Bucklin and Fallback voting rules. Xia~\cite{xia2012computing} studies destructive bribery, 
where the goal of the briber is to change the winner by changing minimum number of votes. 
Dorn et al.~\cite{dorn2012multivariate} studies the 
parameterized complexity of the \textsc{Swap Bribery} problem and Bredereck et al.~\cite{bredereck2014prices} explores the parameterized 
complexity of the \textsc{Shift Bribery} problem for a wide range of parameters.
We recall again that the costs and the budgets involved in all the bribery problems above need not necessarily 
correspond to actual money traded between voters and candidates. They may correspond to any cost in general, for example, the amount of effort or time that the briber needs to spend for each voter.

\subsection{Motivation}
In this work, we propose an effective cost model for the bribery problem. Even the most general cost models that have been studied in the literature fix absolute costs per voter-candidate combination, with no specific consideration to the voters' opinions about the current winner and the distinguished candidate whom the briber wants to be the winner. In our proposed model, a change of vote is relatively easier to effect if the change causes an outcome that the voter would find desirable. Indeed, if the currently winning candidate is, say, $a$, and a voter is (truthfully) promised that by changing her vote from $c \succ d \succ b \succ a$ to $d \succ b \succ c \succ a$, the winner of the election would change from $a$ to $d$, then this is a change that the voter is likely to be  happy to make. While the change does not make her most favorite candidate win the election, it does improve the result from her point of view. Thus, given the circumstances (namely that of her least favorite candidate winning the election), the altered vote serves the voter better than the original one. 

We believe this perspective of voter influence is an important one to study. The cost of a change of vote is proportional to the nature of the outcome that the change promises --- the cost is low or nil if the change results in a better outcome with respect to the voter's original ranking, and high or infinity otherwise. A frugal agent only approaches voters of the former category, thus being able to effectively bribe with minimal or no cost. Indeed the behavior of agents in real life is often frugal. For example, consider campaigners in favor of a relatively smaller party in a political election. They may actually target only vulnerable voters due to lack of human and other resources they have at their disposal.

More formally, let $c$ be the winner of an election and $p$ (other than $c$) the candidate whom the briber 
wishes to make the winner of the election. Now the voters who prefer $c$ to $p$ will be reluctant to change their votes, and we call these votes {\em non-vulnerable with respect to $p$} --- we do not allow these votes to be changed by the briber, which justifies the {\em frugal} nature of the briber. On the other hand, if a voter prefers $p$ to $c$, then it may be very easy to convince her to change her vote if doing so makes $p$ win the election. We name these votes {\em vulnerable with respect to $p$}. When the candidate $p$ is clear from the context, we simply call these votes non-vulnerable and vulnerable, respectively. 

The computational problem is to determine whether there is a way to make a candidate $p$ win the election by changing \emph{only} those votes that are vulnerable with respect to $p$. We call this problem \textsc{Frugal-bribery}. Note that there is no cost involved in the \textsc{Frugal-bribery} problem --- the briber does not incur any cost to change the votes of the vulnerable votes. We also extend this basic model to a more general setting where each vulnerable vote has a certain nonnegative integer price which may correspond to the effort involved in approaching these voters and convincing them to change their votes. We also allow for the specification of a budget constraint, which can be used to enforce auxiliary constraints. This leads us to define the \textsc{Frugal-\$bribery} problem, where we are required to find a subset of vulnerable votes with a total cost that is within a given budget, such that these votes can be changed in some way to make the candidate $p$ win the election. Note that the \textsc{Frugal-\$bribery} problem can be either uniform or nonuniform depending on whether the prices of the vulnerable votes are all identical or different. If not mentioned otherwise, the prices of the vulnerable votes will be assumed to be nonuniform. 
We remind that the briber is not allowed to change the non-vulnerable votes in both the \textsc{Frugal-bribery} and the \textsc{Frugal-\$bribery} problems.

\subsection{Our Contribution}
Our primary contribution in this work is to formulate and study two important and natural models of bribery which turn out to be special cases of the well studied \textsc{\$Bribery} problem in elections. Our results show that both the \textsc{Frugal-bribery} and the \textsc{Frugal-\$bribery} problems are intractable for many commonly used voting rules for weighted as well as unweighted elections, barring a few exceptions. These intractability results can be interpreted as an evidence that the bribery in elections is a hard computational problem in the sense that even many of its important and natural special cases continue to be intractable. Thus bribery, although a possible attack on elections in principle, may be practically not viable. From a more theoretical perspective, our intractability results strengthen the existing hardness results for the \textsc{\$Bribery} problem. On the other hand, our polynomial time algorithms exhibit interesting tractable special cases of the \textsc{\$Bribery} problem.

\subsubsection*{Our Results for Unweighted Elections} ~We have the following results for unweighted elections.
 \begin{itemize}
  \item The \textsc{Frugal-bribery} problem is in \Pshort{} for the $k$-approval, Bucklin, and plurality with runoff voting rules. Also, the \textsc{Frugal-\$bribery} problem is in \Pshort{} for the plurality and veto voting rules. In contrast, the \textsc{Frugal-\$bribery} problem is \NPshort{}-complete for the Borda, maximin, Copeland, and STV voting rules [\Cref{lem:frugalNPC}].
  \item The \textsc{Frugal-bribery} problem is \NPC{} for the Borda voting rule [\Cref{thm:frugalBordaNPC}]. The \textsc{Frugal-\$bribery} is \NPC{} for the $k$-approval for any constant $k \ge 5$ [\Cref{thm:frugalKappNPC}], $k$-veto for any constant $k \ge 3$ [\Cref{thm:frugalKvetoNPC}], and a wide class of scoring rules [\Cref{thm:frugalScrNPC}] even if the price of every vulnerable vote is either $1$ or $\infty$. 
  Moreover, the \textsc{Uniform-frugal-\$bribery} is \NPC{} for the Borda voting rule even if all the vulnerable votes have a uniform
  price of $1$ and the budget is $2$ [\Cref{thm:uniformBordaNPC}].
  \item The \textsc{Frugal-\$bribery} problem is in \Pshort{} for the $k$-approval, Bucklin, and plurality with runoff voting rules when the budget is a constant [\Cref{thm:frugalKappP}]. 
 \end{itemize}
 
\subsubsection*{Our Results for Weighted Elections}
~We have the following results for weighted elections.
\begin{itemize}
 \item The \textsc{Frugal-bribery} problem is in \Pshort{} for the maximin and Copeland voting rules when we have only three candidates [\Cref{lem:wfrugalEasy}], and for the plurality voting rule for any number of candidates [\Cref{thm:wfrugalP}].
 \item The \textsc{Frugal-bribery} problem is \NPC{} for the STV [\Cref{thm:stv_wt}], plurality with runoff [\Cref{cor:run_wt}], and every scoring rule except the plurality voting rule [\Cref{lem:wfrugalScr}] for three candidates. The \textsc{Frugal-\$bribery} problem is \NPC{} for the plurality voting rule for three candidates [\Cref{thm:wfrugalPluNPC}]. 
 \item When we have only four candidates, the \textsc{Frugal-bribery} problem is \NPC{} for the maximin [\Cref{thm:wfrugalMaxmin}], Bucklin [\Cref{thm:bucklin_wt}], and Copeland [\Cref{thm:wfrugalCopeland}] rules.
\end{itemize}
We summarize the results in the \Cref{tbl:frugal_summary}.

\begin{table}[htbp]\centering
  \resizebox{\linewidth}{!}{
\renewcommand*{\arraystretch}{2}
 \begin{tabular}{|c|c|c|c|c|}\hline 
  \multirow{2}{*}{Voting Rules}			& \multicolumn{2}{c|}{Unweighted}	&\multicolumn{2}{c|}{Weighted}	\\\cline{2-5}
						& \textsc{Frugal-bribery} & \textsc{Frugal-\$bribery}& \textsc{Frugal-bribery} & \textsc{Frugal-\$bribery}\\\hline\hline
  Plurality & \makecell{\Pshort{}\\\relax[\Cref{lem:frugalP}]} & \makecell{\Pshort{}\\\relax[\Cref{lem:frugalPluP}]} & \makecell{\Pshort{}\\\relax[\Cref{thm:wfrugalP}]} & \makecell{\NPC{}\\\relax[\Cref{thm:wfrugalPluNPC}]}\\\hline
  
  Veto & \makecell{\Pshort{}\\\relax[\Cref{lem:frugalP}]} & \makecell{\Pshort{}\\\relax[\Cref{lem:frugalPluP}]} & \makecell{\NPC{}\\\relax[\Cref{lem:wfrugalScr}]} & \makecell{\NPC{}\\\relax[\Cref{lem:wfrugalScr}]}\\\hline
  
  $k$-approval & \makecell{\Pshort{}\\\relax[\Cref{lem:frugalP}]} & \makecell{\NPC{}$^{\star}$\\\relax[\Cref{thm:frugalKappNPC}]} & \makecell{\NPC{}$^{\diamond}$\\\relax[\Cref{lem:wfrugalScr}]} & \makecell{\NPC{}\\\relax[\Cref{lem:wfrugalScr}]}\\\hline
  
  $k$-veto & \makecell{\Pshort{}\\\relax[\Cref{lem:frugalP}]} & \makecell{\NPC{}$^{\bullet}$\\\relax[\Cref{thm:frugalKvetoNPC}]} & \makecell{\NPC{}$^{\diamond}$\\\relax[\Cref{lem:wfrugalScr}]} & \makecell{\NPC{}\\\relax[\Cref{lem:wfrugalScr}]}\\\hline
  
  Borda & \makecell{\NPC{}\\\relax[\Cref{thm:frugalBordaNPC}]} & \makecell{\NPC{}$^{\dagger}$\\\relax[\Cref{thm:frugalScrNPC}]} & \makecell{\NPC{}\\\relax[\Cref{lem:wfrugalScr}]} & \makecell{\NPC{}\\\relax[\Cref{lem:wfrugalScr}]}\\\hline
  
  Runoff & \makecell{\Pshort{}\\\relax[\Cref{lem:frugalP}]} & ? & \makecell{\NPC{}\\\relax[\Cref{cor:run_wt}]} & \makecell{\NPC{}\\\relax[\Cref{cor:run_wt}]}\\\hline
  
  Maximin & ? & \makecell{\NPC{}\\\relax[\Cref{lem:frugalNPC}]} & \makecell{\NPC{}\\\relax[\Cref{thm:wfrugalMaxmin}]} & \makecell{\NPC{}\\\relax[\Cref{thm:wfrugalMaxmin}]}\\\hline
  
  Copeland & ? & \makecell{\NPC{}\\\relax[\Cref{lem:frugalNPC}]} & \makecell{\NPC{}\\\relax[\Cref{thm:wfrugalCopeland}]} & \makecell{\NPC{}\\\relax[\Cref{thm:wfrugalCopeland}]}\\\hline
  
  STV & ? & \makecell{\NPC{}\\\relax[\Cref{lem:frugalNPC}]} & \makecell{\NPC{}\\\relax[\Cref{thm:stv_wt}]} & \makecell{\NPC{}\\\relax[\Cref{thm:stv_wt}]}\\\hline
\end{tabular}}
  \caption{${\star}$- The result holds for $k\ge5$. \newline\hspace{\linewidth}$\bullet$- The result holds for $k\ge3$. \newline\hspace{\linewidth}$\dagger$- The result holds for a much wider class of scoring rules. \newline\hspace{\linewidth}${\diamond}$- The results do not hold for the plurality voting rule. \newline\hspace{\linewidth}?- The problem is open.}\label{tbl:frugal_summary}
\end{table}

\subsection{Related Work} The pioneering work of~\cite{faliszewski2006complexity} defined and studied the \textsc{\$Bribery} problem wherein, the input is a set of votes with prices for each vote and the goal is to 
make some distinguished candidate win the election, subject to a budget constraint of the briber. 
The \textsc{Frugal-\$bribery} problem is the \textsc{\$Bribery} 
problem with the restriction that the price of every non-vulnerable vote is infinite. 
Also, the \textsc{Frugal-bribery} problem is a special case of the \textsc{Frugal-\$bribery} problem. 
Hence, whenever the \textsc{\$Bribery} problem is computationally easy in a setting, both the \textsc{Frugal-bribery} and the 
\textsc{Frugal-\$bribery} problems are also computationally easy (see \Cref{prop:conn} for a more formal proof). 
However, the \textsc{\$Bribery} problem is computationally intractable in most of the settings. This makes the study of important special cases such as \textsc{Frugal-bribery} and \textsc{Frugal-\$bribery}, interesting. We note that a notion similar to vulnerable votes has been studied in the context of dominating manipulation by \cite{conitzer2011dominating}. Hazon et al.~\cite{hazon2013change} introduced and studied \textsc{Persuasion} and $k$-\textsc{Persuasion} problems where an external agent suggests votes to vulnerable voters which are beneficial for the vulnerable voters as well as the external agent. It turns out that the \textsc{Persuasion} and the $k$-\textsc{Persuasion} problems Turing reduce to the \textsc{Frugal-bribery} and the \textsc{Frugal-\$bribery} problems respectively (see \Cref{prop:persu}). Therefore, the polynomial time algorithms we propose in this work imply polynomial time algorithms for the persuasion analog. On the other hand, since the reduction in~\Cref{prop:persu} from \textsc{Persuasion} to \textsc{Frugal-bribery} is a Turing reduction, the existing \NP{}-completeness results for the persuasion problems do not imply \NP{}-completeness results for the corresponding frugal bribery variants. We refer to~\cite{rogers1967theory} for Turing reductions.

\section{Problem Definition}

In all the definitions below, $r$ is a fixed voting rule. We define the notion of vulnerable votes as follows.
Intuitively, the vulnerable votes are those votes whose voters can easily be persuaded to change their votes since doing so will result in an outcome that those voters prefer over the current one.
\begin{definition}(Vulnerable votes)\\
 Given a voting rule $r$, a set of candidates $\CC$, a profile of votes $\succ = (\succ_1, \ldots, \succ_n)$, and a distinguished candidate $p$, we say a vote $\succ_i$ is $p$-vulnerable if $p\succ_i r(\succ)$.
\end{definition} 
 Recall that, whenever the distinguished candidate is clear from the context, we drop it from the notation. With the above definition of vulnerable votes, we formally define the \textsc{Frugal-bribery} problem as follows. Intuitively, the problem is to determine whether a particular candidate can be made winner by changing only the vulnerable votes.
\begin{definition}($r$-\textsc{Frugal-bribery})\\
 Given a preference profile $\succ = (\succ_1, \ldots, \succ_n)$ over a candidate set $\CC$, and a candidate $p$, determine
if there is a way to make $p$ win the election by changing only the vulnerable votes.
\end{definition}
 Next we generalize the \textsc{Frugal-bribery} problem to the \textsc{Frugal-\$bribery} problem which involves prices for the 
vulnerable votes and a budget for the briber. Intuitively, the price of a vulnerable vote $v$ is the cost the briber incurs to change the vote $v$.
\begin{definition}($r$-\textsc{Frugal-\$bribery})\\
Let $\succ = (\succ_1, \ldots, \succ_n)$ be a preference profile over a candidate set $\CC$. We are given a candidate $p$, a finite budget $b\in\mathbb{N}$, and a price function $c:[n]\longrightarrow \mathbb{N}\cup\{\infty\}$ such that $c(i)  = \infty$ if $\succ_i$ is not a $p$-vulnerable vote. The goal is to determine if there exist $p$ vulnerable votes $\succ_{i_1}, \ldots, \succ_{i_\ell}\in\succ$ and votes $\succ_{i_1}^\prime, \ldots, \succ_{i_\ell}^\prime\in\mathcal{L}(C)$ such that:

\begin{enumerate}
	\item[(a)] the total cost of the chosen votes is within the budget, that is, $\sum_{j=1}^\ell c(i_j) \le b$, and
	\item[(b)] the new votes make the desired candidate win, that is, $r(\succ_{[n]\setminus\{i_1, \ldots, i_\ell\}}, \succ_{i_1}^\prime, \ldots, \succ_{i_\ell}^\prime) = p$. 
\end{enumerate}  

The special case of the problem when the prices of all the vulnerable votes are the same is called \textsc{Uniform-frugal-\$bribery}. We refer to the general version as \textsc{Nonuniform-frugal-\$bribery}. If not specified, \textsc{Frugal-\$bribery} refers to the nonuniform version. 
\end{definition}
 The above problems are important special cases of the well studied \textsc{\$Bribery} problem. Also, the \textsc{Coalitional-manipulation} problem~\cite{bartholdi1989computational,conitzer2007elections}, one of the classic problems in computational social choice theory, turns out to be a special case of the \textsc{Frugal-\$bribery} problem [see \Cref{prop:conn}].

For the sake of completeness, we include the definitions of these problems here.
\begin{definition}($r$-\textsc{\$Bribery})~\cite{faliszewski2009hard}\\
 Given a preference profile $\succ = (\succ_1, \ldots, \succ_n)$ over a set of candidates \CC, a distinguished candidate $p$, a price function $c:[n]\longrightarrow \mathbb{N}\cup\{\infty\}$, and a budget $b\in\mathbb{N}$, determine if there a way to make $p$ win the election.
\end{definition}
\begin{definition}(\textsc{Coalitional-manipulation})~\cite{bartholdi1989computational,conitzer2007elections}\\
 Given a preference profile $\succ^t = (\succ_1, \ldots, \succ_n)$ of truthful voters over a set of candidates \CC, an integer $\el$, and a distinguished candidate $p$, determine if there exists a $\el$ voter preference profile $\succ^\el$ such that the candidate $p$ wins uniquely (does not tie with any other candidate) in the profile $(\succ^t,\succ^\el)$.
\end{definition}
The following proposition shows the relationship among the above problems.
\Cref{prop:conn,prop:man_bri_con,prop:persu} below hold for both weighted and unweighted elections.
\begin{proposition}\label{prop:conn}
 For every voting rule, \textsc{Frugal-bribery} $\le_\Pshort{}$ \textsc{Uniform-frugal-\$bribery} $\le_\Pshort{}$ \textsc{Nonuniform-frugal-\$bribery} 
 $\le_\Pshort{}$ \textsc{\$Bribery}. Also, \textsc{Coalitional-manipulation} $\le_\Pshort{}$ \textsc{Nonuniform-frugal-\$bribery}.
\end{proposition}

\begin{proof}
 In the reductions below, let us assume that the election to start with is a weighted election. Since we do not change the weights of any vote in the reduction and since there is a natural one to one correspondence between the votes of the original instance and the reduced instance, the proof also works for unweighted elections.
 
 Given a \textsc{Frugal-bribery} instance, we construct a \textsc{Uniform-frugal-\$bribery} instance by defining 
 the price of every vulnerable vote to be zero and the budget to be zero. Clearly, the two instances are equivalent. Hence, 
  \textsc{Frugal-bribery} $\le_\Pshort{}$ \textsc{Uniform-frugal-\$bribery}.
  
  \textsc{Uniform-frugal-\$bribery} $\le_\Pshort{}$ \textsc{Nonuniform-frugal-\$bribery} $\le_\Pshort{}$ \textsc{\$Bribery} follows from the 
  fact that  \textsc{Uniform-frugal-\$bribery} is a special case of \textsc{Nonuniform-frugal-\$bribery} which in turn is a 
  special case of \textsc{\$Bribery}.
  
  Given a \textsc{Coalitional-manipulation} instance, we construct a \textsc{Nonuniform-frugal-\$bribery} instance as follows. 
  Let $p$ be the distinguished candidate of the manipulators and $ \succ_f = p\succ others$ be any arbitrary but 
  fixed ordering of the candidates given in the \textsc{Coalitional-manipulation} instance. 
  Without loss of generality, we can assume that $p$ does not win if all the manipulators vote $\succ_f$ (Since, this is a 
  polynomially checkable case of \textsc{Coalitional-manipulation}). We define the vote of the manipulators to be $\succ_f$, 
  the distinguished candidate of the campaigner to be $p$, the budget of the campaigner to be zero, 
  the price of the manipulators to be zero (notice that all the manipulators' votes are $p$-vulnerable), and the price of the rest of the vulnerable votes to be one. Clearly, the 
  two instances are equivalent. Hence, \textsc{Coalitional-manipulation} $\le_\Pshort{}$ \textsc{Nonuniform-frugal-\$bribery}.
\end{proof}
Also, the \textsc{Frugal-bribery} problem reduces to the \textsc{Coalitional-manipulation} problem by simply making all vulnerable votes to be manipulators.
\begin{proposition}\label{prop:man_bri_con}
 For every voting rule, \textsc{Frugal-bribery} $\le_\Pshort{}$ \textsc{Coalitional-manipulation}.
\end{proposition}
We can also establish the following relationship between the \textsc{Persuasion} (respectively $k$-\textsc{Persuasion}) problem and the \textsc{Frugal-bribery} (respectively \textsc{Frugal-\$bribery}) problem. The persuasions differ from the corresponding frugal bribery variants in that the briber has her own preference order, and desires to improve the outcome of the election with respect to her preference order. The following proposition is immediate from the definitions of the problems. 
\begin{proposition}\label{prop:persu}
 For every voting rule, there is a Turing reduction from \textsc{Persuasion} (respectively \textsc{$k$-persuasion}) to \textsc{Frugal-bribery} (respectively \textsc{Frugal-\$bribery}).
\end{proposition}

\begin{proof}
 Given an algorithm for the \textsc{Frugal-bribery} problem, we iterate over all possible distinguished candidates to have an algorithm for the persuasion problem. 
 
 Given an algorithm for the \textsc{Frugal-\$bribery} problem, we iterate over all possible distinguished candidates and fix the price of the corresponding vulnerables to be one to have an algorithm for the $k$-persuasion problem.
\end{proof}

\section{Results for Unweighted Elections}

Now we present the results for unweighted elections. We begin with some easy observations that follow from known results. 

\begin{observation}\label{lem:frugalP}
 The \textsc{Frugal-bribery} problem is in \Pshort{} for the $k$-approval voting rule for any $k$, Bucklin, and plurality with runoff voting rules.
\end{observation}

\begin{proof}
 The \textsc{Coalitional-manipulation} problem is in \Pshort{} for these voting rules~\cite{xia2009complexity}. Hence, the result follows from~\Cref{prop:man_bri_con}.
\end{proof}

\begin{observation}\label{lem:frugalPluP}
 The \textsc{Frugal-\$bribery} problem is in \Pshort{} for the plurality and veto voting rules.
\end{observation}

\begin{proof}
 The \textsc{\$Bribery} problem is in \Pshort{} for the plurality~\cite{faliszewski2006complexity} and 
 veto~\cite{faliszewski2008nonuniform} voting rules. Hence, the result follows from~\Cref{prop:conn}.
\end{proof}

\begin{observation}\label{lem:frugalNPC}
 The \textsc{Frugal-\$bribery} problem is \NPC{} for Borda, maximin, Copeland, and STV voting rules.
\end{observation}

\begin{proof}
 The \textsc{Coalitional-manipulation} problem is \NPC{} for the above voting rules. Hence, the result follows from~\Cref{prop:conn}.
\end{proof}

 We now present our main results. We begin with showing that the \textsc{Frugal-bribery} problem for the Borda voting rule and the \textsc{Frugal-\$bribery} problem for various scoring rules are \NPC{}. To this end, we reduce from the \PS problem, which is known to be \NPC{}~\cite{Yu:2004:MMT:1013651.1013680}. The \PS problem is defined as follows.

\defproblem{\PS}{$n$ integers $X_i, i\in[n]$ with $1\le X_i\le 2n$ for every $i\in[n]$ and $\sum_{i=1}^n X_i = n(n+1)$.}{Do there exist two permutations $\pi$ and $\sigma$ of $[n]$ such that $\pi(i)+\sigma(i)=X_i$ for every $i\in[n]?$}


We now prove that the \textsc{Frugal-bribery} problem is \NPC{} for the Borda voting rule, by a reduction from \PS. Our reduction is inspired by the reduction used by Davies et al.~\cite{davies2011complexity} and Betzler et al.~\cite{betzler2011unweighted} to prove \NP-completeness of the \textsc{Coalitional-manipulation} problem for the Borda voting rule.

\begin{theorem}\label{thm:frugalBordaNPC}
 The \textsc{Frugal-bribery} problem is \NPC{} for the Borda voting rule.
\end{theorem}

\begin{proof}
 The problem is clearly in \NPshort{}. To show \NPshort{}-hardness, we reduce an arbitrary instance of the \PS problem to the \textsc{Frugal-bribery} problem for the Borda voting rule. Let $(X_1, \ldots, X_n)$ be an instance of the \PS problem. Without loss of generality, let us assume that $n$ is an odd integer -- if $n$ is an even integer, then we consider the instance $(X_1, \ldots, X_n, X_{n+1}=2(n+1))$ which is clearly equivalent to the instance $(X_1, \ldots, X_n).$
 
 We define a \textsc{Frugal-bribery} instance $(\CC, \PP, p)$ as follows. The candidate set is: 
 
 $$\mathcal{C} = \XX \uplus D\uplus \{p,c\}, \text{ where } \XX =  \{\xxx_i: i\in[n]\} \text{ and }|D| = 3n - 1$$ 
 
 Note that the total number of candidates is $4n+1$, and therefore the Borda score of a candidate placed at the top position is $4n$.
 
Before describing the votes, we give an informal overview of how the reduction will proceed. The election that we define will consist of exactly two vulnerable votes. Note that when placed at the top position in these two votes, the distinguished candidate $p$ gets a score of $8n$ ($4n$ from each vulnerable vote). We will then add non-vulnerable votes, which will be designed to ensure that, among them, the score of $\xxx_i$ is $8n-X_i$ more than the score of the candidate $p$. Using the ``dummy candidates'', we will also be able to ensure that the candidates $\xxx_i$ receive (without loss of generality) scores between $1$ and $n$ from the modified vulnerable votes. 

Now suppose these two vulnerable votes can be modified to make $p$ win the election. Let $s_1$ and $s_2$ be the scores that $\xxx_i$ obtains from these altered vulnerable votes. It is clear that for $p$ to emerge as a winner, $s_1+s_2$ must be at most $X_i$. Since the Borda scores for the candidates in \XX range from $1$ to $n$ in the altered vulnerable votes, the total Borda score that all the candidates in \XX can accumulate from two altered vulnerable votes is $n(n+1)$. On the other hand, since the sum of the $X_i$'s is also $n(n+1)$, it turns out that $s_1+s_2$ must in fact be equal to $X_i$ for the candidate $p$ to win. From this point, it is straightforward to see how the permutations $\sigma$ and $\pi$ can be inferred from the modified vulnerable votes: $\sigma(i)$ is given by the score of the candidate $\xxx_i$ from the first vote, while $\pi(i)$ is the score of the candidate $\xxx_i$ from the second vote. These functions turn out to be permutations because these $n$ candidates receive $n$ distinct scores from these votes. 
 
We are now ready to describe the construction formally. We remark that instead of $8n-X_i$, as described above, we will maintain a score difference of either $8n-X_i$ or $8n-X_i-1$ depending on whether $X_i$ is even or odd respectively --- this is a minor technicality that comes from the manner in which the votes are constructed and does not affect the overall spirit of the reduction. 

Let us fix any arbitrary order $\succ_f$ among the candidates in $\XX \uplus D.$ For any subset $A\subset \XX \uplus D,$ let $\overrightarrow{A}$ be the ordering among the candidates in $A$ as defined in $\succ_f$ and $\overleftarrow{A}$ the reverse order of $\overrightarrow{A}$. For each $i\in[n]$, we add two votes $v_i^j$ and $v_i^{j^\pr}$ as follows for every $j\in[4]$. Let $\el$ denote $|D| = 3n-1$. Also, for $d\in D$, let $D_{i}, D_{\nfrac{\el}{2}}\subset D\setminus\{d\}$ be such that: $$|D_{i}| = \nfrac{\el}{2}+n+1-\lceil\nfrac{X_i}{2}\rceil \mbox{ and } |D_{\nfrac{\el}{2}}| = \nfrac{\el}{2}.$$

 \[ v_i^j : 
 \begin{cases}
 c \suc p \suc d \suc \overrightarrow{\CC\setminus(\{d, c, p, \xxx_i\}\uplus D_{i})} \suc \xxx_i \suc  \overrightarrow{D_{i}} & \text{for } 1\le j\le 2\\
 \xxx_i \suc \overleftarrow{D_{i}} \suc \overleftarrow{\CC\setminus(\{d, c, p, \xxx_i\}\uplus D_{i})} \suc c \suc p \suc d & \text{for } 3\le j\le 4
 \end{cases}
 \]
 
 \[ v_i^{j^\pr} : 
 \begin{cases}
 c \suc p \suc d \suc \overrightarrow{\CC\setminus(\{d, c, p, \xxx_i\}\uplus D_{\nfrac{\el}{2}})} \suc \xxx_i \suc  \overrightarrow{D_{\nfrac{\el}{2}}} & \text{for } 1\le j^\pr \le 2\\
 \xxx_i \suc \overleftarrow{D_{\nfrac{\el}{2}}} \suc \overleftarrow{\CC\setminus(\{d, c, p, \xxx_i\}\uplus D_{\nfrac{\el}{2}})} \suc c \suc p \suc d & \text{for } 3\le j^\pr \le 4
 \end{cases}
 \]
 
It is convenient to view the votes corresponding to $j = 3,4$ as a near-reversal of the votes in $j = 1,2$ (except for candidates $c,d$ and $\xxx_i$).  Let $\PP_1 = \{v_i^j, v_i^{j^\pr} : i\in[n], j\in[4]\}.$ Since there are $8n$ votes in all, and $c$ always appears immediately before $p$, it follows that the score of $c$ is exactly $8n$ more than the score of the candidate $p$ in $\PP_1$. 

We also observe that the score of the candidate $\xxx_i$ is exactly $2(\el+n+1)-X_i = 8n-X_i$ more than the score of the candidate $p$ in $\PP_1$ for every $i\in[n]$ such that $X_i$ is an even integer. On the other hand, the score of the candidate $\xxx_i$ is exactly $2(\el+n+1)-X_i-1 = 8n-X_i-1$ more than the score of the candidate $p$ in $\PP_1$ for every $i\in[n]$ such that $X_i$ is an odd integer. Note that for $i^\pr \in [n] \setminus \{i\}$, $p$ and $\xxx_i$ receive the same Borda score from the votes $v_{i^\pr}^j$ and $v_{i^\pr}^{j^\pr}$ (where $j,j^\pr \in [4]$).

We now add the following two votes $\mu_1$ and $\mu_2$. 
 
\[ \mu_1 : p \succ c \succ \text{others} \]
\[ \mu_2 : p \succ c \succ \text{others} \]

 Let $\PP = \PP_1 \uplus \{\mu_1, \mu_2\}, \XX^o = \{\xxx_i : i\in[n], X_i \text{ is odd}\},$ and $\XX^e = \XX\setminus\XX^o.$ We recall that the distinguished candidate is $p.$ The tie-breaking rule is according to the order $\XX^o \suc p \succ \text{others}.$ We claim that the \textsc{Frugal-bribery} instance $(\CC, \PP, p)$ is equivalent to the \PS instance $(X_1, \ldots, X_n).$
 
 In the forward direction, suppose there exist two permutations $\pi$ and $\sigma$ of $[n]$ such that $\pi(i) + \sigma(i) = X_i$ for every $i\in[n].$ We replace the votes $\mu_1$ and $\mu_2$ with respectively $\mu_1^\pr$ and $\mu_2^\pr$ as follows.
 
 \[ \mu_1^\pr : p \suc D \suc \xxx_{\pi^{-1}(n)} \suc \xxx_{\pi^{-1}(n-1)} \suc \cdots \suc \xxx_{\pi^{-1}(1)} \suc c \]
 \[ \mu_2^\pr : p \suc D \suc \xxx_{\sigma^{-1}(n)} \suc \xxx_{\sigma^{-1}(n-1)} \suc \cdots \suc \xxx_{\sigma^{-1}(1)} \suc c \]
 
 We observe that, the candidates $c$ and every $\xxx\in\XX^e$ receive same score as $p$, every candidate $\xxx^\pr\in\XX^o$ receives $1$ score less than $p$, and every candidate in $D$ receives less score than $p$ in $\PP_1\uplus\{\mu_1^\pr, \mu_2^\pr\}.$ Hence $p$ wins in $\PP_1\uplus\{\mu_1^\pr, \mu_2^\pr\}$ due to the tie-breaking rule. Thus $(\CC, \PP, p)$ is a \YES instance of \textsc{Frugal-bribery}.
 
 To prove the other direction, suppose the \textsc{Frugal-bribery} instance is a \YES{} instance. Notice that the only vulnerable votes are $\mu_1$ and $\mu_2.$ Let $\mu_1^\pr$ and $\mu_2^\pr$ be two votes such that the candidate $p$ wins in the profile $\PP_1\uplus\{\mu_1^\pr, \mu_2^\pr\}.$ We assume, without loss of generality, that candidate $p$ is placed at the first position in both $\mu_1^\pr$ and $\mu_2^\pr.$ Since $c$ receives $8n$ scores more than $p$ in $\PP_1,$ $c$ must be placed at the last position in both $\mu_1^\pr$ and $\mu_2^\pr$ since otherwise $p$ cannot win in $\PP_1\uplus\{\mu_1^\pr, \mu_2^\pr\}.$ We also assume, without loss of generality, that every candidate in $D$ is preferred over every candidate in \XX since otherwise, if $\xxx \suc d$ in either $\mu_1^\pr$ or $\mu_2^\pr$ for some $\xxx\in\XX$ and $d\in D,$ then we can exchange the positions of \xxx and $d$ and $p$ continues to win since no candidate in $D$ receives more score than $p$ in $\PP_1.$ Hence, every $\xxx\in\XX$ receives some score between $1$ and $n$ in both the $\mu_1^\pr$ and $\mu_2^\pr.$ Let us define two permutations $\pi$ and $\sigma$ of $[n]$ as follows. For every $i\in[n]$, we define $\pi(i)$ and $\sigma(i)$ to be the scores the candidate $\xxx_i$ receives in $\mu_1^\pr$ and $\mu_2^\pr$ respectively. The fact that $\pi$ and $\sigma$, as defined above, is indeed a permutation of $[n]$ follows from the structure of the votes $\mu_1^\pr, \mu_2^\pr$ and the Borda score vector. Since $p$ wins in $\PP_1\uplus\{\mu_1^\pr, \mu_2^\pr\},$ we have $\pi(i) + \sigma(i) \le X_i.$ We now have the following.
 
 \[ n(n+1) = \sum_{i=1}^n (\pi(i) + \sigma(i)) \le \sum_{i=1}^n X_i = n(n+1) \]
 
 Hence, we have $\pi(i) + \sigma(i) = X_i$ for every $i\in[n]$ and thus $(X_1, \ldots, X_n)$ is a \YES instance of \PS.
\end{proof}

 We will use \Cref{score_gen} in subsequent proofs.
 
\begin{theorem}\label{thm:frugalKappNPC}
 The \textsc{Frugal-\$bribery} problem is \NPC{} for the $k$-approval voting rule for any constant $k\ge 5$, even if the price of every vulnerable vote is either $1$ or $\infty$.
\end{theorem}
\begin{proof}
 The problem is clearly in \NPshort{}. To show \NPshort{}-hardness, we reduce an arbitrary instance of X3C to \textsc{Frugal-\$bribery}. 
 Let $(U, \{S_1, \dots, S_t\})$ be an instance of X3C. We define a \textsc{Frugal-\$bribery} instance as follows. 
 The candidate set is $\mathcal{C} = U\uplus D\uplus \{p,q\}$, where $|D|=k-1$. For each $S_i, 1\le i\le t$, we add a vote $v_i$ as follows.
 \[ v_i : \underbrace{p \succ q \succ S_i}_{5 \text{ candidates}} \succ D \succ \text{others} \]
 By \Cref{score_gen}, we can add $poly(|U|)$ many additional votes to ensure the following scores (denoted by $s(\cdot)$).
 \[ s(q) = s(p) + \nfrac{|U|}{3}, s(x) = s(p) + 1, \forall x\in U,\]
 \[ s(d) < s(p) - \nfrac{|U|}{3}, \forall d\in D \] 
 The tie-breaking rule is ``$p \succ \text{others}$''. The winner is $q$. The distinguished candidate is $p$ and 
 thus all the votes in $\{v_i : 1\le i\le t\}$ are vulnerable. The price of every $v_i$ is $1$ and the price of every other vulnerable vote is $\infty$. The budget is $\nfrac{|U|}{3}$. We claim that the two instances are equivalent. Suppose there exists an index set $I\subseteq [t]$ 
 with $|I|=\nfrac{|U|}{3}$ such that $\uplus_{i\in I} S_i = U$. We replace the votes $v_i$ with $v_i^{\prime}, i\in I,$ which are defined as follows.
 \[ v_i^{\prime} : \underbrace{p \succ D}_{k \text{ candidates}} \succ \text{others} \]
 This makes the score of $p$ not less than the score of any other candidate and thus $p$ wins. 
 
To prove the result in the other direction, suppose the \textsc{Frugal-\$bribery} instance is a \YES{} instance. Then notice that there will be $\nfrac{|U|}{3}$ votes in $\{v_i : 1\le i\le t\}$ 
 where the candidate $q$ should not be placed within the top $k$ positions since $s(p) = s(q) - \nfrac{|U|}{3}$ and the budget is $\nfrac{|U|}{3}$. 
 We claim that the $S_i$'s corresponding to the $v_i$'s that have been changed must form an exact set cover. Indeed, otherwise, there will be a 
 candidate $x\in U$, whose score never decreases which contradicts the fact that $p$ wins the election since $s(p) = s(x)-1$.
\end{proof}
 We next present a similar result for the $k$-veto voting rule.
\begin{theorem}\label{thm:frugalKvetoNPC}
 The \textsc{Frugal-\$bribery} problem is \NPC{} for the $k$-veto voting rule for any constant $k\ge3$, even if the price of every vulnerable vote is either $1$ or $\infty$.
\end{theorem}

\begin{proof}
 The problem is clearly in \NPshort{}. To show \NPshort{}-hardness, we reduce an arbitrary instance of X3C to \textsc{Frugal-\$bribery}. Let $(U, \{S_1,S_2, \dots, S_t\})$ be any instance of X3C. We define a \textsc{Frugal-\$bribery} instance as follows. The candidate set is $\mathcal{C} = U\uplus Q\uplus \{p, a_1, a_2, a_3, d\}$, where $|Q|=k-3$. For each $S_i, 1\le i\le t$, we add a vote $v_i$ as follows.
 \[ v_i : p \succ \text{others} \succ \underbrace{S_i \succ Q}_{k \text{ candidates}} \]
 By \Cref{score_gen}, we can add $poly(|U|)$ many additional votes to ensure following scores (denoted by $s(\cdot)$).
 \[ s(p) > s(d), s(p) = s(x) + 2, \forall x\in U, s(p) = s(q) + 1, \forall q\in Q,\]
 \[ s(p) = s(a_i) - \nfrac{|U|}{3} + 1, \forall i=1, 2, 3 \] 
 The tie-breaking rule is ``$a_1 \succ \cdots \succ p$''. The winner is $a_1$.
 The distinguished candidate is $p$ and thus all the votes in $\{v_i : 1\le i\le t\}$ are vulnerable. 
 The price of every $v_i$ is one and the price of any other vote is $\infty$. The budget is $\nfrac{|U|}{3}$.

 We claim that the two instances are equivalent. In the forward direction, suppose there exists an index set $I\subseteq \{1,\dots,t\}$ 
 with $|I|=\nfrac{|U|}{3}$ such that $\uplus_{i\in I} S_i = U$. We replace the votes $v_i$ with $v_i^{\prime}, i\in I,$ which are defined as follows.
 \[ v_i^{\prime} : \text{others} \succ \underbrace{a_1 \succ a_2 \succ a_3 \succ Q}_{k \text{ candidates}} \]
 The score of each $a_i$ decreases by $\nfrac{|U|}{3}$ and their final scores are $s(p)-1$, since the score of $p$ is not affected 
 by this change. Also the final score of each $x\in U$ is $s(p)-1$ since $I$ forms an exact set cover. This makes $p$ win the election.
 
To prove the result in the other direction, suppose the \textsc{Frugal-\$bribery} instance is a \YES{} instance. Then, notice that there will be exactly $\nfrac{|U|}{3}$ votes in $v_i, 1\le i\le t$, where every $a_j, j=1, 2, 3$, should come in the last $k$ positions since $s(p) = s(a_j) - \nfrac{|U|}{3} + 1$ and the budget is $\nfrac{|U|}{3}$. Notice that candidates in $Q$ must not be placed within top $m-k$ positions since $s(p) = s(q) + 1$, for every $q\in Q$. Hence, in the votes that have been changed, $a_1, a_2, a_3$ and all the candidates in $Q$ must occupy the last $k$ positions. We claim that the $S_i$'s corresponding to the $v_i$'s that have been changed must form an exact set cover. If not, then, there must exist a candidate $x\in U$ and two votes $v_i$ and $v_j$ such that, both $v_i$ and $v_j$ have been replaced by $v_i^{\prime} \ne v_i$ and $v_j^{\prime} \ne v_j$ and the candidate $x$ was present within the last $k$ positions in both $v_i$ and $v_j$. This makes the score of $x$ at least the score of $p$ which contradicts the fact that $p$ wins.
\end{proof}

 However, we show the existence of a polynomial time algorithm for the \textsc{Frugal-\$bribery} problem for the $k$-approval, Bucklin, and plurality with runoff voting rules, when the budget is a constant. The result below follows from the existence of a polynomial time algorithm for the \textsc{Coalitional-manipulation} problem for these voting rules for a constant number of manipulators~\cite{xia2009complexity}.
\begin{theorem}\label{thm:frugalKappP}
 The \textsc{Frugal-\$bribery} problem is in \Pshort{} for the $k$-approval, Bucklin, and plurality with runoff voting rules, if the budget is a constant.
\end{theorem}

\begin{proof}
 Let the budget $b$ be a constant. Then, at most $b$ many vulnerable votes whose price is not zero can be changed 
 since the prices are assumed to be in $\mathbb{N}$. Notice that we may assume, 
 without loss of generality, that all the vulnerable votes whose price is zero will be changed. 
 We iterate over all the $O(n^b)$ many possible vulnerable vote changes and we can solve each one 
 in polynomial time since the \textsc{Coalitional-manipulation} problem is in \Pshort{} for these voting rules~\cite{xia2009complexity}.
\end{proof}

We show that the \textsc{Frugal-\$bribery} problem is \NPC{} for a wide class of scoring rules as characterized in the following
result. Our next result shows that, the \textsc{Frugal-\$bribery} problem is \NPC{} for a wide class of scoring rules that includes the Borda voting rule. \Cref{thm:frugalScrNPC} can be proved by a reduction from the X3C problem.
\begin{theorem}\label{thm:frugalScrNPC}
 For any positional scoring rule $r$ with score vectors $\{\overrightarrow{s_i} : i\in \mathbb{N}\}$, if there exists a polynomial function $f: \mathbb{N} \longrightarrow \mathbb{N}$ such that, for every $m\in \mathbb{N}, f(m) \ge 2m$ and in the score vector $(\alpha_1, \ldots, \alpha_{f(m)})$, there exists a $1\le \ell\le f(m)-5$ satisfying the following condition:
 \[ \alpha_i - \alpha_{i+1} = \alpha_{i+1} - \alpha_{i+2} > 0, \forall \ell \le i \le \ell+3 \]
 then the \textsc{Frugal-\$bribery} problem is \NPC{} for $r$ even if the price of every vulnerable vote is either $1$ or $\infty$.
\end{theorem}

\begin{proof}
 Since the scoring rules remain unchanged if we multiply every $\alpha_i$ by any constant $\lambda>0$ and/or add any constant $\mu$, 
 we can assume the following without loss of generality.
 \[ \alpha_i - \alpha_{i+1} = \alpha_{i+1} - \alpha_{i+2} = 1, \forall \ell \le i \le \ell+3 \]
 The problem is clearly in \NPshort{}. To show \NPshort{}-hardness, we reduce an arbitrary instance of X3C to \textsc{Frugal-\$bribery}. 
 Let $(U, \{S_1, \dots, S_t\})$ be an instance of X3C. We define a \textsc{Frugal-\$bribery} instance as follows. 
 The candidate set is $\mathcal{C} = U \uplus Q \uplus \{p, a, d\}$, where $ |Q| = f(|U|)-\ell-4$. 
 For each $S_i = \{x,y,z\}, 1\le i\le t$, we add a vote $v_i$ as follows.
 \[ v_i : p \succ d \succ \text{others} \succ \underbrace{a \succ x\succ y\succ z \succ Q}_{l \text{ candidates}} \]
 By \Cref{score_gen}, we can add $poly(|U|)$ many additional votes to ensure the following scores (denoted by $s(\cdot)$). Note that the proof of \Cref{score_gen} in \cite{baumeister2011computational} also works for the normalization of $\alpha$ defined in the beginning of the proof.
 \[ s(d) < s(p), s(x) = s(p) - 2, \forall x\in U,\]
 \[ s(a) = s(p) + \nfrac{|U|}{3} - 1 , s(q) = s(p) + 1 \]
 The tie-breaking rule is ``$\cdots \succ p$''. The distinguished candidate is $p$. The price of every $v_i$ is $1$ and the price of every other vulnerable vote is $\infty$. The budget is $\nfrac{|U|}{3}$.

 We claim that the two instances are equivalent. In the forward direction, there exists an index set $I\subseteq [t], |I|=\nfrac{|U|}{3},$ such that $\uplus_{i\in I} S_i = U$. We replace the votes $v_i$ with $v_i^{\prime}, i\in I,$ which are defined as follows.
 \[ v_i^{\prime} : p \succ d \succ \text{others} \succ x\succ y\succ z \succ a \succ Q \]
 This makes the score of $p$ at least one more than the score of every other candidate and thus $p$ wins. 
 
To prove the result in the other direction, suppose there is a way to make $p$ win the election. Notice that the candidates in $Q$ cannot change their positions in the changed votes and must occupy the last positions due to their score difference with $p$. Now we claim that there will be exactly $\nfrac{|U|}{3}$ votes where the candidate $a$ must be placed at the $(l+4)^{th}$ position since $s(p) = s(a) - \nfrac{|U|}{3} + 1$ and the budget is $\nfrac{|U|}{3}$. We claim that the $S_i$'s corresponding to the changed votes must form an exact set cover. If not, then there must exist a candidate $x\in U$ whose score has increased by at least two contradicting the fact that $p$ wins the election.
\end{proof}

 For the sake of concreteness, an example of a function $f$, stated in \Cref{thm:frugalScrNPC}, that works for the Borda voting rule is $f(m)=2m$. \Cref{thm:frugalScrNPC} shows that the \textsc{Frugal-\$bribery} problem is intractable for the Borda voting rule. However, the following theorem shows the intractability of the \textsc{Uniform-frugal-\$bribery} problem for the Borda voting rule, even in a very restricted setting.
\Cref{thm:uniformBordaNPC} can be proved by a reduction from the \textsc{Coalition manipulation} problem for the Borda voting rule for two manipulators which is known to be \NP{}-complete~\cite{betzler2011unweighted,davies2011complexity}.
\begin{theorem}\label{thm:uniformBordaNPC}
 The \textsc{Uniform-frugal-\$bribery} problem is \NPC{} for the Borda voting rule, even when every vulnerable vote has a price of $1$ and the budget is $2$.
\end{theorem}

\begin{proof}
 The problem is clearly in \NPshort{}. To show \NPshort{}-hardness, we reduce an arbitrary instance of the \textsc{Coalitional-manipulation} problem for the Borda voting rule with two manipulators to an instance of the \textsc{Uniform-frugal-\$bribery} problem for the Borda voting rule. Let $(C, \succ^t, 2, p)$ be an arbitrary instance of the \textsc{Coalitional-manipulation} problem for the Borda voting rule. The corresponding \textsc{Frugal-\$bribery} instance is as follows. The candidate set is $C^{\prime} = C \uplus \{d,q\}$. For each vote $v_i \in \succ^t$, we add a vote $v_i^{\prime} : v_i \succ d \succ q$. Corresponding to the two manipulators', we add two more votes $\nu_1, \nu_2 : \overrightarrow{C\setminus \{p\}} \succ d \succ p \succ q$, where $\overrightarrow{C\setminus \{p\}}$ is an arbitrary but fixed order of the candidates in $C\setminus \{p\}$. We add more votes to ensure following score differences ($s(\cdot)$ and $s^{\prime}(\cdot)$ are the score functions for the \textsc{Coalitional-manipulation} and the \textsc{Uniform-frugal-\$bribery} instances respectively).
 $$s^{\prime}(p) = \lambda + s(p) - 2, s^{\prime}(x) = \lambda + s(x) \text{ for every } x\in C,$$ 
 $$s^{\prime}(q) = s^{\prime}(p) - 2m + 1, s^{\prime}(p) > s^{\prime}(d) + 2m \text{ for some } \lambda \in \mathbb{Z}$$ 
 This will be achieved as follows. For any two arbitrary candidates $a$ and $b$, the following two votes increase the score of $a$ by one more than the rest of the candidates except $b$ whose score increases by one less. This construction has been used before~\cite{xia2010scheduling,davies2011complexity}.
 \[ a \succ b \succ \overrightarrow{C\setminus \{a,b\}} \]
 \[ \overleftarrow{C\setminus \{a,b\}} \succ a \succ b \]
 Also, we can ensure that candidate $p$ is always in $(\nfrac{m-1}{2}, \nfrac{m+1}{2})$ positions and the candidate $q$ never immediately follows $p$ in these new votes.  The tie-breaking rule is ``$\text{others} \succ p$''. The distinguished candidate is $p$. The price of every vulnerable vote is one and the budget is two. We claim that the two instances are equivalent. 
 
 In the forward direction, suppose the \textsc{Coalitional-manipulation} instance is a \YES{} instance. Let $u_1, u_2$ be the manipulators' votes that make $p$ win. In the \textsc{Frugal-\$bribery} instance, we replace $\nu_i$ by $\nu_i^{\prime} : p \succ d \succ u_i \succ q$ for $i=1,2$. This makes $p$ win the election. In the reverse direction, recall that in all the vulnerable votes except $\nu_1$ and $\nu_2$, the candidate $q$ never immediately follows candidate $p$. Therefore, changing any of these votes can never make $p$ win the election since $s^{\prime}(q) = s^{\prime}(p) - 2m + 1$ and the budget is two. Hence, the only way $p$ can win the election, if at all possible, is by changing the votes $\nu_1$ and $\nu_2$. Let a vote $\nu_i^{\prime}$ replaces $\nu_i$ for $i=1,2$. We can assume, without loss of generality, that $p$ and $d$ are at the first and the second positions respectively in both $\nu_1^{\prime}$ and $\nu_2^{\prime}$. Let $u_i$ be the order $\nu_i^{\prime}$ restricted only to the candidates in $C$. This makes $p$ winner of the \textsc{Coalitional-manipulation} instance since, $s^{\prime}(p) = \lambda + s(p) - 2, s^{\prime}(x) = \lambda + s(x)$ for every $x\in C$.
\end{proof}

\section{Results for Weighted Elections}

Now we turn our attention to weighted elections. As before, we begin with some easy observations that follow from known results.

\begin{observation}\label{lem:wfrugalEasy}
 The \textsc{Frugal-bribery} problem is in \Pshort{} for the maximin and the Copeland voting rules for three candidates.
\end{observation}

\begin{proof}
 The \textsc{Manipulation} problem is in \Pshort{} for the maximin, Copeland voting rules for three candidates~\cite{conitzer2007elections}.
 Hence, the result follows from \Cref{prop:man_bri_con}.
\end{proof}

\begin{observation}\label{lem:wfrugalScr}
 The \textsc{Frugal-bribery} problem is \NPC{} for any scoring rule except plurality for three candidates.
\end{observation}

\begin{proof}
 The same proof for Theorem 6 of~\cite{conitzer2007elections} would work here.
\end{proof}

\begin{theorem}\label{thm:wfrugalP} 
 The \textsc{Frugal-bribery} problem is in \Pshort{} for the plurality voting rule.
\end{theorem} 

\begin{proof}
 Let $p$ be the distinguished candidate of the campaigner. 
 The greedy strategy of just replacing every vulnerable vote by $p \succ others$ solves 
 the problem due to the monotonicity property of the plurality voting rule.
\end{proof}

Our hardness results in this section are based on the \textsc{Partition} problem, which is known to be \NPC{} \cite{garey1979computers}, and is defined as follows.
\begin{definition}(\textsc{Partition} Problem)\\
 Given a finite multi-set $W$ of positive integers with $\sum_{w\in W} w = 2K$, 
 does there exist a subset $W^{\prime} \subset W$ such that $\sum_{w\in W^{\prime}} w = K$?
 An arbitrary instance of \textsc{Partition} is denoted by $(W,2K)$. 
\end{definition}
 We define another problem which we call $\frac{1}{4}$-\textsc{Partition} as below. We prove that $\frac{1}{4}$-\textsc{Partition} is also \NPC{}.
We will use this fact in the proof of \Cref{thm:stv_wt}.
\begin{definition}(The $\frac{1}{4}$-\textsc{Partition} Problem)\\
 Given a finite multi-set $W$ of positive integers with $\sum_{w\in W} w = 4K$, 
 does there exist a subset $W^{\prime} \subset W$ such that $\sum_{w\in W^{\prime}} w = K$?
 An arbitrary instance of $\frac{1}{4}$-\textsc{Partition} is denoted by $(W,4K)$. 
\end{definition}

\begin{lemma}\label{lem:part}
$\frac{1}{4}$-\textsc{Partition} problem is \NPC{}.
\end{lemma}

\begin{proof}
 The problem is clearly in \NPshort{}. To show \NPshort{}-hardness, we reduce the \textsc{Partition} problem to it. 
 Let $(W,2K)$ be an arbitrary instance of the \textsc{Partition} problem. We can assume, without loss of generality, that $2K \notin W$, since 
 otherwise the instance is trivially a \textit{no} instance. 
 The corresponding $\frac{1}{4}$-\textsc{Partition} problem instance is defined by $(W_1,4K)$, where $W_1 = W \cup \{2K\}$. 
 We claim that the two instances are equivalent. Suppose the \textsc{Partition} instance is a \YES{} instance and thus 
 there exists a set $W^{\prime}\subset W$ such that $\sum_{w\in W^{\prime}} w = K$. This $W^{\prime}$ gives a solution to 
 the $\frac{1}{4}$-\textsc{Partition} instance. To prove the result in the other direction, 
 suppose there is a set $W^{\prime}\subset W_1$ such that $\sum_{w\in W^{\prime}} w = K$. 
 This $W^{\prime}$ gives a solution to the \textsc{Partition} problem instance since $2K \notin W^{\prime}$.
\end{proof}

In the rest of this section, we present the hardness results in weighted elections for the following voting rules: plurality, maximin, STV, Copeland$^\alpha$, and Bucklin. For plurality, recall that the \textsc{Frugal-bribery} problem is in \Pshort{}, and we will show that \textsc{Frugal-\$bribery} is \NPC{}. For all the other rules, we will establish that even \textsc{Frugal-bribery} is \NPC{}.

\begin{theorem}\label{thm:wfrugalPluNPC}
 The \textsc{Frugal-\$bribery} problem is \NPC{} for the plurality voting rule for three candidates.
\end{theorem}

\begin{proof}
 The problem is clearly in \NPshort{}. We reduce an arbitrary instance of \textsc{Partition} to an instance of \textsc{Frugal-\$bribery} for the plurality voting rule. Let $(W,2K),$ with $W=\{w_1, \ldots, w_n\}$ and $\sum_{i=1}^n w_i=2K$, be an arbitrary instance of the \textsc{Partition} problem. The candidates are $p, a,$ and $b$. We will add votes in such a way that makes $b$ win the election. The distinguished candidate is $p$. For every $i\in [n]$, there is a vote $a\succ p\succ b$ of both weight and price $w_i$. There are two votes $b\succ p\succ a$ of weight $3K$ (we do not need to define the price of this vote since it is non-vulnerable) and $p\succ a\succ b$ of both weight and price $2K+1$. The tie-breaking rule is ``$a\succ b\succ p$''. We define the budget to be $K$.

 We claim that the two instances are equivalent. Suppose there exists a $W^{\prime} \subset W$ such that $\sum_{w\in W^{\prime}} w = K$. We change the votes corresponding to the weights in $W^{\prime}$ to $p\succ a\succ b$. This makes $p$ win the election with a score of $3K+1$. To prove the other direction, for $p$ to win, its score must increase by at least $K$. Also, the prices ensure that $p$'s score can increase by at most $K$. Hence, $p$'s score must increase by exactly by $K$ and the only way to achieve this is to increase its score by changing the votes corresponding to the weights in $W$. Thus, $p$ can win only if there exists a $W^{\prime} \subset W$ such that $\sum_{w\in W^{\prime}} w = K$.
\end{proof}
 Next we show the hardness result for the maximin voting rule.
\begin{theorem}\label{thm:wfrugalMaxmin}
 The \textsc{Frugal-bribery} problem is \NPC{} for the maximin voting rule for four candidates.
\end{theorem}
\begin{proof}
 The problem is clearly in \NPshort{}. We reduce an arbitrary instance of \textsc{Partition} to an instance of \textsc{Frugal-bribery} for the maximin voting rule. Let $(W,2K),$ with $W=\{w_1, \ldots, w_n\}$ and $\sum_{i=1}^n w_i = 2K$, be an arbitrary instance of the \textsc{Partition} problem. The candidates are $p, a, b,$ and $c$. For every $i\in [n]$, there is a vote $p\succ a\succ b\succ c$ of weight $w_i$. There is one vote $c\succ a\succ b\succ p$, one $b\succ c\succ a\succ p$, and one $a\succ c\succ b\succ p$ each of weight $K$. The tie-breaking rule is ``$p\succ a\succ b\succ c$''. The distinguished candidate is $p$. Let $T$ denotes the set of votes corresponding to the weights in $W$ and the rest of the votes $S$. Notice that only the votes in $T$ are vulnerable. We claim that the two instances are equivalent. 
 
 Suppose there exists a $W^{\prime} \subset W$ such that $\sum_{w\in W^{\prime}} w = K$. We change the votes corresponding to the weights in $W^{\prime}$ to $p\succ a\succ b\succ c$. We change the rest of the votes in $T$ to $p\succ b\succ c\succ a$. The maximin score of every candidate is $-K$ and thus due to the tie-breaking rule, $p$ wins the election.
 
 On the other hand, suppose there is a way to change the vulnerable votes, that is the votes in $T$, that makes $p$ win the election. Without loss of generality, we can assume that all the votes in $T$ place $p$ at top position. First notice that the only way $p$ could win is that the vertices $a, b,$ and $c$ must form a cycle in the weighted majority graph. Otherwise, one of $a, b,$ and $c$ will be a Condorcet winner and thus the winner of the election. Now we show that the candidate $b$ must defeat the candidate $c$. If not, then $c$ must defeat $b$ by a margin of $K$ since the maximin score of $p$ is fixed at $-K$. Also, $a$ must defeat $c$ by a margin of $K$, otherwise the maximin score of $c$ will be more than $-K$. This implies that all the votes in $T$ must be $p\succ a\succ c\succ b$ which makes $a$ defeat $b$. This is a contradiction since the vertices $a, b,$ and $c$ must form a cycle in the weighted majority graph. Hence $b$ must defeat $c$ by a margin of $K$. This forces every vote in $T$ to prefer $b$ over $c$. Hence, without loss of generality, we assume that all the votes in $T$ are either $p\succ a\succ b\succ c$ or $p\succ b\succ c\succ a$, since whenever $c$ is right after $a$, we can swap $a$ and $c$ and this will only reduce the score of $a$ without affecting the score of any other candidates. If the total weight of the votes $p\succ a\succ b\succ c$ in $T$ is more than $K$, then $D_E(c,a) < K$, thereby making the maximin score of $a$ more than the maximin score of $p$. If the total weight of the votes $p\succ a\succ b\succ c$ in $T$ is less than $K$, then $D_E(a,b) < K$, thereby making the maximin score of $b$ more than the maximin score of $p$. Thus the total weight of the votes $p\succ a\succ b\succ c$ in $T$ should be exactly $K$ which corresponds to a partition of $W$.
\end{proof}
 We now prove the hardness result for the STV voting rule.
\begin{theorem}\label{thm:stv_wt}
 The \textsc{Frugal-bribery} problem is \NPC{} for the STV voting rule for three candidates.
\end{theorem}
\begin{proof}
 The problem is clearly in \NPshort{}. We reduce an arbitrary instance of $\frac{1}{4}$-\textsc{Partition} to an instance of \textsc{Frugal-bribery} for the STV voting rule. Let $(W,4K),$ with $W=\{w_1, \ldots, w_n\}$ and $\sum_{i=1}^n w_i = 4K$, be an arbitrary instance of the $\frac{1}{4}$-\textsc{Partition} problem. The candidates are $p, a,$ and $b$. For every $i\in [n]$, there is a vote $p\succ a\succ b$ of weight $w_i$. There is a vote $a\succ p\succ b$ of weight $3K-1$ and a vote $b\succ a\succ p$ of weight $2K$. The tie-breaking rule is ``$a\succ b\succ p$''. The distinguished candidate is $p$. Let $T$ denotes the set of votes corresponding to the weights in $W$ and the rest of the votes be $S$. Notice that only the votes in $T$ are vulnerable. We claim that the two instances are equivalent.
 
 Suppose there exists a $W^{\prime} \subset W$ such that $\sum_{w\in W^{\prime}} w = K$. We change the votes corresponding to the  weights in $W^{\prime}$ to $b\succ p\succ a$. We do not change the rest of the votes in $T$. This makes $p$ win the election.
 
For the other direction, suppose there is a way to change the votes in $T$ that makes $p$ win the election. First Notice that $p$ can win only if $b$ qualifies for the second round. Hence, the total weight of the votes in $T$ that put $b$ at the first position must be at least $K$. On the other hand, if the total weight of the votes in $T$ that put $b$ at the first position is strictly more than $K$, then $p$ does not qualify for the second round and thus cannot win the election. Hence the total weight of the votes in $T$ that put $b$ at the first position must be exactly equal to $K$ which constitutes a $\frac{1}{4}$-partition of $W$.
\end{proof}
 For three candidates, the STV voting rule is the same as the plurality with runoff voting rule. Hence, we have the following corollary.
\begin{corollary}\label{cor:run_wt}
 The \textsc{Frugal-bribery} problem is \NPC{} for the plurality with runoff voting rule for three candidates.
\end{corollary}

We turn our attention to the Copeland$^{\alpha}$ voting rule next.
\begin{theorem}\label{thm:wfrugalCopeland}
 The \textsc{Frugal-bribery} problem is \NPC{} for the Copeland$^{\alpha}$ voting rule for four candidates, whenever $\alpha\in[0,1)$.
\end{theorem}
\begin{proof}
 The problem is clearly in \NPshort{}. We reduce an arbitrary instance of \textsc{Partition} to an instance of \textsc{Frugal-bribery} for the Copeland$^{\alpha}$ voting rule. Let $(W,2K),$ with $W=\{w_1, \ldots, w_n\}$ and $\sum_{i=1}^n w_i = 2K$, be an arbitrary instance of the \textsc{Partition} problem. The candidates are $p, a, b,$ and $c$. For every $i\in [n]$, there is a vote $p\succ a\succ b\succ c$ of weight $w_i$. There are two votes $a\succ p\succ b\succ c$ and $c\succ b\succ a\succ p$ each of weight $K+1$.  The tie-breaking rule is ``$a\succ b\succ c\succ p$''. The distinguished candidate is $p$. Let $T$ denotes the set of votes corresponding to the weights in $W$ and the rest of the votes be $S$. Notice that only the votes in $T$ are vulnerable. We claim that the two instances are equivalent.
 
 Suppose there exists a $W^{\prime} \subset W$ such that $\sum_{w\in W^{\prime}} w = K$. We change the votes corresponding to the  weights in $W^{\prime}$ to $p\succ c\succ b\succ a$. We change the rest of the votes in $T$ to $p\succ b\succ c\succ a$. This makes $p$ win the election with a Copeland$^{\alpha}$ score of two.
 
 On the other hand, suppose there is a way to change the votes in $T$ that makes $p$ win the election. Without loss of generality, we can assume that all the votes in $T$ place $p$ at top position. We will show that one of the three pairwise elections among $a, b,$ and $c$ must be a tie. Suppose not, then $a$ must lose to both $b$ and $c$, otherwise $a$ wins the election due to the tie-breaking rule. Now consider the pairwise election between $b$ and $c$. If $b$ defeats $c$, then $b$ wins the election due to the tie-breaking rule. If $c$ defeats $b$, then $c$ wins the election again due to the tie-breaking rule. Hence, one of the pairwise elections among $a, b,$ and $c$ must be a tie. Without loss of generality suppose $a$ and $b$ ties. Then the total weight of the votes that prefer $a$ to $b$ in $T$ must be $K$ which constitutes a partition of $W$.
\end{proof}

Finally, we show that the \textsc{Frugal-bribery} problem for the Bucklin voting rule is \NPC{}.
\begin{theorem}\label{thm:bucklin_wt}
 The \textsc{Frugal-bribery} problem is \NPC{} for the Bucklin voting rule for four candidates.
\end{theorem}
\begin{proof}
 The problem is clearly in \NPshort{}. We reduce an arbitrary instance of \textsc{Partition} to an instance of \textsc{Frugal-bribery} for the Bucklin voting rule. Let $(W,2K),$ with $W=\{w_1, \ldots, w_n\}$ and $\sum_{i=1}^n w_i = 2K$, be an arbitrary instance of the \textsc{Partition} problem. The candidates are $p, a, b,$ and $c$. For every $i\in [n]$, there is a vote $p\succ a\succ b\succ c$ of weight $w_i$. There are two votes $a\succ b\succ p\succ c$ and $c\succ b\succ a\succ p$ each of weight $K$.  The tie-breaking rule is ``$p\succ a\succ b\succ c$''. The distinguished candidate is $p$. Let $T$ denote the set of votes corresponding to the weights in $W$ and the rest of the votes be $S$. Notice that only the votes in $T$ are vulnerable. We claim that the two instances are equivalent.
 
 Suppose there exists a $W^{\prime} \subset W$ such that $\sum_{w\in W^{\prime}} w = K$. We change the votes corresponding to the  weights in $W^{\prime}$ to $p\succ c\succ b\succ a$. This makes $p$ win the election with a Bucklin score of three.
 
 
To prove the result in the other direction, suppose there is a way to change the votes in $T$ that makes $p$ win the election. Without loss of generality, we can assume that all the votes in $T$ place $p$ at the first position. First Notice that the Bucklin score of $p$ is already fixed at three. In the votes in $T$, the candidate $b$ can never be placed at the second position since that will make the Bucklin score of $b$ to be two. Also the total weight of the votes in $T$ that place $a$ in their second position can be at most $K$. The same holds for $c$. Hence, the total weight of the votes that place $a$ in their second position will be exactly equal to $K$ which constitutes a partition of $W$.
\end{proof}

We also have the following results for the Copeland$^\alpha$ and Bucklin voting rules by reducing from \textsc{Partition}.
\begin{theorem}\label{thm:wfrugalCopeland}
 The \textsc{Frugal-bribery} problem is \NPC{} for the Copeland$^{\alpha}$ and Bucklin voting rules for four candidates, whenever $\alpha\in[0,1)$.
\end{theorem}

 From \Cref{prop:conn}, \Cref{lem:wfrugalScr}, \Cref{thm:wfrugalPluNPC,thm:wfrugalMaxmin,thm:wfrugalCopeland,thm:stv_wt,thm:bucklin_wt}, 
and \Cref{cor:run_wt}, we get the following corollary.
\begin{corollary}
 The \textsc{Uniform-frugal-\$bribery} and the \textsc{Nonuniform-frugal-\$bribery} problems are \NPC{} for the scoring rules except 
 plurality, STV, and the plurality with runoff voting rules for three candidates and for the maximin, Copeland, and Bucklin voting rules for four candidates.
\end{corollary}

\section{Conclusion}

We proposed and studied two important special cases of the \textsc{\$Bribery} problem where the briber is frugal. Our results show that even for these special cases, the bribery problem continues to be intractable, thus subsuming known hardness results in the literature. Our results reinforce that bribery is a rather hard computational problem, because of the hardness of several important special cases. This also strengthens the view that bribery, although a possible attack on an election in principle, may be infeasible in practice. 
\chapter{Summary and Future Directions}
\label{chap:summary}

\begin{quote}
{\it In this chapter, we summarize our work in this thesis and provide interesting future directions of research.}
\end{quote}

In this thesis, we studied computational complexity of three fundamental aspects of voting. We began with studying efficient strategies for eliciting preferences of voters in the first part of the thesis. We then moved on to study the problem of finding a winner under various interesting circumstances in the second part of the thesis. Finally, we showed interesting complexity theoretic results for the computational problem of controlling an election system in various forms. We now summarize our work in this thesis.

\section{Summary of Contributions}

\subsection{Part I: Preference Elicitation}

\subsubsection*{\Cref{chap:pref_elicit_peak}: \nameref{chap:pref_elicit_peak}}

In this work, we presented algorithms for eliciting the preferences of a set of voters when the preferences are single peaked on a tree. Moreover, our algorithms ask minimum number of comparison queries up to constant factors. We also presented algorithms for finding a weak Condorcet winner from a set of preferences which are single peaked on a tree by asking minimum number of comparison queries up to constant factors. We observed that, the query complexity of finding a weak Condorcet is much less than the query complexity for preference elicitation.

\subsubsection*{\Cref{chap:pref_elicit_cross}: \nameref{chap:pref_elicit_cross}}

In this work, we presented preference elicitation algorithm for single crossing preference profiles. We studied this problem when an ordering of the voters with respect to which the profile is single crossing is {\em known} and when it is {\em unknown}. We also considered different access models: when the votes can be accessed at random, as opposed to when they are coming in a predefined sequence. In the sequential access model, we distinguished two cases when the ordering is known: the first is that sequence in which the votes appear is also a single-crossing order, versus when it is not. Our algorithms ask for minimum number of comparison queries up to constant factors for all the above situations except one when we have a large number of voters.

\subsection{Part II: Winner Determination}

\subsubsection*{\Cref{chap:winner_prediction_mov}: \nameref{chap:winner_prediction_mov}}

We presented efficient algorithms for predicting the winner of an election as well as estimating the margin of victory of an election. We also showed interesting lower bounds for sample complexity of these problems which establish that our algorithms are often optimal up to constant factors.

\subsubsection*{\Cref{chap:winner_stream}: \nameref{chap:winner_stream}}

In this work, we studied the space complexity for determining approximate winners in the setting where votes are inserted continually into a data stream. We showed that allowing randomization and approximation indeed allows for much more space-efficient algorithms. Moreover, our bounds are tight in certain parameter ranges.

\subsection{Part III: Election Control}

\subsubsection*{\Cref{chap:kernel}: \nameref{chap:kernel}}

In this work, we proved that the possible winner problem does not admit any efficient preprocessing rules, more formally any kernelization algorithms for many common voting rules including scoring rules, maximin, Copeland, ranked pairs, and Bucklin when parameterized by the number of candidates. However, we showed that the coalitional manipulation problem which is an important special case of the possible winner problem does admit polynomial time kernelization algorithms parameterized by the number of candidates.

\subsubsection*{\Cref{chap:partial}: \nameref{chap:partial}}

In this chapter, we pursued a comprehensive study of manipulation with incomplete votes. We proposed three natural extension of manipulation in the usual complete information setting to the incomplete information setting namely, weak manipulation, opportunistic manipulation, and strong manipulation. We completely resolved computational complexity of all the three problems for many common voting rules including plurality, veto, $k$-approval, $k$-veto, Borda, maximin, Copeland, Bucklin, and Fallback voting rules.

\subsubsection*{\Cref{chap:detection}: \nameref{chap:detection}}

In this work, we initiated a promising direction of research namely detecting instances of manipulation in elections. We showed that detecting possible instances of manipulation can often be a much easier computational problem than the corresponding problem of manipulating election itself as seen for the case of the Borda voting rule.

\subsubsection*{\Cref{chap:frugal_bribery}: \nameref{chap:frugal_bribery}}

In this work, we studied the classical problem of bribery under a weak notion of briber namely when the briber if frugal in nature. We proved that the bribery problems remain intractable even with this weak briber thereby strengthening the intractability results from the literature on bribery. Hence, although theoretically possible, bribery may not be easy to do in practice.

\section{Future Directions of Research}

We next discuss some of the interesting directions of research from this thesis.

\subsection*{\Cref{chap:pref_elicit_peak}: \nameref{chap:pref_elicit_peak}}

\begin{itemize}
 \item One can generalize the notion of single peaked profiles on trees for preferences that are not necessarily a complete order. Indeed, voters often are indifferent between two or more candidates and there are interesting domains in this setting \cite{DBLP:conf/ijcai/ElkindL15}. Eliciting such incomplete preference is an interesting direction of research to pursue.
 \item The domain of single peaked profiles on trees can further be generalized to single peaked profiles on forests. It would be interesting to study this both as a domain itself and  from the point of view of existence of efficient preference elicitation strategies.
 \item Can we reduce the query complexity for preference elicitation further by assuming more on their preferences, for example, a social network structure on the voters?
 \item In this work, we assume complete knowledge of single peaked tree. How the query complexity would change if we only assume the knowledge of the structure of the single peaked tree without the identity of the candidate associated with each node of the tree?
\end{itemize}

\subsection*{\Cref{chap:pref_elicit_cross}: \nameref{chap:pref_elicit_cross}}

\begin{itemize}
 \item An immediate direction of research is to close the gap between upper and lower bound on query complexity for preference elicitation when we know a single crossing order and voters are allowed to be queried randomly.
 \item Another interesting direction is to study preference elicitation assuming a social network structure among the voters where neighbors tend to have similar preferences. 
\end{itemize}

\subsection*{\Cref{chap:winner_prediction_mov}: \nameref{chap:winner_prediction_mov}}

\begin{itemize}
\item
Is there an axiomatic characterization of the voting rules for which
the sample complexity is independent of $m$ and $n$? We note that a
similar problem in graph property testing was the subject of intense
study \cite{alon2006acomb,borgs2006graph}. 
\item
Specifically for scoring rules, is the sample complexity determined by
some natural property of the score vector, such as its sparsity? 
\item
Is it worthwhile for the algorithm to elicit only part of the vote
from each sampled voter instead of the full vote? As mentioned in the
Introduction, vote elicitation is a well-trodden area, but as far
as we know, it has not been studied how assuming a margin of victory
can change the number of queries.
\item
How can knowledge of a social network on the voters be used to
minimize the number of samples made? Some initial progress in this
direction has been made by Dhamal and Narahari \cite{dhamal2013scalable} and by Agrawal
and Devanur (private communication).
\end{itemize}

\subsection*{\Cref{chap:winner_stream}: \nameref{chap:winner_stream}}

\begin{itemize}
 \item An immediate future direction of research is to find optimal algorithms for heavy hitters variants for other voting rule.
 \item It may be interesting to implement these streaming algorithms for use in practice (say, for participatory democracy experiments or for online social networks) and investigate how they perform.
 \item Finally, instead of having the algorithms which passive observes a few random votes, could we improve performance of the algorithm by actively querying voters as they appear in the stream?
\end{itemize}

\subsection*{\Cref{chap:kernel}: \nameref{chap:kernel}}

\begin{itemize}
 \item There are other interesting parameterizations of these problems for which fixed parameter tractable algorithms are known but the corresponding kernelization questions are still open. One such parameter is the total number of pairs $s$ in all the votes for which an ordering has not been specified. With this parameter, a simple $O(2^s. \text{ poly($m,n$)})$ algorithm is known \cite{betzler2009multivariate}. However, the corresponding kernelization question is still open.
 \item Another interesting problem in the context of incomplete votes is the necessary winner problem which asks for a candidate which wins in every extension of the partial votes. Necessary winner is known to be intractable for Copeland, ranked pairs, voting trees, and STV voting rules \cite{xia2008determining}. Studying parameterized complexity of the necessary winner problem is also another interesting direction of research to pursue in future.
\end{itemize}

\subsection*{\Cref{chap:partial}: \nameref{chap:partial}}

\begin{itemize}
 \item we leave open the problem of completely establishing the complexity of strong, opportunistic, and weak manipulations for all the scoring rules. It would be interesting to resolve it.
 \item Other fundamental forms of manipulation and control do exist in voting, such as destructive manipulation and control by adding candidates. It would be interesting to investigate the complexity of these problems in a partial information setting.
 \item Another exciting direction is the study of average case complexity, as opposed to the worst case results that we have pursued. These studies have already been carried out in the setting of complete information~\cite{procaccia2006junta,faliszewski2010ai,walsh2010empirical}. Studying the problems that we propose in the average-case model would reveal further insights on the robustness of the incomplete information setting as captured by our model involving partial orders. 
\end{itemize}

\subsection*{\Cref{chap:detection}: \nameref{chap:detection}}

\begin{itemize}
 \item In this work, we considered elections with unweighted voters only. An immediate future research direction is to study the complexity of these problems in  weighted elections.
 \item Verifying the number of false manipulators that this model catches in a real or synthetic data set, where, we already have some knowledge about the manipulators, would be interesting.
\end{itemize}

\subsection*{\Cref{chap:frugal_bribery}: \nameref{chap:frugal_bribery}}

\begin{itemize}
 \item A potential and natural direction for future work is to study these problems under various other settings. One obvious setting is to restrict the campaigner's knowledge about the votes and/or the candidates who will actually turn up. The uncertainty can also arise from the voting rule that will eventually be used among a set of voting rules.
 \item Studying these bribery problems when the pricing model for vulnerable votes is similar to swap bribery would be another interesting future direction.
\end{itemize}

\renewcommand{\bibname}{References}
\bibliographystyle{alpha}
\bibliography{refThesis}

\end{document}